\documentclass[12pt]{article}
\usepackage{graphicx}
\usepackage[usenames]{xcolor}
\usepackage{amsmath}
\usepackage{amssymb}
\usepackage{amscd}
\usepackage{amsthm}
\usepackage{amsopn}
\usepackage{xspace}
\usepackage{verbatim}
\usepackage{amsmath}
\usepackage{amsfonts}
\usepackage{epic}
\usepackage{eepic}
\usepackage{empheq}
\usepackage[active]{srcltx}
\usepackage[a4paper,margin=3cm]{geometry}
\definecolor{red}{rgb}{1,0,0}
\definecolor{green}{rgb}{0,1,0}
\definecolor{SeaGreen}{RGB}{46,139,87}
\definecolor{Maroon}{RGB}{128,0,0}
\definecolor{blush}{rgb}{0.87, 0.36, 0.51}

\newcommand{\N}{\mathbb{N}}
\newcommand{\Z}{\mathbb{Z}}

\newcommand{\C}{{\mathbb{C}}}
\newcommand{\R}{{\mathbb{R}}}

\newcommand{\A}{{\mathcal A}}
\newcommand{\B}{\mathcal B}

\def\Hg {{\mathcal H}}

\def\Jg {{\mathcal J}}

\def\Kg {{\mathcal K}}

\newcommand{\LL}{\mathcal L}

\newcommand{\OO}{\mathcal O}

\def\PP{\mathcal P}
\newcommand{\QQ}{\mathcal Q}

\def\Sg {{\mathcal S}}

\def\Wg {{\mathcal W}}
\def\Vg {{\mathcal V}}
\def\Tg {{\mathcal T}}
\def\Tf {{\mathfrak T}}

\def\Df {{\mathfrak D}}
\def\Sf {{\mathfrak S}}
\def\mg{{\bf m}}

\def\curl{\text{\rm curl}}

\def\curl{\text{\rm curl\,}}

\def\Div{\text{\rm div\,}}

\def\sign{\text{\rm sign\,}}

\newcommand {\pa}{\partial}

\def\Ai{\text{\rm Ai\,}}

\def\0{\mathbf  0}

\def\XXint#1#2#3{{\setbox0=\hbox{$#1{#2#3}{\int}$ }
\vcenter{\hbox{$#2#3$ }}\kern-.6\wd0}}

 

\errorcontextlines=0 \numberwithin{equation}{section}
\theoremstyle{plain}
\newtheorem{theorem}{Theorem}[section]
\newtheorem{lemma}[theorem]{Lemma}

\newtheorem{proposition}[theorem]{Proposition}
\newtheorem{assumption}[theorem]{Assumption}

\newtheorem{remark}[theorem]{Remark}
\newtheorem{corollary}[theorem]{Corollary}

\title{On the stability of laminar flows between plates}
  \author{ Y. Almog$^*$, Department of
  Mathematics, \\ Ort Braude College, \\ 
    Carmiel 2161002, Israel \\~\\
  and \\~\\
\noindent   B. Helffer, Laboratoire de Math\'ematiques Jean Leray, \\CNRS and Universit\'e de Nantes, \\
  2 rue de la Houssini\`ere, 44322 Nantes Cedex France}
\date{}

\begin{document}
\maketitle
\tableofcontents
\newpage 
\begin{abstract}
  Consider a two-dimensional laminar flow between two plates, so that \break
  $(x_1,x_2)\in\R \times[-1,1]$, given by ${\mathbf v}(x_1,x_2)=(U(x_2),0)$, where
  $U\in C^4([-1,1])$ satisfies $U^\prime\neq0$ in $[-1,1]$.  We prove that the
  flow is linearly stable in the large Reynolds number limit, in two
  different cases:
  \begin{itemize}
  \item   $\sup_{x\in[-1,1]} |U''(x)|   + \sup_{x\in[-1,1]} |U'''(x)| \ll
  \min_{x\in[-1,1]}|U^\prime(x)|$ \\  (nearly Couette flows), 
  \item 
  $U^{\prime\prime}\neq0$ in $[-1,1]$. 
  \end{itemize}
   We assume either no-slip or fixed
  traction force conditions on the plates, and an arbitrary large (but
  much smaller than the Reynolds number) period in the $x_1$ direction.  
\end{abstract}

\section{Introduction}
Consider the incompressible Navier-Stokes equations in the
two-dimensional pipe  $D=\R\times (-1,1)$
\begin{equation}
  \label{eq:1}
\begin{cases}
\partial_t {\mathbf v} - \epsilon \Delta {\mathbf v} + {\mathbf v} \cdot \nabla {\mathbf v} = -
\nabla p & \text{in } \R_+\times D  \\
{\mathbf v}=v_b\; \hat{i}_1  & \text{on } \R_+\times\partial D  \,,
\end{cases} 
\end{equation}
where  $\hat{i}_1 = (1,0)$ and ${\mathbf v}=(v_1,v_2)$.\\
We require periodicity in the $x_1$ direction, i.e., for all
$(t,x_1,x_2)\in \R_+\times D$ we have  
\begin{equation}
  \label{eq:2}
{\mathbf v}(t,x_1+L,x_2)={\mathbf v}(t,x_1,x_2)\,,
\end{equation}
for some $L>0$, which may be arbitrary large,  but must satisfy
$L\ll\epsilon^{-1}$ as $\epsilon\to0$. An initial condition on ${\mathbf v}$ at
$t=0$ must be placed as well.  The vector ${\mathbf v}=(v_1,v_2)$
denotes the fluid velocity field which belongs, for all $T>0$, to
\begin{displaymath}
\begin{array}{l}
 \Wg_{T,L}  = \{{\mathbf u} \in L^2(0,T;H^2_{loc}(\overline{D},\mathbb R^2))\,,\, \partial_t {\mathbf u} \in L^2(0,T;L^2_{loc}(\overline{D},\mathbb R^2))\,,  \\
\qquad \qquad \qquad  \qquad \qquad \qquad \qquad  \, \Div {\mathbf u} =0 \,, \,
 {\mathbf u}(t,x_1+L,x_2)={\mathbf u}(t,x_1,x_2)\,\} \,, 
 \end{array}
\end{displaymath}
where the divergence is taken in the spatial coordinates.\\
Here $u\in
H^k_{loc}(\overline{D}))$  for some $k\in\N$ means that for any $\phi \in
C_0^\infty (\mathbb R)$, it holds that \break $(x_1,x_2) \mapsto \phi(x_1) u (x_1,x_2) \in H^k( D)$.
Similarly, $u\in H^k_{0,loc}(\overline{D})$ means that for any $\phi \in
C_0^\infty (\mathbb R)$, $(x_1,x_2) \mapsto \phi(x_1) u (x_1,x_2) \in H^k_0( D)$.
Recall that
\begin{displaymath}
  {\mathbf v} \cdot \nabla {\mathbf v} =v_1\frac{\partial{\mathbf v}}{\partial x_1} + v_2\frac{\partial{\mathbf v}}{\partial x_2} \,.
\end{displaymath}
The pressure $p$ belongs, for all $T>0$, to
\begin{displaymath}
  \QQ_{T,L}=\{p\in L^2(0,T;H^1_{loc}(\overline{D})) \,| \, \nabla p(t,x_1+L,x_2)=\nabla p(t,x_1,x_2)\,\} \,. 
\end{displaymath}
The trace of $\bf v$ on the
boundary is constant on each connected component of $\R_+\times\partial D$:
\begin{displaymath}
  \{v_b(-1),v_b(1)\}\in\R^2\,.
\end{displaymath}
The parameter $\epsilon >0$
denotes the inverse of the flow's Reynolds number ${\bf \rm Re} > 0$.
Beyond the above no-slip boundary condition we shall also consider a
prescribed constant  traction force on the boundary  (or Navier-slip
conditions \cite{chen2018transition}), i.e.,
\begin{equation}
 \label{eq:3}
 \frac{\partial v_1}{\partial x_2}=s_b \quad , \quad  v_2 = 0\,,
\end{equation}
where $s_b(\pm 1)\in\R$ denote the prescribed traction force on $\partial D=\R\times\{\pm 1\}$.

We consider the stability of a stationary pair $({\bf v}, p)$ where
the flow ${\bf v}$ is of the form, for $(x_1,x_2)\in D$,
\begin{displaymath}
  {\mathbf v}(x_1,x_2)=U(x_2)\, \hat{i}_1 \,.
\end{displaymath}
For such flows, \eqref{eq:1} is satisfied for $({\mathbf v},p)$ 
  if  and only if there exists some constant $a$ and $p_0$ such that, for $(x_1,x_2)\in D$, 
\begin{equation}
\label{eq:4}
  U^{\prime\prime}(x_2)\equiv a \quad ; \quad p_0 (x_1,x_2)=\epsilon\, a \,x_1+p_0\,.
\end{equation}
Depending on the boundary condition (no-slip or fixed traction force),
$U$ should satisfy an inhomogeneous Dirichlet condition $U(\pm  1) =v_b(\pm 
1)$ or an inhomogeneous Neumann condition $U^\prime(\pm  1)= s_b(\pm  1)$.

We shall not confine ourselves, however, in the sequel to such
unperturbed velocity fields and discuss a more general class of
motions (cf. also \cite[p.  154]{drre04}),  which can be obtained if we
add a non-uniform body force ${\bf b}(x_1,x_2)=b(x_2)\,\hat{i}_2$ to the
right-hand-side of \eqref{eq:1}.  Such a generalization can be useful
if one attempts to examine the stability of a flow in an arbitrary 2D
cross-section but uniform in the longitudinal direction.  The
linearized operator associated with \eqref{eq:1}, at the flow $({\bf
  v},p)$, assumes the form
\begin{equation} \label{eq:4b}
(u,q) \mapsto  \mathcal T_0 ({\bf u} , q ):=  \mathfrak T {\mathbf u} - \nabla q
\end{equation}
where
\begin{equation} \label{eq:4a}
\mathfrak T {\mathbf u}=-  \epsilon \,\Delta {\mathbf u} + U\, \frac{\partial{\mathbf
    u}}{\partial x_1}+ \, u_2\, U^\prime\, \hat{i}_1\,,  \quad \text{in } D \,,
    \end{equation}
where $q \in\QQ_L$ given by 
\begin{displaymath}
  \QQ_L=\{p\in H^1_{loc}(\overline{D}) \,| \, \nabla p(\cdot+L,\cdot)=
 \nabla  p(\cdot,\cdot) \mbox{ and } \int _{(0,L)\times (-1,+1)} p (x_1,x_2) dx_1
    dx_2=0\,\} \,, 
\end{displaymath}  
and ${\mathbf u}=(u_1,u_2)$ belongs  to either
\begin{displaymath}
  \Wg_{\mathfrak D}= \Big\{{\mathbf u} \in\Wg\,\Big|\,  u_1\Big|_{\partial D}=0\,,\; u_2 \Big|_{\partial D}=0\,\Big\}\,,
\end{displaymath}
or to 
\begin{displaymath} 
  \Wg_{\mathfrak S}= \Big\{{\mathbf u} \in\Wg\,\Big|\, \;\frac{\partial u_1}{\partial
    x_2}\Big|_{\partial D}=0\,,\,u_2\Big|_{\partial D}=0\,\Big\}\,,
\end{displaymath}
 where
\begin{displaymath}
   \Wg = \{{\bf u} \in H^2_{loc}(\overline{D},\R^2) \,| \, \Div {\bf u} =0 \,;\; {\mathbf u}(\cdot+L,\cdot)={\mathbf u}(\cdot,\cdot)\} \,.
\end{displaymath}
\pagebreak

\begin{remark}\label{rem:1.1}~
\begin{itemize}
\item Note that when no-slip boundary conditions, introduced in
  \eqref{eq:1}, are applied, then ${\mathbf u}\in\Wg_{\mathfrak D}$.
Otherwise, if we select \eqref{eq:3} instead, then  ${\mathbf
  u}\in\Wg_{\mathfrak S}$. 
\item Note that that for any $p\in \QQ_L$ there exists $A\in \mathbb R$
  and $\tilde p$ such that $p(x_1,x_2)= A x_1 +\tilde p(x_1,x_2)$ with $\tilde p$ satisfying the periodicity condition $ \tilde p(\cdot+L,\cdot)= \tilde p(\cdot,\cdot)$.
   \end{itemize}
   \end{remark}

We now attempt to define a spectral problem for $\Tg_0$. 
We seek an estimate for the solution $({\bf u},q)$ for some $\Lambda \in
\mathbb C$ and $F$ in a suitable space of the equation
   \begin{equation} \label{neweq1}
 \mathcal T_0 ({\bf u},q) - \Lambda {\bf u} = {\bf F}\,.
 \end{equation}
 To this end we need to define the function space in which the
 solution should reside, and then to formulate an effective spectral
 problem involving only ${\bf u}$, so that $q$ is recovered in a final
 step directly from ${\bf u}$.

 The local stability of the flow \eqref{eq:4} has been addressed
 mostly by physicists and and engineers
 \cite{drre04,joseph2013stability,Yaglom12,SchmidHenningson2001}.
 The Poiseuille flow ($U(x)=1-x^2$), which falls outside the scope of this
 work, and the Couette flow ($U(x)=x$), 
, have received some special
 attention. Thus, for Couette flows it has been established, using a
 mix of numerical and analytical techniques that Couette flow is
 always stable \cite{wa53,ro73}, and established numerically
 \cite{orszag1971accurate} that Poiseuille flow looses its stability
 for $\epsilon^{-1}\approx5772$. In \cite{sh04} a survey of results published in
 Russian is presented where the locus of part of the spectrum (but not
 its left margin) for Couette and Poiseuille flows is approximated in
 the limit $\epsilon\to0$.

 In the recent years a number of rigorous mathematical works
 addressing the spectral problem associated with \eqref{neweq1} have
 been published. In particular, in \cite{chen2018transition}, both
 resolvent and semigroup estimates have been obtained for the case of
 Couette flow ($U=x_2$).  We shall relate the results in
 \cite{chen2018transition} to the present work in more detail in the
 sequel. More recently in \cite{bedrossian2019inviscid} it has been
 established that the semigroup can be represented as a sum of the one
 generated by an unbounded Couette flow and an exponentially fast
 decaying boundary term. For symmetric flows in a channel (including
 the Poiseuille flow) it has been proved in \cite{grenier2016spectral}
 that the laminar flow looses stability for sufficiently small $\epsilon$
 and sufficiently large $L$. The stability analysis of a Prandtl
 boundary layer, leading to a similar problem in $D=\R\times\R_+$, was
 studied in \cite{GV,grenier2017green,grenier2019green}.  In three
 dimensions, the stability of radially symmetric Poiseuille flow for
 sufficiently small $\epsilon$ has been proved in \cite{chen2019linear}. It
 is also worthwhile mentioning
 \cite{ibrahim2019pseudospectral,wei2019enhanced} addressing
 Kolmogorov flows.

   The case $\epsilon=0$, or the inviscid case, has been studied much more
   extensively in the Mathematics literature. Some of the recent works
   include \cite{wei2018linear,jia2019linear,jia2019Gevrey} which
   study both the resolvent and the semigroup associated with the
   linearized operator. Naturally, resolvent estimates in the inviscid
   case are crucial while attempting to obtain resolvent estimates for
   $0<\epsilon\ll1$ (see also
   \cite{grenier2016spectral,grenier2017green,grenier2019green}), but
   since the results in the above references are specifically designed
   to obtain stability results for the Euler equation we derive our
   own resolvent estimates in Section \ref{sec:2}.

   It should be emphasized that experimental observations (cf.
   \cite{chapman2002subcritical}) conclude that Couette and Poiseuille
   flows loose their stability for Reynolds numbers that are much
   lower than $\epsilon^{-1}=5772$. It is commonly believed that these
   instabilities arise due to finite, though small, initial
   conditions.  Thus, it has been established in
   \cite{bedrossian2015dynamics,bedrossian2016enhanced} that
   unbounded Couette flows (in $\R$ instead of $[-1,1]$ as is presently
   considered), assuming a period $L=1$, are finitely stable for
   sufficiently initial data.  Similar results are obtained in
   \cite{Mas17} in a three dimentional settings. 
 
   To obtain non-linear stability one needs in
   \cite{bedrossian2015dynamics,bedrossian2016enhanced,Mas17} to use
   semigroup estimates (and not only the locus of the eigenvalues as
   in \cite{sh04}), associated with time dependent equation
\begin{displaymath}
  u_t- \mathcal T_0 ({\bf u},q) =0
\end{displaymath}
in $\Wg_\Df$ or $\Wg_\Sf$.\\

Unlike the unbounded Couette flow in
\cite{bedrossian2015dynamics,bedrossian2016enhanced,Mas17}, the
semigroup associated with more general laminar flows in $[-1,1]$ is
not explicitly known. In the present contribution we thus consider, in
the limit $\epsilon\to0$, velocity fields $U\in C^4([-1,1])$ satisfying $U^\prime\neq0$
in $[-1,1]$, of two different types: 
\begin{itemize}
\item nearly Couette flows, so that
\begin{displaymath}
  \sup_{x\in[-1,1]} |U''(x)|   + \sup_{x\in[-1,1]} |U'''(x)| \ll
  \min_{x\in[-1,1]}|U^\prime(x)|\,,
\end{displaymath}
\item velocity fields for which $U^{\prime\prime}\neq0$ in $[-1,1]$.
\end{itemize}
 We prove that
these laminar flows are stable and provide $C(\R_+;\LL(L^2))$ estimates for
their associated semigroup norm. We believe that these linear
estimates would be useful when considering the nonlinear stability of
these flows in a bounded interval.

The rest of this contribution is arranged as follows.\\  In the next
section we formulate the spectral problem by using Hodge
decomposition. Since the standard Hodge decomposition \cite{gira79} is
not a direct sum in this periodic setup, we define a space of
zero-flux perturbations, and then formulate our main results in this
space.\\
 In Section \ref{sec:orr-somm-oper} we present the problem in terms of the stream
function and its Fourier coefficients and arrive at the Orr-Sommerfeld
operator.\\
 In Section \ref{sec:2} we consider the inviscid problem (where
$\epsilon=0$) in Fourier space. \\
Section \ref{sec:3-prem} includes some resolvent estimates
obtained for one-dimensional Schr\"odinger operators on the entire real
line and for their Dirichlet realization in $(-1,1)$. \\
 In Section \ref{s6} we
consider the same Schr\"odinger operator on $(-1,1)$ and $\R_+$ but this
time in a Sobolev  space of functions satisfying certain orthogonality
conditions (cf. \cite{sh04}).\\
 Section \ref{s7new} provides inverse estimates for
the Orr-Sommerfeld operator for the fixed-traction problem, whereas
Section \ref{sec:no-slip} provides the same estimates for the no-slip realization. \\
In Section \ref{sec:semigroup-estimates}  we prove the main results. \\
Finally in the appendix we bring
some auxiliary estimates obtained for  Airy functions and generalized Airy
functions.

 \section{Spectral problem formulation and main results}\label{s2}
We now amend the spectral question presented in  \eqref{neweq1} to a
more standard  spectral problem. To this end we use a variant of  Hodge
 decomposition adapted to our periodic setting  (see \cite[Theorem
 3.4]{gira79}  for the standard case), which allows us to 
 eliminate $q$.  Though expert readers are probably familiar
   with these ideas, we find it useful to  recall  them for the
   general reader's convenience.

 \subsection{Hodge theory}
 Let 
  \begin{equation}\label{defHg}
     \Hg=\{{\mathbf u} \in L^2_{loc}(\overline{D},\R^2) \,|\, {\mathbf u}(\cdot+L,\cdot)={\mathbf u}(\cdot,\cdot)\}  \,, 
  \end{equation}
  where the scalar product for the Hilbert space $\Hg$ is given by 
  \begin{displaymath}
      \Hg\times \Hg \ni (\mathbf{u},\mathbf{v}) \mapsto \int_{(0,L)\times (-1,1)} \bar {\mathbf u} \cdot {\mathbf  v}\,  dx_1 dx_2\,,
  \end{displaymath}
and the two closed subspaces of $\Hg$
  \begin{displaymath}
    \Hg_{\curl}=\{ {\mathbf u} \in L^2_{loc}(\overline{D},\R^2) \,|\, \curl {\mathbf u}=0 \,; \; {\mathbf u}(\cdot+L,\cdot)={\mathbf u}(\cdot,\cdot)\}  \,,
  \end{displaymath}
and
\begin{displaymath}
  \Hg_{\Div}=\{ {\mathbf u} \in L^2_{loc}(\overline{D},\R^2) \,|\, \Div {\mathbf u} =0 \,; \; {\mathbf
    u}\cdot{\mathbf n}|_{\partial D}=0\,; \; {\mathbf u}(\cdot+L,\cdot)={\mathbf u}(\cdot,\cdot)\}  \,.
\end{displaymath}
We have  
\begin{lemma}
\label{lem:intersection}

\begin{displaymath}
 \Hg_{\Div} \cap   \Hg_{\curl} = {\rm span}(1,0)\,.
\end{displaymath}
\end{lemma}
\begin{proof}
  Let $u\in \Hg_{\Div} \cap \Hg_{\curl}$, and ${\bf
    \hat{u}}=(\hat{u}_1,\hat{u}_2)$ denote its partial Fourier 
  transform with respect to $x_1$, i.e.,
\begin{equation}
\label{eq:6}
  {\bf \hat{u}}(n,x_2)=\frac{1}{L}\int_0^L{\bf u}(x_1,x_2)e^{-i2\pi nx_1/L}\,dx_1 \,.
\end{equation}
It can be easily verified that, for any $n\in \mathbb
Z$,
\begin{displaymath}
\begin{array}{l}
i \frac{2\pi}{L}n {\hat u}_2(n, x_2) - \frac{d}{dx_2}{\hat u}_1(n,x_2)=0\\
i\frac{2\pi}{L}n {\hat u}_1(n,x_2) + \frac{d}{dx_2} { \hat u}_2(n,x_2) =0\\
{\hat u}_2(n,-1)={\hat u}_2(n,+1)=0.
\end{array}
\end{displaymath}
From the above we can conclude that $\hat{\bf u}(n,x_2)=0$ for $n\neq 0$,
${ \hat u}_2(0,x_2)=0$ and ${\hat u}_1(0,x_2)= {\rm Const.}$
\end{proof}
With the above in mind, we introduce 
\begin{equation}
 \Hg_{\Div}^0:=\{ {\mathbf u} \in  \Hg_{\Div}\,,\, \langle \mathbf u, (1,0) \rangle =0\}.
 \end{equation}
The orthogonality condition reads
 \begin{equation}
\label{eq:7}
   \int_{(0,L)\times (-1,+1)} u_1(x_1,x_2)\, dx_1 dx_2 =0\,.
 \end{equation}
We can now prove  the following Hodge decomposition:
 \begin{lemma}
\label{lem:hodge}
   $\Hg_{\Div}^0$ and $\Hg_{\curl}$  are orthogonal subspaces of $\Hg$
   and

\begin{displaymath}
  \Hg = \Hg_{\curl}\oplus\Hg_{\Div}^0\,.
\end{displaymath}
  \end{lemma}
  \begin{proof}
    Let ${\bf u}\in\Hg$. Let further $\phi_c\in H^1_{loc}(\bar{D})$
    and $H^1_{loc}(\bar{D})$ denote the weak
    solutions of
    \begin{equation}
\label{eq:8}
      \begin{cases}
              -\Delta\phi_c= \Div {\bf u} & \text{in } D \\
      \frac{\partial\phi_c}{\partial n}=0 & \text{on } \partial D \\
\phi_c(\cdot,\cdot) = \phi_c(\cdot+L,\cdot) +A\,L & \text{in } D  \,,
      \end{cases}
    \end{equation}
and
\begin{equation}
\label{eq:9}
      \begin{cases}
        -\Delta\phi_d= \curl {\bf u} & \text{in } D \\
      \phi_d=0 & \text{on } \partial D \\
 \phi_d(\cdot,\cdot) = \phi_d(\cdot+L,\cdot) & \text{in } D  \,.
      \end{cases}
\end{equation}
In the above
\begin{equation}
\label{eq:10}
  A=\frac{1}{L}\langle{\bf u},(1,0)\rangle\,.
\end{equation}
It can be easily verified, by using the Lax-Milgram Lemma, that there
exists a unique solution for \eqref{eq:9}.  Similarly, by using (see
Remark \ref{rem:1.1}) the ansatz
\begin{equation}
\label{eq:11}
  \phi_c = Ax_1+ \tilde{\phi}_c\,,
\end{equation}there
where $\tilde{\phi}_c(\cdot,\cdot) = \tilde{\phi}_c(\cdot+L,\cdot) $, it follows that there exists a unique,
solution of \eqref{eq:8}. (Note that the Neumann condition in
the second line of \eqref{eq:8} is satisfied in $H^{-1/2}_{loc}(\partial D)$
sense.)  Equivalently, we can say that $\tilde \phi_c$ is the unique periodic solution,
orthogonal to the constant function, of 
\begin{displaymath}
\int_{(-1,+1)\times (0,L)} (\nabla \tilde \phi_c -{\bf u})\cdot \nabla v\, dx_1 dx_2 =
0\,,
\end{displaymath}
 for every $L$-periodic $v\in H^1((-1,+1)\times (0,L)) \,.$

Clearly, $\nabla\phi_c\in\Hg_\curl$,
$\nabla_\perp\phi_d\in\Hg_\Div^0$, $\langle{\bf u}-\nabla\phi_c-\nabla_\perp\phi_d,(1,0)\rangle=0$, and, by the periodicity of $\phi_d$ and $\nabla \phi_c$,
it holds that
\begin{displaymath}
  \langle\nabla\phi_c,\nabla_\perp\phi_d\rangle= 0\,.
\end{displaymath}
Finally we set $u=v+\nabla\phi_c+\nabla_\perp\phi_d$ to obtain that $v\in
\Hg_{\curl} \cap\Hg_{\Div}^0$ and hence, by Lemma
\ref{lem:intersection}, $v\equiv0$.
  \end{proof}

\subsection{Zero flux solution}
Since $\Wg_\#\not \subset\Hg_{\Div}^0$ for $\#\in\{\Df,\Sf\}$, we need to introduce the following
spaces, as the domain of the operator in the spectral formulation
\begin{subequations}
\label{eq:12}
    \begin{equation}
  \Wg^0_{\mathfrak D}= \Big\{{\mathbf u} \in\Wg^0_L\,\Big|\,  u_1\Big|_{\partial D}=0\,,\; u_2 \Big|_{\partial D}=0\,\Big\}\,,
\end{equation}
and
\begin{equation} 
  \Wg^0_{\mathfrak S}= \Big\{{\mathbf u} \in\Wg^0_L \,\Big|\, \;\frac{\partial u_1}{\partial
    x_2}\Big|_{\partial D}=0\,,\,  u_2\Big|_{\partial D}=0\,\Big\}\,,
\end{equation}
 where
\begin{displaymath}
   \Wg^0_L = \{{\bf u} \in H^2_{loc}(\overline{D},\R^2) \,| \, \Div {\bf u} =0 \,;\; {\mathbf u}(\cdot+L,\cdot)={\mathbf u}(\cdot,\cdot), \langle {\bf u} , (1,0)\rangle_{\Hg}=0\} \,.
\end{displaymath}
\end{subequations}
\begin{remark}
 We note that the orthogonality requirement \eqref{eq:7} is in
  accordance with the requirement ${\bf u} \in L^2(D)$ which should be
  applied if the periodicity requirement is dropped.  Formally, therefore,
  $ \Wg^0_L$ should serve as a good approximation for the space
\begin{displaymath}
   \Wg^0_\infty = \{{\bf u} \in H^2(\overline{D},\R^2) \,| \, \Div {\bf u} =0\} \,,
\end{displaymath}
in the limit $L\to\infty$.\\
We later explain in 
Remark \ref{remgen} how one can  more generally determine all the solutions of
  (\ref{neweq1}) in $\Wg_{\#}$ from its solution in $\Wg_\#^0$.
  \end{remark}

 Let then $P:\Hg\to\Hg_{\Div}^0$ denote the orthogonal projection on $
 \Hg_{\Div}^0$. We may express $P$ explicitly for some ${\bf u} \in\Hg$ by 
 \begin{equation}\label{defP}
 P {\bf u}=  \nabla_\perp\phi_d\,,
 \end{equation}
 where $\phi_d$ is the solution of \eqref{eq:9}. \\
Rewriting \eqref{neweq1} in the form 
\begin{displaymath}
  \mathfrak T {\mathbf u} -\Lambda {\bf u} -{\bf F}=\nabla q\,,
\end{displaymath}
we observe that $\nabla q \in \Hg_{\curl}$. Then,  
projecting on $\Hg_{\Div}^0$, we may write, for ${\mathbf u} \in \Wg_\sharp^0$ 
\begin{displaymath}
  P\big((\Tf-\Lambda){\mathbf u}-{\mathbf F})=0  \,. 
\end{displaymath}
With the above  in mind, we now define as an unbounded operator on
$\Hg_{\Div}^0$ whose domain is $\Wg_\sharp^0$ (which is clearly dense in $\Hg_\Div^0$)
\begin{equation}\label{defTfP}
\Tf_P^{\sharp}:= P \Tf \,.
\end{equation}
By this definition  we have
\begin{equation} 
\label{eq:13}
  (\Tf_P^\sharp -\Lambda){\mathbf u}=P{\mathbf F}\,,
\end{equation}
which appears to be a proper formulation of the resolvent equation.
\begin{proposition}\label{proppropTsharp}
 $\Tf_{P}^{\sharp}$ is semi-bounded on $\Hg_\Div^0$ and has compact
 resolvent.  Furthermore, 
   \begin{equation}
     \label{eq:367}
\| e^{-t \,\Tf_P^\sharp}\|\leq e^{\frac 12 \|U^\prime\|_\infty t} \,.
   \end{equation}
 \end{proposition} 
\begin{proof}
Let ${\bf u}\in\Wg_\#^0$. As
${\bf u} \perp\Hg_{\curl}$ we have
\begin{equation}\label{eq:newineq}
  \Re\langle {\bf u},\Tf_{P}^{\sharp}{\bf u} \rangle =  \Re\langle {\bf u},\Tf {\bf u} \rangle= \epsilon\|\nabla{\bf u}\|_2^2 +
  \Re\langle u_2,U^\prime u_1\rangle \geq  \epsilon\|\nabla{\bf u}\|_2^2 - \frac 12 \|U^\prime\|_\infty\|{\bf
    u}\|_2^2 \,,
\end{equation}
verifying, thereby, semi-boundedness.  More precisely, the resolvent
set of $\Tf_{P}^{\sharp}$ contains $\{\lambda\in\C \,|\, \Re\Lambda<- \frac 12
\|U'\|_\infty\,\}$. The semigroup estimate \eqref{eq:367} is then a
consequence of the Hille-Yosida theorem. The compactness of the
resolvent is proved by observing that $\Wg_\#^0$ is compactly embedded
in $\Hg_\Div^0$.
\end{proof} 
\pagebreak 
\begin{remark} \label{lem:exponential-rate}\strut\\
 \begin{itemize}
 \item Suppose that for some $\Lambda_0 \in\R$ and $C>0$, we have
  \begin{displaymath}
  \sup_{ \Re \Lambda \leq \Lambda_0 \,,\,
      \Lambda\in\rho(\Tf_P^\#)} \|(\Tf_P^\#-\Lambda)^{-1}\|\leq C\,.
  \end{displaymath}
Then, by the compactness of the resolvent the spectrum is discrete,
and hence it holds that   
\begin{displaymath}
  \sup_{  \Re\Lambda\leq \Lambda_0 } \|(\Tf_P^\#-\Lambda)^{-1}\|\leq C\,.
  \end{displaymath}
\item It results from \eqref{eq:newineq}  that there exists $C>0$, such
  that for any $\epsilon\in(0,1]$, ${\mathbf F}\in\Hg$, and $\#\in\{\Sf,\Df\}$, it holds that
\begin{equation}
  \label{eq:14}
\sup_{\Re\Lambda\leq -\frac 12 \|U^\prime\|_\infty-1} \|(\Tf_P^\#-\Lambda)^{-1}{\bf F}\|_{1,2} \leq
\frac{C}{\epsilon^{1/2}} \|{\bf F}\|_2\,,
\end{equation}
 where $\|\cdot \|_{k,p}$ denotes the
  $W^{k,p}((0,L)\times (-1,1))$ norm. In the sequel we use the
  same notation, depending on context, also for the $W^{k,p}(-1,1)$
  norm.  
\end{itemize}
\end{remark}
If $\Lambda$ is in the resolvent set of $\Tf_{P}^{\sharp}$,
we recover $\mathbf u$ by
\begin{equation}\label{eq:138a}
{\mathbf u}=    (\Tf_{P}^{\sharp} -\Lambda)^{-1} P{ \mathbf F}\,.
\end{equation}
Once we have derived ${\bf u}$ we can obtain $q$ in the following manner
\begin{proposition}
\label{prop:exist}
  Let $({\bf u},\Lambda, {\mathbf F} )\in\Wg_\#^0
  \times\C\times\Hg $ satisfy \eqref{eq:138a}.  Then, there exists a unique
  $q\in\QQ_L$ such that \eqref{neweq1} holds for $({\bf u},q, \Lambda,
  {\mathbf F} )$.
\end{proposition}
\begin{proof}
From \eqref{eq:13} it follows that
\begin{displaymath}
  (\Tf-\Lambda){\mathbf   u}-{\mathbf F} = {\mathbf G} \in\Hg_{curl} \,.
\end{displaymath}
It remains to prove the existence  of a unique $q_{\mathbf G} \in\QQ_L$ satisfying
\begin{equation}\label{eq:1.10}
\nabla q_{\mathbf G}  ={\mathbf G}  \,.
\end{equation}
This, however, easily follows from the proof of Lemma
\ref{lem:hodge}, i.e., one obtain $q_G$ as the unique solution of
\eqref{eq:8} with ${\mathbf G}$ in the place of ${\mathbf u}$. 
\end{proof}
\begin{corollary}\label{cor2.4}
If $\Lambda$ is in the resolvent set of $\Tf_{P}^{\sharp}$, then for any $F\in\Hg$ there
 exists a unique pair $({\bf u}_0,q_0)\in \Wg_\#^0\times \QQ_L$ such that \eqref{neweq1}
  holds.
  \end{corollary}
  We use the term  ``the zero flux  solution of \eqref{neweq1}''
  for this solution. 
 
  \begin{remark}\label{remgen}~
    \begin{itemize}
    \item From the proof of Proposition \ref{prop:exist} we learn, in
    addition, that if $\Lambda \in \sigma(\Tf_{P}^{\sharp})$ and ${\bf u}_\Lambda$
    is a corresponding eigenfunction in $\Wg_\#^0 $ then there exists
    $q_\Lambda$ such that $({\bf u}_\Lambda,q_\Lambda)$ satisfies \eqref{neweq1} with
    $F=0$. We cannot exclude, at the moment, the possibility that
    $\Lambda$ is not a simple eigenvalue.
  \item The proof shows also that, for any $\Lambda\in\rho(\Tf_{P}^{\sharp})$, the
    map $F \mapsto ({\bf u}_0, q_0)$ (as defined in the corollary) is
    continuous from $\Hg$ onto $\Wg_\#^0 \times \QQ_L$.
   \item 
 Once the zero flux solution of \eqref{neweq1}  has been found
 in $\Wg_\#^0 \times \QQ_L$, we can also solve the problem more generally  in  $\Wg_\# \times \QQ_L$. More precisely, if $\Lambda\in\rho(\Tf_{P}^{\sharp})$, then for any ${\bf F}$ and any $\gamma \in \mathbb R$, there
 exists a unique pair $({\bf u}_\gamma,q_\gamma)\in \Wg_\# \times \QQ_L$ satisfying \eqref{neweq1} 
  and 
  \begin{equation}
    \label{eq:15a}
   \langle {\bf u}_\gamma\,,\, (1,0)\rangle =\gamma\,.
  \end{equation}
   Let $({\bf u},q)$ denote the solution of  \eqref{neweq1}
  in $\Wg_\#^0 \times \QQ_L$. The proof is obtained by adding to ${\bf u}$
  an appropriate function of $x_2$ only. We omit the rather standad
  details in the interest of brevity.
\end{itemize}
   \end{remark}
 
\subsection{Longitudinal average}
We begin by defining the projection $ {\mathfrak p} :L^2((0,L)\times (-1,1))\to
L^2((0,L)\times (-1,1))$ by 
\begin{equation}
\label{eq:16} 
  {\mathfrak p} u  (x_1,x_2)=\frac{1}{L}\int_0^L u (s,x_2)\,ds \,,
\end{equation}
and then extend it to $\Pi:L^2((0,L)\times (-1,1);\mathbb R^2)\to L^2((0,L)\times
(-1,1);\mathbb R^2)$ by writing
\begin{equation}
  \Pi {\bf u} = ({\mathfrak p} u_1, {\mathfrak p} u_2) \mbox{ for } { \bf u} = (u_1,u_2)\,.
  \end{equation}

We first show
\begin{lemma} \label{Lemma2.9}~\\
$\Pi$ is a projection on $\Hg_\Div^0$.
 Moreover for any $\#\in \{\Df,\Sf\}$, we have
 $\Pi  \, \Wg_\#^0 \subset  \Wg_\#^0$.
\end{lemma}
\begin{proof}
  Let ${\bf u}\in \Hg_\Div^0$. Then ${\bf u}=\nabla_\perp \phi_d$ where $\phi_d$ is a
  solution of \eqref{eq:9}. We may then write, using the periodicity
  of $\phi_d$,
 \begin{equation}\label{eq:commpinabla}
   \Pi{\bf u}=\partial_{x_2} ({\mathfrak p}\phi_d)\, \hat{i}_1\,.
 \end{equation}

 Obviously, $\Div \Pi {\bf u}=0$, and the orthogonality of $\Pi {\bf u}$
 to $(1,0)$ in $\Hg$ follows from
\begin{displaymath}
 \int_{-1}^1 \partial_{x_2} ({\mathfrak p}\phi_d) \, dx_2 =0 \,.
\end{displaymath}
Hence $\Pi{\bf u}\in \Hg_\Div^0$. It can now be easily verified that  $\Pi
\, \Wg_\#^0 \subset  \Wg_\#^0$. 
\end{proof}

 \begin{lemma}
The projector $P$ commutes with $\Pi$.
  \end{lemma}
  \begin{proof}
 Let ${\bf F}\in\Hg$. Then, by the proof of Lemma \ref{lem:hodge}
\begin{displaymath}
  {\bf F}=\nabla_\perp\phi_d+\nabla\phi_c \,,
\end{displaymath}
where $\phi_d$ is a solution of \eqref{eq:9} and $\phi_c$ is a solution of
\eqref{eq:8}. \\
Clearly,
\begin{displaymath}
  \Pi\,{\bf F} = \big(\partial_{x_2} ({\mathfrak p}\phi_d)+A)\hat{i}_1 + \partial_{x_2} ({\mathfrak p} \tilde \phi_c) \,\hat{i}_2 \,,
\end{displaymath}
where $A$ is given by \eqref{eq:10} (with $\bf F$ instead of
$\bf u$). \\
Hence
\begin{displaymath}
  \Pi\,{\bf F} = \nabla_\perp ({\mathfrak p}\phi_d) + A \hat{i}_1 + \nabla ({\mathfrak p} \tilde \phi_c)  \,,
\end{displaymath}
and by uniqueness of the Hodge decomposition we obtain
\begin{displaymath}
 P\Pi\, {\bf F} =  \partial_{x_2} {(\mathfrak p}\phi_d) \,  \hat{i}_1 \,.  
\end{displaymath}
Next, we compute $\Pi P\, {\bf F}$.  Observing (see \eqref{defP}) that
\begin{displaymath}
  P\, {\bf F} = \nabla_\perp\phi_d \,,
\end{displaymath}
we get
\begin{equation}
\label{eq:17}
  \Pi P \, {\bf F}= \partial_{x_2} {(\mathfrak p}\phi_d) \, \hat{i}_1=P\Pi\,{\bf F}\,.
\end{equation}
\end{proof}

We can now prove  the following commutation result
\begin{lemma}
\label{lem:commute}
For any $\#\in\{\Sf,\Df\}$, $\Tf_P^{\sharp}$ commutes with $\Pi$.
\end{lemma}
\begin{proof}~\\
We simply observe that for all ${\bf u} \in \Wg_\#^0$

\begin{equation}\label{eq:comm}
   \Tf \, \Pi {\bf u}  = \Pi \, \Tf   {\bf u} \,,
\end{equation}
   and use the commutation of $P$ and $\Pi$.
       \end{proof}
   An immediate consequence follows
\begin{proposition}
  For any $\#\in\{\Df,\Sf\}$ it holds that
  \begin{equation}
    \label{eq:18}
\Pi \,e^{-t\, \Tf_P^\#}=e^{-t\, \Tf_P^\#}\,\Pi\,.
  \end{equation}
\end{proposition}

\subsection{The semigroup $\Pi \,e^{-t\, \Tf_P^\#}$} 
\label{sec:zeroeth-order}
 Since our
  main results are stated for $(I-\Pi)e^{-t\, \Tf_P^\#}$ we bring, for
  the convenience of the reader, the following, rather
  straightforward, estimate for $\|\Pi \,e^{-t\, \Tf_P^\#}\|$.
  \begin{proposition}
\label{prop:zeroeth}
    Let $U\in C^1([-1,1])$. Then, 
    \begin{equation}
      \label{eq:376}
\|e^{-t\, \Tf_P^\#  (U,\epsilon,L)}\Pi\|\leq e^{-\epsilon t\pi^2/4} \,.
    \end{equation}
  \end{proposition}
  \begin{proof}
    Let ${\bf u}\in L^2(0,T;\Wg_\#^0)$, s.t $\partial_t {\bf u} \in
      L^2(0,T, ; L^2_{loc}(\overline{D}))$,  where $\#\in\{\Sf,\Df\}$,
    and 
    $q(\cdot,\cdot)\in\QQ_{T,L}$
satisfy
    \begin{displaymath}
      {\bf u}_t - \Tf {\bf u} = \nabla q \,.
    \end{displaymath}
Since $\Div {\bf u} =0$, we can conclude, as in
\eqref{eq:commpinabla} that
\begin{equation}
\label{eq:377}
    \Pi {\bf u} = ({\mathfrak p}u_1,0) \,.
\end{equation}
Hence, using \eqref{eq:4a} and \eqref{eq:comm}, we conclude that
\begin{displaymath}
  \Pi \Tf {\bf u} =\epsilon\Big(\frac{\partial^2{(\mathfrak p}u_1)}{\partial x_2^2},0\Big)  \,.
\end{displaymath}
Thus, since $q(\cdot,\cdot)\in\QQ_{T,L}$,there exists, by \eqref{eq:11}, a
function $A \in L^2(0,T)$ such that
\begin{equation}\label{eq:2.28}
       ({\mathfrak p}u_1)_t - \epsilon\frac{\partial^2({\mathfrak p}u_1)}{\partial x_2^2}  = A(t) \,,
\end{equation}
in $L^2((0,T)\, \times (-1,+1))$. \\
Taking the inner product with ${\mathfrak p}u_1$ in $L^2(-1,1)$ then
yields, for any $\#\in\{\Sf,\Df\}$, in view of \eqref{eq:7}
\begin{equation}
  \frac{1}{2}\frac{\partial\|{\mathfrak p}u_1\|_2^2}{\partial t} +
  \epsilon\Big\|\frac{\partial{(\mathfrak p}u_1) }{\partial x_2} \Big\|_2^2 = 0 \,.
\end{equation}
In the case $\#=\Df$ we use Poincar\'e  inequality to obtain
\begin{displaymath}
  \frac{1}{2}\frac{\partial\|{\mathfrak p}u_1\|_2^2}{\partial t} +
  \epsilon \lambda_1^D \Big\| { \mathfrak p}u_1 \Big\|_2^2 \leq  0 \,,
\end{displaymath} 
where $\lambda_1^D$ is the first eigenvalue of the Dirichlet problem in
$(-1,+1)$, or,
\begin{displaymath}
\lambda_1^D= \frac {\pi^2}{4}\,.
\end{displaymath}
In the case $\#=\Sf$ we have $({\mathfrak p}u_1)^\prime(\pm1)$, and hence we can write
\begin{displaymath}
  \frac{1}{2}\frac{\partial\|{\mathfrak p}u_1\|_2^2}{\partial t} +
  \epsilon \lambda_2^N \Big\| { \mathfrak p}u_1 \Big\|_2^2 \leq  0 \,,
\end{displaymath} 
relying on the fact that $x_2 \mapsto ({\mathfrak p}u_1)(x_2) $ is
orthogonal, by (\ref{eq:7}), to the first eigenfunction of the Neumann
problem in $(-1,+1)$. Note that
\begin{displaymath}
  \lambda_2^N =\frac{\pi^2}{4}\,.
\end{displaymath}
From the above, together with \eqref{eq:377}, we conclude
\eqref{eq:376}.
  \end{proof} 

 \subsection{Main results}
 \label{sec:main-results}

 Throughout this work we assume that:
\begin{assumption}\label{ass:3.2} 
 $U^\prime$ does not vanish in $[-1,+1]$, or
 \begin{equation}
\label{eq:19}
  \mg:=\inf_{x\in[-1,1]}|U^\prime(x)|>0 \,.
\end{equation}
\end{assumption}

The statement of the main results below involves the spectral
properties of the complex Airy operator on $\R_+$
\begin{subequations}
\label{eq:20}
  \begin{equation}
  \LL_+ = -\frac{d^2}{dx^2}+ix\,,
\end{equation}
defined on
\begin{equation}
  D(\LL_+) = \{u\in H^2(\R_+)\cap H^1_0(\R_+) \,| \, xu\in L^2(\R_+)\,\}\,.
\end{equation}
We denote its leftmost eigenvalue \cite{al08} by $\nu_1$. We
  further set
  \begin{equation}
    {\mathfrak J}_m(U)=\min(|U^\prime(-1)|,|U^\prime(1)|)
  \end{equation}
\end{subequations}
We also need below
\begin{equation}\label{defdelta3} 
\delta_2(U):=     \Big\|U^{\prime\prime}\Big\|_{1,\infty}\,,
\end{equation}
where $\|u\|_{1,\infty}=\|u\|_\infty+\|u^\prime\|_\infty$\,.\\
Finally we define, for any $r>1$ and $k\geq2$,
\begin{equation}
\label{eq:21}
\Sg_r^k=\{ v\in C^k([-1,+1]), \inf_{x\in [-1,+1]} |v^\prime(x)| \geq r^{-1}\mbox{ and } \|v\|_{k,\infty} \leq r\}\,,
\end{equation}
and then set for convenience of notation
\begin{equation}
  \label{eq:22}
  \Sg_r=\Sg_r^4\,.
\end{equation}
For  $U\in \Sg_r$, $\epsilon >0$ and $L>0$,  we recall that  
\begin{displaymath}
\Tf_P^\Sf :=\Tf_P^\Sf(U,\epsilon,L)
\end{displaymath} 
is  defined in (\ref{eq:12}b) and \eqref{defTfP}
(where $\epsilon$ appears in the definition of $\Tf$ and $L$ is the $x_1$
periodicity).  \\
For $\beta >0$, we introduce
\begin{displaymath}
\Omega(\beta) :=\{(\epsilon, L)\in (0,1]\times \mathbb R_+\,, (L\epsilon)^{-1}\geq \beta/(2\pi)\}\,.
\end{displaymath}

 \begin{theorem}\label{thm:traction}
The following statements hold for any $r>1$. 
\begin{enumerate}
   \item 
Let $U\in \Sg_r$ satisfy
\begin{equation}\label{condsurinf}
   \inf_{x\in[-1,1]} |U^{\prime\prime}(x)| \geq 1/r\,.
   \end{equation}  
Then, for any $\hat{\delta}>0$, 
there exist $\Upsilon>0$, $\beta_0>0$ and $ C>0$ such that  for all $(\epsilon,L)\in
\Omega(\beta_0)$ and $t>0$ we have
\begin{equation}
\label{eq:23}
  \|e^{-t\, \Tf_P^\Sf (U,\epsilon,L)}(I-\Pi)\|\leq  C \, L^{\frac{1}{3}-\check \delta}{ \epsilon^{-\frac{7}{6} -\check \delta}} \, e^{-\epsilon\Upsilon[L\epsilon]^{-2/3}\, t} \,,
\end{equation}
 where $\Pi$ is given by \eqref{eq:16}. 

\item For all $\Upsilon< {\mathfrak J}_m^{2/3}\Re \nu_1$, there exist $ \delta>0$, $\beta_0 >0$ and $C>0$
  such that, for any $U\in\Sg_r$ satisfying $ \delta_2(U) <\delta \,,$,  $(\epsilon,L)\in
\Omega(\beta_0)$ and
  $t>0\,,$ it holds that
  \begin{equation}
    \label{eq:337}
  \|e^{-t\, \Tf_P^\Sf (U,\epsilon,L)}(I-\Pi)\|\leq  C L^{2/3}  \epsilon^{-\frac{5}{6}} \,  e^{-\epsilon\Upsilon[L\epsilon]^{-2/3}\, t} \,.
  \end{equation}
   \end{enumerate}
 \end{theorem}
 
For the case $\#=\Df$,  we first define, for some
$\theta>0$,  the operator $\LL^\theta$ whose differential operator is given by
(\ref{eq:20}a) and its domain by
\begin{displaymath}
  D( \LL^\theta)= \{\,u\in H^2(\R_+)\,| \, \langle e^{-\theta\cdot},u\rangle=0 \,,\;
  xu\in L^2(\R_+)\,\} \,.
\end{displaymath}
We show later (see Proposition \ref{lem:semi-infinite-spectrum} and
Corollary \ref{corhatmum})  that $ \LL^\theta$ is
a closed operator and that 
\begin{equation}
\label{defhatmu}  
  \hat{\mu}_m:=\inf_{\theta\geq0}(\inf \Re\sigma( \LL^\theta)+\frac {\theta^2}{2})\,.
\end{equation}
is finite and positive.  For $U\in \Sg_r$, $\epsilon >0$ and $L>0$,  we recall that
\begin{displaymath}
\Tf_P^\Df := \Tf_P^\Df(U,\epsilon,L)
\end{displaymath} 
is defined in (\ref{eq:12}a) and
\eqref{defTfP}.
\begin{theorem}
\label{thm:no-slip}
 For all $r>1$, the following properties hold.
\begin{enumerate}
\item Let $U\in \Sg_r$ satisfy \eqref{condsurinf}.
Then, for any $\hat{\delta}>0$, 
there exist $\Upsilon>0$, $\beta_0>0$ and $ C>0$ such that  for all $(\epsilon,L)\in
\Omega(\beta_0)$ and $t>0$ we have
\begin{equation}
\label{eq:24}
  \|e^{-t\, \Tf_P^\Df (U,\epsilon,L)}(I-\Pi)\|\leq  C \, L^{\frac{1}{3}-\check \delta} \epsilon^{-\frac{7}{6} -\check \delta}\, e^{-\epsilon\Upsilon[L\epsilon]^{-2/3}\, t} \,.
\end{equation}

 \item For all $\Upsilon< {\mathfrak J}_m^{2/3}\hat{\mu}_m$, there exist $ \delta>0$, $\beta_0 >0$ and $C>0$
  such that, for any $U\in\Sg_r$ satisfying $ \delta_2(U) <\delta \,,$ $(\epsilon,L)\in
\Omega(\beta_0)$ and
  $t>0\,,$ it holds that
  \begin{equation}
 \label{eq:371}
  \|e^{-t\, \Tf_P^\Df (U,\epsilon,L)}(I-\Pi)\|\leq  C\, L^{2/3} \epsilon^{-\frac{5}{6}} \, e^{-\epsilon\Upsilon[L\epsilon]^{-2/3}\, t} \,.
  \end{equation}
   \end{enumerate}
 \end{theorem}
 Note that together with  Proposition \ref{prop:zeroeth}, Theorems
  \ref{thm:traction} and \ref{thm:no-slip} provide stability of the
  semigroup $e^{-t\, \Tf_P^\# }$ for any $\#\in\{\Sf,\Df\}$.

 \begin{remark}
  Note that for Couette flow, where $\delta_2(U)=0$ and $ {\mathfrak J}_m=1$, one
  obtains that \eqref{eq:371} is true for any $ \Upsilon<\hat{\mu}_m$, and that 
  \eqref{eq:337} is true for all  $\Upsilon<\Re \nu_1$. This provides
  better estimate for the exponential rate of decay than in
  \cite[Proposition 6.1]{chen2018transition} which proves semigroup decay only for
  sufficiently small $\Upsilon$. 
 \end{remark}

\section{The Orr-Sommerfeld operator}
\label{sec:orr-somm-oper}
We focus attention in the sequel on $\Tf_{P}^{\sharp}$ and its
resolvent. 
\subsection{Stream Function}
\label{sec:stream-function}
When considering a two-dimensional incompressible fluid flow, it is
customary to introduce a stream function, i.e., to let ${\mathbf
  u}=\nabla_\perp\psi$.  Its introduction is again related to Hodge 
  decomposition theory.
\begin{lemma}
  Let ${\mathbf u}\in \Hg_{\Div}^0$.  Then, there exists a unique $\psi \in
  H^1_{0, loc}(\overline{D},\mathbb R^2))$ such that $ \psi(x_1+L,x_2)=\psi
  (x_1,x_2)$ and ${\mathbf u}=\nabla_\perp\psi$. 
   If in addition, ${\mathbf u} \in \Wg_\#^0$, then
\break  $\psi \in H^3_{loc}(\overline{D},\mathbb R^2))$ and $\psi$ satisfies
  $\partial_{x_2} \psi=0$ on $\partial D$ if $\#=\mathfrak D$ and $\partial_{x_2}^2 \psi=0$
  on $\partial D$ if $\#=\mathfrak S$.
\end{lemma}
\begin{proof}
Existence and uniqueness of $\psi$ follow from the proof of Lemma
\ref{lem:hodge}. In particular, for any ${\bf u} \in\Hg_\Div^0$ we have
$\psi=\phi_d$ where $\phi_d$ is a solution of \eqref{eq:9}. 
The second part of the lemma is immediate.
\end{proof}

We next substitute ${\bf u}=\nabla_\perp \psi$ into \eqref{eq:4a} and take the
curl of the ensuing equation, which leads to the following equation,
in the distributional sense,
\begin{equation}
\label{eq:26}
\mathcal P_{\Lambda,\epsilon}^\#  \,  \psi
= \curl {\mathbf F}  \quad \text{in } D
  \,,
\end{equation}
with $\#\in\{\mathfrak D,\mathfrak S\}$, 
\begin{equation}\label{eq:19aa}
\mathcal P_{\Lambda,\epsilon}^\#  :=  -  \epsilon \Delta^2\ + U\frac{\partial}{\partial x_1}\Delta -
  U^{\prime\prime}\frac{\partial}{\partial x_1}  - \Lambda \,  \Delta \,.
\end{equation}

We treat $\mathcal P_{\Lambda,\epsilon}^\#$ as an unbounded operator on $L^2_{per}(D)$, where
\begin{displaymath}
 L^2_{per}(D):=\{ u \in L^2_{loc} (\overline{D})\,,\, u(x_1+L,x_2)=u(x_1,x_2)\}\,,
\end{displaymath}
is equipped with the $L^2([-1,1]\times[0,L])$ norm.
Note that additional regularity is needed while attempting to use
 results obtained on $L^2_{per}(D)$ for spectral problem for $\Tf^\#_P$.
 We shall obtain the necessary regularity at a later stage. 
 Similarly, we introduce for $k=1,2$ and $s\geq 0$
\begin{displaymath}
 H^s_{per}(D,\C^k):=\{ u \in H^s_{loc} (\overline{D},\C^k)\,,\, u(x_1+L,x_2)=u(x_1,x_2)\}\,,
\end{displaymath}
and assign to it the $H^1([-1,1]\times[0,L])$ norm.  In the interest of
brevity we write in the sequel $H^s_{per}(D,\C)=H^s_{per}(D)$. For
no-slip boundary conditions we take
 \begin{displaymath}
   D(\mathcal P_{\Lambda,\epsilon}^{\mathfrak D})=\{ \psi \in H^4_{per}(D), \psi =0 \mbox{ and } \partial_{x_2} \psi =0 \mbox{ on } \partial D \}\,.
 \end{displaymath}
  For  fixed traction, the domain of $\PP_{\Lambda,\epsilon}^{\mathfrak S}$ is given
 by
 \begin{displaymath}
   D(\PP_{\Lambda,\epsilon}^{\mathfrak S})=\{ \psi\in H^4_{per}(D), \psi =0 \mbox{ and } \partial^2_{x_2} \psi =0 \mbox{ on } \partial D \} \,.
 \end{displaymath}
 We can now make the following statement
\begin{lemma}
  \label{lem:inverse-stream}
The operator $\mathcal P_{\Lambda,\epsilon}^\#$ is invertible for each
$\#\in\{\Sf,\Df\}$ and $\Lambda\in\rho(\Tf_P^\#)$. 
\end{lemma}
\begin{proof}
  Let $\psi\in D(P_{\Lambda,\epsilon}^\#)$ and $f\in L^2_{per}(D)$ satisfy $\mathcal
  P_{\Lambda,\epsilon}^\#\psi=f$ for some $\#\in\{\Sf,\Df\}$ and $\Lambda\in\rho(\Tf_P^\#)$.
  Let ${\bf F}$ denote the unique vector field in $\Hg_\Div^0$
  satisfying $\curl {\bf F}=f$. As
  \begin{displaymath}
    (\Tf-\Lambda)\nabla_\perp\psi-{\bf F}\in\Hg_\curl \,,
  \end{displaymath}
it follows that 
\begin{displaymath}
  \|\nabla\psi\|_2\leq\|(\Tf_P^\#-\Lambda)^{-1}\|\, \|{\bf F}\|_2 \leq \|(\Tf_P^\#-\Lambda)^{-1}\|\, \|f\|_2 \,.
\end{displaymath}
\end{proof}
  
 Due to the periodicity of our function spaces and the fact that the
 coefficients of the differential operator $\mathcal P_{\Lambda,\epsilon}^\#$ and the
 associated boundary conditions do not depend on $x_1$, it is natural
 to consider the operator in a Fourier space. 
 Hence, we introduce $ L^2_{per} (D) \ni \psi \mapsto (\mathcal F \psi ) \in
 \ell^2(\frac{2\pi}{L} \mathbb Z) \otimes L^2(-1,+1)$ satisfying
\begin{equation} 
\label{eq:27}
(  \mathcal F \psi)(\alpha_n,x_2)= \frac{1}{L}\int_0^L e^{-i\alpha_nx_1}\psi(x_1,x_2)\,dx_1 \,,
\end{equation}
for  $\alpha_n=2\pi n/L$ ($n\in \mathbb Z$).\\
We then obtain the Hilbertian sum
\begin{displaymath}
\mathcal F (I-\Pi)\PP_{\Lambda,\epsilon}^\# \mathcal F^{-1} = \epsilon\left(  \oplus_{n\in\mathbb
    Z\setminus \{0\}}  \B_{\lambda,\alpha_n,\beta_n}^\# \right)\,.
\end{displaymath}
The unbounded operator $\B_{\lambda,\alpha,\beta}^\#$ on $L^2(-1,+1)$, which is
commonly referred to as the Orr-Sommerfeld operator is given by
\begin{equation}
  \label{eq:28}
\phi \mapsto     \B_{\lambda,\alpha,\beta}^\#\phi=(\LL_\beta -\beta\lambda)\Big(\frac{d^2}{dx^2}-\alpha^2\Big)\phi-i\beta U^{\prime\prime}\phi \,,
\end{equation}
in which
\begin{equation} 
\label{defbeta}
\beta=\alpha\, \epsilon^{-1}\,,
\end{equation} 
\begin{equation}
\label{eq:29}
  \LL_\beta  = -\frac{d^2}{dx^2}+i\beta U\,,
\end{equation}
and, for $\beta \neq 0$,
\begin{equation}
\label{eq:30}
\lambda=\beta^{-1}(\hat{\Lambda}-\alpha^2)\,.
\end{equation}
in which
\begin{equation}
\label{eq:hatLambda}
  \hat{\Lambda}=\frac{\Lambda}{\epsilon}\,.
\end{equation}
In the sequel, unless stated otherwise, we consider $\beta$ and
$\alpha$ as independent parameters.\\

We now define two different realizations associated with the
differential operator appearing in \eqref{eq:28}.  The domain of
$\B_{\lambda,\alpha,\beta}^{\mathfrak S}$, corresponding to the prescribed traction
force boundary condition, is given by
\begin{equation}
  \label{eq:31}
D(\B_{\lambda,\alpha,\beta}^{\mathfrak S})=\{u\in H^4(-1,1)\,,\, u(\pm  1)=0 \mbox{ and } u^{\prime\prime}(\pm  1)=0\,\}
\,, 
\end{equation}
whereas the operator $\B_{\lambda,\alpha,\beta}^{\mathfrak D}$, corresponding to the no-slip condition, is defined
on 
\begin{equation}\label{eq:10v}
D(\B_{\lambda,\alpha,\beta}^{\mathfrak D})=\{u\in H^4(-1,1)\,,\, u(\pm  1)=0 \mbox{ and } u^\prime(\pm  1)=0\,\}.
\end{equation}
The domain of $\B_{\hat \Lambda}^\#$ is  similarly defined by \eqref{eq:31}
for $\#=\mathfrak S$ and \eqref{eq:10v} for $\#=\mathfrak D$. \\

\subsection{Inverse estimates}

The Orr-Sommerfeld operator given by \eqref{eq:28} has extensively
been studied in the Physics literature
\cite{drre04,Yaglom12,SchmidHenningson2001}. Very few rigorous
studies, however, address its spectrum (cf. \cite{ro73} in the Couette
case $U^{\prime\prime}=0$) and none, to the best of our knowledge provide
estimates for it inverse norm $\|(\B_{\lambda,\alpha,\beta}^\#)^{-1}\|$ in
$\mathcal L (L^2(-1,+1))$. Assuming that $\Lambda \in \rho(\mathcal T_P^\#)$,
the inverse of  $\PP_{\Lambda,\epsilon}(I-\Pi)$ is bounded. To estimate its norm, one
needs a proper uniform bound of $\|(\B_{\lambda,\alpha,\beta}^\#)^{-1}\|$ for all
$\alpha > 0$. Hence we write,
\begin{displaymath} 
  \|(\PP_{\Lambda,\epsilon}^\#)^{-1}(I-\Pi)\|\leq\epsilon^{-1}\sup_{n\in\Z\setminus\{0\}}\|(\B_{\lambda_n,\alpha_n,\beta_n}^\#) ^{-1}\|\,,
\end{displaymath}
where, for $n \in\Z\setminus\{0\}$,  $ \Lambda = \epsilon \hat  \Lambda$, 
\begin{displaymath}
 \lambda_n = \beta_n^{-1} (\hat \Lambda -\alpha_n^2)= \epsilon \alpha_n^{-1} (\hat \Lambda -\alpha_n^2)= \alpha_n^{-1} \Lambda -\epsilon \alpha_n\,.
\end{displaymath}
Clearly, for any $\epsilon >0$ and $\hat \Lambda \in \mathbb R$, 
\begin{displaymath}
 \sup_{n\in\Z\setminus\{0\}}\|(\B_{\lambda_n,\alpha_n,\beta_n}^\#)^{-1}\|\leq  \sup_{
   \begin{subarray}{c}
     \alpha\geq\alpha_1 \\
     \lambda= \alpha^{-1}\epsilon (\hat \Lambda -\alpha^2)
   \end{subarray}}\|(\B_{\lambda,\alpha,\alpha \epsilon^{-1}}^\#)^{-1}\|\leq  \sup_{\beta\geq\beta_1}\sup_{
   \begin{subarray}{c}
     \alpha\geq\alpha_1 \\
     \lambda= \beta^{-1}(\hat \Lambda -\alpha^2)
   \end{subarray}}\|(\B_{\lambda,\alpha,\beta}^\#)^{-1}\|\,,
\end{displaymath}
 with 
\begin{displaymath}
\beta_1(\epsilon,L)= \epsilon^{-1} \alpha_1= 2\pi/(L\epsilon) \,.
\end{displaymath}

 Consequently, for any $\Upsilon\in \mathbb R_+$ and $\epsilon >0$, we have,
the following inequality
\begin{equation}
\label{eq:34}
 \sup_{
   \begin{subarray}{c}
\Re\Lambda\leq \epsilon\Upsilon\beta_1(\epsilon,L)^{2/3} \\
     \Lambda\in\rho(T_P^\#)
   \end{subarray}}\|(\PP_{\Lambda,\epsilon}^\#)^{-1}(I-\Pi)\|\leq 
  \epsilon^{-1}\sup_{ \beta\geq\beta_1(\epsilon,L)}\sup_{
    \begin{subarray}{c}
      \alpha\geq 0  \\
      \Re\lambda\leq \beta^{-1}(\Upsilon\beta_1(\epsilon,L)^{2/3}-\alpha^2)
    \end{subarray}}\|(\B_{\lambda,\alpha,\beta}^\#)^{-1}\|\,.
\end{equation}
Note that $\B_{\lambda,\alpha,\beta}^\#=\overline{\B}_{\lambda,-\alpha,\beta}^{\,\#}$ and hence, it
is sufficient to consider $\alpha\geq0$ in the above. We emphasize that the supremum with
respect to $\lambda$, $\beta$ and $\alpha$ of $\|(\B_{\lambda,\alpha,\beta}^\#)^{-1}\|$ is
obtained while ignoring the dependence of $\beta$ on $\alpha$. Note also that
$\beta_1=\beta_1(\epsilon)$ tends to $+\infty$ as $L\epsilon \to 0$. Hence we shall attempt to
obtain a bound on $\|(\B_{\lambda,\alpha,\beta}^\#)^{-1}\|$ for large values of
$\beta$, since we are interested in the small $\epsilon$ limit.

Throughout this work we recall that we always  assume Assumption \ref{ass:3.2}.
Without any loss of generality we can assume that
\begin{equation}
\label{eq:35}
U^\prime>0 \mbox{ on }[-1,+1]\,.
\end{equation}
Indeed, the case $U^\prime<0$ can similarly be treated after applying
the transformation $x\to-x$.
\\[1.5ex]
In view of \eqref{eq:34}, we attempt to obtain, in the large $\beta$
limit, a bound on $\|\B_{\lambda,\alpha,\beta}^{-1}\|$.  To this end, we introduce
\begin{equation}
\label{eq:36}
\A_{\lambda,\alpha} \overset{def}{=} (U+i\lambda)\Big(-\frac{d^2}{dx^2}+ \alpha^2\Big)+U^{\prime\prime}\,.
\end{equation}
We define $\A_{\lambda,\alpha}$, for $\Re\lambda\neq0$ or when $\Re\lambda=0$ and
$\Im\lambda\not\in[U(-1),U(1)]$, on 
\begin{equation}
\label{eq:37}
  D(\A_{\lambda,\alpha})= H^2(-1,1)\cap H^1_0(-1,1)\,.
\end{equation}
In the case $\Re\lambda=0$ and $\Im\lambda\in[U(-1),U(1)]$ we set
\begin{equation}
\label{eq:38}
  D(\A_{i\nu,\alpha})= H^2((-1,1);|U-\nu|^2dx)\cap H^1_0(-1,1)\,,
\end{equation}
where $\nu\in[U(-1),U(1)]$.\\
 Note that $\A_{\lambda,\alpha}=\A_{\lambda,-\alpha}$ and hence
we consider it only in the case $\alpha\geq0\,$. \\
One can formally obtain $\A_{\lambda,\alpha}$ from $\B_{\lambda,\alpha,\beta}$ by dividing
it by $\beta$ and taking the limit $\beta\to\infty$ which corresponds to the
limit $\alpha^{-1}\epsilon \to 0$. This is why it has been commonly referred to
as the ``inviscid operator'' in \cite{drre04}.

We note that the formal limit of the Orr-Sommerfeld operator as
$\beta\to\infty$ is very different from that of the Schr\"odinger operator
$-\frac{d^2}{dx^2} +i\beta(U-\nu)$ ($\nu=\Im \lambda$).  In the latter case, we
expect the resolvent to be small away from the set were $U= \nu $.
This fact was used in \cite{al08,AGH}, for instance, to obtain
resolvent estimates via localization techniques. For $\B_{\lambda,\alpha,\beta}\,$, the
best one can expect is that $v=\A_{\lambda,\alpha}\phi$ would be small outside a close
neighborhood of the set where $U=\Im\lambda$. We note that $\A_{\lambda,\alpha}$ raises
  considerable interest independently of the viscous operator
  $\B_\lambda^\#$ (cf. \cite[Sections
  21-24]{drre04} and \cite{li05}). 
    
\section{The inviscid operator}\label{sec:2} 
We consider here the inviscid operator $\A_{\lambda,\alpha}$
(often called the Rayleigh operator ) associated with the differential
operator \eqref{eq:36}, whose domain of definition is given either by
\eqref{eq:37} where $\Re\lambda\neq0$ or  $\Im\lambda\not\in[U(-1),U(1)]$, or
by \eqref{eq:38} in the case $\lambda=i\nu$ for $\nu\in[U(-1),U(1)]$.  
  The spectrum and the inverse of $\A_{\lambda,\alpha}$ have been studied in the
  context of both inviscid (Euler) and viscous flows and their
  stability 
  \cite{drre04,li05,grenier2016spectral,grenier2017green,wei2018linear,jia2019Gevrey,jia2019linear}.
  In particular, the inverse norm has been estimated in
  \cite{grenier2016spectral} for symmetric shear flows in a channel in
  the limit $\alpha \ll 1$. Similar estimates are obtained in
  \cite{grenier2017green} for boundary-layer type flows in a
  semi-infinite domain. In
  \cite{wei2018linear,jia2019Gevrey,jia2019linear} weighted norm
  estimates are obtained for $\A_{\lambda,\alpha}^{-1}$. These norms are hard
  to implement when seeking inverse estimates for the Orr-Sommerfeld
  operator \eqref{eq:28}. The purpose of this section is therefore to
  offer a systematic study of $\|\A_{\lambda,\alpha}^{-1}\|$ using standard
  Sobolev norms, with emphasis on the limit $\Re\lambda\to0$ for
  $\nu\in[U(-1),U(1)]$.

\subsection{The case $\Re\lambda=0\,$: preliminaries}
We begin by showing that $\A_{\lambda,\alpha}$ is a closed operator on
$L^2(-1,1)$.  In the cases where $D(\A_{\lambda,\alpha}$ is given by
\eqref{eq:37} the proof is standard and will therefore be omitted.
\begin{proposition}
\label{propclosed} 
Let $U$ satisfy Assumption \ref{ass:3.2} and  $\nu \in [U(-1),U(1)]$ . Then $\A_{i\nu,\alpha}$,
whose domain is given by \eqref{eq:38}, is closed as an unbounded
  operator on $L^2(-1,+1)$  and the space 
\begin{displaymath}
\mathcal V:= \{ u\in C^\infty ([-1,+1]), \mbox{ s.t.  } u(-1)=u(1)=0\}
\end{displaymath}
is dense in $D(\A_{i\nu,\alpha})$ under the graph norm.
\end{proposition}
  \begin{proof} 
Let $\{u_n\}_{n=1}^\infty$ denote 
   a sequence in $\Vg$ such that
   \begin{equation}\label{suite}
  u_n \to \hat u \mbox{ in } L^2 (-1,+1) \mbox{ and } \A_{i\nu,\alpha} u_n \to \hat v\,\mbox{ in } L^2 (-1,+1) 
  \end{equation}
To prove the proposition we need to show that
   $\hat u\in D(\A_{i\nu,\alpha})$ and that $\A_{i\nu,\alpha} \hat u =\hat v$.
 
 Let $U(x_\nu)=\nu$. We use Hardy's inequality for weighted Sobolev 
 spaces associated with the intervals  $(-1,x_\nu)$ and $(x_\nu,1)$ separately.\\
 Set then, for $k\in\{1,2\}$

\begin{displaymath}
 W^{k}_{1,\nu} (x_\nu,1):=  \{u\in {H^{k-1}(x_\nu,1)} \,,\, (U-\nu) u^{(k)} \in L^2(x_\nu,1), u(1)=0\}\,.
\end{displaymath}
  Recall the one-dimensional Hardy inequality (see for example
 \cite[Eq. 0.6]{kupe03} or \cite[Lemma 2.1]{BC1}) which holds for any
 $ v\in W^{1}_{1,\nu} (x_\nu,1)$\,,
 \begin{equation}
\label{eq:hardy1}
   \|v\|_{L^2(x_\nu,1)} \leq  2 \,  \| (x-x_\nu )v ^\prime\|_{L^2(x_\nu,1)}\,.
\end{equation}
Note that  \eqref{eq:hardy1} follows by extension from Hardy's inequality in
 $(0,+\infty)$, given for any $v\in L^2(\mathbb R_+)$ s.t.  $xv' \in L^2(\mathbb R_+)$,  we have
 \begin{equation}\label{eq:hardy2}
  \|v\|_{L^2(0,+\infty)} \leq  2  \, \| x \, v^\prime \|_{L^2(0,+\infty)}\,.
 \end{equation}
 Hence, for any $u\in W^{2}_{1,\nu} (x_\nu,1)$ (and hence for the
 restriction to $[x_\nu,1]$ of any $u\in\Vg$) 
 \begin{equation}
\label{eq:hardy0}
 \|u^\prime\|_{L^2(x_\nu,1)} \leq  2 \| (x-x_\nu) u^{\prime\prime} \|_{L^2(x_\nu,1)}\leq \frac{C}{\mg} \| (U-\nu ) u^{\prime\prime}\|_{L^2(x_\nu,1)}\,.
 \end{equation}

From \eqref{suite} and \eqref{eq:hardy0} we deduce that $u_n$ is a
Cauchy sequence in $H^1(x_\nu,1)$ and in $H^1(-1,x_\nu)$. Hence there are
two corresponding limits $u_+ \in H^1(x_\nu,1)$ and $u_- \in H^1(-1,x_\nu)$,
with $\hat u_{/(-1,x_\nu)}=u_-$ and $\hat u_{/(x_\nu,1)} =u_+$.  By
continuity we have $u_-(-1)=0$, $u_+(1) =0$ and $u_-(x_\nu)=u_+(x_\nu)$.
This shows that $\hat u\in H_0^1(-1,+1)$.  Finally it is clear that
$\hat v = \A_{i\nu,\alpha} \hat u$ in $\mathcal D^\prime(-1,+1)$ and hence
$\hat{u}\in D(\A_{i\nu,\alpha})$.  

The density argument is a consequence, after localization, of
Proposition 2.1 in \cite{BC2}.
   \end{proof}
   The cases $x_\nu=\pm  1$ are easier, since one can apply Hardy's
   inequality only once, in $(-1,1)$.
   
   Consider the case of Couette flow $U=x$  in $(-1,1)$, which is one
   of the most popular examples of uniaxial flows (\cite{wa53,ro73,BGM}.  In this case
\begin{displaymath}
  \A_{i\nu,\alpha}^c\overset{def}{=}(x-\nu)\big(-\frac{d^2}{dx^2} + \alpha^2\big) \,.
\end{displaymath}
We can construct in some cases the explicit solution of the
inhomogeneous problem
\begin{equation}
  \label{eq:39}
\A_{i\nu,\alpha}^c u =f\,.
\end{equation}
For example, when  $\alpha=0$, $\nu \in (-1,+1)$ and $f\equiv1$. It can be easily verified that
\begin{subequations}
\label{eq:40}
  \begin{equation}
  u_\nu (x)=A_1(x-\nu) + A_2+ (x-\nu)\log\, |x-\nu| \,,
\end{equation}
where
\begin{equation}
\begin{array}{ll}
  A_1&= -\frac{1}{2}\big[(1-\nu)\log\,|1-\nu|+(1+\nu)\log\, |1+\nu|\big] \,,\\
  A_2&=\frac{1}{2}\big[(1-\nu)\log\,|1-\nu|-(1+\nu)\log\, |1+\nu|\big]\,.
  \end{array}
\end{equation}
\end{subequations}
It can be easily verified that $u_\nu$ satisfies \eqref{eq:39}. 

While it seems at a first glance that $ \A_{i\nu,0}^c:D(\A_{i\nu,0}^c)\to
L^2(-1,1)$ is injective for all $\nu\in\R$, it is wrong for $\nu\in (-1,+1)$. 
It has indeed a non-trivial solution
of the form
\begin{displaymath}
 \varphi_\nu(x) =
  \begin{cases}
    \frac{x-\nu}{1-\nu}-1 &\mbox{ if } \nu\leq x\,, \\
    -\frac{x-\nu}{1+\nu}-1 & \mbox{ if } x<\nu \,.
  \end{cases}
\end{displaymath}
While $\varphi_\nu\not\in H^2(-1,1)$ for all $\nu\in(-1,1)$, it does satisfy
\eqref{eq:42} in the sense of distributions.\\ 

We now look at the injectivity of $\A^c_{i\nu,\alpha}$ when $\nu \in
(-1,+1)$. We have already explicitly  obtained a non trivial element in
${\rm Ker} \A^c_{i\nu,0}$. More generally, we prove
   \begin{lemma}\label{lem:4.2}
   For all  $\nu \in (-1,1)$, $\alpha \geq 0$, it holds that
   \begin{equation}
     \label{eq:41}
   {\rm Dim}\, {\rm Ker} \A_{i\nu,\alpha}^c=1\,.
   \end{equation}
    \end{lemma}
 \begin{proof}
We observe that if $u$ is in the kernel, we have $u\in H_0^1(-1,+1)$
and  

\begin{displaymath}
\Big(-\frac{d^2}{dx^2} + \alpha^2\Big) u = c\,  \delta(x-\nu)\,,
\end{displaymath}
where $c$ is the jump $u^\prime$ undergoes through $x=\nu$.  Since
$\delta(x-\nu)$ belongs to $H^{-1} (-1,+1)$, we may use the
Lax-Milgram Lemma for the Dirichlet problem to obtain a unique
solution $u_1\in H_0^1(-1,+1)$ for
\begin{displaymath}
\Big(-\frac{d^2}{dx^2} + \alpha^2\Big) u_1 =   \delta(x-\nu)\,.
\end{displaymath} 
It follows that $u = c u_1$. 
\end{proof}

Using \cite[Proposition 3.1]{BC2} and the same argument as in the
proof of \cite[Theorem 3.1]{BC2} one can show that $\A_{i\nu,\alpha}^c$,
which is a bounded operator from the Hilbert space $ D(\A_{i\nu,\alpha}^c)$
into $L^2(-1,+1)$, is a Fredholm operator of index $1$ for $\nu \in
(-1,+1)$. With the aid of Lemma~\ref{lem:4.2}, we can then conclude
the surjectivity of $\A_{i\nu,\alpha}^c$ for any $\nu \in (-1,+1)$ or more
precisely the existence of a right inverse. We note that 
surjectivity of $\A_{i\nu,\alpha}^c$ follows from the  surjectivity
$\A_{i\nu,\alpha}$ we prove in the sequel. 

   \begin{remark}
   When $f\in C^1([-1,1])$, we can prove \eqref{eq:41} in the following
   alternative manner.  Recall the solution of $A_{i\nu,0}^c u_\nu= 1$ we
   have obtained in \eqref{eq:40}. If $f(\nu)=0$, we can solve
   \eqref{eq:39} by dividing it by $(x-\nu)$. 
   When $f(\nu)\neq0$ we write $u =f(\nu) u_\nu + \tilde u\,,$ to obtain
\begin{displaymath}
   \A_{i\nu,\alpha}^c\tilde u = \tilde{f}\,,\, \mbox{ with } \tilde{f}(x):= f(x)-f(\nu) -\alpha^2 (x-\nu) u_\nu\,,
\end{displaymath}
which can be easily solved as $\tilde{f}(\nu)=0$. 
   It is not clear, however, how to extend by density the above
   solution to any $f\in L^2(-1,+1)$. 
    \end{remark}

    In the next subsection, we obtain the surjectivity of $A_{i\nu,\alpha}$
    (for any $U$ satisfying \eqref{eq:35}) via a non-explicit
    compactness argument.
 
  \subsection{Construction  of a right inverse of $\A_{i\nu,\alpha}$}
We begin by establishing the surjectivity of  $\A_{i\nu,\alpha}$
\begin{lemma}
\label{lem:surjective}
Suppose that \eqref{eq:35} holds, and that $\lambda=i\nu$ for some  $\nu\in
[U(-1),U(1)]$. Then, for any $v\in L^2(-1,1)$ and $\alpha\geq0$ there exists some 
$\phi\in D(A_{\lambda,\alpha})$ satisfying $A_{\lambda,\alpha}\phi=v$. Furthermore, there
exists $C>0$, such that for all $\nu$, $\alpha$ and $v\in L^2$,  $\phi$
satisfies 
\begin{equation}
  \label{eq:42}
\|\phi\|_{1,2}+[1+\alpha^2]^{-1}\|(U- \nu)\phi^{\prime\prime}\|_2\leq C \, \|v\|_2\,,
\end{equation}
where  $\|\phi\|_{1,2}$ denotes the norm of $\phi$  in the Sobolev space
$H^1(-1,1)$. 
\end{lemma}

\begin{proof}
  Let $\kappa>0$. With \eqref{eq:36} in mind, we can use the following
  alternative form of $A_{\lambda,\alpha}\phi=v$, which is valid wherever
  $U\neq\nu$,
 \begin{equation}
\label{eq:43}
   -\Big((U-\nu)^2\Big(\frac{\phi}{U-\nu}\Big)^\prime\Big)^\prime+ \alpha^2(U-\nu)\phi = v\,.
 \end{equation}
We look first for some $\phi_\kappa\in H^2(-1,+1)\cap H_0^1(-1,+1) \subset D(\A_{i\nu,\alpha})$  satisfying the regularized
 equation
 \begin{equation}
\label{eq:44}
     -\Big([(U-\nu)^2+\kappa^2]\Big(\frac{\phi_\kappa}{U-\nu+i\kappa}\Big)^\prime\Big)^\prime+
     \alpha^2(U-\nu-i\kappa)\phi_\kappa = v\,.
 \end{equation}
 We may set $w_\kappa=(U-\nu+i\kappa)^{-1}\phi_\kappa$ to obtain 
 \begin{displaymath}
   -([(U-\nu)^2+\kappa^2]w^\prime_\kappa)^\prime + \alpha^2[(U-\nu)^2+\kappa^2]w_\kappa=v
 \end{displaymath}
 Taking the inner product with $w_\kappa$ then yields
 \begin{equation}
\label{eq:45}
   \|[(U-\nu)^2+\kappa^2]^{1/2}w^\prime_\kappa\|_2^2 + 
   \alpha^2\|[(U-\nu)^2+\kappa^2]^{1/2}w_\kappa\|^2_2= \langle w_\kappa,v\rangle\,.
 \end{equation}
 This immediately implies, by the Lax-Milgram Lemma, that $w_\kappa$ 
(and hence also $\phi_\kappa$) uniquely exists in $H^2(-1,1)\cap H^1_0(-1,1)$. \\
Since $\nu\in[U(-1),U(1)]$, there exists $x_\nu\in[-1,1]$ such that
$U(x_\nu)=\nu$. Let $\mg>0$ be defined by \eqref{eq:19}. Clearly, we have
\begin{displaymath}
  (U-\nu)^2+\kappa^2 \geq \mg ^2(x-x_\nu)^2 \,.
\end{displaymath}
Hence
\begin{displaymath}
\begin{array}{ll}
  \|[(U-\nu)^2+\kappa^2]^{1/2}w^\prime_\kappa\|_2^2 & \geq \mg^2\|(x-x_\nu)w^\prime_\kappa\|_2^2\\ & =
  \mg^2\left[\|(x-x_\nu)w^\prime_\kappa\|_{L^2(-1,x_\nu)}^2+\|(x-x_\nu)w^\prime_\kappa \|_{L^2(x_\nu,1)}^2 \right]\,. 
  \end{array}
\end{displaymath}

Upon translation we apply Hardy's inequality \eqref{eq:hardy2},  and
the fact that $w_\kappa(1) =0$ (permitting the extension $w_k=0$ in
$(1,+\infty)$ so that $w_k\in H^1(x_\nu,\infty)$),  it holds that
\begin{displaymath}
  \|(x-x_\nu)w^\prime_\kappa\|_{L^2(x_\nu,1)}^2\geq\frac{1}{4}\|w_\kappa\|_{L^2(x_\nu,1)}^2 \,.
\end{displaymath}
A similar bound can be established on $(-1,x_\nu)$.  Hence, by
\eqref{eq:45},
\begin{equation}
  \label{eq:46}
\|w_\kappa\|_2\leq\frac{4}{\mg^2}\|v\|_2\,.
\end{equation}
Substituting once again into \eqref{eq:45} yields, in addition,
\begin{equation}
  \label{eq:47}
\|[(U-\nu)^2+\kappa^2]^{1/2}w^\prime_\kappa\|_2 \leq\frac{2}{\mg}\|v\|_2 \,.
\end{equation}

We note that \eqref{eq:46} and \eqref{eq:47} do not imply convergence
of $w_\kappa$ in $L^2(-1,1)$ to  $(U-\nu)^{-1}\phi$ where $\phi$ is a solution
of \eqref{eq:36} (with $\lambda=i\nu$). As a matter of fact it
is expected that, in most cases,  $(U-\nu)^{-1}\phi$ will be unbounded in
$L^2$. We thus use \eqref{eq:47} and \eqref{eq:46} to obtain that
\begin{equation}
\label{eq:48}
  \|\phi^\prime_\kappa\|_2 = \|([U-\nu+i\kappa]w_\kappa)^\prime\|_2 \leq
  \|[(U-\nu)^2+\kappa^2]^{1/2}w_\kappa^\prime\|_2 +  \|U^\prime w_\kappa\|_2 \leq
  C\|v\|_2 \,.
\end{equation}
It follows, either by Poincar\'e's inequality  or by \eqref{eq:46}, that
$\{\phi_\kappa\}_{k=1}^\infty$ is bounded in a ball of size $\hat C\|v\|_2$ in
$H^1_0(-1,1)$. By weak compactness,  there exists a sequence $\kappa_n >0$
tending to $0$ such that 
$\phi_{\kappa_n} $ converges weakly in $H^1_0(-1,1)$ to a limit
$\phi_0\in H^1_0(-1,1)$ as $n\to+\infty$ in the same ball, and hence
\begin{equation}  \label{eq:49}
\| \phi_0\|_{ H^1_0(-1,1)} \leq \hat C\, \|v\|\,.
\end{equation}
It remains to establish that
$\phi_0$ is a weak solution of \eqref{eq:36} for $\lambda=i\nu$. We thus write
\eqref{eq:44} in its weak form for some $\psi\in C_0^\infty(-1,1)$
\begin{equation}
  \label{eq:50}
\langle\psi^\prime,(U-\nu-i\kappa)\phi_\kappa^\prime\rangle-
\Big\langle\psi^\prime,\frac{(U-\nu-i\kappa)}{(U-\nu+i\kappa)}U^\prime\phi_\kappa\Big\rangle +
\alpha^2\langle\psi,(U-\nu-i\kappa)\phi_\kappa\rangle= \langle\psi,v\rangle\,.
\end{equation}
Letting $\kappa=\kappa_n$ and then taking the limit $n\to +\infty$ yields by the
above established weak convergence in $H^1$ that
\begin{displaymath}
  \langle\psi^\prime,(U-\nu)\phi_0^\prime\rangle-
\langle\psi^\prime,U^\prime\phi_0\rangle + \alpha^2\langle\psi,(U-\nu)\phi_0\rangle= \langle\psi,v\rangle\,.
\end{displaymath}
Consequently, $\phi_0$ is a weak solution of \eqref{eq:36} satisfying \eqref{eq:35}.
Finally, since for all $x\neq x_\nu$ we have
\begin{displaymath}
  (U-\nu)[\phi^{\prime\prime}_0-\alpha^2\phi_0]=U^{\prime\prime}\phi_0 \,,
\end{displaymath}
we easily obtain from \eqref{eq:49} that
\begin{displaymath}
  \|(U-\nu)\phi^{\prime\prime}_0\|_2 \leq C(1+\alpha^2)\|v\|_2 \,.
\end{displaymath}
The lemma is proved.
\end{proof}

\begin{remark}
\label{rem4.5}
  As in the case of a Couette flow, one can show that the index of
  $\A_{i\nu,\alpha}$ as a Fredholm operator from the Hilbert space
  $D(\A_{i\nu,\alpha})$ (equipped with the graph norm) into $L^2(-1,1)$ is
  one for $\nu \in [U(-1),U(1)]$.  Since the multiplication by
  $U^{\prime\prime}$ is a compact operator from $D(\A_{i\nu,\alpha})$ into
  $L^2(-1,1)$ we may conclude from Fredholm theory that
  \begin{displaymath} \tilde{A}_{\nu,\alpha}=
    (U-\nu)\Big(-\frac{d^2}{dx^2}+\alpha^2\Big)
   \end{displaymath}
and $\A_{i\nu,\alpha}$ have the same index.  Then,  we observe that
   the multiplication operator by $(U(x)-\nu)/(x-x_\nu)$ has index $0$,
   and hence $\tilde{A}_{\nu,\alpha}$ and $\A^c_{i\nu,\alpha}$ again have the same
   index. Consequently, the index of $\A_{i\nu,\alpha}$ equals to one,
   and by Lemma \ref{lem:surjective} it holds that ${\rm dim}\, {\rm ker}
   \A_{i\nu,\alpha}=1$. We note that one can construct a direct proof of injectivity
   with much greater difficulty.
\end{remark}

We have proved above, the surjectivity of $\A=\A_{i\nu,\alpha}$. By
Fredholm theory, there is a natural right inverse $\mathcal
E_{i\nu,\alpha}^{\rm right}$ which associates with $v$ the
solution $\phi$ of $\A \phi=v$ which is orthogonal to ${\rm
  Ker}A_{i\nu,\alpha}$ in $D(\A_{i\nu,\alpha})$ for the scalar product

\begin{displaymath}
  (\phi,\psi) \mapsto \langle \phi\,,\, \psi \rangle_{\nu,\alpha}:= \langle \A \phi,\A\psi\rangle_{L^2} +  \langle \phi,\psi\rangle_{L^2}\,.
\end{displaymath}
Note that $\langle \phi\,,\, \psi \rangle_{\nu,\alpha}$ coincides with the ordinary
  $L^2$ product whenever $\psi\in{\rm
  Ker}A_{i\nu,\alpha}$.

  Employing the estimates of Lemma \ref{lem:surjective} we now prove:
\begin{proposition}
Suppose that \eqref{eq:35} holds, and that $\lambda=i\nu$ for some  $\nu\in
[U(-1),U(1)]$. Then, there
exists $C>0$, such that for all $\nu$, $\alpha$ 
\begin{equation}
  \label{eq:30a}
\| \mathcal E_{i\nu,\alpha}^{\rm right} \|_{\mathcal L (L^2, H_0^1)}+  \| \mathcal E_{i\nu,\alpha}^{\rm right} \|_{\mathcal L (L^2,D(\A_{i\nu,\alpha}))} \leq C \,.
\end{equation}
\end{proposition}

\begin{proof}
We first observe that  the solution constructed in Lemma \ref{lem:surjective} satisfies

\begin{displaymath}
\| \phi \|_{\alpha,\nu} \leq C \, \|v\|_2
\end{displaymath}
To obtain $\mathcal E_{i\nu,\alpha}^{\rm right} v$ we now need to subtract
$\pi_{\alpha,\nu} \phi$, where $\pi_{\alpha,\nu}$ is the orthogonal projector from
$D(\A_{i\nu,\alpha})$ onto ${\rm Ker} \A_{i\nu,\alpha}$. Obviously,
\begin{displaymath}
   \begin{array}{ll}
\| \mathcal E_{i\nu,\alpha}^{\rm right}  v\|_{\alpha,\nu}^2& = \| \phi  -\pi_{\alpha,\nu} \phi \|_{\alpha,\nu}^2\\
 &  \leq \| \phi \|_{\alpha,\nu}^2 
\\ &\leq  C \|v\|_2^2\,,
\end{array}
\end{displaymath}
establishing, thereby, a bound on the right inverse in $\mathcal L
(L^2,D(\A_{i\nu,\alpha}))$.  To estimate $\mathcal E_{i\nu,\alpha}^{\rm right} $
in $\mathcal L (L^2, H_0^1)$ we observe that by \eqref{eq:hardy0}
there exists $C_0>0$ such that for any $\alpha\in\R$, $\nu\in[U(-1),U(1)]$,
and $\phi\in D(A_{i\nu,\alpha})$

\begin{displaymath}
\|\phi \|_{1,2} \leq C_0 \|\phi\|_{\alpha,\nu}\,.
\end{displaymath}
\end{proof}

 \subsection{Nearly Couette velocity fields}
We now attempt to estimate $\|\A_{\lambda,\alpha}^{-1}\|$ in the case where
$\Re\lambda\neq0$.  We shall begin with the case where $\delta_2(U)$, defined in
\eqref{defdelta3}, is small. To this end we introduce
\begin{subequations}
\label{eq:51}
    \begin{equation}
    I(\phi,\lambda) = \frac{1}{2}\|\phi^\prime\|_2^2 +\Big\langle\frac{U^{\prime\prime}(U-\nu)\phi}{(U-\nu)^2+\mu^2},\phi\Big\rangle  \,,
  \end{equation}
and
\begin{equation}
  \gamma_m(\lambda,U) = \inf_{\phi\in H^1_0(-1,1)\setminus\{0\}} \frac{I(\phi,\lambda)}{\|\phi\|_2^2}\,.
\end{equation}
\end{subequations}
The following result proves that under suitable assumptions on $U$,
the infimum of $\gamma_m(\lambda)$ over $\mathbb
  C\setminus{\mathfrak J}$\,, where ${\mathfrak J}=\{\lambda \in\C \,| \, \Re\lambda=0\,,
\; \nu\in[U(-1),U(1)]\}$ is positive.
\begin{lemma}
\label{lem:positive}  For any $r>1$, there exists $\delta_0>0$ and $\gamma_0
>0$ such that for  any $ U\in \Sg_r^3$ satisfying 
 $\delta_2(U) \leq \delta_0$,  
we have 
\begin{equation}
  \label{eq:52}
\inf_{
    \lambda\in\C\setminus{\mathfrak J} }\gamma_m(\lambda,U)\geq \gamma_0>0 \,.
\end{equation}
\end{lemma}
\begin{proof}
Writing $I(\phi,\lambda)$ in the form
\begin{displaymath}
 I(\phi,\lambda) = \frac{1}{2}\|\phi^\prime\|_2^2 + \Re \Big\langle\phi\,,\,\frac{U^{\prime\prime} \phi}{(U +i \lambda)}\Big\rangle\,.
\end{displaymath}
we attempt to estimate the second term on the right hand side.  For
some $\nu_0 >0$ to be chosen later, we consider two different cases
depending on the size of $|\lambda|$.

{\bf Case I:$
\sqrt{|\mu|^2+|\nu|^2}<\nu_0\,$.}\\
Integration by parts yields, accounting 
for the Dirichlet condition $\phi$ satisfies  at $x=\pm  1$,
  \begin{displaymath}
   \Big\langle\phi,\frac{U^{\prime\prime}\phi}{U+i\lambda}\Big\rangle=  \Big\langle\frac{U^{\prime\prime}}{U^\prime}\, |\phi|^2 ,\frac{U^\prime}{U+i\lambda}\Big\rangle=
  - \Big\langle\Big(\frac{U^{\prime\prime}|\phi|^2}{U^\prime}\Big)^\prime,\log\,(U+i\lambda) \Big\rangle\,.
  \end{displaymath}
  We estimate the right-hand-side as follows
\begin{equation}
\label{eq:53}
  \Big|\Big\langle\Big(\frac{U^{\prime\prime}|\phi|^2}{U^\prime}\Big)^\prime,\log\,(U+i\lambda)\Big\rangle\Big| \leq  
C(\nu_0) \, \delta_2(U) [\|\phi^\prime\|_2\|\phi\|_\infty+\|\phi\|_\infty^2]\,.  
\end{equation}
Sobolev's embedding and   Poincar\'e's  inequality  then yield that for some
$\widehat C(\nu_0,r)>0$, 
\begin{equation}
  \label{eq:54}
\Big|\Re\Big\langle\phi,\frac{U^{\prime\prime}\phi}{U+i\lambda}\Big\rangle\Big|\leq \widehat C (\nu_0,r) \,\delta_2(U)\|\phi^\prime\|_2^2\,.
\end{equation}

{\bf Case II:}$\sqrt{|\mu|^2+|\nu|^2}\geq\nu_0$.\\
 As $U\in\Sg_r$ we conclude that there
exists $C(r)>0$ such that
\begin{equation}\label{eq:47a}
\Big|\Re\Big\langle\phi,\frac{U^{\prime\prime}\phi}{U+i\lambda}\Big\rangle\Big|\leq
\Big|\Big\langle\phi,\frac{U^{\prime\prime}\phi}{U+i\lambda}\Big\rangle\Big|\leq 
  \frac{C}{|\nu_0|}\|\phi\|_2^2 \leq \frac{\hat C(r)}{\nu_0}   \|\phi^\prime\|_2^2 \,.
\end{equation}
Poincar\'e's inequality was applied to obtain the last estimate.  We
can  now set $\nu_0 = 8 \,\hat C$ in \eqref{eq:47a} and $\delta_2(U)\leq \delta_0 =
\frac{1}{8 \widehat C (\nu_0,r)}$ in \eqref{eq:47a}, to obtain
\begin{displaymath}
I(\phi,\lambda) \geq \frac 14 \|\phi^\prime\|_2^2\,.
\end{displaymath}
The lemma is proved by using Poincar\'e's inequality once again.
\end{proof}
\begin{remark}
  Let $\alpha\geq 0$. Suppose that there exist $\lambda_0 \in \mathfrak J$ and
  $\phi\in D(\A_{\lambda_0,\alpha})$ satisfying $\A_{\lambda_0,\alpha}\phi=0$ or, in other
  words, that there exists an embedded eigenvalue $\lambda_0 \in \mathfrak
  J$. Taking the scalar product of $\A_{\lambda_0,\alpha} \phi$ with
  $(U-i\lambda_0)^{-1} \phi \,$, we now observe that
  \begin{equation}\label{eqcancel}
   I(\phi, \lambda_0)\leq 0\,.
 \end{equation}  
 On the other hand, since $I(\phi,\lambda)$ depends continuously on $\lambda$ , we
 obtain a contradiction, for a sequence $\lambda_n$ in $\mathbb C \setminus
 \mathfrak J$ tending to $\lambda_0$, between \eqref{eqcancel} and
 \eqref{eq:52} for sufficiently small $\delta_2(U)$. Hence $\A_{\lambda,\alpha}$
 does not possess any eigenvalue in $\mathfrak J$ \cite[Section
 3]{wei2018linear}.
\end{remark}

Without the assumption that $\delta_2(U)$ is small we may still show
\begin{lemma}\label{rem:gamma-finite}
 For any $r>1$, there exists $\gamma_0\in \mathbb R$ such that for  any $ U\in \Sg_r^3$ 
we have 
\begin{displaymath}
 \inf_{
    \lambda\in\C\setminus{\mathfrak J} }\gamma_m(\lambda,U)\geq \gamma_0 \,.
\end{displaymath}
\end{lemma}
\begin{proof}
 Indeed, we obtain by \eqref{eq:53}, Poincar\'e's
  inequality and Sobolev's embeddings, that for any $\varepsilon>0$ and any $\nu_0 >0$ there exists
  $C_{\varepsilon,\nu_0}$ such that, for any $\lambda \in \mathbb C\setminus \mathfrak J$ and
  $\phi\in H^1_0(-1,1)$, 
  \begin{displaymath}
    I(\phi,\lambda) \geq \frac{1}{2}
    \|\phi^\prime\|_2^2-C\|\phi^\prime\|_2^{3/2}\|\phi\|_2^{1/2}\geq
    \Big(\frac{1}{2}-\varepsilon\Big)
    \|\phi^\prime\|_2^2-C_{\varepsilon,\nu_0} \|\phi\|_2^2 \,.
  \end{displaymath}
For $|\lambda|>\nu_0$, semiboundedness of $I$  follows immediately from \eqref{eq:47a}. 
\end{proof}

We now obtain an estimate for $\|\A_{\lambda,\alpha}^{-1}\|$ which is neither
singular as $|\mu|\to0$ (unlike \eqref{eq:77}), nor does it necessitate any
assumption on the injectivity of $\A_{i\nu,\alpha}\,$. Instead, we simply
assume \eqref{eq:52}. We also observe that
  $\mg\geq1/r$ for any $U\in S^2_r$ where $\mg$ is defined in \eqref{eq:19}.
\begin{proposition}
  \label{prop:small}
  For any $r>1$ and $p>1$, there exists a constant $C$ such that for
  any  $U\in S^2_r$ satisfying \eqref{eq:52} we have: 
\begin{enumerate}
\item  For all $v\in H^1(-1,1)$, 
\begin{equation}
  \label{eq:55}
\sup_{
  \begin{subarray}{c}
    |\Re\lambda|\leq (2r)^{-1}\setminus\{\Re\lambda=0\} \\
    0\leq\alpha 
  \end{subarray}}
\frac{1}{\log\, (|\Re \lambda|^{-1} )}  \|\A_{\lambda,\alpha}^{-1}v\|_{1,2} \leq C \, \|v\|_\infty \,.
\end{equation}
\item For  all
$v\in W^{1,p}(-1,1)$,  
\begin{equation}
\label{eq:56}
\sup_{\begin{subarray}{c}
    |\Re\lambda|\leq (2r)^{-1}\setminus\{\Re\lambda=0\} \\
    0\leq\alpha 
  \end{subarray}}
\|\A_{\lambda,\alpha}^{-1}v\|_{1,2} \leq C \, (\|v^\prime\|_p+\|v\|_\infty).
\end{equation}
\item For all $v\in L^p(-1,+1)$,
\begin{equation}
  \label{eq:57}
\sup_{\begin{subarray}{c} 
    |\Re\lambda|\leq (2r)^{-1}\setminus\{\Re\lambda=0\} \\
    0\leq\alpha 
  \end{subarray}}|\Re \lambda|^{1/p}\|\A_{\lambda,\alpha}^{-1}v\|_{1,2} \leq C\, \|v\|_p\,.
\end{equation}
\end{enumerate}
\end{proposition}
\begin{proof}~
\paragraph{Proof of \eqref{eq:55}.}
  Let $v\in H^1(-1,1)$ and $\phi=\A_{\lambda,\alpha}^{-1}v$ for some $\lambda=\mu+i\nu$
  with $0<|\mu|\leq \mg/2$. We begin by observing that
  \begin{equation}
\label{eq:58}
    \Re\Big\langle\phi,\frac{v}{U-\nu+i\mu}\Big\rangle= \|\phi^\prime\|_2^2 +
  \alpha^2\|\phi\|_2^2 +\Re\Big\langle\phi,\frac{U^{\prime\prime}\phi}{U+i\lambda}\Big\rangle\,,
  \end{equation}
and note that by \eqref{eq:52}, which is assumed to hold  here, 
    \begin{equation}
\label{eq:51a}
    \Re\Big\langle\phi,\frac{v}{U-\nu+i\mu}\Big\rangle=   \frac 12 
     \|\phi^\prime\|_2^2 + \alpha^2  | |\phi\|^2 +  I(\phi,\lambda)   \geq  \frac 12 
     \|\phi^\prime\|_2^2\,.
  \end{equation}
  The left hand side of \eqref{eq:51a}  can be estimated as follows
\begin{equation}\label{eq:51b}
 \Big|\Re\Big\langle\phi,\frac{ v}{U+i\lambda}\Big\rangle\Big| \leq \|
 \phi\|_\infty\|v\|_\infty\Big\|\frac{1}{U+i\lambda}\Big\|_1 \leq C\, {\log\, (\mg |\mu|^{-1})}\|
 \phi\|_\infty\|v\|_\infty\,.
\end{equation}
To obtain \eqref{eq:51b}, we consider two different cases.\\
We  first apply the following computation, 
valid whenever $\nu\in[U(-1),U(1)]$, and $0 <  | \mu |  \leq \mg/2$,
\begin{equation}\label{contra}
\begin{array}{ll}
  \Big\|\frac{1}{U+i\lambda}\Big\|_1&  = \int_{-1}^{+1} \frac{1}{ \sqrt{(U(x)-\nu)^2 + \mu^2}} dx \\
  &  \leq  
  \mg^{-1}\int_{-1}^1\frac{1}{\sqrt{|x-x_\nu|^2+|\mu/\mg|^2}}\, dx \\ & 
  \leq  \mg^{-1} \int_{-\frac{2m}{\mu}}^{\frac{2m}{\mu}} \frac{1}{\sqrt{1+ y^2}}\, dy \\
   &  \leq  C\,  \mg^{-1}  \log\, (\mg |\mu|^{-1})\,.
   \end{array}
\end{equation}
In the case $\nu<U(-1)$ we may write
\begin{displaymath}
\begin{array}{ll}
   \Big\|\frac{1}{U+i\lambda}\Big\|_1 & = \int_{-1}^{+1} \frac{1}{ \sqrt{(U(x)-\nu)^2 + \mu^2}} dx \\
    & \leq \int_{-1}^{+1} \frac{1}{ \sqrt{(U(x)- U(-1))^2 + \mu^2}} dx  \\
    &  \leq \mg^{-1}\int_{-1}^{+1} \frac{1}{\sqrt{(x+1)^2 +(\mu/\mg)^2} }\, dx \,,
    \end{array}
\end{displaymath}
and proceed as before. The case $U(1)>\nu$ is similar.  Hence \eqref{eq:51b} is proved in both cases.\\

Sobolev's embedding and Poincar\'e's inequality then yield
\begin{equation}
\label{eq:59}
   \Big|\Big\langle\phi,\frac{v}{U+i\lambda}\Big\rangle\Big|\leq C\,  \log\, (\mg |\mu|^{-1})
   \|\phi^\prime\|_2 \|v\|_\infty\,.
\end{equation}
Together with \eqref{eq:51a}, it  implies for $0<|\mu|\leq \mg/2\,$,
\begin{equation}
\| \phi\|_{1,2}  \leq C \, \log\, (\mg |\mu|^{-1}) \| v\|_\infty\,,
\end{equation}
and hence also \eqref{eq:55}.\\

\paragraph{Proof of \eqref{eq:56}.}~\\
Let $|\nu|>2\max(|U(1)|,|U(-1)|)$. Then we have
\begin{equation}
  \label{eq:60}
\Big|\Re\Big\langle\phi,\frac{v}{U+i\lambda}\Big\rangle\Big| \leq C\|\phi\|_2\|v\|_2 \,. 
\end{equation}
If $|\nu|\leq2\max(|U(1)|,|U(-1)|)$ we integrate by parts to obtain
\begin{displaymath}
 \Big\langle\phi,\frac{v}{U+i\lambda}\Big\rangle =
 - \Big\langle\log\,(U+i\lambda),\Big(\frac{\bar{\phi}\, v}{U^\prime}\Big)^\prime\Big\rangle\,.
\end{displaymath}
Use of H\"older's inequality and Sobolev embedding lead, for every $p>1$, to the conclusion
that there exists $C_p>0$ such that 
\begin{displaymath}
    \Big|\Big\langle\Big(\frac{\bar{\phi}\,v}{U^\prime}\Big)^\prime,\log\,(U+i\lambda)\Big\rangle\Big|\leq C_p(\|\phi\|_\infty\|v^\prime\|_p + \|\phi^\prime\|_p\|v\|_\infty)\,. 
\end{displaymath}
where we have used the fact that the $L^q$-norm ($q$ being the H\"older
conjugate) of $\log\,(U+i\lambda)$ can be uniformly bounded for
\begin{displaymath}
(|\Re\lambda|,|\Im\lambda|)\in[0,\mg/2]\times[0,2\max(|U(1)|,|U(-1)|)]\,.
\end{displaymath}

Combining the above with \eqref{eq:60}, and making use of Sobolev
embedding once again yield
\begin{equation}
\label{eq:61}
  \Big|\Big\langle\phi,\frac{v}{U+i\lambda}\Big\rangle\Big|\leq
  C_p(\|v^\prime\|_p\|\phi\|_\infty+\|v\|_\infty\|\phi^\prime\|_p)\,.
\end{equation}
By \eqref{eq:51a} we 
obtain that 
\begin{equation}
\label{eq:62}
  \frac{1}{2}\|\phi^\prime\|_2^2 \leq 
  C \|\phi\|_{1,p}\min(\|v^\prime\|_p,\|v\|_\infty)\,,
\end{equation}
from which \eqref{eq:56} easily follows for $1<p\leq2$, and then for 
$p>2$ by the inequality $\|v^\prime\|_2\leq2^{1-2/p}\|v^\prime\|_p$. \\
\paragraph{Proof of \eqref{eq:57}.}
We first observe that 
\begin{equation}
\label{eq:63}
  \Big|\Big\langle\phi,\frac{v}{U+i\lambda}\Big\rangle\Big|\leq 
  \Big\|\frac{1}{U+i\lambda}\Big\|_q\, \|\phi\|_\infty \,\|v\|_p\,,
\end{equation}
where $q=p/(p-1)$. \\
If $\nu\in[U(-1),U(1)]$, we estimate
the right-hand-side as follows
\begin{equation}
\label{eq:64}
\begin{array}{ll}
  \Big\|\frac{1}{U+i\lambda}\Big\|_q & \leq
  \mg^{-1}\Big[\int_{-1}^1\frac{dx}{([(x-x_\nu)^2+|\mu/\mg|^2]^{q/2}}\Big]^{1/q} \\ & \leq  
\mg^{-1}\Big[\int_{-\infty}^\infty\frac{dx}{[x^2+|\mu/\mg|^2]^{q/2}}\Big]^{1/q} 
 \\ &  = \mg^{-1} (\frac{|\mu|}{m})^{ \frac 1q  -1 } \Big[\int_{-\infty}^\infty\frac{dx}{[x^2+1]^{q/2}}\Big]^{1/q}\\
&  = C_q \mg^{-\frac 1p } |\mu|^{-1/p}\,.
\end{array}
\end{equation}
If $\nu<U(-1)$,  we write
\begin{displaymath} 
   \Big\|\frac{1}{U+i\lambda}\Big\|_q \leq
  \mg^{-1}\Big[\int_{-1}^1\frac{dx}{[ (x+1)^2+|\mu/\mg|^2]^{q/2}}\Big]^{1/q} 
\end{displaymath}
and proceed as before. The case $\nu>U(-1)$ is similar. Thus, without
any restriction on $\nu$, we get
\begin{equation}\label{eq:65}
     \Big\|\frac{1}{U+i\lambda}\Big\|_q \leq \hat C_q \, |\mu|^{-1/p}\,.
\end{equation}
Substituting the above into \eqref{eq:63} then yields
\begin{equation}
\label{eq:66}
    \Big|\Big\langle\phi,\frac{v}{U+i\lambda}\Big\rangle\Big|\leq  C|\mu|^{-1/p}\|\phi\|_\infty\, \|v\|_p\,.
\end{equation}
Combining the above with \eqref{eq:51a}, we obtain
\begin{equation}\label{eq:67}
 \frac{1}{2}\|\phi^\prime\|_2^2 \leq 
  C|\mu|^{-1/p}\|\phi^\prime\|_2\, \|v\|_p\,,
\end{equation}
from which \eqref{eq:57} easily follows 
\end{proof}
We shall need in Section \ref{sec:no-slip} the following immediate consequence of Proposition~\ref{prop:small}.
 \begin{corollary}
 For any $r>1$ there exists
  $C>0$ such that for all $v\in L^2(-1,1)$,  $U\in \Sg_r$ satisfying
  \eqref{eq:52},  $\lambda\in\C\setminus{\mathfrak J}$ and  $\alpha \geq 0$  it holds that
  \begin{equation}
    \label{eq:68}\
\|\A_{\lambda,\alpha}^{-1}v\|_{1,2} \leq
C\Big\|(1\pm x)^{1/2}\frac{v}{U+i\lambda}\Big\|_1 \,.
  \end{equation}
\end{corollary}
\begin{proof}
  Let $v\in L^2(-1,1)$ and $\phi=\A^{-1}_{\lambda,\alpha}v$. By \eqref{eq:51a} we have that
  \begin{equation}
\label{eq:69}
    \frac{1}{2}  \|\phi^\prime\|_2^2 \leq \Re\Big\langle\phi,\frac{v}{U+i\lambda}\Big\rangle\,.
  \end{equation}
As $\phi\in H^1_0(-1,1)$ we may write
\begin{displaymath}
  \Re\Big\langle\phi,\frac{v}{U+i\lambda}\Big\rangle=
  \Re\Big\langle\phi-\phi(\pm1),\frac{v}{U+i\lambda}\Big\rangle \leq
  \|\phi^\prime\|_2\,\Big\|(1\pm x)^{1/2}\frac{v}{U+i\lambda}\Big\|_1  \,.
\end{displaymath}
Combining the above with \eqref{eq:69} yields \eqref{eq:68} with the
aid of Poincar\'e's  inequality. 
\end{proof}

In the next lemma we address the optimality of \eqref{eq:57}.
\begin{lemma}
Let $U\in C^3([-1,1])$ satisfy \eqref{eq:35}. Then, \eqref{eq:57} is optimal, i.e., there
exists a sequence  $\{\lambda_k,v_k\}_{k=1}^\infty\in \C^\N\times [L^p(-1,1)]^{\N}$ and
$\alpha\in\R_+$  such that $\|v_k\|_p=1$, $\mu_k=\Re\lambda_k\to0$, and
 \begin{equation}
 \label{eq:70}
\liminf_{k\to\infty}|\mu_k|^{1/p}\|\A_{\lambda_k,\alpha}^{-1}v_k\|_{1,2} >0 \,.
\end{equation}
\end{lemma}
\begin{proof} 
We prove \eqref{eq:70} for $\alpha=0$. 
  Consider then $\A_{\lambda,0}$ for some $\lambda=\mu+i\nu$ with $\mu\neq0$
  and $\nu\in(U(-1),U(1))$.
  Let $(\phi , v )$ satisfy 
  \begin{displaymath}
\A_{\lambda,0}\, \phi=v\,,\,  \mbox{ with } v\in L^p(-1,1) \mbox{ and }
\|v\|_p=1\,.
  \end{displaymath}
   We may rewrite  this equation in the form
  \begin{displaymath}
    \Big[(U+i\lambda)^2\Big(\frac{\phi}{U+i\lambda}\Big)^\prime\Big]^\prime=v \,.
  \end{displaymath}
Integrating once yields
\begin{displaymath}
(U(t)+i\lambda)^2 \left( \frac{\phi}{U+i\lambda}\right)^\prime(t) = A_1 + \int_{-1}^t v(\tau) d\tau\,.
\end{displaymath}
Integrating again leads to
\begin{displaymath}
  \frac{\phi(x)}{U(x)+i\lambda}= A_1 \int_{-1}^x \frac{dt}{(U(t)+ i\lambda)^2} +
  \int_{-1}^x \frac{1}{(U(t)+ i\lambda)^2} \Big(\int_{-1}^t v(\tau) d\tau\Big)
  dt\,.
\end{displaymath}
The Dirichlet boundary condition at $x=1$ is then satisfied through
the requirement that $A_1$ satisfies
\begin{displaymath}
   A_1 \int_{-1}^1 \frac{dt}{(U(t)+ i\lambda)^2} + \int_{-1}^1\frac{1}{(U(t)+ i\lambda)^2} \Big(\int_{-1}^t v(\tau) d\tau\Big) dt =0\,.
\end{displaymath}
  Making use of Fubini's Theorem, we finally obtain 
  \begin{subequations}
\label{eq:71}
  \begin{equation}
  \phi(x)= (U(x)+i\lambda)\Big[\int_{-1}^xv(t)\int_t^x\frac{d\tau}{(U(\tau)+i\lambda)^2}dt- A_1\int_{-1}^x\frac{d\tau}{(U(\tau)+i\lambda)^2}\Big]\,,
\end{equation}
where
\begin{equation}
  A_1= \frac{\int_{-1}^1v(t)\int_t^1\frac{d\tau}{(U(\tau)+i\lambda)^2}dt}{\int_{-1}^1\frac{d\tau}{(U(\tau)+i\lambda)^2}}\,.
\end{equation}
\end{subequations}
We now write
\begin{equation}
\label{eq:72}
  \int_t^x\frac{d\tau}{(U(\tau)+i\lambda)^2}= - \frac{1}{U^\prime(U+i\lambda)}\Big|_t^x -
  \int_t^x\frac{U^{\prime\prime}(\tau)\,d\tau}{|U^\prime(\tau)|^2(U(\tau)+i\lambda)} \,.
\end{equation}
It can be verified (as in the proof of \eqref{contra})  that for some
positive, independent of $\mu$, constants $C_1$ and $C$, 
\begin{displaymath} 
  \Big|\int_t^x\frac{U^{\prime\prime}(\tau)d\tau}{|U^\prime(\tau) |^2(U(\tau)+i\lambda)}\Big|\leq C_1 \int_1^1| U(\tau)+i\lambda)|^{-1}\,d\tau \leq C\big(1+\big|\log\,|\mu|\,\big|\big)\,.
\end{displaymath}
Consequently, as $\mu\to0$ we have
\begin{equation}
\label{eq:73}
  \Big|\int_{-1}^xv(t)\int_t^x\frac{d\tau}{(U(\tau)+i\lambda)^2}dt-
  \int_{-1}^xv(t)\frac{1}{U^\prime(U+i\lambda)}\Big|_t^x \,dt\Big| \leq C\, |\log\,
  |\mu|\,| \;\|v\|_1 \,. 
\end{equation}
We seek a sequence $\{ \lambda_k,v_k\}$ with $v_k\in L^p(-1,1)$ and $\mu_k>0$ a
decreasing sequence tending to $0$ for which \eqref{eq:70} holds. Let
then 
\begin{displaymath}
 \lambda_k=\mu_k+i\nu \,.
\end{displaymath}
 For sufficiently large $k$ we have that
$x_\nu+2\mu_k < 1$. We then define $v_k$ by
\begin{displaymath}
  v_k(x)=
  \begin{cases}
    \mu_k^{-1/p} & x\in[x_\nu,x_\nu+\mu_k] \\
0 & \text{otherwise} \,.
  \end{cases}
\end{displaymath}
Clearly $\|v_k\|_p=1$, and $\|v_k\|_1= \mu_k ^{1-1/p}$.\\
 By \eqref{eq:73} we
then have as $\mu_k$ tends to $0\,$,
\begin{equation}
  \label{eq:74}
  \int_{-1}^xv_k(t)\int_t^x\frac{d\tau}{(U(\tau)+i\lambda_k)^2}dt=
  \int_{-1}^x v_k(t)\, \left(\frac{1}{U^\prime(U+i\lambda_k)}\Big|_t^x\right) \,dt +\OO(|\mu_k|^{1-1/p}|\log\,
  \mu_k|) \,. 
\end{equation}
We now write, for $(x_\nu+1)/2<x\leq1$,
\begin{displaymath}
\begin{array}{l}
  \int_{-1}^xv_k(t)\,\left(\frac{1}{U^\prime(U+i\lambda_k)}\Big|_t^x\right)  \,dt = \\
\qquad \qquad \quad   = \frac{1}{U^\prime(x)(U(x)+i\lambda_k)}\mu_k^{1-1/p} - \mu_k^{-1/p}
  \int_{x_\nu}^{x_\nu+\mu_k}\frac{1}{U^\prime(t)(U(t) +i\lambda_k)}\, dt \,.
  \end{array}
\end{displaymath}
As
\begin{displaymath}
  \Big|\int_{x_\nu}^{x_\nu+\mu_k}\frac{d\tau}{U^\prime(U+i\lambda_k)} -
  \int_{x_\nu}^{x_\nu+\mu_k}\frac{d\tau}{J_\nu[J_\nu(\tau-x_\nu)+i\mu_k]}\Big|\leq
 C\, \mu_k\, \big(1+\log\, \mu_k^{-1} \big)\,,
\end{displaymath}
where $J_\nu=U^\prime(x_\nu)$,  we obtain, for  $(x_\nu+1)/2<x\leq1$,
\begin{displaymath}
   \int_{-1}^xv_k(t)\,\left(\frac{1}{U^\prime(U+i\lambda_k)}\Big|_t^x \right)\,dt =
   -\frac{\mu_k ^{-1/p}}{J_\nu^2}\log\,
   [J_\nu(\tau-x_\nu)+i\mu_k]\Big|_{x_\nu}^{x_\nu+\mu_k}+ \OO(\mu_k^{1-1/p}) \,.
\end{displaymath}
Consequently,
\begin{equation}
  \label{eq:75}
\int_{-1}^xv_k(t)\,\left(\frac{1}{U^\prime(U+i\lambda_k)}\Big|_t^x\right) \,dt =
   -\frac{\mu^{-1/p}_k}{J_\nu^2} \log\, (1-iJ_\nu) + \OO(\mu_k^{1-1/p}|\log\,  \mu_k|) \,.
\end{equation}
Substituting the above into (\ref{eq:71}a) and (\ref{eq:71}b) yields,
for all $(x_\nu+1)/2<x\leq1$,
\begin{displaymath}
 \phi_k (x) = -(U(x)+i\lambda_k)\frac{\mu_k^{-1/p}}{J_\nu^2} \log\, (1-iJ_\nu) 
  \bigg[ \frac{\int_x^1(U(\tau)+i\lambda_k)^{-2}d\tau}{\int_{-1}^1(U(\tau)+i\lambda_k)^{-2} d\tau}\bigg]
   + \OO(\mu^{1-1/p}_k|\log\,  \mu_k|) \,,
\end{displaymath}
where $  \phi_k := \A_{\lambda_k,0}^{-1}v_k\,.$

Clearly, for all $(x_\nu+1)/2<x\leq1$,
\begin{displaymath}
  \lim_{k\to\infty}\int_x^1(U(\tau)+i\lambda_k)^{-2}d\tau= \int_x^1(U(\tau)-\nu)^{-2}d\tau>0 \,,
\end{displaymath}
the convergence being uniform in $x$.\\
 We now prove \eqref{eq:70} by establishing that 
\begin{displaymath}
  \liminf_{k\to\infty}|\mu_k|^{1/p}\|\phi_k\|_{L^2((x_\nu+1)/2,1]))}>0\,,
\end{displaymath} 
which will immediately imply $\liminf_{k\to\infty}|\mu_k|^{1/p}\|\phi_k\|_{L^2(-1,1)}>0$   and consequently \eqref{eq:70}.\\
To this end we need to prove that
\begin{equation}
\label{eq:76}
  \Big|\int_{-1}^1(U(\tau)+i\lambda_k)^{-2} d\tau\Big| \leq C\,.
\end{equation}
We now use \eqref{eq:72} together with the fact that $U\in C^3([-1,1])$ to
obtain that
\begin{displaymath}
  \int_{-1}^1\frac{d\tau}{(U(\tau)+i\lambda_k)^2}= \Big[- \frac{1}{U^\prime(U+i\lambda_k)}-\frac{U^{\prime\prime}\log\,(U+i\lambda_k)}{|U^\prime|^3} \Big]_{-1}^1 +
  \int_{-1}^1\Big(\frac{U^{\prime\prime}}{|U^\prime|^3}\Big)^\prime \log\,(U+i\lambda_k)\,d\tau\,.
\end{displaymath}
Since $U^\prime >0$ and $|U(\pm 1)+i\lambda_k)| \geq C(1\pm  x_\nu)$, for some $C>0$,  we can conclude \eqref{eq:76} from
the above.
\end{proof}

\subsection{Non-vanishing $U^{\prime\prime}$}

We dedicate this subsection to the case when $U\in C^2([-1,+1]$ and
$U^{\prime\prime}\neq0$ which may result from a combination of non-vanishing
pressure gradient and relative velocity between the plates at
$x_2=\pm1$. Note that in this case there are no eigenvalues of
  $\A_{\lambda,\alpha}$ embedded in $\mathfrak J$ (see \cite[Section
  3]{wei2018linear}). We begin by establishing the following, rather
  straightforward, result.
\begin{proposition}
\label{lem:inviscid-boundedness-1} Suppose that $U^{\prime\prime}\neq 0$ on $[-1,+1]$, 
then, for any $\lambda\in\C$ for which
$\Re \lambda \neq 0$ and $\alpha\geq0$, $\A_{\lambda,\alpha}$ is invertible. Moreover, for any  $r>1$ there
exists $C>0$, such that, for any $\lambda$ with $\Re \lambda \neq 0$, $\alpha\geq0$, and
$U\in\Sg^2_r$ satisfying \eqref{condsurinf},  it holds that 
  \begin{equation}
    \label{eq:77} 
\|\A_{\lambda,\alpha}^{-1}\|+ \Big\| \frac{d}{dx} \A_{\lambda,\alpha}^{-1}\Big\|\leq C\frac{1+|\Re \lambda|^{1/2}}{|\Re \lambda|} \,.
  \end{equation}
\end{proposition}
\begin{proof}
For a pair $(\phi,v)$ such
  that $v = \A_{\lambda,\alpha} \phi$, with $\lambda=\mu+i\nu$, 
 we write, 
\begin{equation}
\label{eq:78}
  \Im\Big\langle\phi,\frac{v}{U-\nu+i\mu}\Big\rangle=-\mu\, \Big\langle\frac{U^{\prime\prime}\phi}{(U-\nu)^2+\mu^2},\phi\Big\rangle\,.
\end{equation}
Consequently, since $U^{\prime\prime}\neq0$ (hence has constant sign)  we obtain 
\begin{equation}
\label{eq:79}
  \Big\|  \frac{\phi}{U-\nu+i\mu} \Big\|_2 \leq\frac{C}{|\mu|}\|v\|_2 \,.
\end{equation}

Let $\nu_0>0$ be chosen at a later stage and consider the following two cases. \\
In the case $ \sqrt{|\mu|^2+|\nu|^2}<\nu_0\,$, we immediately deduce from
\eqref{eq:79} that there exists $C(\nu_0)$ such that
\begin{equation}
\label{eq:80} 
  \|\phi\|_2 \leq\frac{C(\nu_0)}{|\mu|}\|v\|_2 \,.
\end{equation}
In the case $\sqrt{|\mu|^2+|\nu|^2}\geq \nu_0\,$, as
\begin{equation}
\label{eq:81}
  \Re\Big\langle\phi,\frac{v}{U-\nu+i\mu}\Big\rangle= \|\phi^\prime\|_2^2 +
  \alpha^2\|\phi\|_2^2 +\Big\langle\frac{U^{\prime\prime}(U-\nu)\phi}{(U-\nu)^2+\mu^2},\phi\Big\rangle\,,
\end{equation}
we can use \eqref{eq:79} once again to obtain that,  for  $\nu_0 \geq 2 \sup|U|$, 
there exists $C>0$ such that
\begin{displaymath}
  \|\phi^\prime\|_2^2 +   \Big(\alpha^2-\frac{C}{\nu_0}\Big)\|\phi\|_2^2 \leq
  \frac{C}{|\mu|}\|v\|_2^2 \,.
\end{displaymath}
Since $\phi\in H^1_0(-1,1)$ we can use Poincar\'e's inequality to
 obtain for sufficiently large $\nu_0$ that
there exists $C>0$ such that, for any $\alpha \geq 0$ (and $|\lambda| \geq \nu_0$),  
\begin{equation}
\label{eq:82}
  \|\phi\|_{1,2}\leq \frac{C}{|\mu|^{1/2}} \|v\|_2 \,,
\end{equation}
which, combined with \eqref{eq:80} yields \eqref{eq:77}.\\

Once injectivity of $\A_{\lambda,\alpha}$ is established, we may apply Fredholm theory
to prove its surjectivity. By the compactness of the multiplication
with $U^{\prime\prime}$ from  $D(\A_{\lambda,\alpha})$ into $L^2(-1,+1)$, we can conclude, as in Remark
\ref{rem4.5}, that the index of $ \A_{\lambda,\alpha}$ is the same as the index of
$(U+i\lambda)\Big(-\frac{d^2}{dx^2}+ \alpha^2\Big)\,.$ Since for $\mu\neq0$, 
$U+i\lambda\neq 0$ on $[-1,+1]$, it follows that the indices of $\A_{\lambda,\alpha}$ and
$-\frac{d^2}{dx^2} + \alpha^2$  from $H^2(-1,1)\cap H_0^1 (-1,1)$ onto
$L^2(-1,1)$ are the same. Consequently, the index of $\A_{\lambda,\alpha}$ is $0$ and
surjectivity follows from injectivity.
\end{proof}

It should be noted that \eqref{eq:43} is unsatisfactory. Clearly, it
is significantly inferior to  (\ref{eq:55})-(\ref{eq:57}),
where a bound of $\OO(|\Re\lambda|^{-1/2})$
for $\|\A_{\lambda,\alpha}^{-1}\|$ is obtained.  We seek, therefore, a better estimate
for $\|\A_{\lambda,\alpha}^{-1}\|$ that will be applicable in Sections \ref{s7new}
and \ref{sec:no-slip}. 

\begin{proposition}
\label{prop:constant-sign}
Let  $r>1$ and $p>1$. There
exist $\mu_0>0$ and $C>0$ such that for all $v\in W^{1,p}(-1,1)$ and
$U\in\Sg^3_r$ satisfying \eqref{condsurinf} we have
\begin{subequations}
\label{eq:83}
  \begin{equation}
\sup_{
  \begin{subarray}{c}
    |\Re\lambda|\leq\mu_0\setminus\{\Re\lambda=0\} \\
    0\leq\alpha 
  \end{subarray}}
\frac{1}{\log\, (|\Re \lambda|^{-1} )}  \|\A_{\lambda,\alpha}^{-1}v\|_{1,2} \leq C \, \|v\|_\infty \,,
\end{equation}
\begin{equation}
\sup_{\begin{subarray}{c}
    |\Re\lambda|\leq \mu_0\setminus\{\Re\lambda=0\} \\
    0\leq\alpha 
  \end{subarray}}
\|\A_{\lambda,\alpha}^{-1}v\|_{1,2} \leq C \, (\|v^\prime\|_p+\|v\|_\infty)\,,
\end{equation}
and
\begin{equation}
\sup_{\begin{subarray}{c} 
    |\Re\lambda|\leq \mu_0\setminus\{\Re\lambda=0\} \\
    0\leq\alpha 
  \end{subarray}}|\Re \lambda|^{1/p}\|\A_{\lambda,\alpha}^{-1}v\|_{1,2} \leq C\, \|v\|_p\,.
\end{equation}
\end{subequations}
\end{proposition}

\begin{proof}
  In the case where $\nu\in[U(-1),U(1)]$ we (uniquely) select
  $x_\nu\in[-1,1]$ where $U(x_\nu)=\nu$. Otherwise
  if $\nu>U(1)$ ($\nu<U(-1)$) we set $x_\nu=1$ ($x_\nu=-1$). \\

{\em Step 1:} {\it For  $p>1$ and $\Re \lambda \neq 0$ define $N_{m,p}^\pm$ by
\begin{displaymath}
  v \mapsto N_{m,p}^\pm(v,\lambda) := \min \Big(\Big\|(1\pm\cdot)^{1/2}\frac{v}{U+i\lambda}\Big\|_1,\|v\|_{1,p}\Big)\,.
\end{displaymath}
We prove that there exists $C>0$ such that,
for all $\varepsilon>0$ and $0<|\mu|\leq 1$ it holds that 
 \begin{equation}
\label{eq:84} 
   |\phi(x_\nu)| \leq C\big(\varepsilon^{-1/2}N_{m,p}^\pm (v,\lambda) +
   (|\mu|^{1/2}+\varepsilon^{1/2})\|\phi^\prime\|_2\big)\,,
 \end{equation}
 for all pairs $(\phi,v)\in D(\A_{\lambda,\alpha})\times W^{1,r}(-1,1)$ satisfying
 $\mathcal A_{\lambda,\alpha} \phi =v$.\\ \vspace{2ex}}

As 
\begin{displaymath}
  |\phi (x)|^2\geq \frac{1}{2}|\phi(x_\nu)|^2 - |\phi(x)-\phi(x_\nu)|^2 \,,
\end{displaymath}
we may use \eqref{eq:78} to obtain
\begin{equation} \label{eq:69a}
 |  \Im\Big\langle\phi,\frac{v}{U-\nu+i\mu}\Big\rangle| \geq |\mu| \,
   \Big\langle\frac{|U^{\prime\prime}|}{(U-\nu)^2+\mu^2},\frac 12 |\phi(x_\nu)|^2 - |\phi(x)-\phi(x_\nu)|^2 \Big\rangle\,.
\end{equation}
We note that, for any $1<p<2$, there exists $C>0$ such that
\begin{displaymath}
\begin{array}{ll}
    \Big|\Big\langle\phi,\frac{v}{U+i\lambda}\Big\rangle\Big|&  =
    \Big|\Big\langle\Big(\frac{\phi\bar{v}}{U^\prime}\Big)^\prime,\log\,(U+i\lambda)\Big\rangle\Big| \\
   & \leq C\, \left(\|\phi^\prime\|_2\|v\|_\infty+\|\phi\|_\infty\|v^\prime\|_p \right)\\ & \leq C\, \|\phi^\prime\|_2\|v\|_{1,p}\,.
   \end{array}
  \end{displaymath}
Note that
\begin{displaymath}
  |\log(U+i\lambda)|\leq |\log \mg |+|\log|x-x_\nu||+|\log \|U+i\lambda\|_\infty| \,,
\end{displaymath}
and consequently,   the constant $C$, which depends on 
of $\|\log (U+i\lambda)\|_{p/(p-1)}$ is uniformly bounded  for all $U\in
\Sg^3_r$ satisfying \eqref{condsurinf} whenever $|\lambda|<2\|U\|_\infty$. In
the case where $|\lambda|\geq2\|U\|_\infty$ we may use \eqref{eq:60} and Sobolev
embeddings.\\ 
On the other hand, 
\begin{displaymath}
|\phi(x)|=|\phi (x)-\phi(\pm 1)| \leq \| \phi^\prime\|_2 (1\pm x)^\frac
    12\,, 
\end{displaymath} 
we may conclude that
\begin{displaymath}
      \Big|\Big\langle\phi,\frac{v}{U+i\lambda}\Big\rangle\Big| \leq 
      \|\phi^\prime\|_2 \,
      \Big\|(1\pm\cdot)^{1/2}\, \frac{v}{U+i\lambda}\Big\|_1 
\end{displaymath}
and hence, there exists $C>0$, such  that
\begin{equation}
  \label{eq:85}  
\Big|\Big\langle\phi,\frac{v}{U+i\lambda}\Big\rangle\Big| \leq  C\, \|\phi^\prime\|_2 \, N_{m,p}^\pm(v,\lambda) \,.
\end{equation}
Substituting the above into \eqref{eq:69a} yields
\begin{equation*}  
  \frac{|\mu| }{2} |\phi(x_\nu)|^2 \Big\|\frac{|U^{\prime\prime}|^{1/2}}{U+i\lambda}\Big\|_2^2 \leq
  |\mu| \sup |U^{\prime\prime}| \Big\|\frac{\phi-\phi(x_\nu)}{U+i\lambda}\Big\|_2^2 +
 C\|\phi^\prime\|_2 N_{m,p}^{\pm} (v,\lambda)\,.
\end{equation*}
We now observe, as in \eqref{eq:64} (but with a lower bound in mind),
  that, for some positive $C$ and  $\hat C$ (note that $|\mu|\leq 1$), it holds
\begin{displaymath}
  \Big\|\frac{|U^{\prime\prime}|^{1/2}}{U+i\lambda}\Big\|_2^2 \geq \frac 1C
  \Big\|\frac{1}{(x-x_\nu)^2+\mu^2}\Big\|_1 \geq  \frac {1}{\hat C|\mu|} \,.
\end{displaymath}
 Hence, for another constant $C>0$, we get
\begin{equation}\label{eq:86}
   |\phi(x_\nu)|^2 \leq C\Big[ |\mu| \Big\|\frac{\phi-\phi(x_\nu)}{U+i\lambda}\Big\|_2^2 +\|\phi^\prime\|_2N_{m,p}^\pm(v,\lambda)\Big]\,.
\end{equation}
To estimate the first  term on the right-hand-side of
\eqref{eq:86} we first observe that for some $C>0$ we have
\begin{displaymath}
  \Big|\frac{1}{U(x)+i\lambda}\Big| \leq \frac{C}{|x-x_\nu|}\,,\, \forall x\in (-1,+1) \,.
\end{displaymath}
Then we notice  that for any $w\in L^2(\R_+)$ such that $xw'\in
L^2(\mathbb R_+)$ we have by
\eqref{eq:hardy2} and some integration by parts
\begin{displaymath}
  \|(xw)^\prime\|_2^2 = \|x w^\prime\|_2^2 \geq \frac{1}{4}\|w\|_2^2 \,.
\end{displaymath}
Recalling that $\phi(-1)=\phi(1)=0$, we thus apply the above inequality to 
\begin{displaymath}
  w(x)=
  \begin{cases}
    \frac{\phi(x) -\phi(x_\nu)}{x-x_\nu} & -1<x<1 \\
    \frac{-\phi(x_\nu)}{x-x_\nu} & 1\leq |x|
  \end{cases}
\end{displaymath} 
in $(x_\nu,+ \infty)$ and $(-\infty,x_\nu)$ to obtain 
\begin{displaymath}
\Big\|\frac{\phi-\phi(x_\nu)}{U+i\lambda}\Big\|_2^2 \leq C  \| \phi^\prime\|_2^2\,,
\end{displaymath}
which when substituted into \eqref{eq:86} readily yields \eqref{eq:84}
via Cauchy's inequality.  Note that, for $\nu\not\in[U(-1),U(1)]$,
\eqref{eq:84} is trivial as $\phi(x_\nu)=0\,$.
\vspace{1ex}

{\em Step 2:} {\it Let $d=\min(1-x_\nu,1+x_\nu)$. We prove that for any $A
  >0$, and $d_1 >0$, there exists $C_{A,d_1}$ and $\mu_{A,d_1}$ such that, for $\alpha^2 \leq A$,
  $x_\nu\in(-1,1)$, $|\mu|\leq \mu_{A,d_1}$, and   $d\geq d_1$\,,
\begin{equation}
  \label{eq:87}
\|\phi\|_{1,2}\leq  C_{A,d_1}\, N_{m,p}^\pm(v,\lambda)\,.
\end{equation}
holds for any pair $(\phi,v)$ satisfying $\mathcal A_{\lambda,\alpha}\phi=v$.}
\vspace{3ex}

\noindent Let
$\chi\in C_0^\infty(\R,[0,1])$ satisfy 
\begin{displaymath}
  \chi(x)=
  \begin{cases}
    1 & |x|<1/2 \\
    0 & |x|>1 \,.
  \end{cases}
\end{displaymath}
Let $\chi_d(x) = \chi((x-x_\nu)/d)$ and set
\begin{equation}
\label{eq:88}
  \phi=\varphi + \phi(x_\nu)\chi_d \,.
\end{equation}
Note that by the choice of $d$, $\varphi$ satisfies also the boundary condition at $\pm 1\,$.\\
It can be easily verified that
\begin{displaymath}
  \A_{\lambda,\alpha}\varphi =v + \phi(x_\nu)\big((U+i\lambda)(\chi_d^{\prime\prime}-\alpha^2\chi_d) -U^{\prime\prime}\chi_d\big)
  \,.
\end{displaymath}
By construction we have that $w=(U-\nu)^{-1}\varphi\in H^2(-1,1)$, and hence we
can rewrite the above equality (using \eqref{eq:36} twice)  in the form 
\begin{multline*}
  -\Big((U-\nu)^2\Big(\frac{\varphi}{U-\nu}\Big)^\prime\Big)^\prime+\alpha^2(U-\nu)\varphi 
\\ \quad = v  +  \phi(x_\nu)\big((U-\nu)(\chi_d^{\prime\prime}-\alpha^2\chi_d) -U^{\prime\prime}\chi_d\big) +
  i\mu(\phi^{\prime\prime}-\alpha^2\phi) \\ \quad = \frac{(U-\nu)v}{U+i\lambda}+
  \phi(x_\nu)\big((U-\nu)(\chi_d^{\prime\prime}-\alpha^2\chi_d) -U^{\prime\prime}\chi_d\big)
  +i\mu\frac{U^{\prime\prime}\phi}{U+i\lambda} \,.
\end{multline*}
Taking the scalar product with $w$ and  integrating by parts then yield
\begin{multline} 
\label{eq:89}
  \|(U-\nu)w^\prime\|_2^2 + \alpha^2\|\varphi\|_2^2 =\Big\langle \varphi, \frac{v}{U+i\lambda}\Big\rangle -\langle w,\phi(x_\nu)U^{\prime\prime}\chi_d\rangle\\
  +  \phi(x_\nu)\langle\varphi,\chi_d^{\prime\prime}-\alpha^2\chi_d\rangle  +i\mu\Big\langle w,\frac{U^{\prime\prime}\phi}{U+i\lambda}\Big\rangle\,.
\end{multline}

As in the proof of \eqref{eq:85}, the first term on the right-hand side is estimated as follows
\begin{equation}
\label{eq:90}
\Big|\Big\langle \varphi, \frac{v}{U+i\lambda}\Big\rangle\Big|\leq \|\varphi^\prime\|_2\, N_{m,p}^\pm(v,\lambda) \leq C(\|\phi^\prime\|_2+d^{-1/2}|\phi(x_\nu)|)N_{m,p}^\pm(v,\lambda) \,.
\end{equation}
To estimate the second term on the right-hand-side, we note that by Hardy's inequality, we have 
\begin{equation} 
\label{eq:91}
  \|w\|_2\leq C \|\varphi^\prime\|_2\leq \hat C \Big(\|\phi^\prime\|_2+
  \frac{1}{d^{1/2}}|\phi(x_\nu)|\Big)\,.
\end{equation}
From \eqref{eq:91} we get
\begin{equation}\label{eq:126a}
| \langle w,\phi(x_\nu)U^{\prime\prime}\chi_d\rangle| \leq \check C\,d^\frac 12 \,  |\phi(x_\nu)|  \Big(\|\phi^\prime\|_2+
  \frac{1}{d^{1/2}}|\phi(x_\nu)|\Big)\,.
\end{equation}

Then, we write for the third term on the right-hand-side of
\eqref{eq:89}, using integration by parts and the upper bound on $\alpha^2$
\begin{displaymath}
  |\langle\varphi,\chi_d^{\prime\prime}-\alpha^2\chi_d\rangle|\leq\|\varphi^\prime\|_2\|\chi_d^\prime\|_2 +C \|\varphi\|_2 \leq \hat C \left(
  \frac{1}{d^{1/2}}\|\varphi^\prime\|_2 +\|\varphi\|_2 \right)\,.
\end{displaymath}
Consequently, by \eqref{eq:88},
\begin{equation*}
\begin{array}{l}
  |\phi(x_\nu)| \,|\langle\varphi,\chi_d^{\prime\prime}-\alpha^2\chi_d\rangle|\leq \\ \qquad \leq C |\phi(x_\nu)| \left(   \frac{1}{d^{1/2}}\Big(\|\phi^\prime\|_2+
  \frac{1}{d^{1/2}}|\phi(x_\nu)|\Big)+ (\|\phi\|_2+d^{1/2}  |\phi(x_\nu)|)\right)\,.
  \end{array}
\end{equation*}
Hence, using Poincar\'e's inequality, 
\begin{equation*} 
  |\phi(x_\nu)| \,|\langle\varphi,\chi_d^{\prime\prime}-\alpha^2\chi_d\rangle|\leq  C  \frac{1}{d^{1/2}} |\phi(x_\nu)|  \Big(\|\phi^\prime\|_2+
  \frac{1}{d^{1/2}}|\phi(x_\nu)|\Big) \,,
\end{equation*}
from which we conclude the existence of $C$ such that for any $\varepsilon_1\in (0,1)$, we have
\begin{equation}
\label{eq:92} 
  |\phi(x_\nu)| \,|\langle\varphi,\chi_d^{\prime\prime}-\alpha^2\chi_d\rangle|\leq  C    \Big(\varepsilon_1\|\phi^\prime\|_2^2+ 
  \frac{1}{\varepsilon_1d}|\phi(x_\nu)|^2\Big) \,.
\end{equation}

To estimate the last term on the right-hand-side of \eqref{eq:89},  we use 
\eqref{eq:65} and \eqref{eq:91} to obtain
\begin{equation}
\label{eq:93}
  \Big|\Big\langle w,U^{\prime\prime}\frac{\phi}{U+i\lambda}\Big\rangle\Big|\leq
  C\|w\|_2\|\phi\|_\infty\Big\|\frac{1}{U+i\lambda}\Big\|_2 \leq C|\mu|^{-1/2}\|\phi\|_\infty\Big(\|\phi^\prime\|_2+
  \frac{1}{d^{1/2}}|\phi(x_\nu)|\Big)
\end{equation}

Substituting \eqref{eq:93} together with \eqref{eq:90}, \eqref{eq:126a}, and
  \eqref{eq:92} into \eqref{eq:89} yields that there exists $C>0$
  such that for every $0<\varepsilon_1<1$ it holds that
\begin{equation}
  \label{eq:94} 
  \|(U-\nu)w^\prime\|_2^2 + \alpha^2\|\varphi\|_2^2 \leq  C \, \Big((\varepsilon_1+|\mu|^{1/2})\|\phi^\prime\|_2^2+ 
  \frac{1}{\varepsilon_1d}|\phi(x_\nu)|^2 +\varepsilon_1^{-1}N_{m,p}^\pm(v,\lambda)^2\Big)
\end{equation}

By Hardy's inequality \eqref{eq:hardy1}, Poincar\'e's inequality, and
\eqref{eq:84}  we obtain, for $0 < |\mu| \leq 1$,  $\varepsilon\in (0,1)$, and $\varepsilon_1\in (0,1)$, 
that
\begin{displaymath}
  \|w\|_2 \leq C \Big( \Big[ |\mu|^{1/4}+\varepsilon_1^{1/2}+\frac{\varepsilon^{1/2}+|\mu|^{1/2}}{[\varepsilon_1d]^{1/2}}\Big] \|\phi^\prime\|_{2}
 + [\varepsilon_1d\varepsilon]^{-1/2}N_{m,p}^\pm(v,\lambda)\Big)\,.
\end{displaymath}
Selecting $\varepsilon=\varepsilon_1^2$ then yields, for any $|\mu|\leq \mu_{A,d_1}  $, and $\varepsilon_1\in (0,1)$,
\begin{displaymath}
  \|w\|_2 \leq C  \big(|\mu|^{1/4}(1+|\mu|^{1/4}[\varepsilon_1d]^{-1/2})+ d^{-1/2}\varepsilon_1^{1/2}\big) \|\phi^\prime\|_{2} 
 + C[\varepsilon_1^3d]^{-1/2}N_{m,p}^\pm(v,\lambda)\,.
\end{displaymath}
As $d\geq d_1$ we then write
\begin{equation}
\label{eq:95}
  \|w\|_2 \leq C  \big(|\mu|^{1/4}(1+|\mu|^{1/4}[\varepsilon_1d_1]^{-1/2})+ d_1^{-1/2}\varepsilon_1^{1/2}\big) \|\phi^\prime\|_{2} 
 + C[\varepsilon_1^3d_1]^{-1/2}N_{m,p}^\pm(v,\lambda)\,.
\end{equation}
Recalling the definitions of $\phi$, $\varphi$ and $w$, we immediately
conclude that
\begin{equation*}
  \|\phi^\prime\|_2\leq \|(U-\nu)w^\prime\|_2 + \|U^\prime w\|_2 + C d^{-1/2}|\phi(x_\nu)|\,.
\end{equation*}
which together with \eqref{eq:84}   gives
\begin{equation}
\label{eq:96}
  \|\phi^\prime\|_2\leq \|(U-\nu)w^\prime\|_2 + \|U^\prime w\|_2 + C\big(
  [d_1^{1/2}\varepsilon_1]^{-1} N_{m,p}^\pm(v,\lambda)  + d_1^{-\frac 12}  (|\mu|^\frac 12 + \varepsilon_1)  \big)\|\phi^\prime\|_2 \,.
\end{equation}
Substituting \eqref{eq:95} and \eqref{eq:94} into \eqref{eq:96}
yields
\begin{displaymath}
    \|\phi^\prime\|_2\leq C(d_1) \big(\varepsilon_1^{-3/2} N_{m,p}^\pm(v,\lambda)+  (|\mu|^\frac 14 \,+  |\mu|^{\frac 12} \epsilon_1^{-\frac 12}  + \varepsilon_1^{1/2})\|\phi^\prime\|_2\big)
\end{displaymath}
Hence, we can choose first $\epsilon_1$ and  then $\mu_{A,d_1} >0$ such that 
\eqref{eq:87} follows for  $|\mu|\leq \mu_{A,d_1} $ and $d\geq d_1$. \\

{\em Step 3:} {\it We prove \eqref{eq:87} under the assumption that 
\begin{equation}
\alpha^2 \geq C_U:= \frac 14 C_0^4 \delta_2(U)^4\,,
\end{equation}
with 
\begin{equation}
\label{eq:97}
  C_0 = 2\sup_{
    \begin{subarray}{c}
      |\mu|\leq 1  \\
      \nu\in[U(-1),U(1)]
    \end{subarray}}\|\log\,(U+i\lambda)\|_2\,.
\end{equation}
}
\vspace{2ex}

\noindent We recall \eqref{eq:58} which reads
\begin{displaymath}
  \Re\Big\langle\phi,\frac{v}{U+i\lambda}\Big\rangle= \|\phi^\prime\|_2^2 + \alpha^2\|\phi\|_2^2- \Re\Big\langle U^{\prime\prime}\phi,\frac{\phi}{U+i\lambda}\Big\rangle\,.
\end{displaymath}
For the last term we have, using Poincar\'e's inequality and Sobolev's
embeddings 
\begin{equation*}\begin{array}{ll}
  \Big|\Big\langle U^{\prime\prime}\phi,\frac{\phi}{U+i\lambda}\Big\rangle\Big| 
 &  \leq  \Big|\Big\langle
  \Big(\frac{U^{\prime\prime}}{U^\prime}|\phi|^2\Big)^\prime,\log\, (U+i\lambda)\Big\rangle\Big| \\ &\leq 
  \|\log\,(U+i\lambda)\|_2\|\phi\|_\infty\Big(2\Big\|\frac{U^{\prime\prime}}{U^\prime}\Big\|_\infty
  \|\phi^\prime\|_2+
  \Big\|\Big(\frac{U^{\prime\prime}}{U^\prime}\Big)^\prime\Big\|_\infty\|\phi\|_2\Big) \\ &\leq 
  C_0 \, \delta_2(U) \|\phi^\prime\|_2^{3/2}\|\phi\|_2^{1/2} \,,
  \end{array}
\end{equation*}
where $C_0$ is given by \eqref{eq:97}.

Consequently, by \eqref{eq:85}
\begin{displaymath}
  \|\phi^\prime\|_2^2 +\alpha^2\|\phi\|_2^2\leq C_0\delta_2(U) \|\phi^\prime\|_2^{3/2}\|\phi\|_2^{1/2} +C
  \|\phi^\prime\|_2N_{m,p}^\pm(v,\lambda)\,.
\end{displaymath}
Using Young's inequality we obtain 
\begin{equation}
\label{eq:98} 
  \frac{1}{8}\|\phi^\prime\|_2^2\leq \Big(
  \frac{1}{4}C_0^4\delta_2(U)^4-\alpha^2\Big)\|\phi\|_2^2 + \hat CN_{m,p}^\pm(v,\lambda)^2
\end{equation}
Hence, for $\alpha^2\geq C_U$  \eqref{eq:87} follows immediately from the above
inequality in conjunction with Poincar\'e's inequality.\\

{\em Step 4:} {\it We prove that there exist
  $d_0>0$, $\mu_0>0$  and
$C >0$ such that, for all $d \leq d_0$, $\nu\in[U(-1),U(1)]$, $\alpha \geq 0$  and $|\mu| \leq \mu_0$, 
\begin{equation}
\label{eq:99} 
\|\phi\|_{1,2} \leq  CN_{m,p}^\pm(v,\lambda)\,.
\end{equation}
holds for any pair $(\phi,v)$ such that $\mathcal A_{\lambda,\alpha} \phi=v$.}
\vspace{2ex}

Without any loss of generality we can assume that $d=1-x_\nu$. As
\begin{displaymath}
  (U-U(1))(-\phi^{\prime\prime}+\alpha^2\phi)-U^{\prime\prime}\phi= v- (U(1)-\nu+i\mu)(-\phi^{\prime\prime}+\alpha^2\phi)
\end{displaymath}
or equivalently, by \eqref{eq:36},
\begin{displaymath}
\begin{array}{l}
  -\Big((U-U(1))^2\Big(\frac{\phi}{U-U(1)}\Big)^\prime\Big)^\prime+\alpha^2(U-U(1))\phi\\
  \qquad \qquad \qquad \qquad \qquad \qquad = \frac{(U-U(1)) v }{U+i\lambda}+
  (U(1)-\nu+i\mu)\frac{U^{\prime\prime}\phi}{U+i\lambda}\,.
\end{array}
\end{displaymath}
Taking the scalar product with $(U-U(1))^{-1} \phi$  and integrating by parts then yield
\begin{multline}
\label{eq:100}
  \Big\|(U-U(1))\Big(\frac{\phi}{U-U(1)}\Big)^\prime\Big\|_2^2 + \alpha^2\|\phi\|_2^2 \\
=  \Big\langle\phi,\frac{v}{U+i\lambda}\Big\rangle +
  (U(1)+ i\lambda )\Big\langle\frac{\phi}{U-U(1)},\frac{U^{\prime\prime}\phi}{U+i\lambda}\Big\rangle \,. 
\end{multline}

For the first term on the right-hand-side of \eqref{eq:100} we
  use \eqref{eq:85}.\\
  
Next, we estimate the second inner product on the right-hand-side of
\eqref{eq:100} by splitting the
domain of integration in two sub-intervals: $(1-2d,1)$ and
$(-1,1-2d)$. \\

\paragraph{The integral over $(1-2d,1)$.} To estimate the integral over $(1-2d,1)$
we use the identity
\begin{displaymath}
  \frac{1}{[U(1)-U](U+i\lambda)}=\frac{1}{U(1) +i\lambda}\Big[\frac{1}{U+i\lambda}+\frac{1}{U(1)-U}\Big] \,,
\end{displaymath}
to obtain that
\begin{multline}
\label{eq:101} 
   (U(1)+i\lambda)\int_{1-2d}^1
    \frac{U^{\prime\prime}|\phi|^2}{[U(1)-U](U+i\lambda)}\,dx =\\
   \int_{1-2d}^1\frac{U^{\prime\prime}|\phi|^2}{U+i\lambda}\,dx +\int_{1-2d}^1\frac{U^{\prime\prime}|\phi|^2}{U(1)-U}\,dx\,.
\end{multline}
As
\begin{displaymath}
  \int_{1-2d}^1\frac{U^{\prime\prime}|\phi|^2}{U+i\lambda}\,dx =-\frac{U^{\prime\prime}}{U^\prime}
  |\phi|^2\,|\log\,(U+i\lambda)\Big|_{1-2d}^1-
\int_{1-2d}^1\Big(\frac{U^{\prime\prime}}{U^\prime}|\phi|^2\Big)^\prime \,\,|\log\,(U+i\lambda)|\,dx \,,
\end{displaymath}
we may conclude, having in mind that $\mu$ and $d$ 
are bounded,  that 
\begin{equation}
\begin{array}{ll}
 \Big|\int_{1-2d}^1\frac{U^{\prime\prime}|\phi|^2}{U+i\lambda}\,dx\Big|& \leq \, C\, 
   \Big\|\frac{U^{\prime\prime}}{U^\prime}\Big\|_{1,\infty}\Big(d\log d \|\phi^\prime\|_2^2+ 
   \int_{1-2d}^1\left( |\phi|\,|\phi^\prime| +|\phi|^2\right) \,|\log\,(U+i\lambda)|\,dx\Big)\\ & \leq  \hat C\, d^{1/2} \|\phi^\prime\|_2^2 \,.  
\end{array}
\end{equation}
Furthermore, employing Hardy's inequality and Cauchy-Schwarz inequality yields
\begin{displaymath}
   \Big|\int_{1-2d}^1\frac{U^{\prime\prime}|\phi|^2}{U-U(1)}\,dx\Big|\leq  C\, d^{1/2}\|\phi^\prime\|_2\|\phi\|_\infty\,.
\end{displaymath}
Substituting the above inequalities into \eqref{eq:101} yields 
\begin{equation}
\label{eq:102}
  \Big|(U(1)+i\lambda)\int_{1-2d}^1
    \frac{U^{\prime\prime}|\phi|^2}{[U(1)-U](U+i\lambda)}\,dx\Big| \leq Cd^{1/2}\|\phi^\prime\|_2^2 \,. 
\end{equation}

\paragraph{The integrals over $(-1,1-2d)$.} We now estimate the
integrals over $[-1,1-2d]$ for the inner products on the right-hand-side of
\eqref{eq:100}. To this end we write, using Hardy's inequality
\eqref{eq:hardy1} and lower bounds of $|U+i \lambda|$, 
\begin{multline}\label{eq:4.71a}
  \Big|\int_{-1}^{1-2d}
    \frac{U^{\prime\prime}|\phi|^2}{[U(1)-U](U+i\lambda)}\,dx \Big| \\ \leq
    C\|\phi^\prime\|_2\|\phi\|_\infty \Big[\int_{-1}^{1-2d}\frac{1}{|U+i\lambda|^2}\Big]^{1/2}
    \\  \leq 
    \hat C\, ( \max (d,|\mu|))^{-1/2}\,\|\phi^\prime\|_2^2\,. 
\end{multline}
Returning to the estimate of the right hand side of \eqref{eq:100}, we use
\eqref{eq:4.71a}, and the fact that $|U(1)-\nu| \leq  \hat C  d$
(following from by definition of $d$) together with Poincar\'e's
inequality, to obtain
\begin{equation*}
  |U(1)-\nu|\,\Big|\int_{-1}^{1-2d}
    \frac{U^{\prime\prime}|\phi|^2}{[U(1)-U](U+i\lambda)}\,dx \Big|\leq
    C d^{1/2}\|\phi^\prime\|_2^2  \,. 
\end{equation*}
and 
\begin{equation*}
  |\mu|\,\Big|\int_{-1}^{1-2d}
    \frac{\bar{\phi}\,U^{\prime\prime}\phi}{[U(1)-U](U+i\lambda)}\,dx \Big| 
    \leq 
    C|\mu|^{1/2}\|\phi^\prime\|_2^2\,.   
  \end{equation*}
Combining the above and \eqref{eq:102} yields 
\begin{equation}
\label{eq:103}
  \Big|
  (U(1)-\nu+i\mu)\Big\langle\frac{\phi}{U-U(1)},\frac{U^{\prime\prime}\phi}{U+i\lambda}\Big\rangle\Big|\leq
  C[d^{1/2}+|\mu|^{1/2}]\|\phi^\prime\|_2^2\,. 
\end{equation}
 Substituting \eqref{eq:103} together with \eqref{eq:85}
into \eqref{eq:100}, yields
\begin{displaymath}
   \Big\|(U-U(1))\Big(\frac{\phi}{U-U(1)}\Big)^\prime\Big\|_2^2
   \leq C\, \left([d^{1/2}+|\mu|^{1/2}]\|\phi^\prime\|_2^2+\|\phi^\prime\|_2N_{m,p}^\pm(v,\lambda)\right)\,.  
\end{displaymath}
From Hardy's inequality in the form \eqref{eq:hardy1} we then conclude
\begin{displaymath}
   \Big\|\frac{\phi}{U-U(1)}\Big\|_2^2 \leq
   C([d^{1/2}+|\mu|^{1/2}]\|\phi^\prime\|_2^2+\|\phi^\prime\|_2N_{m,p}^\pm(v,\lambda))\,.
\end{displaymath}
Combined with the following straightforward observation
\begin{equation}\label{eq:4.75a}
  \|\phi^\prime\|_2 \leq  \Big\|(U-U(1))\Big(\frac{\phi}{U-U(1)}\Big)^\prime\Big\|_2 +
  \Big\|U^\prime\frac{\phi}{U-U(1)}\Big\|_2\,,
\end{equation}
this yields,
\begin{displaymath}
  \|\phi^\prime\|_2^2\leq C([d^{1/2}+|\mu|^{1/2}]\|\phi^\prime\|_2^2+\|\phi^\prime\|_2N_{m,p}^+(v,\lambda))\,.  
\end{displaymath}
 Hence, there exists $d_0>0$ and $\mu_0>0$ such that \eqref{eq:99} holds for all
$d\leq d_0$ and $|\mu|\leq \mu_0$. \\

{\em Step 5:} {\it We prove that there exist $C>0$ and $\mu_0>0$ such that \eqref{eq:99}
holds for all  $\nu\in\R\setminus[U(-1),U(1)]$ and $|\mu|\leq \mu_0$. }
\vspace{2ex}

Without any loss of generality we assume $\nu>U(1)$. We begin by
rewriting $\A_{\lambda,\alpha}\phi=v$ in the form
\begin{displaymath}
  -\Big((U-\nu)^2\Big(\frac{\phi}{U-\nu}\Big)^\prime\Big)^\prime+\alpha^2(U-\nu)\phi
    = v-i\mu\frac{v+U^{\prime\prime}\phi}{U+i\lambda}=\frac{(U-\nu)v-i\mu U^{\prime\prime}\phi}{U+i\lambda}\,.
\end{displaymath}
Taking the inner product with $\phi/(U-\nu)$ on the left yields
\begin{equation}
\label{eq:104}
   \Big\|(U-\nu)\Big(\frac{\phi}{U-\nu}\Big)^\prime\Big\|_2^2 + \alpha^2\|\phi\|_2^2 
=  \Big\langle\phi,\frac{v}{U+i\lambda}\Big\rangle -i\mu\Big\langle\frac{\phi}{U-\nu},\frac{U^{\prime\prime}\phi}{U+i\lambda}\Big\rangle \,. 
\end{equation}
Let
\begin{displaymath}
  \tilde{\phi}(x)=
  \begin{cases}
    \phi(x) & x\in[-1,1] \\
    0 & |x|>1 \,.
  \end{cases}
\end{displaymath}
Let 
\begin{displaymath}
  \hat{U}=
  \begin{cases}
    U(x) & x\in[-1,1] \\
    U(1) + U^\prime(1)(x-1) & x>1 \,.
  \end{cases}
\end{displaymath}
Let $\hat{x}_\nu>1$ denote the unique zero of $\hat{U}-\nu$. We may now use
Hardy's inequality to obtain the existence of $C >0$ such that for all $\nu \in \mathbb R\setminus [U(-1), U(+1)]$, all $\phi \in H_0^1 (-1,+1)$, 
\begin{equation}\label{eq:4.75z}
\begin{array}{ll}
  \Big\|(U-\nu)\Big(\frac{\phi}{U-\nu}\Big)^\prime\Big\|_2^2 &  =
  \Big\|(\hat{U}-\nu)\Big(\frac{\tilde{\phi}}{\hat{U}-\nu}\Big)^\prime\Big\|_{L^2(-1,\hat{x}_\nu)}^2\\  & \geq
  \frac 1 C \Big\|\frac{\tilde{\phi}}{\hat{U}-\nu}\Big\|_{L^2(-1,\hat{x}_\nu)}^2 \\ & =
  \frac 1 C\Big\|\frac{\phi}{U-\nu}\Big\|_2^2  \,. 
  \end{array}
\end{equation}
Next we use the analog of \eqref{eq:4.75a}
  \begin{equation*}
  \|\phi^\prime\|_2^2 \leq 2  \Big\|(U-\nu)\Big(\frac{\phi}{U- \nu}\Big)^\prime\Big\|_2^2 + 2
  \Big\|U^\prime\frac{\phi}{U-\nu}\Big\|_2^2  \,,
\end{equation*}
which leads, together with \eqref{eq:4.75z},  to
\begin{equation}
\label{eq:105}
\|\phi^\prime\|_2 \leq \hat C \|(U-\nu)\Big(\frac{\phi}{U-\nu}\Big)^\prime\Big\|_2\,.
\end{equation}
On the other hand we have from \eqref{eq:104} and \eqref{eq:65}
\begin{equation*} \begin{array}{ll}
 \Big\|(U-\nu)\Big(\frac{\phi}{U-\nu}\Big)^\prime\Big\|_2^2  &  \leq
 \|\phi^\prime\|_2N_{m,p}^\pm(v,\lambda)
 +|\mu|\Big\|\frac{\phi}{U-\nu}\Big\|_2\|\phi\|_\infty\Big\|\frac{1}{U+i\lambda}\Big\|_2\\ 
  & \leq C \|\phi^\prime\|_2 \Big(N_{m,p}^\pm(v,\lambda)+|\mu|^{1/2}\Big\|\frac{\phi}{U-\nu}\Big\|_2 \Big)\,.
  \end{array}
\end{equation*}
Combining the above with \eqref{eq:4.75z} yields first
\begin{equation*}
 \Big\|(U-\nu)\Big(\frac{\phi}{U-\nu}\Big)^\prime\Big\|_2 \leq \check C \,
 (N_{m,p}^\pm(v,\lambda)+ |\mu|^\frac 12
 \Big\|(U-\nu)\Big(\frac{\phi}{U-\nu}\Big)^\prime\Big\|_2) \,, 
 \end{equation*}
 hence, for sufficiently small $|\mu|$ , 
 \begin{equation} \label{eq:262a}
 \Big\|(U-\nu)\Big(\frac{\phi}{U-\nu}\Big)^\prime\Big\|_2 \leq 2 \check C\, N_{m,p}^\pm(v,\lambda)\,.
 \end{equation}
  Finally, using \eqref{eq:105} once again leads to the existence of
  $C>0$  and $\mu_0>0$ such that if $|\mu| \leq  \mu_0$  
\begin{equation}\label{eq:262z}
  \|\phi^\prime\|_2 \leq C \, N_{m,p}^\pm(v,\lambda)\,, 
\end{equation}
and we obtain \eqref{eq:99} in this case as
well. \\

{\em Step 6:} {\it Prove \eqref{eq:83}. }\\
\vspace{2ex}
By \eqref{eq:99} (established in steps 4 and  5 for $d\leq d_0$)  and \eqref{eq:87} 
 (proved in steps 2 and 3 for $d\geq d_1= d_0$),  there exist $C>0$ and $\mu_0 >0$ 
such that for $|\mu| \leq \mu_0$ we have 
\begin{equation}
\label{eq:106}
  \|\phi\|_{1,2}\leq  C\, N_{m,p}^\pm(v,\lambda)\,,
\end{equation}
As $N_{m,p}^\pm(v,\lambda)\leq \|v\|_{1,p}$, we can immediately conclude
(\ref{eq:83}b).  To conclude (\ref{eq:83}a,c) we first
observe that
\begin{displaymath}
N_{m,p}^\pm(v,\lambda) \leq 2\Big\| \frac{v}{U+i\lambda}\Big\|_1
\end{displaymath}
 and then use  H\"older inequality
\begin{displaymath}
  \Big\|\frac{v}{U+i\lambda}\Big\|_1 \leq 
  \|v\|_p \, \Big\|\frac{1}{U+i\lambda}\Big\|_q \,,
\end{displaymath}
valid for any $1 < p\leq\infty$, together with \eqref{eq:65}, and
\eqref{contra}. 
\end{proof}

For later reference we also need the following estimate which can be
deduced  from the proofs of Propositions~
\ref{lem:inviscid-boundedness-1} and \ref{prop:constant-sign}. 
\begin{proposition}
\label{prop:unbounded-mu}
Let $r>1$. Then,  there exists $C>0$ such that, for all \break $v\in
  L^2(-1,1)$,  $\lambda$ such that $\Re \lambda \neq0$,  and $U\in\Sg^3_r$
  satisfying \eqref{condsurinf} it holds that
 \begin{equation}
    \label{eq:107}
\|\A_{\lambda,\alpha}^{-1}v \|_{1,2} \leq C\Big\|(1\pm x)^{1/2}\frac{v}{U+i\lambda}\Big\|_1 \,.
  \end{equation}
\end{proposition}
\begin{proof}
  Let $\mu_0$ be as in the statement of Proposition
  \ref{prop:constant-sign}.  Let $\lambda=\mu+i\nu$, \break $v\in L^2(-1,1)$, and
  $\phi=\A_{\lambda,\alpha}^{-1}v$. For $|\mu|<\mu_0$,  \eqref{eq:107} is an
  immediate result of
(\ref{eq:106}). 

Consider then the case $|\mu|\geq\mu_0$. 
By \eqref{eq:78} and \eqref{eq:85},  we 
obtain that
\begin{displaymath}
  |\mu|\min_{x\in[-1,1]}|U^{\prime\prime}|\, \Big\langle\frac{\phi}{(U-\nu)^2+\mu^2},\phi\Big\rangle \leq
  \Big|\Big\langle\phi,\frac{v}{U-\nu+i\mu}\Big\rangle\Big| \leq  C  \|\phi^\prime\|_2N_{m,p}^\pm(v,\lambda) \,.
\end{displaymath}
Consequently, as $|\mu|>\mu_0$, there exists $C(\mu_0)>0$ such that 
\begin{equation}
  \label{eq:108}
\Big\|\frac{\phi}{U+i\lambda}\Big\|_2^2 \leq C \|\phi^\prime\|_2N_{m,p}^\pm(v,\lambda)\,.
\end{equation}
Then, we use \eqref{eq:81} to establish that
\begin{displaymath}
  \|\phi^\prime\|_2^2 + \alpha^2\|\phi\|_2^2 \leq
  \Big|\Big\langle\phi,\frac{v}{U-\nu+i\mu}\Big\rangle\Big|+\Big\|\frac{\phi}{U+i\lambda}\Big\|_2\|\phi\|_2 \,, 
\end{displaymath}
from which we conclude, with the aid of \eqref{eq:108}, that
\begin{equation}
\label{eq:375}
    \|\phi^\prime\|_2^2 +
  \alpha^2\|\phi\|_2^2 \leq  C \left( \big[\|\phi^\prime\|_2N_{m,p}^\pm(v,\lambda)\big]^{1/2}\|\phi\|_2 +
  \|\phi^\prime\|_2N_{m,p}^\pm(v,\lambda)\right)  \,.
\end{equation}
Using Poincar\'e's inequality we can now establish \eqref{eq:107}.
\end{proof}

\begin{corollary}
Under the assumptions of Proposition \ref{prop:unbounded-mu} there
exists $C>0$ such
that 
  \begin{subequations}
\label{eq:374}
  \begin{equation}
\sup_{
  \begin{subarray}{c}
    \Re\lambda\neq0 \\
    0\leq\alpha 
  \end{subarray}}
\frac{1}{\log\, |\Re \lambda|^{-1} }  \|\A_{\lambda,\alpha}^{-1}v\|_{1,2} \leq C \, \|v\|_\infty \,,
\end{equation}
\begin{equation}
\sup_{\begin{subarray}{c}
    \Re\lambda\neq0 \\
    0\leq\alpha 
  \end{subarray}}
\|\A_{\lambda,\alpha}^{-1}v\|_{1,2} \leq C \, (\|v^\prime\|_p+\|v\|_\infty)\,,
\end{equation}
and
\begin{equation}
\sup_{\begin{subarray}{c} 
    \Re\lambda\neq0 \\
    0\leq\alpha 
  \end{subarray}}|\Re \lambda|^{1/p}\|\A_{\lambda,\alpha}^{-1}v\|_{1,2} \leq C\, \|v\|_p\,.
\end{equation}
\end{subequations}
\end{corollary}
The proof follows immediately from \eqref{eq:375} and step 6 of
  the proof of Proposition \ref{prop:constant-sign}.

\section{Some Schr\"odinger operators and their resolvents} 
\label{sec:3-prem} 
In this section we derive several refinements of estimates obtained in
\cite{Hen2,AGH} for the resolvent of $\LL_\beta= -\frac{d^2}{dx^2} +i \beta
U $, (as in \eqref{eq:29}) defined over different domains. As in the
rest of this contribution, we are assuming \eqref{eq:19}. These
estimates will be used in Sections \ref{s7new} and \ref{sec:no-slip}.

\subsection{The entire real line}
We begin by stating the following result on $\R$.
\begin{proposition}
\label{lem:model-entire}
For $\tilde{U}\in C^1(\R)$, let 
  $\widetilde{\LL}_{\beta,\R}$  be given by
  \begin{equation}
    \label{eq:109}
\widetilde{\LL}_{\beta,\R} =-\frac{d^2}{dx^2} +i \beta \tilde{U} \,,
  \end{equation}
with domain
  \begin{equation}
    \label{eq:110}
D(\widetilde{\LL}_{\beta,\R})= \{ u\in H^2(\R)\,|\, xu\in L^2(\R)\,\} \,. 
  \end{equation}
Then, for all  positive $\Upsilon$,  $m$, $M$, $C$, and $\epsilon$, there exist $\beta_0>0$ and
  $\hat C $, 
  such that for all $\beta\geq \beta_0$ and $\tilde U\in C^1(\mathbb R)$   satisfying  
  \begin{equation}\label{asstildeU}
  \left\{
  \begin{array} {l} 
  0<m\leq \tilde {U}^\prime(x) \leq M \mbox{  for all } x\in\R\\
  |\tilde U^\prime (x) -\tilde U^\prime(y)| \leq C |x-y|^\epsilon\,,  \mbox{ for all } x,y \in \mathbb R \mbox{ s.t. } |x-y|\leq 1\,,
  \end{array}
  \right.
  \end{equation}
  it holds that
\begin{multline}
  \label{eq:111}
\sup_{\Re \lambda \leq\Upsilon\beta^{-1/3}}{ (1+\beta^{1/3}|\Re \lambda|)}
\|(\widetilde{\LL}_{\beta,\R}-\beta\lambda)^{-1}\|  \\ +  \beta^{-1/3}
\sup_{\Re \lambda \leq \Upsilon\beta^{-1/3}} \| \frac{d}{dx} (\widetilde{\LL}_{\beta,\R}-\beta\lambda)^{-1}\| \leq
\frac{\hat C}{\beta^{2/3}} \,,
\end{multline}
and
\begin{equation}
  \label{eq:112}
\sup_{\Re \lambda \leq\Upsilon\beta^{-1/3}}\|(\tilde{U}-\Im \lambda )(\widetilde{\LL}_{\beta,\R}-\beta\lambda)^{-1}\| \leq \frac{\hat C}{\beta} \,.
\end{equation}
\end{proposition}
\begin{proof}~\\
  The estimation of the first term on the left-hand-side of
  \eqref{eq:111}, can be obtained,  for $-1<\beta^{1/3}\mu<\Upsilon$  (with $\mu =\Re \lambda$) as in
    the proof of \cite[Theorem 1.1 (ii)]{Hen2}.  The difference is
    that the interval is infinite and that we impose here less
    regularity on the potential. To accommodate $C^{1,\epsilon}$ potential in
    the proof in \cite{Hen2}, we construct a partition of unity
    composed of intervals of size $\beta^{-\frac{\rho}{2}}$), and select $\rho \in
    (\frac{2} {3(1+\epsilon)}, \frac 23)$ instead of $\rho \in (\frac 13, \frac
    23)$ (as in p. 16, line 6 in \cite{Hen2}).  The remaining details
    are skipped.\\
  We now observe that
  \begin{equation}\label{eq:97a}
    \Re\langle u,(\widetilde{\LL}_{\beta,\R}-\beta\lambda)u\rangle=\|u^\prime\|_2^2 -\beta\mu\|u\|_2^2 \,,
  \end{equation}
   For $\mu\leq-\beta^{-1/3}$, we deduce 
   \begin{displaymath}
   \|u\|_2\leq \frac{C}{\beta|\mu|}\| (\widetilde{\LL}_{\beta,\R}-\beta\lambda)u\|_2 \leq
   \frac{2C}{\beta^{2/3}(1+\beta^{1/3}|\mu|)}\| (\widetilde{\LL}_{\beta,\R}-\beta\lambda)u\|_2\,,
  \end{displaymath}
  which gives the estimate of the first term on the left-hand-side of
  \eqref{eq:111}  for \break $\beta^{1/3}\mu \leq -1$.\\
  To estimate the second term on the right-hand-side of \eqref{eq:111},
  we return to \eqref{eq:97a}
  to conclude from the bound of
  $\|(\widetilde{\LL}_{\beta,\R}-\beta\lambda)^{-1}\|$, we have just obtained, that
\begin{equation}
  \label{eq:113}
\|u^\prime\|_2\leq \frac{C}{\beta^{1/3}}\| (\widetilde{\LL}_{\beta,\R}-\beta\lambda)u\|_2 \,.
\end{equation}
Finally, to prove \eqref{eq:112} we use the identity
\begin{displaymath}
   \Im \langle(\tilde{U}-\nu)u,(\widetilde{\LL}_{\beta,\R}-\beta\lambda)u\rangle= \beta\|(\tilde{U}-\nu)u\|_2^2+ \Im\langle \tilde{U}^\prime u,u^\prime\rangle\,,
\end{displaymath}
to obtain with the aid of \eqref{eq:111} that
\begin{displaymath}
  \beta\|(\tilde{U}-\nu)u\|_2^2 \leq
  \|(\widetilde{\LL}_{\beta,\R}-\beta\lambda)u\|_2\|(\tilde{U}-\nu)u\|_2+\frac{C}{\beta}\|(\widetilde{\LL}_{\beta,\R}-\beta\lambda)u\|_2^2 \,,
\end{displaymath}
from which \eqref{eq:112}  easily follows.
\end{proof}

\subsection{A Dirichlet problem}
We now obtain some resolvent estimate for the Dirichlet realization $\LL_\beta^\Df$  of
$\LL_\beta$ in $(-1,1)$.
\begin{proposition}
\label{Schrodinger-Dirichlet}
 For any  $r>1$ and $\Upsilon< {\mathfrak J}_m^{2/3}\Re\nu_1$, there exist $C>0$ and $\beta_0>0$ such that for all $U\in \Sg^2_r$ and 
  $\beta\geq \beta_0$
  \begin{multline}
    \label{eq:114}
\sup_{\Re\lambda\leq\Upsilon\beta^{-1/3}}
\big[\|(\LL_\beta^\Df-\beta\lambda)^{-1}\|+\beta^{-1/3}\|\frac{d}{dx}\,(\LL_\beta^\Df-\beta\lambda)^{-1}\|
\\ +
\beta^{-2/3}\|\frac{d^2}{dx^2}\, (\LL_\beta^\Df-\beta\lambda)^{-1}\|\big]\leq\frac{C}{\beta^{2/3}}\, ,
  \end{multline}
and 
\begin{equation}
  \label{eq:115}
\sup_{\Re\lambda\leq\Upsilon\beta^{-1/3}} \|(U-\Im\lambda)(\LL_\beta^\Df-\beta\lambda)^{-1}\|\leq\frac{C}{\beta}\,.
\end{equation}
Furthermore, for every $1<p<2$ there exists $C_p>0$ such that for all $f\in L^2(-1,1)$
\begin{equation}
\label{eq:116}
\sup_{\Re\lambda\leq\Upsilon\beta^{-1/3}}
\Big\| \frac{d}{dx} (\LL_\beta^\Df-\beta\lambda)^{-1}f\Big\|_p\leq\frac{C_p}{\beta^{\frac{2+p}{6p}}}\|f\|_2\,.
\end{equation}
\end{proposition}
\begin{proof}
 By \cite[Theorem 1.1]{Hen2}  we have, under the assumptions of the
 proposition,  
  \begin{equation}
    \label{eq:117}
\|(\LL_\beta^\Df-\beta\lambda)^{-1}\|\leq \frac{C}{\beta^{2/3}}\,.
  \end{equation}
  We next observe that for any $u\in
  D(\LL_\beta^\Df)$
  \begin{displaymath}
    \Re\langle u,(\LL_\beta^\Df-\beta\lambda)u\rangle=\|u^\prime\|_2^2 -\beta\mu\|u\|_2^2 \,,
  \end{displaymath}
which together  with \eqref{eq:117}  yields
\begin{equation}
\label{eq:118}
  \|u^\prime\|_2\leq \frac{C}{\beta^{1/3}}\| (\LL_\beta^\Df-\beta\lambda)u\|_2 \,.
\end{equation}
To complete the proof of \eqref{eq:114} we write, with $\mu=\Re \lambda$, 
\begin{displaymath}
  -\Re\langle u^{\prime\prime},(\LL_\beta^\Df-\beta\lambda)u\rangle=\|u^{\prime\prime}\|_2^2 -\beta\mu\|u^\prime\|_2^2- \beta  \, \Im\langle U^\prime u,u^\prime\rangle \,.
\end{displaymath}
From \eqref{eq:117} and \eqref{eq:118} we then obtain that
\begin{displaymath}
  \|u^{\prime\prime}\|_2 \leq C\|(\LL_\beta^\Df-\beta\lambda)u\|_2 \,,
\end{displaymath}
completing, thereby, the proof of \eqref{eq:114}. 
We establish \eqref{eq:115} in the same way we have
established \eqref{eq:112}.\\

It remains to  prove \eqref{eq:116}. Let $f\in L^2(-1,1)$ and $u\in D(\LL_\beta^\Df)$ satisfy
\begin{equation}
\label{eq:119}
  (\LL_\beta^\Df-\beta\lambda)u=f\,.
\end{equation}
Let $\eta$  and $\tilde \eta$   denote $C^\infty$ functions such that 
  \begin{equation}
    \label{eq:120} 
\tilde \eta(t)=
\begin{cases}
  0 & |t|<1/2 \\
  1 & |t|>1\,,
\end{cases}  \mbox{ and } 
\eta(t) =\sqrt{1-\tilde \eta(t)^2}\,.
  \end{equation} 
Let $x_\nu\in[-1,1]$ be such that $U(x_\nu)=\nu$ (otherwise, if $U(x)\neq\nu$
for all $x\in[-1,1]$,  we arbitrarily set $x_\nu=-2$).\\
 Let
 \begin{equation}
\label{eq:121}
  \eta_\nu(x)=\tilde{\eta}(\beta^{1/3}(x-x_\nu)){\mathbf
    1}_{\R_+}(x-x_\nu)\,.
 \end{equation}
 As
\begin{displaymath}
\begin{array}{l}
  \Re \langle\eta_\nu^2(U-\nu)u,(\LL_\beta^\Df-\beta\lambda)u\rangle\\
  \quad =  \|\eta_\nu|U-\nu|^{1/2}u^\prime\|_2^2 +
  \Re\langle \eta_\nu(U^\prime \eta_\nu+2(U-\nu)\eta_\nu^\prime)u ,u^\prime\rangle-\mu\beta\|\eta_\nu|U-\nu|^{1/2}u\|_2^2\,,
  \end{array}
 \end{displaymath}
we obtain, observing that $|(U-\nu) \eta^\prime_\nu|$ is uniformly bounded,
\begin{multline}
\label{eq:122}
  \|\eta_\nu|U-\nu|^{1/2}u^\prime\|_2^2 \leq \|(U-\nu)u\|_2\|f\|_2 + \\
  C\,\left(\beta^{2/3}\|\eta_\nu|U-\nu|^{1/2}u\|_2^2+ \|u\|_2\|u^\prime\|_2\right)\,.
\end{multline}
Furthermore, since
\begin{equation}\label{eq:107a}
  \Im \langle\eta_\nu^2u,(\LL_\beta^\Df-\beta\lambda)u\rangle= \beta\|\eta_\nu|U-\nu|^{1/2}u\|_2^2+ 2\Im\langle\eta^\prime_\nu u,\eta_\nu u^\prime\rangle\,,
\end{equation}
we obtain
\begin{displaymath}
  \|\eta_\nu|U-\nu|^{1/2}u\|_2^2 \leq
  \frac{1}{\beta}(\|u\|_2\|f\|_2+\beta^{1/3}\|u\|_2\|u^\prime\|_2 \,,
\end{displaymath}
and hence, by \eqref{eq:114},
\begin{displaymath}
   \|\eta_\nu|U-\nu|^{1/2}u\|_2 \leq \frac{C}{\beta^{5/6}}\|f\|_2 \,.
\end{displaymath}
Substituting the above into \eqref{eq:122} yields, with the aid of
\eqref{eq:114} and \eqref{eq:115},
\begin{displaymath}
\|\eta_\nu|U-\nu|^{1/2}u^\prime\|_2 \leq \frac{C}{\beta^{1/2}}\|f\|_2\,.
\end{displaymath}
Setting $\eta_\nu^-(x)=\eta_\nu(-x)$, we obtain in a similar manner
\begin{displaymath}
  \|\eta_\nu^-|U-\nu|^{1/2}u^\prime\|_2 \leq \frac{C}{\beta^{1/2}}\|f\|_2\,,
\end{displaymath}
and hence, with the aid of \eqref{eq:114}, we can conclude that
\begin{displaymath} 
\begin{array}{ll}
  \||U+i\lambda|^{1/2}u^\prime\|_2 & \leq  \|\eta_\nu |U-\nu|^{1/2}u^\prime\|_2 +
  \|\eta_\nu^-|U-\nu|^{1/2}u^\prime\|_2  + C\beta^{-1/6}\|u^\prime\|_2\\ &  \leq  \hat C \beta^{-1/2}\, \|f\|_2\,.
  \end{array}
\end{displaymath}
We can now conclude, assuming first $|\mu|\geq\Upsilon\beta^{-1/3}/2$,  with the aid \eqref{eq:64},
\begin{equation}
\label{eq:123}
  \|u^\prime\|_p\leq
  \||U+i\lambda|^{1/2}u^\prime\|_2\||U+i\lambda|^{-1/2}\|_{\frac{2p}{2-p}}\leq
  C\beta^{-\frac{2+p}{6p}}\|f\|_2 \,.
\end{equation}
 Otherwise, if $|\mu|<\Upsilon\beta^{-1/3}/2$, we rewrite \eqref{eq:119} in the
form 
\begin{displaymath}
 (\LL_\beta^\Df-\beta[\Upsilon \beta^{-1/3}+i\nu])u=f - (\Upsilon \beta^{-\frac 13}-\mu )\beta\, u 
\end{displaymath}
to obtain from \eqref{eq:123} that
\begin{displaymath}
    \|u^\prime\|_p\leq  C\beta^{-\frac{2+p}{6p}}(\|f\|_2+\beta^{2/3}\|u\|_2)  \,.
\end{displaymath}
The proof of \eqref{eq:116} can now be completed using \eqref{eq:114}.  
\end{proof}
We can now deduce the following corollary
\begin{corollary}
   For any  $r>1$ and $\Upsilon< {\mathfrak J}_m^{2/3}\Re\nu_1$, there exist $C>0$ and $\beta_0>0$ such that for all $U\in \Sg^2_r$,
  $\beta\geq \beta_0$,  $2<p\leq\infty$, and $f\in L^2(-1,1)$ it holds that
\begin{equation} 
\label{eq:124}
\sup_{\Re\lambda\leq\Upsilon\beta^{-1/3}}  \|(\LL_\beta^\Df-\beta\lambda)^{-1}f\|_p \leq C\beta^{-\frac{3p+2}{6p}}\|f\|_2\,.
\end{equation}  
\end{corollary}
\begin{proof}
  Let $f\in L^2(-1,1)$ and set $v=(\LL_\beta^\Df-\beta\lambda)^{-1}f$.  By
  \eqref{eq:114}, there exist $C>0$ and $\beta_0$ such that, for $\beta \geq \beta_0$
\begin{equation}
\label{eq:125}
  \|v\|_2\leq\frac{C}{\beta^{2/3}}\|f\|_2 \,,
\mbox{ and }
  \|v^\prime\|_2 \leq C\beta^{-1/3}\|f\|_2\,.
\end{equation}
Consequently, using the interpolation inequality
\begin{displaymath}
\|v\|_p \leq \|v\|_2^{\frac 2 p} \, \| v\|_\infty^{1-\frac 2p}\,,
\end{displaymath}
and the Sobolev embedding
\begin{displaymath}
\|v\|_\infty \leq \sqrt{2} \,  \|v\|_2^\frac 12 \, \| v^\prime\|_2^\frac 12\,,
\end{displaymath}
we can conclude that 
\begin{displaymath}
   \|v\|_p \leq 2^{\frac {p-2}{2p}}\, \|v\|_2^{\frac{p+2}{2p}}
   \|v^\prime\|_2^{\frac{p-2}{2p}}\leq C\beta^{-\frac{3p+2}{6p}}\|f\|_2\,. 
\end{displaymath}
\end{proof}

For later reference we also need the following refined estimate.
\begin{proposition}
\label{prop5.3}  
For any $r>1$ and $\Upsilon<\Re\nu_1$, , there exist $C>0$ and $\beta_0>0$ such that, for all  $U\in \Sg_r$, 
   $\beta\geq \beta_0$, and $f\in L^\infty(-1,1)$, 
  \begin{equation}
\label{eq:126}
\sup_{\Re\lambda\leq\Upsilon  {\mathfrak J}_m^{2/3}\beta^{-1/3}} \|(\LL_\beta^\Df-\beta\lambda)^{-1}f\|_2\leq\frac{C}{\beta^{5/6}}\|f\|_\infty\,.
  \end{equation}
\end{proposition}
\begin{proof}
  Let $\eta$ be given by \eqref{eq:120} and for any $x\in[-1,1]\,,$

\begin{displaymath}
\eta_x(t)=  \eta(\beta^{1/3}(t-x)){\mathbf 1}_{[-1,1]}(t)\,.
\end{displaymath}
Let further
\begin{displaymath}
  \eta_{x,2}(t)=\eta(\beta^{1/3}(t-x)/2){\mathbf  1}_{[-1,1]}(t)\,,
\end{displaymath}
implying that
\begin{displaymath}
   \eta_x \,\eta_{x,2} =\eta_x\,.
\end{displaymath}
    
  {\em Step 1:} { \it We  prove that, for any  $\nu \in [U(-1),U(1)]$, there exist $C>0$
and $\beta_0$ such that, for all  $\beta \geq
\beta_0$, $x\in[-1,1]$, $f\in L^\infty(-1,1)$, $\Re \lambda \leq \Upsilon \beta^{-\frac 13}$,
\begin{equation}
\label{eq:127}
  \|(U-\nu)\eta_xv\|_2^2 \leq
  \frac{C}{\beta^2}\|\eta_xf\|_2^2+C\frac{\beta^{1/3}|x-x_\nu|+1}{\beta^{2/3}}\|\eta_{x,2}v\|_2^2\,,
\end{equation}
where $x_\nu\in [-1,+1]$ is chosen so that $U(x_\nu)=\nu$ and 
\begin{equation}
\label{defpairfv} 
v=(\LL_\beta^\Df-\beta\lambda)^{-1}f\,.
\end{equation}}
Clearly,
\begin{equation*}
  \Im\langle\eta_x^2(U-\nu)v,(\LL_\beta^\Df-\beta\lambda)v= \beta\|(U-\nu)\eta_xv\|_2^2 +
  \Im\langle[\eta_x^2(U-\nu)]^\prime v,v^\prime\rangle\,,
\end{equation*}
which implies
\begin{equation}
\label{eq:112a} 
  \|(U-\nu)\eta_xv\|_2^2 \leq   \frac{C}{\beta^2}\|\eta_xf\|_2^2+
  \frac{C}{\beta} \left(\sup_t |\eta^\prime_x(t) (U(t)-\nu)| \right) \|\eta_xv^\prime\|_2\|\eta_{x,2}v\|_2 \,.
\end{equation}
To estimate $\|\eta_xv^\prime\|_2$ we use the identity 
\begin{equation}
\label{eq:129}
   \Re\langle\eta_x^2v,(\LL_\beta^\Df-\beta\lambda)v\rangle= \|(\eta_xv)^\prime\|_2^2-\|\eta_x^\prime v\|_2^2-\mu\beta\|\eta_xv\|_2^2\,,
\end{equation}
from which we easily conclude, using the fact that $\beta^\frac 13 \mu$ is
bounded by assumption, that
\begin{equation}\label{eq:113a}
  \|\eta_xv^\prime\|_2^2\leq C(\beta^{-2/3}\|\eta_xf\|_2^2+\beta^{2/3}\|\eta_{x,2}v\|_2^2)\,.
\end{equation}
Finally we note that
\begin{equation}\label{eq:113b}
 C^{-1} |x-x_\nu| \leq|U(x)-\nu|\leq C|x-x_\nu| \,.
 \end{equation}
{\bf Suppose first that $|x-x_\nu | <3\beta^{-1/3}$.}  \\
Since by \eqref{eq:113b} it holds that $ \left(\sup_t |\eta^\prime_x(t)
  (U(t)-\nu)| \right) \leq C $  we may use 
(\ref{eq:112a})  to obtain
\begin{equation*}
  \|(U-\nu)\eta_xv\|_2^2 \leq   \frac{C}{\beta^2}\|\eta_xf\|_2^2+
  \frac{C}{\beta}\|\eta_xv^\prime\|_2\|\eta_{x,2}v\|_2 \,.
\end{equation*}
which together with \eqref{eq:113a} yields \eqref{eq:127}. \\
{\bf Suppose now that
$|x-x_\nu | \geq 3\beta^{-1/3}$.}\\
This time we have by \eqref{eq:113b} that $ \left(\sup_t |\eta^\prime_x(t)
  (U(t)-\nu)| \right) \leq C \beta^\frac 13 |x-x_\nu | $, and hence we get
 from (\ref{eq:112a}) 
\begin{equation}
\label{eq:130} 
  \|(U-\nu)\eta_xv\|_2^2 \leq   \frac{C}{\beta^2}\|\eta_xf\|_2^2+C\frac{|x-x_\nu | }{\beta^{2/3}}\|\eta_{x,2}v\|_2\|\eta_xv^\prime\|_2\,.
\end{equation}
To estimate  the last term on the right-hand-side of
\eqref{eq:130}, we use \eqref{eq:129} once again and get instead of \eqref{eq:113a}
\begin{equation}
\label{eq:106a}
  \|\eta_xv^\prime\|_2^2\leq C\Big(\beta^{2/3}\|\eta_{x,2}v\|_2^2 +
  \|(U-\nu)\eta_xv\|_2\Big\|\eta_x\frac{f}{U-\nu}\Big\|_2\Big)\,,
\end{equation}
or, alternatively, for any $\delta>0$\,,
\begin{displaymath}
   \|\eta_xv^\prime\|_2\leq C\Big(\beta^{1/3}\|\eta_{x,2}v\|_2 +
  \delta\frac{\beta^{1/2}}{|x-x_\nu|^{1/2}}\|(U-\nu)\eta_x v\|_2 +\frac{|x-x_\nu|^{1/2}}{\delta \beta^{1/2}} \Big\|\eta_x\frac{f}{U-\nu}\Big\|_2\Big)\,.
\end{displaymath}
Substituting the above into \eqref{eq:130} then yields
\begin{multline*}
   \|(U-\nu)\eta_xv\|_2^2 \leq
   C \, \Big(\frac{1}{\beta^2}\|\eta_xf\|_2^2
   +\frac{|x-x_\nu|}{\beta^{1/3}}\|\eta_{x,2}v\|_2^2 + \\
  + \frac{|x-x_\nu|^{3/2}}{\delta\beta^{7/6}}\|\eta_{x,2}v\|_2 \Big\|\eta_x\frac{f}{U-\nu}\Big\|_2
   + \frac{\delta |x-x_\nu|^{1/2}}{\beta^{1/6}}\|\eta_{x,2}v\|_2\|(U-\nu)\eta_xv\|_2 
  \Big)  \,.
\end{multline*}
Observing that
\begin{displaymath}
  \Big\|\eta_x\frac{f}{U-\nu}\Big\|_2 \leq C |x-x_\nu|^{-1}\Big\|\eta_x f \Big\|_2\,,
\end{displaymath}
yields \eqref{eq:127} by choosing a sufficiently small value of $\delta$.\\

To proceed to the next step, we need to define, yet, two
additional $\beta$ dependent cutoff functions. Let then for $s\geq 2$,
$\eta_s$ and $\tilde \eta_s$ in $C^\infty(\R,[0,1])$ satisfy
\begin{displaymath}
  \eta_s(t)=
  \begin{cases}
    1 & |t-x_\nu|\leq s\beta^{-1/3} \\
    0 & |t-x_\nu|\geq(s+1)\beta^{-1/3} 
  \end{cases}
  \mbox{ and } \tilde{\eta}_s=\sqrt{1-\eta_s^2}\,.
\end{displaymath}
We further require that there exists $C$ and $\beta_0$ such that for any $s\geq 2$ and $\beta \geq \beta_0$
\begin{displaymath}
  \|\eta_s^\prime\|_\infty \leq C\beta^{1/3} \quad ; \quad \|\eta_s^{\prime\prime}\|_\infty \leq C\beta^{2/3} \,.
\end{displaymath}

{\em Step 2:} {\it We prove that there exist $s_0>0$,   and  $C>0$ such that, for
all $s \geq s_0$,  there exists $\beta_s$ such that $\beta \geq \beta_s$
\begin{equation}
  \label{eq:131}
\|\tilde{\eta}_sv\|_2^2 \leq C( \beta^{-5/3}\|f\|_\infty^2 + s^{-1}\|v\|_2^2)\,,
\end{equation}
for any pair $(f,v)$ satisfying \eqref{defpairfv}.}\\
By \eqref{eq:127} we have 
\begin{equation}
\label{eq:132}
  \|\eta_xv\|_2^2 \leq
  \frac{C}{|x-x_\nu|^2\beta^2}\|\eta_xf\|_2^2+\frac{C}{|x-x_\nu|\beta^{1/3}}\|\eta_{x,2}v\|_2^2\,.
\end{equation}
We now integrate the above inequality with respect to $x$ over
$(-1,x_\nu-s\beta^{-1/3}/2) \cup (x_\nu +s\beta^{-1/3}/2,1)$. By changing the order of integration
we obtain that for all $s>4$
\begin{equation*}
\begin{array}{l}
  \int_{s\beta^{-1/3}/2<|x-x_\nu|}{\mathbf 1}_{[-1,1]}(x)\|\eta_xv\|_2^2\,dx 
  \\
\qquad = \int_{-1,1}|v|^2(t) \,dt\int_{s\beta^{-1/3}/2<|x-x_\nu|}{\mathbf
  1}_{[-1,1]}(x)\eta^2(\beta^{1/3}(x-t)) \,dx\\
\qquad  \geq   \int_{-1,1}|\tilde{\eta}_s v|^2(t) \,dt\int_{s\beta^{-1/3}/2<|x-x_\nu|}{\mathbf
  1}_{[-1,1]}(x)\eta^2(\beta^{1/3}(x-t)) \,dx \\
 \qquad  \geq \beta^{-1/3}\|\tilde{\eta}_sv\|_2^2 \,.
\end{array}
\end{equation*}
Note that 
\begin{displaymath}
  \int_{s\beta^{-1/3}/2<|x-x_\nu|}{\mathbf
  1}_{[-1,1]}(x)\eta^2(\beta^{1/3}(x-t)) \,dx=\beta^{-\frac 13}
\Big(\int_{s/2<|\tau+\beta^{1/3}(t-x_\nu)|}
\eta^2(\tau)d\tau \Big) \geq \beta^{-\frac 13}\,,
\end{displaymath}
for any $t$ in the support of $\tilde \eta_s$.
We rewrite the above in the form
\begin{equation}\label{eq:133}
\|\tilde{\eta}_sv\|_2^2 \leq \beta^{\frac 13}   \int_{s\beta^{-1/3}/2<|x-x_\nu|}{\mathbf 1}_{[-1,1]}(x)\|\eta_xv\|_2^2\,dx \,.
\end{equation}
 As $\|\eta_xf\|_2^2\leq 2\beta^{-1/3}\|f\|_\infty^2$ we have 
\begin{equation}
\label{eq:134}
\begin{array}{l}
  \int_{s\beta^{-1/3}/2<|x-x_\nu|} {\mathbf 1}_{[-1,1]}(x)
  \,\frac{1}{|x-x_\nu|^2\beta^2}\|\eta_xf\|_2^2\,dx \\
\qquad  \qquad \leq   C \beta^{-7/3}\|f\|_\infty^2\int_{s\beta^{-1/3}/2<|x-x_\nu|} {\mathbf 1}_{[-1,1]}(x)
  \,\frac{1}{|x-x_\nu|^2}\,dx \\ \qquad \qquad \leq \hat C \beta^{-2} \,\|f\|_\infty^2 \,.
\end{array}
\end{equation}
Finally, we have for all $s \geq16\,$,
\begin{displaymath}
\begin{array}{l}
   \int_{s\beta^{-1/3}/2<|x-x_\nu|}\, {\mathbf 1}_{[-1,1]}(x)\,
   \frac{1}{|x-x_\nu|\beta^{1/3}}\,  \|\eta_{x,2}v\|_2^2 \,dx \\ \qquad  \leq 
   \int_{-1,1}|v|^2(t) \,dt\int_{s\beta^{-1/3}/2<|x-x_\nu|}{\mathbf
  1}_{[-1,1]}(x)\frac{\eta^2(\beta^{1/3}(x-t)/2)}{|x-x_\nu|\beta^{1/3}} \,dx\\
\qquad  \leq  \beta^{-1/3} \int_{s\beta^{-1/3}/4<|t-x_\nu|}{\mathbf
    1}_{[-1,1]}(t)\,\log\,\frac{|t-x_\nu|+2\beta^{-1/3}}{|t-x_\nu|-2\beta^{-1/3}} \; 
  |v(t)|^2 \,dt \\ \qquad  
   \leq  \frac{C\beta^{-1/3}}{s}\,\|v\|_2^2 \,.  
    \end{array}
\end{displaymath}
Note that $\eta(\beta^{1/3}(x-t)/2)$ vanishes for $s\beta^{-1/3}/2<|x-x_\nu|$
(for $r\geq 16$) and  \break $|t-x_\nu|<s\beta^{-1/3}/4\,$. Note further that
\begin{displaymath}
  0<\log\,\frac{|t-x_\nu|+2\beta^{-1/3}}{|t-x_\nu|-2\beta^{-1/3}}\leq\log\,\frac{1+8/s}{1-8/s} \leq \frac{C}{s}\,.
\end{displaymath}
 Combining the above with \eqref{eq:134}, \eqref{eq:133},
and \eqref{eq:132} easily yields \eqref{eq:131}.

{\em Step 3:} We now prove \eqref{eq:126}. \\
Writing
\begin{displaymath}
  (\LL_\beta^\Df-\beta\lambda)(\eta_sv)=\eta_sf+ 2\eta_s^\prime v^\prime + \eta_s^{\prime\prime}v \,,
\end{displaymath}
we obtain from \eqref{eq:114} and the definitions and properties of the
cut-off functions $\eta_s$ and $\eta_x$, 
\begin{equation}
\label{eq:135}
  \|\eta_sv\|_2^2\leq C(\beta^{-4/3}\|\eta_sf\|_2^2 + \beta^{-2/3}\|\eta_{x_s}v^\prime\|_2^2+
\|\tilde{\eta}_{s/2}v\|_2^2)\,,
\end{equation}
where $x_s=x_\nu+ \frac{s+1}{2} \beta^{-1/3}$. \\ 
We now use \eqref{eq:106a} (with $x=x_s$) together with
\eqref{eq:127} to obtain that 
\begin{equation*}
\begin{array}{ll}
  \|\eta_{x_s}v^\prime\|_2^2&\leq C\Big(\beta^{2/3}\|\eta_{x_s,2}v\|_2^2 +
  \|(U-\nu)\eta_{x_s}v\|_2\Big\|\eta_{x_s}\frac{f}{U-\nu}\Big\|_2\Big)  \\ &\leq
\hat C\big(\beta^{2/3}\|\eta_{x_s,2}v\|_2^2 + s^{-1}\beta^{1/3} \big[\beta^{-1}\|\eta_{x_s}f\|_2 +
s\beta^{-1/3}\|\eta_{x_s,2}v\|_2\big] \|\eta_{x_s}f\|_2\big)\,,
\end{array}
\end{equation*}
from which we can conclude, using the inequality $\tilde{\eta}_{s/2}\geq
\eta_{x_s,2}$ for $s>16$, that
\begin{displaymath}
  \|\eta_{x_s}v^\prime\|_2^2\leq C(\beta^{2/3}\|\tilde{\eta}_{r/2}v\|_2^2+\beta^{-2/3}\|\eta_{x_s}f\|_2^2)\,.
\end{displaymath}
Substituting the above into \eqref{eq:135}, yields that, for some $C>0$, 
it holds:
\begin{displaymath}
  \|\eta_sv\|_2^2\leq C_s(\beta^{-5/3}\|f\|_\infty^2 + \|\tilde{\eta}_{s/2}v\|_2^2)\,.
\end{displaymath}
Using \eqref{eq:131} (used with $r$ replaced by   $\frac s 2$ and for $s$ large
enough) then yields
\begin{displaymath}
   \|\eta_sv\|_2^2\leq C(\beta^{-5/3}\|f\|_\infty^2 + s^{-1}\|v\|_2^2)\,.
\end{displaymath}
Combining the above with again \eqref{eq:131} yields \eqref{eq:126},
by choosing a sufficiently large value of $s$.
\end{proof}

\subsection{$L^1$ estimates}
In this subsection we establish new $L^1$ estimates for the resolvent
of $\widetilde\LL_{\beta,\R}$, defined by \eqref{eq:109} and
\eqref{eq:110}.  We first observe that the proof of
Proposition~\ref{prop5.3} can be applied to the entire real line case,
replacing the estimates of resolvent for the Dirichlet problem
$(\LL_\beta^\Df-\beta\lambda)^{-1}$ by the corresponding estimates of the
resolvent $(\widetilde{\LL}_{\beta,\R}-\beta\lambda)^{-1}$ for $\Re \lambda \leq
\beta^{-\frac 13} \Upsilon$. Hence, we may state the following
  \begin{lemma}
For any $r>1$, any $\Upsilon>0$ and $a>0$ there
  exists $C>0$ and $\beta_0>0$ such that, for all $U\in \Sg_r$, 
  $g\in L^2(\R)\cap L^\infty(\R)$ and $\beta \geq \beta_0$,
  \begin{equation}
    \label{eq:136}
\sup_{\Re\lambda\leq\Upsilon\beta^{-1/3}}\|(\widetilde{\LL}_{\beta,\R}-\beta\lambda)^{-1}g\|_{L^2(-a,a)} \leq
  \frac{C}{\beta^{5/6}} \|g\|_\infty \,.
  \end{equation}
\end{lemma}
 We continue this subsection with the following $L^1$ estimate 
\begin{lemma}
  \label{lem:entire-L1} 
 For any $r>1$ and $\Upsilon>0$,  there exist $C>0$ and $\beta_0>0$
such that,  for all $\beta\geq \beta_0$, $U\in \Sg_r^2$, $\Re\lambda\leq
\Upsilon  {\mathfrak J}_m^{2/3}\beta^{-1/3}$, and any $g\in L^\infty(\R)$  supported in $[-1,1]$,  we have
\begin{equation}
  \label{eq:137}
\|  (\widetilde\LL_{\beta,\R}-\beta\lambda)^{-1} g \|_{L^1(-1,1)} \leq C \min(\beta^{-5/6}\|g\|_2,\beta^{-1}\log\beta\|g\|_\infty) \,.
\end{equation}
\end{lemma}
\begin{proof}
  Let $\Upsilon$ denote a fixed positive value and $\Re\lambda\leq \Upsilon\beta^{-1/3}$. Let
  further  $g\in L^\infty(\R)$, supported on $[-1,1]$, and
  $v\in D(\widetilde {\mathcal L} _{\beta,\R})$ satisfy
\begin{equation}\label{eq:138}
(\widetilde\LL_{\beta,\R}-\beta\lambda) v=g\,.
\end{equation} 
By  \eqref{eq:111}-\eqref{eq:112} and \eqref{eq:65}, we have 
  \begin{displaymath}
   \|v\|_{L^1(-1,1)} \leq \|(\widetilde U-\nu+i\beta^{-1/3})^{-1}\|_{L^2(-1,+1)}
   \|(\widetilde U-\nu+i\beta^{-1/3})v\|_2 \leq C\beta^{-5/6}\|g\|_2 \,,
  \end{displaymath}
which  proves the first inequality of \eqref{eq:137}.

Let $x_\nu\in\R$ satisfy $\widetilde U(x_\nu)=\nu$.  Let $ \eta$ and $\eta_\nu$ be given by \eqref{eq:120} and 
\eqref{eq:121}  and 
\begin{displaymath}
  \zeta_\nu(x)= \eta_\nu(x)\,\eta (|x-x_\nu|/6)\,.
\end{displaymath}
Taking the inner product of \eqref{eq:138} with $\zeta_\nu^2v$ yields for
the imaginary part (see \eqref{eq:107a} with $\eta_\nu$ instead of $\zeta_\nu$)
\begin{equation}
\label{eq:139}
  \beta\|\zeta_\nu| \widetilde U -\nu|^{1/2}v\|_2^2 + 2\Im\langle\zeta_\nu^\prime v,(\zeta_\nu v)^\prime\rangle =
  \langle\zeta_\nu v,\zeta_\nu g\rangle\,.
\end{equation}
By \eqref{eq:136} we have that
\begin{equation}\label{eq:299a}
  \|\zeta_\nu^\prime v \|_2\leq \frac{C}{\beta^{1/2}}\|g\|_\infty \,.
\end{equation}
To estimate $\|(\zeta_\nu v)^\prime\|_2$ we use the identity
\begin{displaymath}
  \|(\zeta_\nu v)^\prime\|_2^2-\Re\lambda\, \beta\, \|\zeta_\nu v \|_2^2 -\|\zeta_\nu^\prime v \|_2^2 =
  \Re \langle\zeta_\nu v,\zeta_\nu g\rangle\,,
\end{displaymath}
to obtain from \eqref{eq:111}, \eqref{eq:136}, and \eqref{eq:299a}
\begin{displaymath}
  \|(\zeta_\nu v)^\prime\|_2^2\leq \frac{C}{\beta}\|g\|_\infty^2+ \|\zeta_\nu v \|_1\|g\|_\infty \,.
\end{displaymath}
Substituting the above into \eqref{eq:139} yields
\begin{equation}\label{eq:300a}
  \|\zeta_\nu|\widetilde U-\nu|^{1/2}v \|_2 \leq C \Big(\frac{1}{\beta}\|g\|_\infty +
  \beta^{-1/2}\|\zeta_\nu v \|_1^{1/2}\|g\|_\infty^{1/2} \Big) \,.
\end{equation}
As
\begin{equation}
\label{eq:300b}
\begin{array}{ll}
 \|\zeta_\nu v \|_1& \leq \|\zeta_\nu|\widetilde U-\nu|^{1/2} v \|_2\|{\mathbf
   1}_{(x_\nu+\frac 12 \beta^{-1/3},x_\nu+6)} | \widetilde U -\nu|^{-1/2}\|_2\\[1.5ex] & \leq C(\log\beta)^{\frac 12}\,  \|\zeta_\nu|\widetilde U -\nu|^{1/2} v \|_2\,,
\end{array}
\end{equation}
we obtain from \eqref{eq:300a} that
\begin{displaymath}
   \|\zeta_\nu|\widetilde U -\nu|^{1/2}v \|_2 \leq \widetilde C\frac{(\log \beta)^{\frac 12}}{\beta} \|g\|_\infty \,.
\end{displaymath}
Then, using \eqref{eq:300b} once again yields
\begin{displaymath}
 \|\zeta_\nu v \|_1 \leq \widehat C\frac{\log \beta}{\beta} \|g\|_\infty \,,
\end{displaymath}
from which we easily conclude that
\begin{displaymath}
  \| {\mathbf 1}_{(x_\nu,x_\nu+3)}v \|_1 \leq
  \|\zeta_\nu v \|_1 + C\beta^{-1/6}\|{\mathbf
    1}_{(x_\nu,x_\nu+\beta^{-1/3})}v \|_2 \leq \check C\, \frac{\log \beta}{\beta} \|g\|_\infty \,.
\end{displaymath}
In a similar manner we obtain that
\begin{displaymath}
  \| {\mathbf 1}_{(x_\nu-3,x_\nu)}v \|_1 \leq \check C\, \frac{\log \beta}{\beta} \|g\|_\infty \,.
\end{displaymath}
This proves \eqref{eq:137} in the case where $(-1,+1) \subset (x_\nu-3,x_\nu +3)$.\\

It remains to prove \eqref{eq:137} in the case $|x_\nu| \geq 2\,$. By
\eqref{eq:112} and the fact that $g$ is supported on $[-1,1]$, we have
that
\begin{displaymath}
  \| v \|_{L^1(-1,1)}\leq  \sqrt{2}\, \| v \|_{L^2(-1,1)}\leq
  C\|(\widetilde U -\nu) v \|_{L^2(-1,1)} \leq \frac{C}{\beta} \|g\|_2 \leq  \frac{C}{\beta}
  \|g\|_\infty \,.
\end{displaymath}
The lemma is proved
\end{proof}

 A similar statement can be proved in the Dirichlet case.
\begin{lemma}
  \label{lem:Dirichlet-L1} 
For any $r>1$ and $\Upsilon\in (0,\Re\nu_1)$,  there exist $C>0$ and $\beta_0>0$
such that, for all $\beta\geq \beta_0$, $U\in \Sg_r^2$  and $\Re\lambda\leq
 {\mathfrak J}_m^{2/3}\Upsilon\beta^{-1/3}$,  and for any $g\in L^\infty(-1,1)$  we have
\begin{equation}
  \label{eq:296d} 
  \| (\LL_\beta^\Df-\beta\lambda)^{-1} g  \|_1 \leq C \min(\beta^{-5/6}\|g\|_2,\beta^{-1}\log\beta\|g\|_\infty) \,.
\end{equation}
\end{lemma}
The proof is similar to the proof of Lemma \ref{lem:entire-L1} and is
therefore skipped.

\section{No-slip resolvent estimates}\label{s6}
\subsection{A no-slip Schr\"odinger operator}
\label{sec:no-slip-schrodinger}
We begin by providing a short explanation of the difficulties arising
when the no-slip boundary condition \eqref{eq:10v} is prescribed. Complete
details will be given in  Section \ref{sec:no-slip}.\\ 
In the zero-traction case,  estimating $\phi\in
D(\B_{\lambda,\alpha,\beta}^\Sf)$ satisfying $\B_{\lambda,\alpha,\beta}^\Sf \, \phi=f$ for some
$f\in L^2(-1,1)$,  we may write, by \eqref{eq:28},
\begin{displaymath}
  \Big(-\frac{d^2}{dx^2} + i \beta U-\beta\lambda)\Big)(-\phi^{\prime\prime}+\alpha^2\phi)= i\beta
  U^{\prime\prime}\phi + f\,.
\end{displaymath}
Since by \eqref{eq:31}, it holds that $-\phi^{\prime\prime}+\alpha^2\phi $ satisfies a
Dirichlet condition at $\pm 1$, one can now use, for instance,
\eqref{eq:114} and \eqref{eq:126} to obtain
\begin{displaymath}
  \|-\phi^{\prime\prime}+\alpha^2\phi\|_2 \leq C\beta^{1/6}\|\phi\|_\infty + \beta^{-2/3}\|f\|_2 \,.
\end{displaymath}
Such an estimate is particularly useful in the case $\alpha\gg\beta^{1/6}$,
but also in other cases  (detailed in Section \ref{sec:no-slip}).\\
 Similar estimates can
be obtained for $v=\A_{\lambda,\alpha}\phi$ and $\tilde{v}=(U+i\lambda)^{-1}v$. 

If we now consider the same problem in the no-slip case the above
approach is inapplicable. Thus, for $\phi\in D(\B_{\lambda,\alpha,\beta}^\Df)$ satisfying
$\B_{\lambda,\alpha,\beta}^\Df\,\phi=f$, we can no longer use neither \eqref{eq:114}
nor \eqref{eq:126}, as $-\phi^{\prime\prime}+\alpha^2\phi$ does not satisfy a
Dirichlet condition at $x=\pm1$. However,
integration by parts easily yields that for all $\phi\in
D(\B_{\lambda,\alpha,\beta}^\Df)$ we have 
\begin{displaymath}
  \langle e^{\pm\alpha x},-\phi^{\prime\prime}+\alpha^2\phi\rangle=0\,.
\end{displaymath}
If we consider $v$ or $\tilde{v}$ instead of $-\phi^{\prime\prime}+\alpha^2\phi$ we
can still obtain similar orthogonality conditions (see  \eqref{eq:264} and 
  \eqref{eq:268}). These conditions read 
  \begin{equation}
\label{eq:140}
    \langle v \,,\, \zeta_{\pm}\rangle_{L^2(-1,+1)}=0 \,,
  \end{equation}
  where $\zeta_{-}$ and $\zeta_+$ are linearly independent, $\beta$ dependent,
 and  belong to $H^1(-1,+1)$.

 With \eqref{eq:28} in mind, we let $\LL_\beta^\zeta$ be the differential
 operator $-d^2/dx^2 + i \beta U$ with domain
  \begin{equation}
\label{eq:141}
    D(\LL_\beta^\zeta)= \{ u\in H^2(-1,1)\,| \, \langle\zeta_\pm ,u\rangle=0\,\} \,.
  \end{equation}
For convenience we require that $\zeta_\pm$ satisfy 
  \begin{equation}
\label{eq:142}   
    \begin{bmatrix}
      \zeta_+(1) & \zeta_-(1) \\
      \zeta_+(-1) & \zeta_-(-1)
    \end{bmatrix}
=
\begin{bmatrix}
  1 & 0 \\
  0 & 1
\end{bmatrix}
\,.
  \end{equation}  
  Note that $\zeta_\pm=e^{\pm\alpha x}$ do not satisfy the above requirement, and
  we shall therefore need to replace them by a pair of proper linear
  combinations of them \cite{sh04} (a more detailed explanation is brought
    in Section \ref{sec:no-slip}).  We seek resolvent estimates for
  $\LL_\beta^\zeta$ in the following.  In the absence of a Dirichlet
  boundary condition, it seems reasonable to approximate the solution
  of
\begin{equation}
\label{eq:5.31}
 (\LL_\beta^\zeta -\beta \lambda) v=g \,,
\end{equation}
by a sum of a solution in $\R$ of the inhomogeneous equation and a
linear combination of two independent approximate solutions of the
homogeneous equation whose coefficients will be
determined by the above integral conditions. Using affine approximations of $U$
in $(-1,+1)$ or extensions outside of $(-1,+1)$, the approximate
solutions can be described by a pair of dilated and translated Airy
functions in $(-1,+\infty)$ and $(-\infty,+1)$.   

{\bf The solution in $\mathbb R$.\\}
We now explain  our construction of an approximate inverse in $\mathbb
R$ by defining first a natural $C^1$-extension $\tilde U$ of $U$
outside of $[-1,+1]$, satisfying \eqref{asstildeU},  
by 
\begin{displaymath}
  \tilde{U}(x)=
  \begin{cases}
    U(x) & x\in[-1,1] \\
    U(1)+U^\prime(1)(x-1) & x>1 \\
   U(-1)+U^\prime(-1)(x+1) & x<-1\,.
  \end{cases}
\end{displaymath}
We note that $\tilde U$ satisfies the conditions of Proposition
\ref{lem:model-entire}.
We also extend $g$ by 
\begin{displaymath}
  \tilde{g}(x) =
  \begin{cases}
    g(x) & x\in[-1,1] \\
    0 & \text{otherwise}\,,
  \end{cases}
\end{displaymath}
and set 
\begin{equation}
\label{eq:143}
  u = \Gamma_{(-1,1)}(\widetilde{\LL}_{\beta,\R}-\beta\lambda)^{-1}\tilde{g}\,.
\end{equation}
where $\widetilde{\LL}_{\beta,\R}$ is defined by \eqref{eq:109} and \eqref{eq:110} and $
\Gamma_{(-1,1)}$ denotes the restriction to $(-1,1)$. \\
 
{\bf Boundary terms.}\\
To obtain the boundaries effect, we replace $U(x)$ by its affine
approximation at $\pm 1$ and consider the $L^2\cap L^1$ solutions $\psi_\pm$
of the approximate problems
  \begin{equation}\label{defpsi+}
      \left\{
      \begin{array}{l}
              (-d^2/dx^2 + i \beta[U(- 1) +  J_- (x+1)] -\beta\lambda)\psi_-  =0 \mbox{ in } (-1,+\infty)   \\
              \int_{-1}^{+\infty} \psi_- (x) \,dx =  (J_- \beta)^{-1/3} \,,
      \end{array}
      \right.
    \end{equation}
    and 
    \begin{equation}\label{defpsi-}
    \left\{
    \begin{array}{l}
              (-d^2/dx^2 + i \beta[U(1) + J_+(x-1))] -\beta\lambda)\psi_+  =0  \mbox{ in } (-\infty,1)   \\
              \int_{ -\infty}^1  \psi_+(x) \,dx =  (J_+\beta)^{-1/3} \,,
              \end{array}
              \right.
    \end{equation}
    with 
\begin{displaymath}
J_\pm = U^\prime(\pm 1).
\end{displaymath}

    Except, perhaps for some particular values of $\lambda$, the above
    solutions are unique, and $\psi_\pm$ rapidly decays as $x\to\mp \infty$,
    but their existence (due to the additional integral condition)
    could depend on $(\beta,\lambda, J_\pm)$ as is clarified below.  We express
    $\psi_\pm$ using Airy functions.  Having in mind the definition
  of the generalized Airy functions \cite[eq.  (39)]{wa53} or
  \cite[Lemma 2]{ro73} (for more details see \cite[Appendix]{drre04}
  or our short review in Appendix A).
  These solutions are given, assuming that the denominator does
    not vanish, by
    \begin{subequations}
\label{eq:132a} 
    \begin{equation} 
\psi_-(x)= e^{i \pi/ 6} \, \frac{{\rm Ai}\big((J_- \beta)^{1/3}e^{ i\pi/6}\big[(1+x)+iJ_-^{-1}(\lambda-iU(-1))\big]\big)}
{A_0\big(i\beta^{1/3}J_-^{-2/3}[\lambda-iU(-1)]\big)}\,,
\end{equation}
and 
\begin{equation}
\overline{\psi_+ (x)}=-   e^{ i \pi/6}\frac{{\rm Ai}\big((J_+ \beta)^{1/3}e^{ i\pi/6}\big[(1- x)+ iJ_+ ^{-1}(\bar \lambda+ i U( 1))\big]\big)}
{A_0\big( i \beta^{1/3}J_+ ^{-2/3}[\bar \lambda+iU(1)]\big)}\,.
\end{equation}
\end{subequations}
where  $A_0$ is the holomorphic extension to $\mathbb C$ of 
\begin{equation}
\label{eq:144}
 x \mapsto  A_0(x)=e^{i\pi/6}\int_x^{+\infty}\Ai(e^{i\pi/6}t)\,dt\,.
\end{equation}
Much of the properties of $A_0$ are recalled (mainly from Wasow's
paper \cite{wa53}) in Appendix \ref{AppA}.  It has been established in
\cite{wa53} (see also Appendix \ref{AppA}) that The zeroes of $z \mapsto
A_0 (i z)$ are located in the sector $\arg z\in(\frac \pi 6,\frac \pi
2)$.  Let
\begin{displaymath}
  \Sg_\lambda =\{\,z\,|\,A_0 (i z)=0\,\}\,, 
\end{displaymath}
and  further define
\begin{equation}\label{deftheta1r}
\vartheta_1^r:=  \inf_{z\in\Sg_\lambda} \Re z\,.
\end{equation}
In addition, we prove  in Appendix  \ref{sec:definition-a_0-locus}
that $\Sg_\lambda\neq\emptyset$ and(relying on \cite{wa53}) that
$
\vartheta_1^r >0\,.
$
It follows that  the denominators in (\ref{eq:132a}b) and (\ref{eq:132a}a) do not vanish, 
if 
\begin{displaymath}
\Re\lambda<\vartheta_1^r\, \beta^{-1/3}\,  {\mathfrak J}_m^{2/3}\,,
\end{displaymath}
where $ {\mathfrak J}_m$ is given by (\ref{eq:20}c).

The functions $\psi_\pm$ are not exact solutions of $(-d^2/dx^2+i\beta U -\beta \lambda)\psi=0$
and hence we must introduce a correction term. We thus consider
\begin{equation}
\label{eq:145}
  g_\pm  =
  \begin{cases}
    (-\frac{d^2}{dx^2} +i\beta(U+i\lambda))\psi_\pm    &\mbox{ for }  x\in (-1,1) \\
    0 & \text{otherwise}\,,
  \end{cases}
\end{equation}
and then introduce 
\begin{equation}
\label{eq:146}
  \tilde{v}_\pm  = \Gamma_{(-1,1)} (\widetilde{\LL}_{\beta,\R}-\beta\lambda)^{-1}g_\pm  \,.
\end{equation}
This correction term can be estimated as follows
\begin{lemma} \label{lemma6.1}
For  any $r>1$ and $ \Upsilon <\vartheta_1^r$, there exist  $C$ and $\beta_0$ such that, 
 for all $U\in \Sg_r^2$,  $\lambda\in \C$ satisfying 
\begin{equation}
  \label{eq:147}
\beta^{1/3}\Re\lambda\leq\Upsilon  {\mathfrak J}_m^{2/3} \,,
\end{equation}
and  $\beta \geq \beta_0$,   we have
\begin{equation}
  \label{eq:148}
\|(U+i\lambda)\tilde{v}_\pm \|_2 +\beta^{-1/3}\|\tilde{v}_\pm \|_2\leq C\beta^{-5/6} \,.
\end{equation}
\end{lemma}
\begin{proof}~\\
A simple computation shows that: 
    \begin{displaymath} 
g_\pm = i\beta\, [U- U(\pm 1) -J_\pm(x\mp 1))]\, \psi_\pm \mbox{ in }
      (-1,1)\,.
\end{displaymath}
Let
\begin{equation}\label{eq:deflambdapm}
  \lambda_\pm = \mu - i(U(\pm 1)-\nu)\,.
\end{equation}
We note that
\begin{equation}\label{eq:expsup}
\psi_\pm (x) = e^{i\frac \pi 6} \Psi_{(J_\pm\beta)^{\frac 13} J_\pm^{-1} \lambda_\pm}((J_\pm\beta)^{\frac 13} (1 \mp x))\,
\end{equation}
where $\Psi_\lambda$ is defined in the appendix (see \eqref{eq:362}).\\
Using translation, dilation, and (\ref{eq:363}a), we can conclude 
that, under the assumptions of Lemma \ref{lemma6.1}, it holds, for
$k\in[0,4]$, that 
\begin{equation}
\label{eq:149}
  \|(1\mp x)^k\, \psi_\pm \|_2 \leq C\, [1+|\lambda_\pm|\beta^{1/3}]^{\frac{1-2k}{4}} \beta^{-(1+2k) /6}\,.
\end{equation}
Hence, as 
\begin{displaymath}
  |g_\pm(x)|\leq C\beta(1\mp x)^2 |\psi_\pm(x)|\,,
\end{displaymath}
we have 
\begin{equation}
\label{eq:150}
  \|g_\pm \|_2\leq C\, \beta^{1/6} [1+|\lambda_\pm|\beta^{1/3}]^{-3/4}\,. 
\end{equation}
By  \eqref{eq:111} and  \eqref{eq:112}  we have 
\begin{equation}
  \label{eq:151}
\|(U+i\lambda)\tilde{v}_\pm \|_2 +\beta^{-1/3}\|\tilde{v}_\pm \|_2\leq
C\beta^{-5/6}\, [1+|\lambda_\pm|\beta^{1/3}]^{-3/4} \,,
\end{equation}
establishing thereby \eqref{eq:148}. \\
\end{proof}
We are now ready for introducing  a solution of \eqref{eq:5.31} in the form
\begin{equation}
  \label{eq:152}
v = A_+(\psi_+-\tilde{v}_+) + A_-(\psi_--\tilde{v}_-) + u \,.
\end{equation}
We observe that $(\LL_\beta -\beta \lambda) v =g$ for any pair $(A_-,A_+)$.
Therefore, one can attempt to find two linear forms $g \mapsto A_-(g)$ and
$g \mapsto A_+(g)$ such that $v$ belongs to the domain of $\LL_\beta^\zeta\,$,
hence is the solution of \eqref{eq:5.31}. This is the object of the
next lemma.

\begin{lemma}
\label{lem:integral-conditions}
Let  $\theta>0$ and $C_\zeta>0$, and suppose that
$\zeta_-,\zeta_+\in H^1(-1,1)$ satisfy the conditions  \eqref{eq:142},
\begin{equation} 
\label{condzeta1}
\|\zeta_\pm \|_\infty\leq C_\zeta\,,
\end{equation}
and
\begin{equation} \label{condzeta2}
\|\zeta_\pm ^\prime\|_2\leq\theta\, \beta^{1/6}\,.
\end{equation}
Let further $r>1$ and $\Upsilon <\vartheta_1^r\,$. Then, there exist $\beta_0>0$ and
$\theta_0>0$ such that for all $\beta\geq \beta_0$, $0< \theta\leq \theta_0$, $U\in \Sg_r^2$ and $\lambda\in
\C$ satisfying
\begin{equation}
  \label{eq:153}
\beta^{1/3}\Re\lambda\leq  {\mathfrak J}_m^{2/3}\Upsilon \,,
\end{equation}
\eqref{eq:5.31} and \eqref{eq:152}   hold true with $v\in D(\mathcal
L_\zeta^\beta)$ and  $A_\pm= A_\pm (g)$ denoting a pair of linear 
forms  $A_\pm (g):L^\infty(-1,1)\to\C$.
 Furthermore, there exists $C>0$ 
such that, for all $\beta \geq \beta_0$, $U\in \Sg_r^2$ and $g \in L^\infty(-1,+1)$,  we have
\begin{equation}
  \label{eq:154}
|A_\pm (g) |\leq C \min(\beta^{-1/2}\|g\|_2,\beta^{-2/3}\log\beta \,  \|g\|_\infty)\,.
\end{equation}
\end{lemma}
\begin{proof}
 In view of the discussion preceding the statement of the lemma, it
  remains to show the existence of $A_\pm (g)$ satisfying \eqref{eq:154}. 

Taking the inner product of \eqref{eq:152} in $L^2(-1,+1)$, first by $\zeta_+$ and then by
$\zeta_-$ while having \eqref{eq:140}  in mind yields the following system
\begin{equation}
\label{eq:156}
      \begin{bmatrix}
      \langle\zeta_+,(\psi_+-\tilde{v}_+)\rangle &   \langle\zeta_+,(\psi_--\tilde{v}_-)\rangle \\
       \langle\zeta_-,(\psi_+-\tilde{v}_+)\rangle &   \langle\zeta_-,(\psi_--\tilde{v}_-)\rangle
    \end{bmatrix}
\begin{bmatrix}
  A_+ (g)\\
  A_- (g)
\end{bmatrix}
=
\begin{bmatrix}
  \langle\zeta_+,u\rangle \\ 
  \langle\zeta_-,u\rangle
\end{bmatrix}\,.
\end{equation}
We now write
\begin{equation}
\label{eq:157}
  \langle\zeta_\pm ,\psi_\pm \rangle = \langle1,\psi_\pm \rangle+\langle\zeta_\pm -1,\psi_\pm \rangle  \,.
\end{equation}
For the first term on the right-hand-side we have 
\begin{displaymath} 
  \langle1,\psi_-\rangle=(J_- \beta)^{-1/3}-\int_1^\infty\psi_-(x)\,dx \,.
\end{displaymath}
The integral on the other side can be estimated as follows: we first write
\begin{displaymath} 
 \Big| \int_1^\infty\psi_-(x)\,dx\Big|   \leq\|(1+x)^3\psi_-\|_1 \,.
\end{displaymath}
Then, using (\ref{eq:363}b), \eqref{eq:expsup} and dilation, we
obtain for all $s\in [0,3]$,
\begin{equation} \label{eq:148a}
\|(1\mp x)^s \psi_\pm \|_1 \leq C\,  [1+|\lambda_\pm|\beta^{1/3}]^{-s/2} \, \beta^{-(s+1)/3} \,.
\end{equation}
The above estimate for $s=3$ yields,
\begin{displaymath} 
 \Big| \int_1^\infty\psi_-(x)\,dx\Big|   \leq C\beta^{-4/3}\,.
\end{displaymath}
A similar estimate can be obtained for $ \langle1,\psi_+\rangle$. Consequently we
have
\begin{equation}
\label{eq:158}
 \langle1,\psi_\pm\rangle =(J_\pm \beta)^{-1/3}[1+\OO(\beta^{-1})] \,.
\end{equation}
For the second term on the right-hand-side of \eqref{eq:157}  we use the fact that,
for all $x\in[-1,1]$, we have by \eqref{eq:142}
\begin{displaymath} 
  |\zeta_\pm(x)-1|\leq [1\mp x]^{1/2} \|\zeta^\prime_\pm \|_2\,.
\end{displaymath}
We obtain, using \eqref{eq:148a} with $s=\frac 12$ and \eqref{condzeta2}, 
\begin{equation}
\label{eq:159}
  |\langle\zeta_\pm -1, \psi_\pm \rangle| \leq
  \|\zeta^\prime_\pm \|_2\|[1\mp x]^{1/2}\psi_\pm \|_1\leq C(\Upsilon)\theta_0\beta^{-1/3} \,.
\end{equation}
Furthermore, by  \eqref{eq:64}  and  \eqref{eq:151},  we have 
\begin{equation}
\label{eq:160}
\begin{array}{ll}
  |\langle\zeta_\pm  ,\tilde{v}_\pm \rangle| & \leq 2 
  \|(U-\nu+i\beta^{-1/3})^{-1}\|_2\,\|(U-\nu+ i\beta^{-1/3})\tilde{v}_\pm \|_2\\ & 
  \leq C\, \beta^{-2/3}\, [1+|\lambda_\pm|\beta^{1/3}]^{-3/4} \,.
  \end{array}
\end{equation}
By the above, \eqref{eq:158}, and \eqref{eq:159}, there exists $C>0$
and $\beta_0$, such that, for any $\beta \geq \beta_0$ and any $\lambda$ satisfying
$\Re \lambda \leq \Upsilon\beta^{-1/3}$, we have
\begin{equation}
  \label{eq:161}
|\langle\zeta_\pm , (\psi_\pm -\tilde{v}_\pm ,)\rangle-(J_\pm \beta)^{-1/3}|\leq C(\theta_0\, \beta^{-1/3}+\beta^{-2/3}) \,. 
\end{equation}
As $\zeta_\pm (\mp1)=0$,  we obtain as in \eqref{eq:159}
\begin{equation}
\label{eq:162}
   |\langle\zeta_\pm , \psi_\mp\rangle| \leq   C\, \theta_0\, \beta^{-1/3} \,.
\end{equation}
Furthermore, as in \eqref{eq:160} we obtain that 
\begin{equation}
\label{eq:163} 
  |\langle\zeta_\pm  ,\tilde{v}_\mp\rangle| \leq C\, \beta^{-2/3}\, [1+|\lambda_\pm|\beta^{1/3}]^{-3/4} \,.
\end{equation}
Substituting the above, together with \eqref{eq:161} and \eqref{eq:162}
into \eqref{eq:156} then yields, for $\theta_0$ small enough,  and $\beta$ large enough, the invertibility of \eqref{eq:156} together with the estimate
\begin{equation}
\label{eq:164}
  |A_\pm(g) |\leq  |\langle\zeta_\pm ,u\rangle|
  (J_\pm \beta)^{1/3}[1+C\theta_0]+  
C\theta_0\beta^{-1/3} |\langle\zeta_\mp ,u\rangle| \,.  
\end{equation}
By \eqref{eq:137}  
we obtain that
\begin{displaymath}
  |\langle\zeta_\pm ,u\rangle|\leq C\min(\beta^{-1/2}\|g\|_2,\beta^{-2/3}\log\beta\|g\|_\infty)\,.
\end{displaymath}
Substituting the above into \eqref{eq:164} yields \eqref{eq:154}. 
\end{proof}

\begin{remark}
\label{rem:dirichlet-estimate}
We may replace in Lemma \ref{lem:integral-conditions},
$\widetilde{\LL}_{\beta,\R}$ by $\LL_\beta^\Df$ in \eqref{eq:143} and
\eqref{eq:146}, but under the condition $\Upsilon < {\mathfrak J}_m^{2/3}\Re\nu_1$. In this
case, we have by (\ref{eq:152}) (with $\LL_\beta^\Df$ instead of
$\widetilde{\LL}_{\beta,\R}$) and having in mind the Dirichlet condition
at $\pm 1\,$,
\begin{displaymath}
  v(\pm1) = A_\pm\psi_\pm(\pm1)+A_\mp\psi_\mp(\pm1)\,.
\end{displaymath}
By (\ref{eq:355}) we have 
\begin{displaymath}
  |\psi_\pm(\pm1)| \leq C[1+|\lambda_\pm|^{1/2}\beta^{1/6}]\,,
\end{displaymath}
and  by  (\ref{eq:363}c) 
\begin{displaymath}
   |\psi_\mp(\pm1)| \leq  C\beta^{-4/3}[1+|\lambda_\pm|\beta^{1/3}]^{-2}\,.
\end{displaymath}
Combining the above with \eqref{eq:154} yields that
\begin{equation}
\label{eq:33}
  |v(\pm1)|\leq C\, [1+|\lambda_\pm|^{1/2}\beta^{1/6}] \, \min(\beta^{-1/2}\|g\|_2,\beta^{-2/3}\log\beta \,  \|g\|_\infty)\,.
\end{equation}
\end{remark}

\subsection{A no-slip Schr\"odinger in $\R_+$}
\label{sec:no-slip-schrodinger-1}

In the previous subsection we have considered a space of functions
satisfying the orthogonality condition \eqref{eq:140}. We have assumed
that the functions spanning the orthogonal space $\zeta_+$ and $\zeta_-$
satisfy the bound 
\begin{displaymath}
\|\zeta_\pm^\prime\|_2\leq\theta \beta^{1/6}\,,
\end{displaymath} 
where $\theta \in (0,\theta_0]$ for some sufficiently small
$\theta_0>0$\, .

 Of particular interest is the example $\zeta_\pm=e^{-\alpha(1\mp x)}$ (or
a proper linear combination of them satisfying \eqref{eq:142}).  In
this case, we have \begin{displaymath}
\|(e^{-\alpha(1\mp x)})^\prime\|_2\cong\sqrt{\alpha/2} \,,
\end{displaymath} for
sufficiently large $\alpha$.\\
 Consequently, as long as $\alpha\ll\beta^{1/3}$,
Lemma \ref{lem:integral-conditions} is applicable in this case.
We, however, need to consider also the case where $\alpha\sim\beta^{1/3}$, or
even $\alpha\gg\beta^{1/3}$. These cases can, nevertheless, be treated using
localization techniques as in \cite{hen14,AGH}. To this end we have to
consider a localized version of $\LL_\zeta$ near $x=\pm1$. This subsection
is devoted therefore to the study of the ensuing linearized operator.

We begin by establishing a proper spectral formulation for the
problem. 
\begin{proposition}
  \label{lem:semi-infinite-spectrum}
Let, for some $\theta>0$,  
\begin{subequations}
\label{eq:165}
  \begin{equation}
    \LL^\theta= -\frac{d^2}{dx^2 }+ix\,, 
  \end{equation}
be defined on
\begin{equation}
  D( \LL^\theta)= \{\,u\in H^2(\R_+)\,| \, \langle e_\theta,u\rangle=0 \,,\;
  xu\in L^2(\R_+)\,\} \,.
\end{equation}
where 
\begin{displaymath}
e_\theta(x):= e^{-\theta x} \,.
\end{displaymath}
\end{subequations}
Then, $ \LL^\theta$ is a closed operator with non empty resolvent set and compact resolvent.\\
Moreover $\LL^\theta$ has index $0$.
\end{proposition}

Before proceeding to the proof of the proposition we establish the
following $H^1$ estimate of any $v\in D(\LL^\theta)$ in term of the
$L^2$-norm of $\LL^\theta v$ and $v$. 
 \begin{lemma}
 There exists some constant $C(\theta)$ such that, for any $\lambda \in \mathbb C$ and any $v\in D(\LL^\theta)$,  we have
 \begin{equation}
\label{eq:166}
  \|v^\prime\|_2 + |v(0)| \leq C(\theta) [(1+|\Re \lambda|^{1/2}\sign \Re \lambda )\|v\|_2 + \|( \LL^\theta-\lambda)v\|_2]\,. 
\end{equation}
 \end{lemma}
 \begin{proof}
  Let $v\in D( \LL^\theta)$ and $g\in
  L^2(\R_+)$ satisfy
  \begin{equation}
    \label{eq:167}
( \LL^\theta-\lambda)v=g \,.
  \end{equation}
Taking the inner product with $v$ yields
\begin{equation}
  \label{eq:168}
 \|v^\prime\|_2^2 -\lambda\|v\|_2^2+i\langle xv,v\rangle +v^\prime(0) \bar{v}(0)=\langle v,g\rangle\,.
\end{equation} 
To obtain an estimate for the fourth  term on the left-hand-side of
\eqref{eq:168} we need an effective bound on $v^\prime(0)$.\\
 Integration by parts yields, with the
aid of the fact that $\langle e_\theta ,v\rangle=0\,$,
\begin{displaymath}
  -\langle e_\theta ,v^{\prime\prime}\rangle=v^\prime(0)+\theta \, v(0) \,.
\end{displaymath}
Taking the inner product of
\eqref{eq:167} with $e_\theta$, 
we then obtain 
\begin{equation}
\label{eq:169}
  v^\prime(0)+\theta v(0)+i\langle xe_\theta ,v\rangle=\langle e_\theta ,g\rangle \,,
\end{equation}
from which we conclude that 
\begin{equation}
  \label{eq:170}
|\bar{v}(0)v^\prime(0)| \leq \theta\, |v(0)|^2  +\frac{1}{2\theta^{3/2}}|v(0)|\,(\|v\|_2 +
\theta\sqrt{2}\|g\|_2)\,.
\end{equation}
As 
\begin{equation}
\label{eq:171}
  |v(0)|^2 \leq \|v^\prime\|_2\|v\|_2 \,,
\end{equation}
we obtain, using \eqref{eq:170} that 
\begin{equation}
  \label{eq:172}
|\bar{v}(0)v^\prime(0)| \leq\frac{\theta}{2}\|v^\prime\|_2\|v\|_2+ \frac{1}{2\theta^{3/2}}|v(0)|\,(\|v\|_2 +
\theta\sqrt{2}\|g\|_2)\,.
\end{equation}
Combining the real part of \eqref{eq:168} with
\eqref{eq:172} yields
\begin{equation}
\label{eq:173} 
  \|v^\prime\|_2^2 -\mu \|v\|_2^2\leq
  \frac{\theta}{2}\|v^\prime\|_2\|v\|_2 + \|v\|_2\|g\|_2 + \frac{1}{2\theta^{3/2}} |v(0)|(
  \|v\|_2 + \theta\sqrt{2}\|g\|_2)\,.
\end{equation}
Consequently, we obtain \eqref{eq:166}.
\end{proof}

\begin{proof}[Proof of Proposition \ref{lem:semi-infinite-spectrum}]~\\
{\em Step1:} We prove that $\LL^\theta$ is a closed operator.\\

Let $\{(v_n,\LL^\theta \, v_n)\}_{n=1}^\infty \in[D(\LL^\theta)]^\N\times
[L^2(\R_+)]^\N$ converge, as $n\to\infty$, in $L^2(0,+\infty) \times L^2(0,+\infty)$ to
$(v_\infty, g_\infty)$. We need to establish that $v_\infty \in D(\LL^\theta)$. The
orthogonality of $v_\infty$ with $e_\theta$ immediately follows from the $L^2$
convergence. From \eqref{eq:166}  (with $\lambda=0$) we conclude 
that $v_n$ is a Cauchy sequence in $H^1(\R_+)$ and hence 
must converge to $v_\infty$ in the $ H^1(\R_+)$ norm.

Let $\chi\in C^\infty(\R)$ be supported on $[\frac 12,+\infty)$ and satisfy
$\chi=1$ for $x>1$. Clearly,
\begin{displaymath}
  \Big(-\frac{d^2}{dx^2}+i x\Big) (\chi v_n) = \chi g_n + 2 \chi^\prime v_n^\prime + \chi^{\prime\prime} v_n\,.
\end{displaymath}
Since we can smoothly extend $\chi v_n$ to $H^2(\R)$, it follows from
\eqref{eq:111} and \eqref{eq:112} (with $\beta=1$ and $\tilde{U}=x$) that
$\chi v_n$ and $\chi x v_n$ are Cauchy sequences in $H^2(\mathbb R) $ and
in $L^2(\mathbb R)$ respectively. Hence, its limit $\chi v_\infty$ satisfies
$\chi v_\infty \in H^2(\R)$ and $\chi x v_\infty\in L^2(\R)$. By the $H^1(\R_+)$
convergence of $\{v_n\}_{n=1}^\infty$ it follows that $v_\infty$ is a weak
solution of
\begin{displaymath}
  v_\infty^{\prime\prime} = -g_\infty + i xv_\infty  \,.
\end{displaymath}
Since the right-hand-side is in $L^2(\R_+)$, it follows that $v_\infty \in
H^2(\R_+)$ and hence $\LL^\theta$ is closed.
 \vspace{.5ex}

{\em Step 2:}  We prove that $\LL^\theta -\lambda$ has index $0$. \\
 
Let $\tilde{\mathcal L}:H^2(\mathbb R_+) \cap L^2 (\mathbb R^+; x^2
dx)\to L^2(\mathbb R^+)$ be associated with the same differential
operator as $\LL^\theta\,$.  Clearly, $\tilde{\LL}$ is a Fredholm operator
of index $1$. Indeed, it is clearly surjective (we can find a unique
solution satisfying a Dirichlet condition at $x=0$) and it is easy to
see that the kernel has dimension $1$ (${\rm span}
\{Ai(e^{i\pi/6}\cdot)\}$.  Consequently, for any $\lambda \in \mathbb C$,
$\tilde{\mathcal L} - \lambda$ has index $1$. We now observe that $\mathcal L_\theta$
is obtained by imposing a single orthogonality condition in the
domain. Hence the index of $\LL^\theta -\lambda $ equals $0$.
 \vspace{.5ex}

{\em Step 3:} We show that  $\rho(\LL_\theta)\neq\emptyset\,$. \\

We prove that there exists $\tilde{\mu}<0$ such that for all $\Re \lambda <
\tilde{\mu}$ the operator $\mathcal L^\theta-\lambda$ is injective. Combined
with the above zero index property, it would yield that the
resolvent set contains the half plane $\Re \lambda < \tilde{\mu}$.  The
injectivity follows from \eqref{eq:171} and \eqref{eq:173}, by which
there exist $\mu_0$ and $C>0$ such that for all $\Re \lambda < \mu_0$, we have
   \begin{equation}
\label{eq:128a}
\|v\|_2 +   \|v^\prime\|_2 + |v(0)| \leq C\,  \|g\|_2\,. 
\end{equation}

Finally, the compactness of the resolvent follows from the fact that
$D(\LL^\theta)$ is compactly embedded in $L^2(\R_+)$.
\end{proof}

The previous proof has also shown to us that $\LL^\theta$ is semi-bounded.
The next proposition provides a more explicit lower bound for the spectrum as a
function of $\theta$.
\begin{proposition}
For all $\theta \in \mathbb R_+$, we have
\begin{equation}
  \label{eq:174}
\mu_0(\theta)=\inf \Re \sigma( \LL^\theta)\geq -\frac{1}{2}\min(\theta^2,\theta^{-2})\,.
\end{equation}
\end{proposition}
\begin{proof}
  Suppose that for some positive $\hat{\theta}_0$ there exist $\lambda_0\in\sigma(
  \LL^{\hat{\theta}_0})$ and $v_0\in D( \LL^{\hat{\theta}_0})$ such that
\begin{equation}
\label{eq:175}
  ( \LL^{\hat{\theta}_0} -\lambda_0)v_0=0\,.
\end{equation}
Since $v_0$ is an $L^2$ solution of the complex Airy equation in
$\R_+$ it is expressible, up to a multiplicative constant, in the form
\begin{subequations}
\label{eq:176}
  \begin{equation}
  v_0(x)=\Ai(e^{i\pi/6}(x+i\lambda_0)) \,.
\end{equation}
The orthogonality condition for $v_0$ reads
\begin{equation}
\label{eq:177}
  F(\lambda_0,\hat{\theta}_0)=0 \,,
\end{equation}
where 
    \begin{equation}
\label{eq:178}
      F(\lambda,\theta)=\int_{\R_+}e_\theta(x)\, \Ai (e^{i\pi/6}(x+i\lambda))\,dx\,.
     \end{equation}
\end{subequations}

{\em Step 1:} We prove that the set $\{\theta\in (0,+\infty)\,,\, \exists \lambda \in \sigma
(\mathcal L_\theta) \mbox{ with } \Re \lambda < \Re \nu_1\}$ is open.

We use the implicit function theorem.    If indeed $ F(\lambda_0,\hat{\theta}_0) =0$, we get after
   integration by parts that
\begin{displaymath}
  \frac{\partial F}{\partial\lambda}(\lambda_0,\hat{\theta}_0)= -i\Ai (e^{i2\pi/3}\lambda_0)\neq0\,.
\end{displaymath}
Hence there exists a neighborhood of $\hat{\theta}_0$ and a $C^1$- solution
$\lambda(\theta)$ in this neighborhood such that $\lambda(\hat{\theta}_0)=\lambda_0\,$.\\

{\em Step 2:} Let $\varepsilon<\Re\nu_1$. Consider  the set 
\begin{displaymath}
\Sigma (\varepsilon) :=\{ (\lambda,\theta)\in \mathbb C \times (0,+\infty)\,,\, \Re \lambda < \varepsilon \mbox{ and } F(\lambda,\theta)=0\} \,,
\end{displaymath}
which can be described as a countable (or finite) union of simple
analytic curves denoted by $\{\lambda_k(\theta)\}_{k\in\Kg}$, each with an
interval of definition $(\theta_k,\theta_{k}^*)$.  Let further
$\mu_k(\theta)=\Re\lambda_k(\theta)$. We prove that for all $k\in\Kg$ and
$\theta\in(\theta_k,\theta_{k}^*)$
\begin{equation}
  \label{eq:179} 
\mu_k(\theta)+\theta^2/2 \geq(\mu_k(\theta_k)+\theta_k^2/2)\exp\bigg\{-\int_{\theta_k}^\theta\frac{
  \|v_k(\cdot,\tau)\|_2^2}{|v_k(0,\tau)|^2}\,d\tau \bigg\}\,,
\end{equation}
where
\begin{displaymath}
v_k(x,\theta) :=\Ai(e^{i\pi/6}(x+i\lambda_k(\theta))) \,.
\end{displaymath}

Let $k\in\Kg$. For convenience of notation we set 
\begin{displaymath}
\lambda(\theta)=\lambda_k(\theta)\,, \, \mu(\theta)=\Re\lambda(\theta)\mbox { and } v(\theta) = v_k(\theta)\,.
\end{displaymath}
By (\ref{eq:176}b,c) we have that
\begin{equation}\label{eq:166-1}
F(\lambda(\theta),\theta)=0\,,\, \forall \theta \in \Sigma(0)\,.
\end{equation}
Differentiating this identity with respect to $\theta$ yields
\begin{subequations}
\label{eq:180}
  \begin{equation}
  e^{2\pi i/3}\frac{d\lambda}{d\theta}(\theta)\,I_1(\theta) -I_2 (\theta) =0 \,,
\end{equation}
where
\begin{equation}
  I_1(\theta) =\int_{\R_+}e_\theta(x)\, \Ai^\prime(e^{i\pi/6}(x+i\lambda(\theta)))\,dx\,,
\end{equation}
and
\begin{equation}
  I_2(\theta) =\int_{\R_+}xe_\theta(x)\, \Ai(e^{i\pi/6}(x+i\lambda(\theta) ))\,dx\,.
\end{equation}
\end{subequations}
Integration by parts yields, in conjunction with (\ref{eq:166-1}),
\begin{equation}
\label{eq:181}
  I_1(\theta) = -e^{-i\pi/6}\Ai(e^{2\pi i/3}\lambda(\theta) ) \,.
\end{equation}
We now write, with the aid of (\ref{eq:166-1}) and Airy's equation
\begin{displaymath}
\begin{array}{ll}
   I_2(\theta) & =\int_{\R_+}(x+i\lambda(\theta) )e_\theta(x) \,\Ai(e^{i\pi/6}(x+i\lambda(\theta) ))\,dx\\[1.5ex]
      & =e^{-i\pi/6}\int_{\R_+}e_\theta(x)\, \Ai^{\prime\prime}(e^{i\pi/6}(x+i\lambda(\theta) ))\,dx\,.
     \end{array}
\end{displaymath}
Integration by parts and \eqref{eq:181} then yield
\begin{displaymath}
  I_2(\theta) =-e^{-i\pi/3}\Ai^\prime(e^{2\pi i/3}\lambda(\theta))+i\theta\Ai(e^{2\pi i/3}\lambda(\theta))  \,.
\end{displaymath}
Substituting the above, together with \eqref{eq:181} into
\eqref{eq:180} yields
\begin{equation}
  \label{eq:182}
\frac{d\lambda}{d\theta}(\theta) = -\frac{\pa_x v(0,\theta)}{v(0,\theta)}-\theta \,.
\end{equation}
Taking the inner product of \eqref{eq:175} with $v(\cdot,\theta)$ we obtain
for the real part
\begin{displaymath}
  \|\pa_x v (\cdot,\theta)\|_2^2 -\mu (\theta) \|v(\cdot,\theta)\|_2^2 +\Re\{\bar{v}(0,\theta) \pa_x v(0,\theta)\}=0\,,
\end{displaymath}
where $\mu(\theta) =\Re\lambda(\theta)$. \\ Combining the above with~\eqref{eq:182} then
yields
\begin{displaymath}
  \frac{d\mu}{d\theta}(\theta) +\theta= \frac{ \|\pa_x v(\cdot,\theta \|_2^2}{|v(0,\theta)|^2}-\mu (\theta)  \frac{ \|v(\cdot,\theta\|_2^2}{|v(0,\theta)|^2}\,.
\end{displaymath}
We then have on the branch
\begin{displaymath}
  \frac{d(\mu+\theta^2/2)}{d\theta}+ \frac{ \|v(\cdot,\theta)\|_2^2}{|v(0,\theta)|^2}(\mu+\theta^2/2)>0 \,.
\end{displaymath}
Solving in $[\hat{\theta}_0,\hat{\theta}_0^*]$ yields \eqref{eq:179}. \\

{\em Step 3:} We prove that along every curve in $\Sigma(0)$
  \begin{equation}
\label{eq:5.83a}
   \Re \lambda(\theta)=\mu(\theta)\geq
-\theta^{-2}/2  \,.
\end{equation}
\\
From \eqref{eq:173} with $g=0$, and \eqref{eq:171} we obtain that
\begin{displaymath}
  \frac{1}{2}\|v^\prime\|_2^2 -\mu(\theta) \|v\|_2^2\leq \frac{1}{2\theta^{3/2}} \|v^\prime\|_2^{1/2}\|v\|_2^{3/2}\,.
\end{displaymath}
The above, in conjunction with Young's inequality yields  $\mu(\theta)\geq-\theta^{-2}/2$
which is precisely \eqref{eq:5.83a}.  In particular it implies that
\begin{equation}
\label{eq:183}
\liminf_{\theta \to+\infty} \inf_{F(\lambda,\theta)=0} \Re \lambda \geq 0\,.
\end{equation}

{\em Step 4:} We prove \eqref{eq:174}. \\

If $ \inf_\theta \inf_{F(\lambda,\theta)=0} \Re \lambda \geq 0$ then
\eqref{eq:174} readily follows.  Hence, we can assume that there
exists $(\lambda_0,\hat{\theta}_0)$ such that $\Re \lambda_0 < 0$ and $F(\lambda_0,\hat{\theta}_0) =0$. \\
We then look at $\{\lambda_k(\theta)\}_{k\in\Kg}$ inside $\Sigma (-\varepsilon) $ where $ \Re
\lambda_0  < -\varepsilon<0$.
By \eqref{eq:183}, all branches exit $\Sigma (-\varepsilon) $ for sufficiently
large $\theta$.\\
We now observe that 
\begin{displaymath}
F(\lambda(0),0)=0 \Rightarrow \ \mu(0)>0 ,.
\end{displaymath}
Indeed, as \begin{displaymath}
 F(\lambda,0)= e^{-\frac{i\pi}{6}} A_0(i \lambda)
\end{displaymath} we can apply
Corollary \ref{corwas}. Hence, these branches must lie outside $  \Sigma
(0) $ for $\theta\in[0,\theta_{inf})$ for sufficiently small $\theta_{inf}>0$.
Assume that $\varepsilon$ is chosen small enough so that
\begin{displaymath}
-\varepsilon+ \frac{1}{2}\theta_{inf}^2 >0\,.
\end{displaymath} 
Consider any branch $\lambda(\theta)$ in  $\Sigma (-\varepsilon)$ with $\theta$  in some
interval $[\theta(\varepsilon),\theta^*(\varepsilon)] \subset (0,+\infty)$ such that 
$\Re \lambda(\theta(\varepsilon))=\Re \lambda (\theta^* (\varepsilon)) = -\varepsilon$\\
We can then use \eqref{eq:179} to obtain
\begin{equation}\label{eq:169-1}
\begin{array}{ll}
  \mu(\theta)+\theta^2/2 & \geq(\mu(\theta(\varepsilon) )+\theta(\varepsilon)^2/2)\exp\Big\{-\int_{\theta(\varepsilon) }^\theta\frac{ \|v\|_2^2}{|v(0)|^2} d\theta \big \}\\
  &  = (-\varepsilon  +\theta(\varepsilon)^2/2)\exp\Big\{-\int_{\theta(\varepsilon)}^\theta\frac{ \|v\|_2^2}{|v(0)|^2}
  \,d\theta\Big\} \\
  & \geq (-\varepsilon  +\theta_{inf}^2/2)\exp\Big\{-\int_{\theta(\varepsilon)}^\theta\frac{ \|v\|_2^2}{|v(0)|^2}
  \,d\theta\Big\} \\
  &  >0\,.
  \end{array}
\end{equation}
This completes the proof of \eqref{eq:174}. 
\end{proof}

\begin{corollary}\label{corhatmum}
 It holds that
\begin{equation}
  \label{eq:184}
\hat{\mu}_m := \inf_{\theta\in \mathbb R_+} \inf_{ \lambda\in\sigma(\LL^\theta)} \Big(\Re \lambda+\frac{\theta^2}{2}\Big) >0\,.
\end{equation}
\end{corollary}
The proof of \eqref{eq:184} with the non strict inequality follows immediately from \eqref{eq:174}. 
Note now that by \eqref{eq:169-1}, taking into account that
$\mu(\theta)+\frac{\theta^2}{2}  \xrightarrow[\theta\to+\infty ]{}+\infty$ and that $\mu(0)>0$ we
can conclude the strict inequality in \eqref{eq:184}.

\paragraph{The adjoint operator}~\\
We note that $D(\LL^\theta)$ is not dense in $L^2(\R_+)$.
We thus
introduce $\mathcal H_0= \overline{D(A)}$, and then define
\begin{equation}
\label{eq:185}
   D(A^*)
   =\{v \in \Hg \mbox{ s.t. } D(A)\ni u \mapsto \langle v, Au\rangle  \mbox{ 
extends as a
    continuous linear map on } \mathcal H_0\}\,,
\end{equation}
and set $A^* v$ to be the unique (by Riesz theorem)
$y\in\mathcal H_0$ for which
\begin{displaymath}
   \langle v,Au\rangle  = \langle y , u\rangle \,.
\end{displaymath}
The standard definition is recovered when $\mathcal H_0=\mathcal H$.
Note that 
\begin{displaymath}
(A-\lambda I)^* = A^* - \bar \lambda \, \Pi_{\Hg_0}\,,
\end{displaymath}
where $ \Pi_{\Hg_0}$ is the projector on $\Hg_0$.
We further note that $A^* A$ is an unbounded
  operator on $\mathcal H_0$. 

In the particular case
$A=\LL^\theta$, $\mathcal H_0$ is the orthogonal complement in
$L^2(\R_+)$ of $e_\theta$. Hence, $ \Pi_{\Hg_0}= I-P_\theta$, where,
for any $u\in L^2(\mathbb R_+)$,
\begin{displaymath}
   P_\theta u(x) = \frac{1}{2\theta}\langle e_\theta ,u\rangle 
e_\theta\,.
\end{displaymath}

We next provide a more explicit representation of  $(\LL^\theta)^*$.\\
     \begin{lemma}
We have
\begin{subequations}
\label{eq:186}
       \begin{equation}
D((\LL^\theta)^*)=H^2_0(\R_+)\cap L^2(\R_+;x^2dx)\,,
\end{equation}
and for any $v \in  D((\LL^\theta)^*)$, 
\begin{equation}
    (\LL^\theta)^* v= (I-P_\theta) \Big(-\frac{d^2}{dx^2} -i x\Big) v\,.
\end{equation}
\end{subequations}
     \end{lemma}
     \begin{proof}
       The proof is reminiscent of the analysis of selfadjointness for
       $(1D)$-problems in \cite{ReSi}.  Let $\phi \in C_0^\infty(\mathbb
       R^+)$ and then set $u=(I-P_\theta) \phi\in D(\LL_\theta)$. Let $v\in
       D(\LL_\theta^*)$, where $D(\LL_\theta^*)$ is defined by
       \eqref{eq:185}.  From the definition we deduce that the
       distribution 
\begin{displaymath}
  C_0^\infty(\mathbb R_+) \ni \phi       \mapsto \langle v,\LL_\theta (I-P_\theta) \phi\rangle_{L^2(\mathbb R_+)}
\end{displaymath}
should extend as a continuous linear map on $L^2(\mathbb R^+)$. We
then observe that
\begin{displaymath}
      \langle v,\LL_\theta (I-P_\theta) \phi\rangle  = \Big\langle  \Big(-\frac{d^2}{dx^2} 
-ix\Big) v, \phi\Big\rangle  -  \Big\langle  v, \Big(-\frac{d^2}{dx^2} +ix\Big) e_\theta\Big\rangle  \langle e_\theta, \phi 
 \rangle \,.
\end{displaymath}
The second term on the right hand side defines a linear form on
$L^2(\mathbb R_+)$. Hence, from \eqref{eq:185} we get that $\phi \mapsto \langle 
(-d^2/dx^2-ix) v, \phi\rangle$ is a distribution  in $L^2(\mathbb R_+)$. Hence, it holds that
$ (-d^2/dx^2-ix) v\in L^2(\mathbb R_+)$.  We can thus conclude
that $v\in H^2(\mathbb R_+)$, $x v\in L^2(\mathbb R_+)$.

We now compute $ \langle v,\LL_\theta u\rangle$ using integration by
parts to obtain
\begin{displaymath} 
      \Big\langle v,  \Big(-\frac{d^2}{dx^2} + ix\Big)u \Big\rangle = - u^\prime(0)  \bar v(0) 
+ u (0) \bar v^\prime(0) +
      \Big\langle \Big( -\frac{d^2}{dx^2} - ix\Big) v, u \Big\rangle \,.
\end{displaymath}
To conform with \eqref{eq:185} $u\mapsto - u^\prime(0) \bar v(0) + u (0) \bar
v^\prime(0) $ must be a continuous map on $(I-P_\theta) L^2$. This, however, is
possible only if $v(0)=v^\prime(0)=0$ (consider the sequence $u_n = \chi_n
-P_\theta \chi_n$ with $\chi_n(x) =\sqrt{n} \chi (nx)$), leading
thereby to (\ref{eq:186}a). Consequently for any $v\in D(\LL_\theta^*)$,
we have
\begin{displaymath} 
      \Big\langle v,  \Big(-\frac{d^2}{dx^2} + ix\Big)u \Big\rangle =
      \Big\langle \Big( -\frac{d^2}{dx^2} - ix\Big) v, u \Big\rangle \,.
\end{displaymath}
Having in mind that
   \begin{equation}
\label{eq:187}
      \langle v,\LL_\theta u\rangle =  \langle \LL_\theta^* v, u\rangle\,,
\end{equation}
leads to
\begin{displaymath}
  \LL_\theta^* v = (I-P_\theta)  \Big(-\frac{d^2}{dx^2}-ix\Big) v\,.
\end{displaymath}
We can then extend \eqref{eq:187} by
density to any $u\in \mathcal H_0$.
\end{proof}

\begin{proposition}
  The eigenfunctions of $\LL^\theta$ are complete in $(I-P_\theta)L^2(\R_+)\,$.~  
\end{proposition}
\begin{proof}
We take a similar approach to the one in \cite{alhe15a}.

{\em Step 1:} By the semi-boundedness of $\mathcal L^\theta$ and
(\ref{eq:186}a) there exists $c_0>0$ and $\hat \mu_0 \in \mathbb R_-$
such that for all $u\in D(\LL^\theta)$
 \begin{equation}
   \label{eq:188}
c_0 \,  \|-u^{\prime\prime}+xu\|_2^2 \leq \,2c_0\big(\|xu\|_2^2+\|u\|_{2,2}^2\big)\leq \, \|-u^{\prime\prime}+(ix-\hat \mu_0)u\|_2^2 \,.
 \end{equation}
 \vspace{2ex}

{\em Step 2:}  We now show that the resolvent of $\LL^\theta$ is in $S_p$ for
 every $p>3/2\,$,  where $S_p$ denotes the Schatten of order $p$. \\ 
 
 By the Max-Min principle the singular values $(\mu_n)_{n\in \mathbb
   N^*}$ of the operator $(\LL^\theta-\hat \mu_0)^{-1}$ satisfy for $k\in
 \mathbb N$ 
\begin{displaymath}
  \mu_{k+1}^{-2}=\max_{U_k\in \mathcal H_0^k}\;\min_{u\in U_k^\perp\cap D(\LL_\theta)}
  \frac{\|-u^{\prime\prime}+(ix-\hat \mu_0)u\|_2^2}{\|u\|_2^2} \,.
\end{displaymath}
Let further
\begin{displaymath}
   \kappa_{k+1}^{-2}=\max_{U_k\in (L^2(\mathbb R_+))^k}\;\min_{u\in
     U_k^\perp \cap D(\LL_\theta) }
  \frac{\|-u^{\prime\prime}+xu\|_2^2}{\|u\|_2^2} \,.
\end{displaymath}
By \eqref{eq:188} we have, for $n\in \mathbb N^*$, 
\begin{displaymath}
   \kappa_{n}^{-2} \leq c_0 \, \mu_n^{-2}\,.
\end{displaymath}
Finally, let 
\begin{displaymath} 
   \tilde{\kappa}_{k+1}^{-2}=\max_{U_k\in L^2(\mathbb R_+) ^k}\; \min_{u\in
     U_k^\perp\cap H^2(\R_+)\cap L^2(\R_+;x^2dx)} \frac{\|-u^{\prime\prime}+xu\|_2^2}{\|u\|_2^2} \,.
\end{displaymath}
In view of the additional constraint embedded in $D(\LL^\theta)$ we have
\begin{displaymath}
 \kappa_n  \leq \tilde{\kappa}_{n}\,.
\end{displaymath}
By the Max-Min principle the $\tilde{\kappa}_{n}^{-2}$ are eigenvalues
of \begin{displaymath}
\A_N:=(-d^2/dx^2+x)^2
\end{displaymath} defined on
\begin{displaymath}
  D(\A_N)=\{u\in H^4(\R_+)\cap L^2(\R_+;x^4dx) \,| \,
  u^{\prime\prime}(0) =(-u^{\prime\prime}+xu)^\prime(0)=0\,\}\,.
\end{displaymath}
 Let $\lambda=\alpha^2$, where $\alpha>0$, denote an eigenvalue of $\A_N$. (Note
that $\lambda=0$ is an eigenvalue.) Let $u_\alpha$ denote the corresponding
eigenfunction. As
\begin{displaymath}
  \Big(-\frac{d^2}{dx^2}+x\Big)^2-\alpha^2 =  \Big(-\frac{d^2}{dx^2}+x+\alpha\Big)\Big(-\frac{d^2}{dx^2}+x-\alpha\Big)\,,
\end{displaymath}
we easily obtain that, up to a product by an arbitrary constant, 
\begin{displaymath}
  u_\alpha =-\frac{1}{2\alpha}\Ai(x+\alpha)+A_1\,  \Ai(x-\alpha)\,,
\end{displaymath}
where $A_1$ has to be  determined from the requirement
$u_\alpha\in D(\A_N)$. It can now be easily verified that $\alpha\in\sigma(\A_N)$ if
and only if
\begin{equation}
\label{eq:189}
\delta(\alpha):=  -\frac{\Ai^\prime(\alpha)}{\Ai(\alpha)} + \frac{\Ai^\prime(-\alpha)}{\Ai(-\alpha)}=0\,.
\end{equation}

Let $\{\omega_n\}_{n=1}^\infty\subset\R_-$ denote the zeroes of Airy's function
$\Ai(x)$.   By computation of its derivative $\delta(\alpha) $ is a monotone increasing for $\alpha\in(-\omega_n,-\omega_{n+1})$
and tends to $\pm\infty$ at the edges. 
Consequently, there is precisely one
solution of \eqref{eq:189} in   $(-\omega_n,-\omega_{n+1})$. \\

As $-\omega_n\sim n^{2/3}$ we may conclude from the foregoing discussion that
$\tilde{\kappa}_{n}^{-1}\sim n^{2/3}$ as well.  Consequently, there exists
$C>0$ such that, for sufficiently large $n\,$,
\begin{displaymath}
  \mu_n \leq \frac{C}{n^{2/3}} \,.
\end{displaymath}
As a result, for all $p>3/2$ it holds that
\begin{displaymath}
  \sum_{n=1}^\infty \mu_n^p <\infty \,. 
\end{displaymath}

{\em Step 3:} We complete the proof of the proposition. \\
\vspace{2ex} 

We take a similar approach to the one in \cite{alhe15a}.
By \eqref{eq:171} and \eqref{eq:173}  we have, for sufficiently
large $-\Re\lambda$,
\begin{displaymath}
  \|(\LL^\theta-\lambda)^{-1}\| \leq \frac{C(\theta)}{|\lambda|} \,.
\end{displaymath}
Let $v=(\LL^\theta-\lambda)^{-1}g$ for some $g\in L^2(\R_+)$. From the imaginary
part of  \eqref{eq:168} we learn that
\begin{displaymath}
  -\Im\lambda\|v\|_2^2 + \|x^{1/2}v\|_2^2 = \Im\langle v,g\rangle
  -\Im\{v^{\prime}(0)\bar{v}(0)\} \,.
\end{displaymath}
With the aid of \eqref{eq:166} and  \eqref{eq:172}   we then obtain that
\begin{displaymath}
   -\Im\lambda\|v\|_2^2 \leq C(\theta)[(1+|\mu|^{1/2})\|v\|_2^2+\|g\|_2^2]\,. 
\end{displaymath}
Hence, there exists $\hat C$ such that if $ -\Im\lambda \geq \hat C\, (1+|\mu|^{1/2})$, then
\begin{displaymath}
  \|(\LL^\theta-\lambda)^{-1}\|_2 \leq \frac{C(\theta)}{|\lambda|} \,.
\end{displaymath}
From the foregoing discussion we may conclude that every direction
where \linebreak $\pi/2<\arg \lambda<2\pi$ is a direction of minimal growth for
$(\LL^\theta-\lambda)^{-1}$. Following the arguments of the proof of
\cite[Theorem 16.4]{ag65} (cf. also \cite[Theorem X.3.1 ]{goetal90} or
\cite[Corollary XI.9.31]{dusc63})) we can conclude that the eigenspace
of $\LL^\theta$ is given by $\overline{D(\LL^\theta)}=(I-P_\theta)L^2(\R_+)$.
\end{proof}

\begin{proposition}
\label{lem:large-theta}
Let $\mu_0(\theta) = \inf_{\lambda \in \sigma(\LL^\theta)} \Re \lambda$.
 Then,
\begin{equation}
\label{eq:190}
  \lim_{\theta\to + \infty}\mu_0(\theta)=\Re\nu_1 \,,
\end{equation}
 where $\nu_1$ the left most eigenvalue of $\LL_+$ the Dirichlet realization of $\LL$ in $\mathbb R_+$.
\end{proposition}
\begin{proof} 
  We begin the proof by applying Rouch\'e's Theorem, in the large $\theta$
  limit, to the holomorphic functions $\theta F(\lambda,\theta)$ and $\Ai(e^{2i\frac \pi
    3} \lambda)$ inside a disk of radius $r>0$ centered at $\nu_1$ and
  containing no other eigenvalue of $\LL_+$. As $\lambda \mapsto \Ai(e^{2i\frac \pi 3}
  \lambda)$ has a unique zero in this disk, Rouch\'e's Theorem would show the
  same for the zeros of $F(\cdot,\theta)$. It is therefore necessary to compare the two
  functions for $\lambda \in \partial B(\nu_1,r)$.  We thus write
  \begin{equation}\label{eq:3.14b}
   \theta  F(\lambda,\theta)- \Ai(e^{i2\pi/3}\lambda)=  \theta \int_{\R_+}e^{-\theta x}[\Ai(e^{i2\pi/3}\lambda+e^{i\pi/6}x)-\Ai(e^{i2\pi/3}\lambda)]\,dx\,.
  \end{equation}
We bound the right-hand-side in the following
manner 
\begin{multline}\label{eq:314f}
\theta  \int_{\R_+}e^{-\theta x}|\Ai(e^{i2\pi/3}\lambda+e^{i\pi/6}x) -
  \Ai(e^{i2\pi/3}\lambda)|\,dx \\
 \leq \theta  \int_{\R_+}e^{-\theta x}x^{1/2}\|\Ai^\prime(e^{i\pi/6}(.+i\lambda)\|_{L^2(0,x)}\,dx\\
  \leq \frac{C}{\theta^{1/2}}  \|\Ai^\prime(e^{i\pi/6}(.+i\lambda)\|_2 \,.
\end{multline}
From this, we obtain the existence of $r_0 >0$ and $C>0$ such that, for any $r\in (0,r_0]$ and any $\lambda \in \partial B(\nu_1,r)$, we have
\begin{displaymath}
 |\theta  F(\lambda,\theta)- \Ai(e^{2i\pi/3}\lambda)| \leq \frac{C}{r \theta^{\frac 12}}  | \Ai(e^{2i\pi/3}\lambda)| \,.
\end{displaymath}
It follows from Rouch\'e's Theorem that for sufficiently large $\theta$,
$F(\lambda,\theta)$  has a unique zero in $B(\nu_1,r)$. \\
At this stage we have obtained
\begin{displaymath}
\limsup_{\theta \to +\infty} \mu_0(\theta) \leq \Re \nu_1\,.
\end{displaymath}
Using the arguments as above and supposing 
 now that $r<|\lambda-\nu_1|<R$ and $\Re\lambda\leq \Re\nu_1$, 
 we can establish that
\begin{equation}
\label{eq:191}
  \Big|\frac{\theta F(\lambda,\theta)}{\Ai(e^{i2\pi/3}\lambda)}-1\Big|\leq
  \frac{C(R,r)}{\theta^{1/2}} \,.
\end{equation}
Consequently,  we obtain that, there exists $\theta_1(R,r)$ such that, for all $\theta>\theta_1(R,r)$
$F(\lambda,\theta)$ does not vanish in $(B(\nu_1,R)\setminus B(\nu_1,r))\cap \{\Re \lambda \leq  \Re \nu_1\}$. 

To complete the proof we need yet to establish that there exists
$R_0>0$, and $\theta_2(R_0)>0$ such that for all $\theta>\theta_2(R_0)$ we have that 
\begin{equation}
\label{eq:192}
  \inf_{
    \begin{subarray}{c}
      \Re\lambda\leq\Re\nu_1\\
      |\lambda-\nu_1|>R_0
    \end{subarray}}\Big|\frac{[\theta+(-\lambda)^{1/2}] F(\lambda,\theta)}{\Ai(e^{i2\pi/3}\lambda)}\Big|>0\,.
\end{equation}
To this end we set as in \eqref{eq:341}
\begin{displaymath}
  \frac{F(\lambda,\theta)}{\Ai(e^{i2\pi/3}\lambda)}=\int_{\R_+}e^{-\theta x}e^{-(-\lambda)^{1/2}x}\,dx+ \int_{\R_+}e^{-\theta x}w(x)\,dx\,.
\end{displaymath}
To bound the second term we use \eqref{uppb3} and \eqref{eq:344} (with $\mu_0=\Re \nu_1$)
together with Sobolev embeddings to obtain
\begin{displaymath}
   \|w\|_\infty^2  \leq \|w\|_2\|w^\prime\|_2\leq C |\lambda|^{-3/2 }\,.
\end{displaymath}
Hence,
\begin{equation}
\label{eq:193}
  \Big|\frac{[\theta+(-\lambda)^{1/2}] F(\lambda,\theta)}{\Ai(e^{i2\pi/3}\lambda)}-1\Big| \leq C |\lambda|^{-1/4} \,, 
\end{equation}
from which \eqref{eq:192} easily follows.
\end{proof}

\subsection{No-slip operator on $(-1,1)$ for large $\alpha$}
\label{sec:large-alpha-no-slip}
Consider $\LL^\zeta_\beta$, defined in \eqref{eq:141}, with $\zeta_\pm=  \mathfrak z_\pm$, where
$  \mathfrak z_\pm \in C^2(-1,1)$ is the solution of
\begin{equation}
\label{eq:194}
  \begin{cases}
      -  \mathfrak z_\pm ^{\prime\prime}+\alpha^2  \mathfrak z_\pm  = 0  \mbox{ for } x\in(-1,1) \\
        \mathfrak z_\pm (\pm 1)=1 \mbox{ and }   \mathfrak z_\pm (\mp1)=0 \,.
  \end{cases}
\end{equation}
An immediate computation gives 
\begin{equation}\label{eq:176a}
  \mathfrak z_+(x) = \frac{\sinh \alpha(1+x)}{\sinh 2\alpha }\,,
\end{equation}
and a similar formula for $  \mathfrak z_-$.\\
We attempt to obtain a resolvent estimate for $\LL^\zeta_\beta$ in the case
$\alpha\geq\theta_1\beta^{1/3}$ where $\theta_1>0\,$.  If we try to use the arguments
of Subsection~\ref{sec:no-slip-schrodinger} we would encounter a
problem while attempting to use \eqref{eq:159}. It can be verified
from \eqref{eq:176a} that
\begin{equation}\label{eq:176c}
\|  \mathfrak z_\pm-e^{-\alpha(1\mp x)}\|_\infty\leq Ce^{-2\alpha}\leq Ce^{-2\theta_1\beta^{1/3}}\,.
\end{equation}
Then, one can deduce in the same manner that for some $C>0$, $C_1\in (0,1)$, and sufficiently large $\beta$,
\begin{displaymath}
\|  \mathfrak z_\pm^\prime\|_2\geq\alpha^{1/2}(1-Ce^{-2\alpha})\geq C_1\theta_1^{1/2}\beta^{1/6}\,.
\end{displaymath}
 Thus, the error
introduced by \eqref{eq:159} is not necessarily small and one needs an
alternate route for the estimation of $\|(\LL^\zeta_\beta-\lambda)^{-1}\|$. 

Since for $\alpha\geq\theta_1\beta^{1/3}$ we need to consider, in the next section,
only the case $\zeta_\pm\approx  \mathfrak z_\pm$, we focus attention here on the resolvent
of $\LL_\beta^\zeta$ in that case. Thus, we no longer approximate  $\zeta_\pm$ near
$x=\pm1$ by $1$, as in Subsection \ref{sec:no-slip-schrodinger}, and use
instead the approximation $\zeta_\pm  \approx e^{-\alpha(1\mp x)}$ as observed in
\eqref{eq:176c}.  Note that $\LL_\beta^\zeta$ depends on $\alpha$ through the
orthogonality conditions appearing in the definition of its domain. 
Consequently, we need
to renormalize $\psi_\pm$ from (\ref{eq:132a}b) and  (\ref{eq:132a}a)  in a
manner that would suit the approximation used for $\zeta_\pm$. \\
For  some $\theta>0$, the renormalization factor will be defined by
\begin{equation}
\omega_\pm ( \beta,\lambda,\theta):= \frac{F(\tilde{\lambda}_\pm
      ,0)}{F(\tilde{\lambda}_\pm ,\theta J_\pm^{-1/3})}\,,
\end{equation}
where,  (see \eqref{eq:deflambdapm} for the definition of $\lambda_\pm$),  
\begin{equation}
\label{eq:195}
  \tilde{\lambda}_\pm  = \beta^{1/3}J_\pm ^{-2/3}[\lambda-iU(\pm 1)] = \beta^{1/3}J_\pm ^{-2/3} \lambda_\pm \,.
\end{equation}
 We now   define
\begin{equation}
\label{eq:196}
    \psi_{\pm ,\theta}= \omega_\pm  ( \beta,\lambda,\theta)  \psi_\pm\,,
\end{equation}
where
$\psi_\pm$ was introduced in (\ref{eq:132a}b)- (\ref{eq:132a}a). \\
The above normalization provides the approximation
$\langle\zeta_\pm,\psi_{\pm,\theta}\rangle\sim(J_\pm \beta)^{-\frac 13}$, in the limit $\beta\to\infty$, as
in Subsection \ref{sec:no-slip-schrodinger} (see below
\eqref{eq:205}).
 
We similarly introduce with  the notation of \eqref{eq:145} and \eqref{eq:146}
\begin{equation}
\label{eq:197}
 g_{\pm,\theta}=  \omega_\pm  ( \beta,\lambda,\theta) g_\pm \mbox{ and }   \tilde{v}_{\pm ,\theta} =  \omega_\pm  ( \beta,\lambda,\theta) \tilde v_\pm \,.
\end{equation}
We can now state:
\begin{proposition}
\label{lem:large-alpha} 
Let $r>1$, $\theta_1 >0$ and  $\varkappa>0\,$. Then,  there
exist $\beta_0>0$ and $C(\varkappa)>0$ such that, for all $U\in\Sg^2_r$, $\beta\geq
\beta_0$ and $\theta=\alpha\beta^{-1/3} 
\geq\theta_1$\,,
   \begin{equation}
    \label{eq:198}
\sup_{\Re\lambda\leq (\hat{\mu}_0(\theta)-\varkappa) \beta^{-1/3}}\|(\LL_\beta^{\zeta}-\beta\lambda)^{-1}\| 
\leq C(\varkappa)\beta^{-2/3} \,,
  \end{equation}
where
\begin{equation}
\label{eq:32}
  \hat{\mu}_0(\theta)=\min(J_-^{2/3}\mu_0(J_-^{-1/3}\theta),J_+^{2/3}\mu_0(J_+^{-1/3}\theta)) \,.
\end{equation}
\end{proposition}
\begin{remark}
  In the sequel we apply Proposition \ref{lem:large-alpha} with
  $\theta_1=\theta_0$ where $\theta_0$ is defined  in the statement of Lemma 
  \ref{lem:integral-conditions}.
\end{remark}
\begin{proof}~\\
  The proof goes along similar lines to the proof of Lemma
  \ref{lem:integral-conditions}.  Let  $\lambda\in\C$ satisfy
\begin{displaymath}
\beta^{1/3}\Re\lambda\leq\Upsilon (\theta,\varkappa):=\hat{\mu}_0(\theta) -\varkappa\,.
\end{displaymath}
Furthermore, let the pair $(g,v)$ in  $ L^2(-1,1)\times 
  D(\LL_\beta^\zeta)$ satisfy the relation 
  \begin{displaymath}
 (\LL_\beta^\zeta-\beta\lambda)v=g\,.
\end{displaymath} 
  We assume, as in \eqref{eq:152},
\begin{equation}
\label{eq:199}
v = A_+(g) (\psi_{+,\theta}-\tilde{v}_{+,\theta}) + A_-(g)(\psi_{-,\theta}-\tilde{v}_{-,\theta}) + u \,,
\end{equation}
 where $u$ is given by \eqref{eq:143},
and then estimate $A_\pm(g) $ in the relevant regime of $\alpha$
values.\\
We first estimate the renormalization factor. We note that
\begin{displaymath}
  \omega_\pm ( \beta,\lambda,\theta) = e^{-i\pi/6} \frac{A_0(i \tilde \lambda_\pm)}{F(\tilde \lambda_\pm,\theta J_\pm^{-1/3})}\,,
\end{displaymath}
  and from \eqref{eq:364} which reads,
\begin{equation*}
\sup_{\Re \tilde \lambda\leq\Upsilon(\hat \theta)}\Big|\frac{A_0(i\tilde \lambda)}{F(\tilde \lambda,\hat \theta)}\Big|
\leq C (\varkappa) (1+\hat \theta) \,,
\end{equation*}
we obtain  that there exist
  $C(\varkappa)>0$ and $\beta_0(\varkappa)$, such that for all
  $\theta\geq \theta_1\,$, \break $\Re\lambda \leq \Upsilon(\theta,\varkappa)\beta^{-1/3} $, and $\beta\geq \beta_0(\varkappa)$ we have 
  \begin{equation}\label{eq:182a}
  |\omega_\pm( \beta,\lambda,\theta)| \leq C \, \theta\,.
  \end{equation}  
  We can now use \eqref{eq:149},  \eqref{eq:150}  and \eqref{eq:151} to obtain that
  \begin{equation}
\label{eq:200}
  \|\psi_{\pm ,\theta}\|_2 \leq C\, \theta\, [1+|\lambda_\pm|^{1/2}\beta^{1/6}]^{1/2}\beta^{-1/6}\,,
\end{equation}
\begin{equation}
\label{eq:201} 
\|g_{\pm,\theta} \|_2 \leq C\,  \theta\,  \frac{\beta^{1/6}}{[1+|\lambda_\pm |^{1/2}\beta^{1/6}]^{3/2}}\,,
\end{equation}
and 
\begin{equation}
 \label{eq:202} 
\|(U-\nu+i\beta^{-1/3})\tilde{v}_{\pm,\theta}  \|_2 +\beta^{-1/3}\|\tilde{v}_{\pm,\theta} \|_2\leq
C\,\theta \, \frac{\beta^{-5/6}}{[1+|\lambda_\pm|^{1/2}\beta^{1/6}]^{3/2}}\,.
\end{equation}
We note that \eqref{eq:156} remain valid in the case $\theta>0$, i.e.,
\begin{equation}
\label{eq:203}
      \begin{bmatrix}
      \langle{\mathfrak z}_+,(\psi_{+,\theta}-\tilde{v}_{+,\theta})\rangle &
      \langle{\mathfrak z}_+,(\psi_{-,\theta}-\tilde{v}_{-,\theta})\rangle \\
       \langle{\mathfrak z}_-,(\psi_{+,\theta}-\tilde{v}_{+,\theta})\rangle &   \langle{\mathfrak z}_-,(\psi_{-,\theta}-\tilde{v}_{-,\theta})\rangle
    \end{bmatrix}
\begin{bmatrix}
  A_+ \\
  A_-
\end{bmatrix}
=
\begin{bmatrix}
  \langle{\mathfrak z}_+,u\rangle \\ 
  \langle{\mathfrak z}_-,u\rangle
\end{bmatrix}\,,
\end{equation}
We now write
\begin{displaymath}
   \langle{\mathfrak z}_\pm ,\psi_{\pm ,\theta}\rangle = \langle e^{-\alpha(1\pm  x)},\psi_{\pm ,\theta}\rangle+\langle{\mathfrak z}_\pm -e^{-\alpha(1\pm  x)},\psi_{\pm ,\theta}\rangle  \,.
\end{displaymath}
Since by \eqref{eq:182a} and (\ref{eq:363}b) we have that 
\begin{equation}
\label{eq:204}
  \|\psi_{\pm ,\theta}\|_1 \leq C(\varkappa)\theta\|\psi_\pm\|_1\leq C\theta\beta^{-1/3} \,,
\end{equation}
we can easily deduce using \eqref{eq:176c} that  for sufficiently large $\beta$
\begin{displaymath} 
  |{\mathfrak z}_\pm -e^{-\alpha(1\pm  x)},\psi_{\pm ,\theta}\rangle| \leq \hat Ce^{-2 \theta \beta^{1/3}} \|\psi_{\pm
    ,\theta}\|_1 \leq \check C \beta^{-\frac 23} e^{-\theta_1\beta^{1/3}}\,.
\end{displaymath}
Furthermore, as 
\begin{multline*}
  \int_{-1}^{+\infty} e^{-\alpha (1+x)} \Ai\big((J_- \beta)^{1/3}e^{
    i\pi/6}\big[(1+x)+iJ_-^{-1}(\lambda-iU(-1))\big]\big)\,dx  \\ =  (J_- \beta)^{-1/3} \int_{\R_+}
  e^{-\theta
    J_-^{1/3}\xi}\Ai\big(e^{i\pi/6}\big[\xi+iJ_-^{-2/3}\beta^{1/3}(\lambda-iU(-1))\big]\big)\,d\xi \\ = 
  (J_- \beta)^{-1/3}  F(\tilde \lambda_\pm,\theta J_\pm^{-1/3})\,,
\end{multline*}
we have that
\begin{displaymath}
  \langle e^{-\alpha(1+x)},\psi_{-,\theta}\rangle=
    (J_-\beta)^{-1/3} - \langle{\mathbf 1}_{1,\infty}e^{-\alpha(1+x)},\psi_{-,\theta}\rangle \,.
\end{displaymath}
Since by \eqref{eq:204} we have
\begin{displaymath} 
  |\langle{\mathbf
    1}_{1,\infty}e^{-\alpha(1+x)},\psi_{-,\theta}\rangle|\leq Ce^{-2\theta\beta^{1/3}}\|\psi_{-
    ,\theta}\|_1 \leq \check C \beta^{-\frac 13} e^{-\theta_1\beta^{2/3}} \,,
\end{displaymath}
and hence we obtain 
\begin{equation} 
\label{eq:205}
  \langle e^{-\alpha(1+  x)},\psi_{-,\theta}\rangle=
    (J_- \beta)^{-1/3}[1+ \mathcal O (e^{-\theta_1\beta^{1/3}})] \,,
\end{equation}
and a  similar estimate can be obtained for  $\langle e^{-\alpha(1-x)},\psi_{+,\theta}\rangle$. \\
We now write
\begin{displaymath}
   |\langle{\mathfrak z}_+,\psi_{-,\theta}\rangle| \leq \|\psi_{-,\theta}\|_{L^1(-1,0)}e^{-\theta\beta^{1/3}}+
   \|\psi_{-,\theta}\|_{L^\infty(0,1)}\|{\mathfrak z}_+\|_{L^1(0,1)}\|\,,
\end{displaymath}
and then use \eqref{eq:204}, \eqref{eq:182a}, and  (\ref{eq:363}c) to
obtain that 
\begin{equation}
\label{eq:206}
  |\langle{\mathfrak z}_+,\psi_{-,\theta}\rangle| \leq C\beta^{-5/3} \,.
\end{equation}

\paragraph{Bounds for  $\langle{\mathfrak z}_\pm ,u\rangle$.}~\\
 Set, 
\begin{equation}
\label{eq:207}
  \hat{\lambda}=
  \begin{cases}
    \lambda  &\mbox{ if }  |\Re \lambda |>\max\Big(1,\frac{\|U^\prime\|_\infty}{\theta}\Big)\beta^{-1/3} \\
\max\Big(1,\frac{\|U^\prime\|_\infty}{\theta}\Big)\beta^{-1/3} +i\Im\lambda & \mbox{ if } |\Re \lambda |\leq
\max\Big(1,\frac{\|U^\prime\|_\infty}{\theta}\Big)\beta^{-1/3}
  \end{cases}\,.
\end{equation}
As $\beta^\frac 13\, |\Re\hat{\lambda}| \geq1$ and $|\Re \hat{\lambda} |\geq |\Re\lambda|$, it can be  verified that
\begin{displaymath}
  1+|\Re\lambda|^2\beta^{2/3}\leq2\, |\Re\hat{\lambda}|^2\, \beta^{2/3}\,,
\end{displaymath}
and hence
\begin{equation}
\label{eq:208}
  1+\beta^{2/3}|\lambda_\pm|^2\leq2\,\beta^{2/3}\, |U(\pm1)+i\hat{\lambda}|^2\,.
\end{equation}
We now verify that 
\begin{equation}
\label{eq:209}
  \|{\mathfrak z}_-^2(U+i\hat{\lambda})^{-2}\|_\infty =  |U(-1)+i\hat{\lambda}|^{-2}
\end{equation}
Indeed, we have
\begin{equation*}
\begin{array}{ll}
  [{\mathfrak z}_-^2|U+i\hat{\lambda}|^{-2}]^\prime &=
  {\mathfrak z}_-^2|U+i\hat{\lambda}|^{-4}\big(-2\theta\beta^{1/3}[(U-\Im\lambda)^2
  +|\Re\hat{\lambda}|^2]+2U^\prime(U-\Im\lambda )\big)\\[1.5ex]& = {\mathfrak z}_-^2|U+i\hat{\lambda}|^{-4}
  \Big(-2\theta\beta^{1/3}\Big[(U-\Im\lambda)
  + \frac{U^\prime}{2\theta\beta^{1/3}}\Big]^2 +\frac{|U^\prime|^2}{2\theta\beta^{1/3}}
  -2\theta\beta^{1/3}|\Re\hat{\lambda}|^2\Big)\,,
  \end{array}
\end{equation*}
 which is non positive by  \eqref{eq:207}.\\
Hence the maximum of ${\mathfrak z}_-^2|U+i\hat{\lambda}|^{-2}$ is obtained at
$x=-1$. A similar inequality can be established for
${\mathfrak z}_+^2|U+i\hat{\lambda}|^{-2}$. 
Combining \eqref{eq:208} and \eqref{eq:209} yields 
\begin{multline}
\label{eq:210}
  \|{\mathfrak z}_\pm(U+i\hat{\lambda})^{-1/2}\|_1\leq  \|{\mathfrak z}_\pm^{1/2}\|_1
  \,\|{\mathfrak z}_\pm^{1/2}(U+i\hat{\lambda})^{-1/2}\|_\infty\\  \leq 
  C\frac{\beta^{-1/6}}{\theta} |\beta^{1/3}(U(\pm1)+i\hat{\lambda})|^{-1/2}\\ \leq \hat C\frac{\beta^{-1/6}}{\theta
    [1+\beta^{1/3}|\lambda_\pm|]^{1/2}} \,. 
\end{multline}
Note that the replacement of $\lambda$ by $\hat \lambda$ avoids the burden of a
vanishing denominator. 

We now write
\begin{equation}
\label{eq:211}
   |\langle{\mathfrak z}_\pm ,u\rangle| \leq \|{\mathfrak z}_\pm (U+i\hat{\lambda})^{-1/2}\|_1
  \|(U+i\hat{\lambda})^{1/2}u\|_\infty\,.
\end{equation}
To estimate $\|(U+i\hat{\lambda} )^{1/2}u\|_\infty$ in the right hand side of
\eqref{eq:211} , we first obtain a bound for $ \||U-\nu|^{1/2}u\|_\infty$ \,.
Thus, integration by parts yields for all $(x,x_0)\in[-1,1]^2$
\begin{displaymath}
  (U-\nu)|u|^2\big|_{x_0}^x = 2\Re\int_{x_0}^x(U-\nu)\bar{u}u^\prime \,dt -
  \int_{x_0}^xU^\prime |u|^2 \,dt \,.
\end{displaymath}
From which we conclude, by integrating the above for $x\in (-1,+1)$  and
Cauchy-Schwarz inequalities, that 
\begin{displaymath}
  \||U-\nu|^{1/2}u\|_\infty^2 \leq C(\||U-\nu|\, u\|_2 \|u^\prime\|_2+\|u\|_2^2+\||U-\nu|^{1/2}\, u\|_2^2)\,.
\end{displaymath}
By \eqref{eq:111} and \eqref{eq:112} we then have
\begin{displaymath}
  \||U-\nu|^{1/2}u\|_\infty^2 \leq\frac{C}{\beta^{4/3}}\|g\|_2^2 \,.
\end{displaymath}
We now write
\begin{displaymath}
\begin{array}{ll}
  \|(U+i\hat{\lambda})^{1/2}u\|_\infty^2& \leq2\big(
  \||U-\nu|^{1/2}u\|_\infty^2+(|\Re\lambda|+2\beta^{-1/3})\,\|u\|_\infty^2\big)\\[1.5ex] & \leq
  C(\beta^{-4/3}\|g\|_2^2+(|\Re\lambda|+2\beta^{-1/3})\,\|u\|_2\|u^\prime\|_2\big) \,.
\end{array}
\end{displaymath}
 By \eqref{eq:111} we then have
\begin{displaymath}
  (|\Re\lambda|+2\beta^{-1/3})\,\|u\|_2\leq \frac{C}{\beta}\|g\|_2 \,, 
\end{displaymath}
and hence, using\eqref{eq:111} once again to estimate $\|u^\prime\|_2$, we
may conclude that
\begin{equation}
\label{eq:212}
  \|(U+i\hat{\lambda})^{1/2}u\|_\infty\leq\frac{C}{\beta^{2/3}}\|g\|_2\,.
\end{equation}
Combining the above with \eqref{eq:210} and \eqref{eq:211} yields 
\begin{equation}
\label{eq:213}
  |\langle{\mathfrak z}_\pm ,u\rangle| \leq 
  \frac{C}{\theta[1+|\lambda_\pm|\beta^{1/3}]^{1/2}} \beta^{-5/6}\|g\|_2  \,. 
\end{equation}

\paragraph{Bounds for  $\langle{\mathfrak z}_\pm
,\tilde{v}_{\pm,\theta}\rangle$.}~\\

The estimation of $\langle{\mathfrak z}_\pm ,\tilde{v}_{\pm,\theta}\rangle$ follows a similar path
to that of $\langle{\mathfrak z}_\pm ,u\rangle$.  We
begin by writing
\begin{displaymath}
  |\langle{\mathfrak z}_\pm ,\tilde{v}_{\pm,\theta}\rangle|\leq  \| {\mathfrak z}_\pm(U+i\hat{\lambda})^{-1/2}\|_1
  \|(U+i\hat{\lambda})^{1/2}\tilde{v}_{\pm,\theta}\|_\infty\,.
\end{displaymath}
Since $\tilde{v}_{\pm,\theta}$, given by \eqref{eq:197} satisfies the same
problem as $u$ with $g$ replaced by $g_{\pm,\theta}$ we may conclude as in
\eqref{eq:212} that
\begin{displaymath}
    \|(U+i\hat{\lambda})^{1/2}\tilde{v}_{\pm,\theta}\|_\infty\leq\frac{C}{\beta^{2/3}}\|g_{\pm,\theta}\|_2\,. 
\end{displaymath}
Consequently, by \eqref{eq:201} and \eqref{eq:210}  
\begin{equation}
\label{eq:214}
   |\langle{\mathfrak z}_\pm ,\tilde{v}_{\pm,\theta}\rangle| \leq C\frac{\beta^{-2/3}}{[1+|\lambda_\pm|\beta^{1/3}]^{5/4}}\,.
\end{equation}
In a similar manner we can obtain that 
\begin{equation}
\label{eq:215}
  |\langle{\mathfrak z}_\mp,\tilde{v}_{\pm,\theta}\rangle|\leq C\frac{\beta^{-2/3}}{[1+|\lambda_\pm|\beta^{1/3}]^{3/4}[1+|\lambda_\mp|\beta^{1/3}]^{1/2}} \leq C\beta^{-2/3} \,.
\end{equation}
As in \eqref{eq:161} we can now write, in view of \eqref{eq:214} and 
\begin{displaymath}
|\langle{\mathfrak z}_\pm , (\psi_{\pm,\theta} -\tilde{v}_{\pm ,\theta})\rangle-(J_\pm \beta)^{-1/3}|\leq C\beta^{-2/3} \,. 
\end{displaymath}
Combining the above with \eqref{eq:206}, \eqref{eq:214}, and  \eqref{eq:215}
 yields
\begin{displaymath} 
  |A_\pm (g)  | \leq C\beta^{1/3}|\langle{\mathfrak z}_\pm ,u\rangle| + C \beta^{-2/3} |\langle{\mathfrak z}_\mp ,u\rangle|\,.
\end{displaymath}
The above, together with \eqref{eq:213} yields
\begin{equation*}
|A_\pm (g)  | \leq \frac{C\beta^{-1/2}} {\theta} \left( \frac{1}{[1+|\lambda_\pm|\beta^{1/3}]^{1/2}} + \beta^{-1} \frac{1}{[1+|\lambda_\mp |\beta^{1/3}]^{1/2}} \right) \, \|g\|_2 \,.
\end{equation*}
As $|\lambda_+-\lambda_-| =  |U(1)-U(-1)|$,     we obtain for sufficiently large
$\beta$ 
\begin{equation}
  \label{eq:216}
|A_\pm (g)  | \leq \frac{\hat C\beta^{-1/2}} {\theta} \, \frac{1}{[1+|\lambda_\pm|\beta^{1/3}]^{1/2}}  \, \|g\|_2 \,.
\end{equation}
Combining the above with \eqref{eq:199}, \eqref{eq:200}, and
\eqref{eq:202} yields (\ref{eq:198}). 
\end{proof}

\section{Zero traction Orr-Sommerfeld operator}
\label{s7new}
\subsection{A short reminder}
We recall for the commodity of the reader that 
$\B_{\lambda,\alpha,\beta}^\Sf$ is  defined by \eqref{eq:28} and  \eqref{eq:10v}
i.e.
\begin{displaymath}
\B_{\lambda,\alpha,\beta}:= (\LL_\beta -\beta \lambda) \Big(\frac{d^2}{dx^2} -\alpha^2\Big) - i \beta U^{\prime\prime}\,,
\end{displaymath} with domain
\begin{displaymath}
  D(\B_{\lambda,\alpha,\beta}^\Sf)=\{\,u\in H^4(-1,1)\cap H^1_0
  (-1,1)\,|\,u^{\prime\prime}\in H^1_0(-1,1)\,\} \,.
\end{displaymath}
Here 
\begin{displaymath}
\LL_\beta = -\frac{d^2}{dx^2} + i \beta U\,,
\end{displaymath}
and $\LL_\beta^\Df$ is the Dirichlet realization of $\LL_\beta$ in $(-1,+1)$.\\
Finally we recall that the inviscid operator $\A_{\lambda,\alpha}$ associated with $U$ is defined by
\begin{displaymath}
\A_{\lambda,\alpha} =(U+i\lambda) \Big(-\frac{d^2}{dx^2} +\alpha^2\Big) + U^{\prime\prime}
\end{displaymath}
with domain $D(\A_{\lambda,\alpha} )= H^2(-1,+1) \cap H_0^1(-1,+1)$.
\subsection{The case $U^{\prime\prime}\neq 0\,$.}
We now prove 
\begin{proposition}
\label{prop:zero-shear-stress}
For all $r>1$ and  $\hat \delta>0$
  there exist positive $\beta_0$, $\Upsilon$, and $C$ such that, for any $\beta\geq \beta_0$
  and   $U\in\Sg_r$  satisfying \eqref{condsurinf}, it holds that
  \begin{equation}
 \label{eq:217}
      \sup_{
        \begin{subarray}{c}
          \Re\lambda\leq\Upsilon\beta^{-1/3} \\
          0\leq\alpha
        \end{subarray}}\big\|(\B_{\lambda,\alpha,\beta}^\Sf)^{-1}\big\|+
      \Big\|\frac{d}{dx}\circ (\B_{\lambda,\alpha,\beta}^\Sf)^{-1}\Big\|\leq  C \, \beta^{-\big(\frac{1}{2}-\hat \delta\big)}\,.
  \end{equation}
\end{proposition}
\begin{proof}  \strut
 Let $\lambda=\mu+i\nu$. Let   further $0<\Upsilon<\Re \nu_1$ and suppose $\mu
 \leq \Upsilon\beta^{-1/3}$.
Let $\phi \in  D(\B_{\lambda,\alpha,\beta}^\Sf)$ and $f\in L^2(-1,+1)$ satisfy
\begin{equation}\label{eq:205z}
\B_{\lambda,\alpha,\beta}^\Sf \phi=f\,
\end{equation}
and 
\begin{equation}
\label{eq:218}
  v:=\A_{\lambda,\alpha}\phi\,. 
\end{equation}
We note that $v\in D( \LL_\beta^\Df)$ and, defining $g$ by
\begin{equation}\label{eq:198z}
  g:= (\LL_\beta^\Df-\lambda\beta)v\,,
\end{equation}
we have 
\begin{equation}
\label{eq:219} 
   g= -(U+i\lambda)f + 2U^\prime(\phi^{(3)}-\alpha^2\phi^\prime)+U^{\prime\prime}(\phi^{\prime\prime}-\alpha^2\phi)- (U^{\prime\prime}\phi)^{\prime\prime}\,.
\end{equation}
By \eqref{eq:124}, there exist $C>0$ and $\beta_0>0$ such that, for $\beta \geq
\beta_0$ and $p\in (2,+\infty]$ we have 
\begin{equation}
\label{eq:220}
   \|v\|_p \leq \hat C\beta^{-\frac{3p+2}{6p}}\|g\|_2\,.
\end{equation}

As
\begin{multline}
\label{eq:221}
\Re\langle(U^{\prime\prime})^{-1}(\phi^{\prime\prime}-\alpha^2\phi),\B_\lambda^\Sf\phi\rangle=
  \|(U^{\prime\prime})^{-1/2}(\phi^{(3)}-\alpha^2\phi^\prime)\|_2^2+ \\
  +\Re\langle\big((U^{\prime\prime})^{-1}\big)^\prime(\phi^{\prime\prime}-\alpha^2\phi),\phi^{(3)}-\alpha^2\phi^\prime\rangle
- \beta \mu \|\phi^{\prime\prime}-\alpha^2\phi\|_2^2 \,,
\end{multline}
we easily obtain that
\begin{equation}
\label{eq:222}
   \|\phi^{(3)}-\alpha^2\phi^\prime \|_2 \leq C(\|f\|_2+|\mu|^{1/2}\beta^{1/2}\|\phi^{\prime\prime}-\alpha^2\phi\|_2)\,.
\end{equation}
We now write,
\begin{equation}
\label{eq:223}
  (\LL_\beta^\Df-\beta\lambda)(\phi^{\prime\prime}-\alpha^2\phi)=i\beta U^{\prime\prime}\phi+f\,.
\end{equation}
With the aid of \eqref{eq:114} and \eqref{eq:126} we then obtain
\begin{equation}
\label{eq:224}
  \|\phi^{\prime\prime}-\alpha^2\phi\|_2\leq C(\beta^{1/6}\|\phi\|_\infty+\beta^{-2/3}\|f\|_2) \,,
\end{equation}
Substituting the above into \eqref{eq:222} then yields
\begin{equation}
\label{eq:225}
   \|\phi^{(3)}-\alpha^2\phi^\prime \|_2 \leq C(\|f\|_2+|\mu|^{1/2}\beta^{2/3}\|\phi\|_\infty)\,.
\end{equation}
To bound $\|(U^{\prime\prime}\phi)^{\prime\prime}\|_2$ we first use the fact that
 \begin{equation}
 \label{eq:226}
   \langle\phi^{\prime\prime},\phi^{\prime\prime}-\alpha^2\phi\rangle= \|\phi^{\prime\prime}\|_2^2 +
   \alpha^2\|\phi^\prime\|_2^2 \,.
 \end{equation}
Then, by \eqref{eq:224} we obtain  that
 \begin{equation}
 \label{eq:227}
   \|\phi^{\prime\prime}\|_2\leq C(\beta^{1/6}\|\phi\|_\infty+\beta^{-2/3}\|f\|_2) \,,
 \end{equation}
and
 \begin{equation}
 \label{eq:207a}
   \|(U^{\prime\prime}\phi)^{\prime\prime}\|_2\leq C(\beta^{1/6}\|\phi\|_\infty+\beta^{-2/3}\|f\|_2 + \|\phi_{1,2}\|) \,.
\end{equation}
 ~\\
 
 We continue the proof by considering $\B_{\lambda,\alpha,\beta}^\Sf$ in a few different regimes of
 $\lambda$ values. 

\paragraph{Case 1: Bounded $|\lambda|$}~\\
Suppose first that
\begin{subequations}
\label{eq:228}
  \begin{equation}
  |\nu|\leq C_0 \quad ; \quad 0<|\mu|\leq
  \Upsilon\beta^{-1/3}\,, \tag{\ref{eq:228}a,b}
\end{equation}
\end{subequations}
where $C_0\in\R_+$. The value of $C_0$ to be selected above will be
determined in a later stage.

Using \eqref{eq:224} , \eqref{eq:225}, and \eqref{eq:207a}, we obtain
from \eqref{eq:219} that
\begin{displaymath}
  \|g\|_2 \leq C(\|f\|_2 + (|\mu|^{1/2}\beta^{2/3}+\beta^{1/6})[\|\phi\|_\infty+\|\phi\|_{1,2}])\,,
\end{displaymath}
or by Sobolev's embedding
\begin{equation}
\label{eq:229}
  \|g\|_2 \leq C(\|f\|_2 + (|\mu|^{1/2}\beta^{2/3}+\beta^{1/6})\|\phi\|_{1,2})\,.
\end{equation}
\eqref{eq:218}, we may
apply Proposition \ref{prop:constant-sign} 
to the pair $(\phi,v)$ to conclude, by \eqref{eq:56},
that for any $q>1$ there exists $C>0$
such that
\begin{equation}\label{eq:208a}
 \|\phi\|_{1,2}\leq C (\|v^\prime\|_q+\|v\|_\infty)\,.
\end{equation}
 Similarly, by  \eqref{eq:57}, for any $p>2$ there exists $C>0$ such
that
\begin{equation}\label{eq:208b}
  |\mu|^{1/2}\beta^{2/3}\|\phi\|_{1,2} \leq |\mu|^{1/2-1/p}\beta^{2/3}\|v\|_p \,.
\end{equation}
Hence,
\begin{displaymath}
   \|g\|_2 \leq C\, \big(\|f\|_2 +
   |\mu|^{1/2-1/p}\beta^{2/3}\|v\|_p+\beta^{1/6}(\|v^\prime\|_q+\|v\|_\infty)\big)\,.
\end{displaymath}
We may now use \eqref{eq:220} to obtain, for $p>2$,
\begin{displaymath}
\begin{array}{ll}
|\mu|^{1/2-1/p}\beta^{2/3}\|v\|_p & \leq C  |\mu|^{\frac 12 -\frac 1p} \beta^{\frac 23  - \frac{3p+2}{6p}} \|g\|_2 \\
& \leq \hat C \Upsilon^{\frac 12 -\frac 1p}  \beta^{\frac 23  - \frac{p-2}{6p} - \frac{3p+2}{6p}} \|g\|_2\\
& =  \hat C \Upsilon^{\frac 12 -\frac 1p} \|g\|_2
\end{array}
\end{displaymath} 
Applying \eqref{eq:220} (which is valid for $p=+\infty$ as well) once again
yields
\begin{displaymath}
\beta^\frac 16 \|v\|_\infty \leq C \beta^{\frac 16 -\frac 12} \|g\|_2 =  C \beta^{ -\frac 13} \|g\|_2\,.  
\end{displaymath}
Finally,  we apply \eqref{eq:116}  (with $q=p$) to the pair
$(v,g)$ satisfying \eqref{eq:198z}) to conclude, for $1<q<2$,  that 
\begin{displaymath}
  \beta^{1/6}\|v^\prime\|_q\leq \frac{C_p}{\beta^{\frac{1}{3r}}}\, \|g\|_2\,.
\end{displaymath}
Combining the above we then obtain, for sufficiently small $\Upsilon$ and
$\beta^{-1}$,
\begin{equation}
\label{eq:330a}
  \|g\|_2\leq C\, \|f\|_2\,.
\end{equation}
From \eqref{eq:116}  and \eqref{eq:330a} we now get, for
any $p\in (1,2)$,
\begin{displaymath}
  \|v^\prime\|_p  \leq C \beta^{- \frac{2+p}{6p }}  \| g\|_2 \leq \hat C  \beta^{- \frac{2+p}{6p }}\, \| f \|_2 \,.
\end{displaymath}
From which we deduce, for $p$ sufficiently close to $1$ 
 and sufficiently small $\delta>0$, the existence of 
$C_\delta>0$  such that
\begin{displaymath}
  \|v^\prime\|_p +\|v\|_\infty \leq C_\delta\beta^{-\big(\frac{1}{2}-\delta\big)}\|f\|_2.
\end{displaymath}
We now return to \eqref{eq:208a} to conclude that
\begin{equation}
\label{eq:230}
  \|\phi\|_{1,2} \leq  \hat C_\delta\beta^{-\big(\frac{1}{2}-\delta\big)}\|f\|_2\,. 
\end{equation}
 Hence \eqref{eq:217} is proven in Case 1 for sufficiently small
 $\Upsilon>0$.
\vspace{1ex}

\paragraph{Case 2: $\Re\lambda$ unbounded negative} \strut  \\
Next, consider the case where (\ref{eq:228}a) is kept in place but
instead of (\ref{eq:228}b), \break $\mu<-\Upsilon\beta^{-1/3}$ is assumed.  In this case
we return to \eqref{eq:221} and use the positivity of $-\beta \mu$ for the
last term on its right hand side. In a similar manner to the one used
to derive \eqref{eq:222} we establish the existence of $\beta(\Upsilon)>0$ such
that for all $\beta \geq \beta(Y)$,
\begin{displaymath}
  \|\phi^{\prime\prime}-\alpha^2\phi\|_2 \leq \frac{C}{\Upsilon\beta^{2/3}}\|f\|_2 \,.
\end{displaymath}

Since $\phi\in H^1_0(-1,1)$ it can be easily verified that
\begin{equation}
\label{eq:231}
   \|\phi\|_{1,2}\leq\frac{ 4}{\pi}\|\phi^{\prime\prime}-\alpha^2\phi\|_2\,, 
\end{equation}
and hence for any $\Upsilon >0$, there exist $C $ and $\beta_0$  such that,
for all $\alpha \geq 0$ and all $\beta \geq \beta_0$
\begin{displaymath}
   \|\phi\|_{1,2}\leq C\beta^{-2/3}\|f\|_2 \,,
\end{displaymath}
establishing \eqref{eq:217} in Case 2.

\begin{remark}
  
\end{remark}

\paragraph{Case 3: $|\Im\lambda|$ unbounded}~\\
Consider next the case, where instead of \eqref{eq:228} we have
\begin{equation}
\label{eq:232}
  |\nu| \geq C_0 \,.
\end{equation}
In this case we write
\begin{displaymath} 
  -\Im\langle\phi^{\prime\prime}-\alpha^2\phi,\B_\lambda^\Sf\phi\rangle=\beta\big[\langle\phi^{\prime\prime}-\alpha^2\phi,(\nu-U)(\phi^{\prime\prime}-\alpha^2\phi)\rangle +
  \langle\phi^{\prime\prime}-\alpha^2\phi,U^{\prime\prime}\phi\rangle\big] 
\end{displaymath}
As $\phi\in H^1_0(-1,1)$ we have that
\begin{displaymath}
  \|\phi\|_2\leq \frac{4}{\pi^2}\|\phi^{\prime\prime}-\alpha^2\phi\|_2 \,,
\end{displaymath}
and hence, under \eqref{eq:232},  
\begin{displaymath}
|  \langle\phi^{\prime\prime}-\alpha^2\phi,(\nu-U)(\phi^{\prime\prime}-\alpha^2\phi)\rangle +
  \langle\phi^{\prime\prime}-\alpha^2\phi,U^{\prime\prime}\phi\rangle|  \geq (C_0-\|U\|_\infty-4\pi^{-2}\|U^{\prime\prime}\|_\infty)  \|\phi^{\prime\prime}-\alpha^2\phi\|_2^2\,.
\end{displaymath}
Selecting $C_0>\|U\|_\infty+4\pi^{-2}\|U^{\prime\prime}\|_\infty$ we then obtain 
\begin{displaymath}
   \|\phi^{\prime\prime}-\alpha^2\phi\|_2 \leq \frac{C}{\beta}\|f\|_2 \,.
\end{displaymath}
We can now conclude  from \eqref{eq:231} that
\begin{displaymath}
   \|\phi\|_{1,2} \leq \frac{C}{\beta}\|f\|_2 \,.
\end{displaymath}

As (\ref{eq:55})- (\ref{eq:57}) are not valid for  $\mu=0$, we need to
establish \eqref{eq:230} for the case $\mu =0$.
 Let 
\begin{displaymath}
  \B_{i\nu,\alpha}^\Sf\phi=f \,.
\end{displaymath}
Then, we may write by \eqref{eq:221} that
\begin{displaymath}
  \|\phi^{\prime\prime}+\alpha^2\phi\|_2 \leq C\|f\|_2\,.
\end{displaymath}
Then, as 
\begin{displaymath}
  \B_{i\nu+\beta^{-1},\alpha}^\Sf\phi=\phi^{\prime\prime}+\alpha^2\phi + f \,,
\end{displaymath}
we may conclude \eqref{eq:217} from \eqref{eq:230}. 
\end{proof}
\subsection{The nearly Couette case}
\label{sec:nearly-couette-case}
We now proceed to consider the nearly Couette case addressed by
  both Theorems~\ref{thm:traction} and \ref{thm:no-slip}.  Let $U\in
  C^4([-1,1])$  satisfy  \eqref{eq:19} 
and recall the definition of $\delta_2(U)$ from (\ref{defdelta3}),
\begin{displaymath}
\delta_2(U):=\|U^{\prime\prime}\|_{1,\infty}\,.
\end{displaymath}
We next recall from \eqref{eq:22}, for some $r>1$,  
\begin{displaymath}
    \Sg_r=\{v\in C^4([-1,1]) \,|\,, \,   \inf_{x\in[-1,1]}v^\prime \geq  1/r  \;\mbox{ and } \; \|v\|_{4,\infty}\leq r\} \,.
  \end{displaymath}

We shall consider the case where $\delta_2(U)$ is small. 
    Unlike the Couette case, where $U(x)=x$, $\B_{\lambda,\alpha,\beta}^\Sf$ is no
    longer accretive when $\delta_2(U)>0$  \cite[Subsection 6.1]{chen2018transition}.
Note that in contrast with the assumptions of Proposition
\ref{prop:zero-shear-stress}, $U^{\prime\prime}$ may change its sign.
\begin{proposition}
  \label{prop:zero-shear-stress-extension}
For any $r>1$ and  $\Upsilon<\Re\nu_1$,  there exist $\delta\in (0,\frac 14) $, and positive  
   $\beta_0$ and $C$, such that, for all  $U\in \Sg_r$   satisfying 
 $0<\delta_2(U)\leq \delta$ and 
 all $\beta \geq \beta_0$,   
 \begin{equation}
\label{eq:233}
      \sup_{
        \begin{subarray}{c}
          \beta^{1/3}\Re\lambda\leq  {\mathfrak J}_m^{2/3}\Upsilon \\
          0\leq\alpha
        \end{subarray}}\Big\|(\B_{\lambda,\alpha,\beta}^\Sf)^{-1}\Big\|+
      \Big\| \frac{d}{dx}\, (\B_{\lambda,\alpha,\beta}^\Sf)^{-1}\Big\|\leq  \frac{C}{\beta^{5/6}}\,.
  \end{equation}
\end{proposition}
\begin{proof}~\\
Note that by \eqref{eq:22}
\begin{equation}
\label{eq:234}
r \geq {\bf m} \geq \frac 1r\,.
\end{equation}
As in the proof of Proposition \ref{prop:zero-shear-stress} we
separately consider different regimes of $\lambda\in\C$.\\

  {\bf  Case 1: $\Re\nu_1/2\leq  {\mathfrak J}_m^{-2/3}\beta^{1/3}\mu\leq\Upsilon$ or
  $\beta^{1/3}\mu\leq- {\mathfrak J}_m^{2/3}\Re\nu_1/2$ }\\

 Let $f\in L^2(-1,1)$. Using the definitions
  \eqref{eq:205z}-\eqref{eq:218} we further set
  \begin{displaymath}
    \tilde{v}=\frac{v}{U+i\lambda} \,,
  \end{displaymath}
Clearly
\begin{equation}
  \label{eq:235}
(\LL_\beta^\Df-\beta\lambda)\tilde{v}= f + \Big(\frac{U^{\prime\prime}\phi}{U+i\lambda}\Big)^{\prime\prime} =h\,.
\end{equation}
We now write
\begin{displaymath}
  \|v\|_2\leq \|(U+i\lambda)(\LL_\beta^\Df-\beta\lambda)^{-1}h\|_2
\end{displaymath}
By \eqref{eq:114} and \eqref{eq:115} we then have
\begin{equation}
  \label{eq:236}
\|v\|_2\leq \frac{C}{\beta}\|h\|_2\,.
\end{equation}

Next, we turn to estimate $\|h\|_2$. Clearly, by \eqref{eq:235},
\begin{multline}
\label{eq:237}
  \|h\|_2\leq \|f\|_2+C(r)\Big(\Big\|\frac{\phi^{\prime\prime}}{|U+i\lambda|}\Big\|_2+
  \Big\|\phi^\prime\Big(1+\frac{1}{|U+i\lambda|^2}\Big)\Big\|_2\\
 + \Big\|\frac{U^{\prime\prime}\phi}{|U+i\lambda|^3}\Big\|_2+\Big\|\phi\Big(1+\frac{1}{|U+i\lambda|^2}\Big)\Big\|_2\Big)  \,.   
\end{multline}
where we have use that fact that
  \begin{displaymath}
     \frac{1}{|U+i\lambda|}\leq \frac{1}{2}\Big(1+\frac{1}{|U+i\lambda|^2}\Big) \,.
  \end{displaymath}
For the third term in the right-hand-side of \eqref{eq:237} we have
\begin{equation}
  \label{eq:238}
 \Big\|\phi^\prime\Big(1+\frac{1}{|U+i\lambda|^2}\Big)\Big\|_2\leq C \beta^{2/3}\|\phi^\prime\|_2\,.
\end{equation}
We now turn to estimate the second term on the right-hand-side of
\eqref{eq:237}. Here we write
\begin{displaymath}
  \Big\|\frac{\phi^{\prime\prime}}{U+i\lambda}\Big\|_2 \leq
  C \beta^{1/3}\|\phi^{\prime\prime}\|_2 \,.
\end{displaymath}
Then, we use \eqref{eq:227}  (which remains valid in the
nearly-Couette case) to obtain
\begin{displaymath}
  \Big\|\frac{\phi^{\prime\prime}}{U+i\lambda}\Big\|_2 \leq
  C(r) (\beta^{1/2}\|\phi\|_\infty+\beta^{-1/3}\|f\|_2) \,.
\end{displaymath}
Sobolev embedding then yields
\begin{equation}
  \label{eq:239}
 \Big\|\frac{\phi^{\prime\prime}}{U+i\lambda}\Big\|_2 \leq C(r) (\beta^{1/2}\|\phi\|_{1,2} +\beta^{-1/3}\|f\|_2) \,.
\end{equation}

Finally, we turn to estimate the last term on the right-hand-side of
\eqref{eq:237}. \\
Suppose first  that for some  $x_\nu\in[-1,1]$ we have
$U(x_\nu)=\nu$. Clearly, as
\begin{equation}\label{mino62}
  |U(x)+i\lambda| \geq \Big|\mg(x-x_\nu)+\frac{1}{2}i\, {\mathfrak J}_m^{2/3}\Re\nu_1\, \beta^{-1/3}\Big|
\end{equation}
it holds by  \eqref{eq:234} that 
\begin{displaymath}
  \Big\|\frac{U^{\prime\prime}\phi}{|U+i\lambda|^3}\Big\|_2\leq C \delta\, 
  \Big\|\frac{\phi}{|(x-x_\nu)+i\, {\mathfrak J}_m^{2/3}\Re\nu_1\, \beta^{-1/3}|^3}\Big\|_2 \,.
\end{displaymath}
Hence,  we obtain
\begin{equation}
\label{eq:240}
  \Big\|\frac{U^{\prime\prime}\phi}{|U+i\lambda|^3}\Big\|_2\leq \hat C \delta\beta^{5/6}\|\phi\|_\infty \,.
\end{equation}
Substituting the above together with \eqref{eq:239} and \eqref{eq:238}
into \eqref{eq:237} yields with the aid of Sobolev embedding and
Poincare inequality
\begin{displaymath} 
  \|h\|_2\leq \|f\|_2 + C(r)  (\delta \beta^{5/6} + \beta^{\frac 23} ) \|\phi\|_{1,2} \,,
\end{displaymath}
which for $\beta \geq \beta_0(\delta)$ gives
\begin{equation} \label{eq:abcd}
  \|h\|_2\leq \|f\|_2 + \hat C(r) \, \delta \,  \beta^{5/6} \,  \|\phi\|_{1,2} \,.
\end{equation}
By \eqref{eq:236} we then have
\begin{equation}
\label{eq:241}
  \|v\|_2\leq \frac{C}{\beta}\|f\|_2 +  C(r)\,  \delta \, \beta^{-1/6} \|\phi\|_{1,2}\,.
\end{equation}

The proof can now be completed, for $\Re\nu_1/2\leq
 {\mathfrak J}_m^{-2/3}\beta^{1/3}\mu\leq\Upsilon$ or \break
$\beta^{1/3}\mu\leq- {\mathfrak J}_m^{2/3}\Re\nu_1/2\,$, by using \eqref{eq:57} with $p=2$.
In the case where $U(x)-\nu \neq 0$ on $[-1,+1]$, we may still apply the
previous arguments by replacing \eqref{mino62} with
  \begin{equation}
\label{mino62pm}
  |U(x)+i\lambda| \geq \Big|\mg(x\pm 1)+\frac{1}{2}i {\mathfrak J}_m^{2/3}\Re\nu_1\beta^{-1/3}\Big|
\end{equation}
where $\pm$ denotes $+$
  for $\nu < \inf U (x)$ and $-$ for $\nu > \sup U(x)\,$.\\
 
{\bf Case 2: $\beta^{1/3}|\mu|< {\mathfrak J}_m^{2/3}\Re\nu_1/2$ }~\\

Let
\begin{displaymath}
  s= \beta^{1/3}\Big(\frac{ {\mathfrak J}_m^{2/3}\Re\nu_1}{2}-\mu\Big) \,,
\end{displaymath}
and then write
\begin{displaymath}
 \B_{\lambda+s\beta^{-1/3},\alpha}^{\mathfrak S} \phi = f -s\beta^{2/3}(\phi^{\prime\prime}-\alpha^2\phi) \,.
\end{displaymath}
Since $\Re(\lambda+s\beta^{-1/3})=\beta^{-1/3} {\mathfrak J}_m^{2/3}\Re\nu_1/2\,,$ we can use \eqref{eq:233}
to obtain that
\begin{equation}
\label{eq:242}
  \|\phi\|_{1,2}\leq C (\beta^{-5/6}\|f\|_2+\beta^{-1/6}\|\phi^{\prime\prime}-\alpha^2\phi\|_2) \,.
\end{equation}
We now recall \eqref{eq:223} 
\begin{displaymath}
   (\LL_\beta^\Df-\beta\lambda)(\phi^{\prime\prime}-\alpha^2\phi)=i\beta U^{\prime\prime}\phi+f\,,
\end{displaymath}
and then use \eqref{eq:114} and  \eqref{eq:126} to establish that
\begin{equation}
\label{eq:243}
  \|\phi^{\prime\prime}-\alpha^2\phi\|_2 \leq C(\beta^{-2/3}\|f\|_2+ \delta\beta^{1/6} \|\phi\|_\infty)
\end{equation}
Substituting the above into \eqref{eq:242} yields
\begin{displaymath}
   \|\phi\|_{1,2}\leq C (\beta^{-5/6}\|f\|_2+\delta\|\phi\|_\infty) \,.
\end{displaymath}
Using Sobolev embedding yields \eqref{eq:233} , for sufficiently small
$\delta$ and $\beta^{1/3}\mu\leq \Upsilon$.
\end{proof}

\section{No slip Orr-Sommerfeld operator}
\label{sec:no-slip}
In contrast with the prescribed traction condition, the auxiliary
function $v=\A_{\lambda,\alpha}\phi$, does not satisfy, if $\phi\in
D(B_{\lambda,\alpha,\beta}^\Df),$ a Dirichlet boundary condition. Consequently,
special attention must be given to the behaviour of $v$ near the
boundary through \eqref{eq:152} for instance. Since $\phi$ satisfies a
Dirichlet boundary condition, we expect that the rapidly decaying
boundary terms $\psi_\pm-\tilde{v}_\pm$ in \eqref{eq:152} should have a
negligible contribution to $\phi$ compared with that of
$u$ (i.e.,
$\|\A_{\lambda,\alpha}^{-1}(\psi_\pm-\tilde{v}_\pm)\|_{1,2}\ll\|\A_{\lambda,\alpha}^{-1}u\|_{1,2}$).
The
next subsection is dedicated to the establishment of such estimates.

  \subsection{Preliminaries} 
  We recall that $J_\pm = U^\prime(\pm 1)$ and that $\vartheta_1^r$ and $\psi_{\pm}$ are
  respectively defined by \eqref{deftheta1r}, \eqref{defpsi+} and
  \eqref{defpsi-}.  The next lemma holds true under the assumptions
  of either Proposition \ref{prop:small} or Proposition
  \ref{prop:constant-sign}.
 \begin{lemma}
\label{lem:inviscid-decay}
For any  $r>1$ and $\Upsilon<\vartheta_1^r$, there exist  positive constants 
$C$ and $\beta_0$ such that, for all $\beta\geq \beta_0$, $\lambda\in\C$ for which
$\Re\lambda\leq\beta^{-1/3}J_\pm^{2/3}\Upsilon$, and $U\in \Sg_r$ satisfying either
\eqref{condsurinf}  or  \eqref{eq:52},  it holds that 
   \begin{equation}
\label{eq:244}
 \|\A_{\lambda,\alpha}^{-1}(U+i\lambda)\psi_\pm \|_{1,2} \leq 
 C\, [1+\beta^{1/3}|\lambda_\pm|]^{-1/4} \, \beta^{-1/2}
 \,.
   \end{equation}
where 
\begin{equation}
  \label{eq:245}
\lambda_\pm=\mu + i(\nu-U(\pm1)) \,.
\end{equation}
 \end{lemma}
 \begin{proof}~\\
 Let $\phi=\A_{\lambda,\alpha}^{-1}(U+i\lambda)\psi_\pm $. We write 
    \begin{equation}
\label{eq:246}
       \Big|\Big\langle\frac{\phi}{U+i\lambda},(U+i\lambda)\psi_\pm \Big\rangle\Big| \leq      
      \,\|(\phi-\phi(\pm 1))\psi_\pm \|_1 \,. 
    \end{equation}
As 
\begin{displaymath}
|\phi(x)-\phi(\pm 1)|\leq\|\phi\|_{1,2}\| (1\mp  x)^{1/2}\,,\, \forall x\in(-1,1)\,,
\end{displaymath}
 we may write
 \begin{displaymath}
   \|(\phi-\phi(\pm 1))\psi_\pm \|_1\leq \|\phi\|_{1,2}\|(1\mp x)^{1/2}\psi_\pm \|_1\,.
 \end{displaymath}
 It follows from (\ref{eq:148a}) (with $s=\frac 12$) that for some
 positive $C$
 \begin{equation}
 \label{eq:247}
   \|(1\mp x)^{1/2}\psi_\pm \|_1 \leq C\,[1+|\lambda_\pm|\beta^{1/3}]^{-1/4}\,\beta^{-1/2}\,.
 \end{equation}
 In the case where \eqref{eq:52} is satisfied, we may use
 \eqref{eq:68} so that
 \begin{displaymath}
   \|\phi\|_{1,2} \leq  C\|(1\mp x)^{1/2}\psi_\pm \|_1 \,,
 \end{displaymath}
 which, combined with \eqref{eq:247} yields \eqref{eq:244}. \\

In the case where \eqref{condsurinf} is satisfied, we may
 use \eqref{eq:107}  for $v=(U+i\lambda)\psi_\pm$ to obtain that
 \begin{displaymath}
    \|\A_{\lambda,\alpha}^{-1}(U+i\lambda)\psi_\pm \|_{1,2} \leq  C \|(1\mp x)^{1/2}\psi_\pm \|_1\,.
 \end{displaymath}
Then, we apply \eqref{eq:247} to obtain \eqref{eq:244}.
 \end{proof}

We shall also need the following
 \begin{lemma}
   Let $r>1$ and $\Upsilon<\vartheta_1^r$.  Let further $(\lambda,\check \lambda)\in\C^2$
   satisfy $\lambda-\check \lambda\in\R$,  $\Re \lambda $ and $ \Re \check \lambda
   $ are in $( -\infty, \beta^{-1/3}J_\pm^{2/3}\Upsilon)$ and $|\Re\lambda-\Re\check{\lambda}|\leq
   2\vartheta_1^r\beta^{-1/3}$. Then, there exist positive $C$ and $\beta_0$ such that,
   for all $\beta\geq \beta_0$, and all $U\in \Sg_r$ satisfying either
   \eqref{condsurinf} or \eqref{eq:52}, it holds that
   \begin{equation}
\label{eq:248}
 \|\A_{\check{\lambda},\alpha}^{-1}(U+i\check{\lambda})\Gamma_{[-1,1]}(\widetilde{\LL}_{\beta,\R}-\beta\check{\lambda})^{-1}\tilde{\psi}_\pm (\lambda)\|_{1,2} \leq 
 C\,[1+|\lambda_\pm|\beta^{1/3}]^{-1}\,\beta^{-7/6}
 \,,
   \end{equation}
where
\begin{displaymath}
  \tilde{\psi}_\pm(x,\lambda) =
  \begin{cases}
    \psi_\pm(x,\lambda) & x\in[-1,1] \\
    0 & |x|>1 \,.
  \end{cases}
\end{displaymath}
 \end{lemma}

 \begin{proof}~\\
For later reference we note that the requirement set above that $|\lambda-\check{\lambda}|\leq
   2\vartheta_1^r\beta^{-1/3}$ implies the existence of $C>0$ such that
\begin{displaymath}
\frac 1C [1+|\lambda_\pm|\beta^{1/3}] \leq [1+|\check \lambda_\pm|\beta^{1/3}]\leq C[1+|\lambda_\pm|\beta^{1/3}]
\end{displaymath}
   and
\begin{displaymath}
 \frac 1C   (1+\beta^{1/3}|\mu|) \leq (1+\beta^{1/3}|\check \mu|) \leq C (1+\beta^{1/3}|\mu|) \,,
\end{displaymath}
where $\check \mu=\Re \check \lambda$.
   
   {\em Step 1:} We prove that there exist positive $C$,
  and $\beta_0$ such that for all $\beta>\beta_0$ and $k\in\{0,1,2\}$
   \begin{equation}
     \label{eq:249}
\|(1\mp
x)^k(\widetilde{\LL}_{\beta,\R}-\beta\check{\lambda})^{-1}\tilde{\psi}_\pm(\lambda)\|_2 \leq
 C \beta^{-(5/6+k/2)}\, [1+\beta^{1/3}|\lambda_\pm|]^{-3/4- k/2}\,.
   \end{equation}
   ~\\
For convenience of notation we prove \eqref{eq:249} only for
   $(\widetilde{\LL}_{\beta,\R}-\beta\check{\lambda})^{-1}\tilde{\psi}_+(\lambda)$. The proof for
   $(\widetilde{\LL}_{\beta,\R}-\beta\check{\lambda})^{-1}\tilde{\psi}_-(\lambda)$ can be obtained
   in a similar manner. \\[2ex]

Let $u_k=(1-x)^k(\widetilde{\LL}_{\beta,\R}-\beta\check{\lambda})^{-1}\tilde{\psi}_+$, for
$k\in\{0,1,2,3\}$.  For convenience we also set $u_k\equiv0$ for all $k\leq-1$. \\
\paragraph{The case $k=0$\,.}~\\
For $k=0$, we observe from  \eqref{eq:111} and \eqref{eq:149},  that we have
\begin{equation}
\label{eq:250}
  \|u_0\|_2  \leq \frac{C}{\beta^{5/6}(1+\beta^{1/3}|\mu|)}[1+\beta^{1/3}|\lambda_+|]^{1/4}\,.
\end{equation}
We note that \eqref{eq:249} for $k=0$ does not follow from
\eqref{eq:250}. Some additional estimates for large values
$\beta^{1/3}|\nu|$ should be obtained to this end.   For
$ k\in\{1,2,3\}$, we now write: \\
\begin{equation}
  (\widetilde{\LL}_{\beta,\R}-\beta\check{\lambda})u_k=(1-x)^k\tilde{\psi}_\pm (\lambda) - 2ku_{k-1}^\prime+k(k-1)u_{k-2}\,.
\end{equation}
By \eqref{eq:111} it holds that
\begin{equation}
\label{eq:251}
  \|u_k\|_2 \leq \frac{C}{\beta^{2/3}(1+\beta^{1/3}|\mu|)}( \|(1-x)^k\tilde{\psi}_\pm (\lambda)\|_2 +
  \|u_{k-1}^\prime\|_2+(k-1)\|u_{k-2}\|_2) \,.
\end{equation}
As
\begin{equation}
\label{eq:252}
  \Re\langle u_{k-1},(\widetilde{\LL}_{\beta,\R}-\beta\check{\lambda})u_{k-1}\rangle =
  \|u_{k-1}^\prime\|_2^2 - \beta\check \mu \|u_{k-1}\|_2^2 \,,
\end{equation}
we obtain, as $\check \mu\leq\beta^{-1/3}J_\pm ^{2/3}\Upsilon$,  that
\begin{multline}
\label{eq:253}
   \|u_{k-1}^\prime\|_2 \leq C\, \Big( \beta^{1/3}\|u_{k-1}\|_2 + \beta^{-1/3} \big( \|(1-x)^{k-1}\tilde{\psi}_\pm (\lambda)\|_2 \\
 + (k-1)  \|u_{k-2}^\prime\|_2   + (k-1)(k-2) \|u_{k-3}\| \big) \Big)\,.
\end{multline}
Substituting the above into \eqref{eq:251} yields, for $k=1$, 
\begin{displaymath}
\|u_1\|_2 \leq C \beta^{-\frac 23} (1+ \beta^\frac 13 |\mu|)^{-1}
\left( \| (1-x) \tilde{\psi}_\pm (\lambda)\|_2 + \beta^\frac 13 \|u_0\|_2 + \beta^{-\frac 13} \|\tilde{\psi}_\pm (\lambda)\|_2\right)\,.
\end{displaymath}
Using \eqref{eq:149} for $k=0$ and $k=1$, which holds as
$\beta^{1/3}\Re\lambda\leq\Upsilon<J_+^{2/3}\vartheta_1^r$,  yields
\begin{equation}
\label{eq:254}
   \|u_1\|_2 \leq \frac{C}{\beta^{2/3}(1+\beta^{1/3}|\mu|)}\big([1+\beta^{1/3}|\lambda_+|]^{1/4}
      \beta^{-1/2} + \beta^{1/3}\|u_0\|_2\big) \,.
\end{equation}
From \eqref{eq:250} and  \eqref{eq:254} we then get
\begin{equation}
\label{eq:255}
   \|u_1\|_2 \leq C \, \beta^{-7/6}\,(1+\beta^{1/3}|\mu|)^{-1} \, [1+\beta^{1/3}|\lambda_+|]^{1/4}\,.
\end{equation}
Using \eqref{eq:255} we  can now complete the proof of  \eqref{eq:249}
for $k=0\,$. To this end we write
\begin{displaymath}
  [\beta^{-1/3}+|\lambda_+|]\|u_0\|_2\leq (\beta^{-1/3}+|\mu|)\|u_0\|_2+\|[U-U(1)]u_0\|_2 + \|(U-\nu)u_0\|_2 \,.
\end{displaymath}
From \eqref{eq:112} and \eqref{eq:149}, we deduce that
\begin{displaymath}
\|(U-\nu)u_0\|_2 \leq C \beta^{-1} \|\tilde{\psi}_\pm (\lambda)\|_2 \leq \hat C \beta^{-7/6} \, [1+\beta^{1/3}|\lambda_+|]^{-3/4}\,.
\end{displaymath}
On the other hand, we have, by  (\ref{eq:255}) 
\begin{displaymath}
\begin{array}{ll}
  \|[U-U(1)]u_0\|_2 & \leq C\|u_1\|_2\\ &  \leq \hat C \beta^{-\frac 7 6} ( 1+\beta^{1/3}|\mu|)^{-1}\, [1+\beta^{1/3}|\lambda_+|]^{1/4}\\
  &\leq \hat C \beta^{-\frac 76}  [1+\beta^{1/3}|\lambda_+|]^{1/4}\,.
  \end{array}
\end{displaymath}
Together with \eqref{eq:250}, we obtain
\begin{equation}
\label{eq:256}
 \|u_0\|_2 \leq \frac{C}{\beta^{5/6}}[1+\beta^{1/3}|\lambda_+|]^{-3/4}\,,
\end{equation}
which proves \eqref{eq:145}  for $k=0\,$.\\
\paragraph{The case $k=1$\,.}~\\
By \eqref{eq:252} we may write, instead of \eqref{eq:253},
\begin{displaymath}
  \|u_0^\prime\|_2 \leq C\big(\beta^{1/3}[1+\beta^{1/3}|\lambda_+|]^{1/2}\|u_0\|_2 + \beta^{-1/3}[1+\beta^{1/3}|\lambda_+|]^{-1/2}\|\tilde{\psi}_\pm (\lambda)\|_2\big)\,.
\end{displaymath}
Hence, by \eqref{eq:149}, 
it holds that
\begin{equation}
\label{eq:257}
   \|u_0^\prime\|_2\leq \frac{C}{\beta^{1/2}}[1+\beta^{1/3}|\lambda_+|]^{-1/4}\,,
\end{equation}
which, when substituted into \eqref{eq:251} for $k=1$\,, yields
\begin{displaymath}
  \|u_1\|_2 \leq \frac{C}{\beta^{7/6}(1+\beta^{1/3}|\mu|)}[1+\beta^{1/3}|\lambda_+|]^{-1/4}\,.
\end{displaymath}
 Substituting into \eqref{eq:253} yields, for $k=2\,$,
 \begin{equation}
 \label{eq:258}
   \|u_1^\prime\|_2 \leq \frac{C}{\beta^{5/6}}[1+\beta^{1/3}|\lambda_+|]^{-1/4}\,,
 \end{equation}
and hence, by \eqref{eq:251} for $k=2\,$,
\begin{equation}
\label{eq:259}
  \|u_2\|_2 \leq \frac{C}{\beta^{3/2}(1+\beta^{1/3}|\mu|)}[1+\beta^{1/3}|\lambda_+|]^{-1/4}\,.
\end{equation}
As above we now write
\begin{equation}
\label{eq:260}
  [\beta^{-1/3}+|\lambda_+|]\|u_1\|_2\leq (\beta^{-1/3}+|\mu|)\|u_1\|_2+C\|u_2\|_2 + \|(U-\nu)u_1\|_2 \,,
\end{equation}
to obtain from  \eqref{eq:259}, \eqref{eq:258},  \eqref{eq:112}, and
(\ref{eq:149}) that
\begin{equation}
\label{eq:261}
  \|u_1\|_2 \leq \frac{C}{\beta^{7/6}}[1+\beta^{1/3}|\lambda_+|]^{-5/4}\,.
\end{equation}
\paragraph{The case $k=2$}~\\
We briefly repeat the same argument as in the case $k=1$. By
\eqref{eq:252} with $k=2$ we may write (instead of using  \eqref{eq:253})
\begin{displaymath}
  \|u_1^\prime\|_2 \leq C\big(\beta^{1/3}[1+\beta^{1/3}|\lambda_+|]^{1/2}\|u_1\|_2 +
  \beta^{-1/3}[1+\beta^{1/3}|\lambda_+|]^{-1/2}[\|(1-x)\tilde{\psi}_\pm (\lambda)\|_2+\|u_0^\prime\|_2]\big)\,.
\end{displaymath}
From which we obtain, with the aid of \eqref{eq:261} and \eqref{eq:257}
\begin{equation}
\label{eq:262}
  \|u_1^\prime\|_2 \leq  \frac{C}{\beta^{5/6}}[1+\beta^{1/3}|\lambda_+|]^{-3/4}\,.
\end{equation}
Substituting \eqref{eq:262} into \eqref{eq:251} with $k=2$ then yields
with the aid of \eqref{eq:256}
\begin{displaymath}
  \|u_2\|_2\leq \frac{C}{\beta^{3/2}(1+\beta^{1/3}|\mu|)}[1+\beta^{1/3}|\lambda_+|]^{-3/4}\,.
\end{displaymath}
Substituting the above, \eqref{eq:262}, and \eqref{eq:256} into
\eqref{eq:253}, with $k=3$ yields 
\begin{displaymath}
  \|u_2^\prime\|_2 \leq \frac{C}{\beta^{7/6}(1+\beta^{1/3}|\mu|)}[1+\beta^{1/3}|\lambda_+|]^{-3/4}\,.
\end{displaymath}
From \eqref{eq:251} with $k=3$ we then obtain
\begin{displaymath}
   \|u_3\|_2\leq \frac{C}{\beta^{11/6}(1+\beta^{1/3}|\mu|)^2}[1+\beta^{1/3}|\lambda_+|]^{-3/4}\,.
\end{displaymath}
As in \eqref{eq:260} and \eqref{eq:261} we can now obtain that
\begin{displaymath}
    \|u_2\|_2 \leq \frac{C}{\beta^{3/2}}[1+\beta^{1/3}|\lambda_+|]^{-7/4}\,.
\end{displaymath}
Combining the above with \eqref{eq:261} and \eqref{eq:256} yields
\eqref{eq:254}. \\

\noindent {\em Step 2:} We prove \eqref{eq:248}.\\
We first observe that by interpolation \eqref{eq:249} holds for any $k\in [0,2]$.
If the conditions of \eqref{prop:small} are met, we may now obtain
\eqref{eq:248} from \eqref{eq:52} as in the proof of
\eqref{eq:244}. \\
Otherwise if the assumptions of Proposition \ref{prop:constant-sign}
are met, we may conclude \eqref{eq:244} from \eqref{eq:107}.
\end{proof}
\subsection{Nearly Couette flows}
We begin by considering the case where the flow $U$ is nearly linear,
 as in Subsection~\ref{sec:nearly-couette-case}.

\subsubsection{The case $0\leq \alpha \leq \tilde \alpha_0\beta^{1/3}$}
Let $\B_{\lambda,\alpha,\beta}^\Df$ be defined by \eqref{eq:28}. We can now state and prove
\begin{proposition}
\label{prop:no-slip}
For every $r>1$ and  $\Upsilon<\vartheta_1^r$, there exist positive $\delta$, $\beta_0$, $\alpha_0$,
and $C$ such that, for all $U\in \Sg_r$ satisfying  $0<\delta_2(U)\leq \delta$  (where $\delta_2$ is given by
\eqref{defdelta3}) and  $\beta\geq \beta_0$, it holds that
  \begin{equation}
 \label{eq:263}
      \sup_{
        \begin{subarray}{c}
         \Re\lambda\leq\Upsilon  {\mathfrak J}_m^{2/3}\beta^{-1/3} \\
          0\leq\alpha\leq \alpha_0\beta^{1/3}
        \end{subarray}}\big\|(\B_{\lambda,\alpha,\beta}^\Df)^{-1}\big\|+ \big\|\frac{d}{dx}\, (\B_{\lambda,\alpha,\beta}^\Df)^{-1}\big\|\leq C\beta^{-5/6}\,.
  \end{equation}
\end{proposition}
\begin{proof}
  Let $\phi\in D(\B_{\lambda,\alpha,\beta}^\Df)$ and $f\in L^2(-1,1)$ satisfy $\B_{\lambda,\alpha,\beta}^\Df\phi=f$. Let
  $\lambda=\mu+i\nu$. Let
  further $v$ be given by \eqref{eq:218}, and set
  $\tilde{v}=(U+i\lambda)^{-1}v$. Without any loss of generality we
    select 
\begin{displaymath}
\Upsilon>\vartheta_1^r/2\,.
\end{displaymath}

{\em Step 1:}  We prove \eqref{eq:263} for $\lambda\in\C$ satisfying either
$\vartheta_1^r/2\leq  {\mathfrak J}_m^{-2/3}\beta^{1/3}\Re\lambda\leq\Upsilon$  or $
-3r\leq\Re\lambda\leq- {\mathfrak J}_m^{2/3}\beta^{-1/3}\vartheta_1^r/2$,  and $|\Im\lambda|\leq3r$.
\vspace{2ex}

Let $\zeta_\pm \in C^2([-1,1])$ satisfy the problem
\begin{equation}
\label{eq:264}
  \begin{cases}
     -(U+i\lambda)\zeta_\pm ^{\prime\prime}+(\alpha^2[U+i\lambda]+U^{\prime\prime})\zeta_\pm = 0 & \text{in }
     (-1,1) \\
     \zeta_\pm (\pm 1) =1 \mbox{ and } \zeta_\pm (\mp1)=0&\,. 
  \end{cases}
\end{equation}
Note that the differential operator on the left-hand-side is
identical with that of $\A_{\lambda,\alpha}$ given by \eqref{eq:36}.  We now write
\begin{equation}\label{deczeta}
\zeta_\pm (x)=\zeta_\pm^0(x) + \exp\{-\alpha(1\mp x)\}\tilde{\eta}_\pm(x)
\end{equation}
in which $\tilde{\eta}_\pm(x)=\eta(1\mp x)$, where $\eta$ is given by
\eqref{eq:120}, and $\zeta_\pm^0\in D(\A_{\lambda,\alpha})$.\\

\paragraph{$\zeta_\pm^0$ and $\zeta_\pm$ estimates} ~\\ We begin by establishing
$L^\infty$ and $W^{1,q}$ estimates for $\zeta_\pm^0$. To this end we first write
\begin{equation}\label{azeta}
  \A_{\lambda,\alpha}\zeta_\pm^0 = -\big(U^{\prime\prime} \tilde{\eta}_\pm
  +2\alpha(U+i\lambda)\tilde{\eta}^\prime_\pm +\tilde{\eta}^{\prime\prime}_\pm
  (U+i\lambda)\big)\exp\{-\alpha(1\mp x)\} \,.
\end{equation}

Note that since the derivatives of $\tilde{\eta}$ are supported on
$[-1/2,1/2]$ we have, since $|\mu|$ is bounded, 
\begin{displaymath}
  \|\big(2\alpha(U+i\lambda)\tilde{\eta}^\prime_\pm +\tilde{\eta}^{\prime\prime}_\pm
  (U+i\lambda)\big)\exp\{-\alpha(1\mp  \cdot)\}\big\|_{1,\infty} \leq C(1+|\nu|)(1+\alpha^2)e^{-\alpha/2} \leq
  C(1+|\nu|)\,.
\end{displaymath}
Note further that, for all $1\leq q\leq 2\,$, we have 
\begin{equation*}
  \|U^{\prime\prime} \tilde{\eta}_\pm \exp\{-\alpha(1\mp \cdot)\}\|_{q} \leq C(1+\alpha)^{-1/q} \,.
\end{equation*}
and 
\begin{equation*} 
  \|U^{\prime\prime} \tilde{\eta}_\pm \exp\{-\alpha(1\mp \cdot)\}\|_{1,q} \leq C(1+\alpha)^{1-1/q} \,.
\end{equation*}
Consequently, for all $1\leq q \leq 2\,$, we have 
\begin{subequations}  \label{eq:265} 
\begin{equation}
  \big\|\A_{\lambda,\alpha}\zeta_\pm^0\big\|_{q}\leq  C(1+|\nu| + (1+\alpha)^{-1/q}) \leq \hat C (1+|\nu|)\,.
\end{equation}
and 
\begin{equation}
  \big\|\A_{\lambda,\alpha}\zeta_\pm^0\big\|_{1,q}\leq  C(1+|\nu| +\alpha^{1-1/q}) \,.
\end{equation}
\end{subequations}

Observing that 
\begin{displaymath}
  \Re\langle\zeta_\pm^0 ,(U+i\lambda)^{-1}\A_{\lambda,\alpha}\zeta_\pm^0 \rangle = \frac 12 | (\zeta_\pm^0)^\prime |^2 + \alpha^2 | \zeta_\pm^0|^2 + I (\zeta_\pm^0,\lambda)\,,
\end{displaymath}
we deduce from  Lemma \ref{lem:positive}  that, for
sufficiently small $\delta$ there exist $\gamma_0>0$ and $C>0$  such that, for
any $\alpha \geq 0$ and any $\lambda \in \mathbb C\setminus \Jg$, (recall that
$|\Re\lambda|> {\mathfrak J}_m^{2/3}\beta^{-1/3}\vartheta_1^r/2$ in this step)
\begin{multline}\label{equid}
   \Re\langle\zeta_\pm^0 ,(U+i\lambda)^{-1}\A_{\lambda,\alpha}\zeta_\pm^0 \rangle\geq \\  \frac{1}{2}\|(\zeta_\pm^0)^\prime \|_2^2 +
   (\gamma_0+\alpha^2)\|\zeta_\pm^0\|_2^2 \geq
   C(1+\alpha)^{1/2}\|(\zeta_\pm^0)^\prime\|_2^{3/2}\|\zeta_\pm^0\|_2^{1/2}  \,.
\end{multline}
Consequently, by \eqref{eq:61} (which is applied with $q=p$,
$\phi=\zeta_\pm^0$ and $v=A_{\lambda,\alpha}\zeta_\pm^0$) and \eqref{eq:265}, for any
$1<q<2$ there exists $C_q>0$ such that
\begin{multline*}
  \|(\zeta_\pm^0)^\prime\|_2^{3/2}\|\zeta_\pm^0\|_2^{1/2} \leq
  C_q  \,(1+\alpha)^{-1/2}\, (\|(A_{\lambda,\alpha}\zeta_\pm^0)^\prime\|_q\|\zeta_\pm^0\|_\infty+
  \|A_{\lambda,\alpha}\zeta_\pm^0\|_\infty\|(\zeta_\pm^0)^\prime\|_q) \\[1.5ex] \leq  
   C_q  (1+|\nu|)\, (1+\alpha)^{-1/q+1/2} \, \|(\zeta_\pm^0)^\prime\|_2\,. 
\end{multline*}
(We use the fact that $q \leq 2$ and Poincar\'e's inequality to obtain the last inequality.)
We may thus conclude that
\begin{displaymath}
 \|(\zeta_\pm^0)^\prime\|_2^{1/2}\|\zeta_\pm^0\|_2^{1/2} \leq 
    C_q (1+|\nu|)\, (1+\alpha)^{-1/q+1/2}\,.
\end{displaymath}
Hence, for $q=\frac 32$,  we get 
\begin{displaymath}
  \|\zeta_\pm^0\|_\infty\leq \|(\zeta_\pm^0)^\prime\|_2^{1/2}\|\zeta_\pm^0\|_2^{1/2} \leq C_q (1+|\nu|)\, (1+\alpha)^{-1/6}\,,
\end{displaymath}
  which finally  implies 
\begin{equation}
  \label{eq:266}
 \|\zeta_\pm^0 \|_\infty \leq C \, (1+|\nu|) \,.
\end{equation}
By \eqref{equid} we also have that
\begin{equation*}
   \Re\langle\zeta_\pm^0 ,(U+i\lambda)^{-1}\A_{\lambda,\alpha}\zeta_\pm^0 \rangle\geq \frac{1}{2}\|(\zeta_\pm^0)^\prime \|_2^2 \,,
\end{equation*}
and hence
\begin{displaymath}
\|(\zeta_\pm^0)^\prime \|_2\leq  C_q  (1+|\nu|)\, (1+\alpha)^{-1/q +1} \,.
\end{displaymath}
For $q=\frac 32 $ we may thus conclude 
\begin{equation} \label{eq:336a}
  \|(\zeta_\pm^0)^\prime\|_2 \leq  C (1+|\nu|) (1+\alpha)^{\frac 13}\,.
\end{equation}

As \eqref{eq:266} and \eqref{eq:336a} are unsatisfactory for large
values of $\nu$, we use the fact that
\begin{displaymath}
\begin{array}{ll}
   \Re\langle(U-\nu)^2\zeta_\pm^0 ,(U+i\lambda)^{-1}\A_{\lambda,\alpha}\zeta_\pm^0 \rangle & =
   \|((U-\nu)\zeta_\pm^0)^\prime\|_2^2 -\|U^\prime\zeta_\pm^0\|_2^2+\alpha^2\|(U-\nu)\zeta_\pm^0\|_2^2 \\[1.5ex]
  &  \qquad \qquad  - \Re \langle(U-\nu)^2\zeta_\pm^0 ,(U+i\lambda)^{-1}U^{\prime\prime}\zeta_\pm^0 \rangle \,.
  \end{array}
\end{displaymath}
Assuming $|\nu| \geq 2 |U|_\infty$, we first obtain, for any $\varepsilon >0$
\begin{displaymath}
\varepsilon | |(U-\nu) \zeta_\pm^0 \|_2^2 + \frac{4}{\varepsilon} \| A_{\lambda,\alpha} \zeta_\pm^0| |_2^2 
\geq    \|((U-\nu)\zeta_\pm^0)^\prime\|_2^2  - C \,  \|(U-\nu) \zeta_\pm^0 \|_2\,\|\zeta_\pm^0\|_2\,.
\end{displaymath}
Poincar\'e's inequalit then yields, for sufficiently small $\varepsilon>0$
\begin{displaymath}
  \nu\|\zeta_\pm^0\|_{1,2} \leq C(\|\A_{\lambda,\alpha}\zeta_\pm^0\|_2+\|\zeta_\pm^0\|_2)\,,
\end{displaymath}
which implies, for sufficiently large $|\nu|$ 
\begin{equation}\label{eq:336c}
\|\zeta_\pm^0\|_{1,2} \leq  C\,.
\end{equation}
We can now use (\ref{eq:265}a), (\ref{eq:266}) and   (\ref{eq:336c})  to obtain, for any $\nu$, 
\begin{equation}\label{eq:336d}
\|\zeta_\pm^0\|_{1,2} \leq  C \, (1+\alpha)^\frac 13\,.
\end{equation}
Combining \eqref{eq:336c} with \eqref{eq:266} implies, in addition, the
following refinement of \eqref{eq:266} to any $\nu$
\begin{equation}\label{eq:336b}
\|\zeta_\pm^0\|_{\infty} \leq  C\,.
\end{equation}
Substituting \eqref{eq:336d} and \eqref{eq:336b} into \eqref{deczeta}
respectively yields
\begin{subequations}   \label{eq:267}
\begin{equation}
\|\zeta_\pm \|_{1,2}\leq C (1+\alpha_0^{1/2}\beta^{1/6})\,.
\end{equation}
and
\begin{equation}
\|\zeta_\pm\|_{\infty} \leq  C\,.
\end{equation}

\paragraph{Application of Lemma \ref{lem:integral-conditions}}\strut\\
It can be  easily verified that 
\begin{equation}
  \label{eq:268}
\langle\zeta_\pm ,\tilde{v}\rangle=0\,.
\end{equation}
\end{subequations}
We can thus consider the problem 
\begin{displaymath}
  \Big(-\frac{d^2}{dx^2}+i\beta(U+i\lambda)\Big)\tilde{v}=h\,,
\end{displaymath}
where $\tilde{v}$ satisfies \eqref{eq:268} and $h $ is given by (see
\eqref{eq:235})
\begin{displaymath}
h = f + \Big(\frac{U^{\prime\prime}\phi}{U+i\lambda}\Big)^{\prime\prime} \,.
\end{displaymath}
By \eqref{eq:267} we may now use Lemma \ref{lem:integral-conditions}
to conclude from \eqref{eq:152} (applied for $(g,u,v) = (h, \tilde u,
\tilde v$)) that $\tilde{v}$ admits the decomposition
\begin{equation}
\label{eq:269}
  \tilde{v}=A_+(h)(\psi_+-\tilde{v}_+) +A_- (h)(\psi_--\tilde{v}_-) + \tilde{u}\,.
\end{equation}
where $A_\pm(h)$ and $\tilde{v}_\pm $ respectively satisfy \eqref{eq:154}  and
\eqref{eq:148}, and $\tilde{u}$ is given by
\begin{equation}
\label{eq:270}
  \tilde{u}=\Gamma_{[-1,1]}(\widetilde{\LL}_{\beta,\R}-\beta\lambda)^{-1}\tilde{h} \,,
\end{equation}
where
\begin{displaymath}
  \tilde{h}(x) =
  \begin{cases}
    h(x) & x\in[-1,1] \\
    0 & \text{otherwise}\,.
  \end{cases}
\end{displaymath}
By \eqref{eq:111} we then have 
\begin{equation}
  \label{eq:271}
\|\tilde{u}\|_2 \leq C\beta^{-2/3}\|h\|_2\,.
\end{equation}

\paragraph{Estimate of $h$.}~\\
Since \eqref{eq:237},   \eqref{eq:238}, and \eqref{eq:240}
are  still valid, we obtain
\begin{equation*}
  \|h\|_2\leq \|f\|_2+ C \, (\delta \,\beta^{5/6}+ \beta^\frac 23) \|\phi\|_{1,2}+
  \Big\|\frac{\phi^{\prime\prime}}{U+i\lambda}\Big\|_2\,.
\end{equation*}
which leads for $\beta \geq \beta_0(\delta)$ to 
\begin{equation}
\label{eq:272}
  \|h\|_2\leq \|f\|_2+ \hat C \, \delta \,\beta^{5/6} \|\phi\|_{1,2}+
  \Big\|\frac{\phi^{\prime\prime}}{U+i\lambda}\Big\|_2\,.
\end{equation}
To bound the last term on the right-hand-side we write, recalling that
$|\mu| \geq  {\mathfrak J}_m^{2/3}\beta^{-1/3}\vartheta_1^r/2$ in this step,
\begin{equation}
\label{eq:273}
  \Big\|\frac{\phi^{\prime\prime}}{U+i\lambda}\Big\|_2\leq C\beta^{1/3}\|\phi^{\prime\prime}\|_2\,.
\end{equation}
Then, we use the fact that by \eqref{eq:218} and \eqref{eq:36}
\begin{displaymath}
  \|\phi^{\prime\prime}-\alpha^2\phi\|_2\leq \|\tilde{v}\|_2+
  \|U^{\prime\prime}(U+i\lambda)^{-1}\phi\|_2\,.
\end{displaymath}
As in the proof of  \eqref{eq:240}, we get
\begin{displaymath}
 \|U^{\prime\prime}(U+i\lambda)^{-1}\phi\|_2 \leq C \delta  \beta^{1/6}\|\phi\|_\infty\,.
\end{displaymath}
We now obtain a bound for the norm of $\tilde v$ by estimating each
term appearing in the right hand side of \eqref{eq:269}.  By
\eqref{eq:148} we get for $\tilde v_\pm$,
\begin{displaymath}
\| \tilde v_\pm\| \leq  C \beta^{-\frac 12}\,,
\end{displaymath}
By 
 \eqref{eq:149} (for $k=0$) we get for $\psi_\pm$,
 \begin{equation}\label{eq:psipm}
\| \psi_\pm\| \leq  C\,  \beta^{-\frac 16}[1+ |\lambda_\pm|\beta^{1/3}]^{\frac{1}{4}}\,,
\end{equation}
and by 
\eqref{eq:154} we get for $A_\pm(h)$
\begin{equation}\label{eq:apmh}
|A_\pm (h)| \leq C\,  \beta^{-\frac 12} \|h\|_2\,.
\end{equation}
Together with \eqref{eq:271} for $\tilde u$, we then
have
\begin{displaymath}
  \|\tilde v\|_2\leq C \beta^{-2/3}[1+|\lambda_\pm|\beta^{1/3}]^{\frac{1}{4}} \| h\|_2 \,,
\end{displaymath}
where $\lambda_\pm$ is given by \eqref{eq:deflambdapm}.\\
Since \eqref{eq:226} remains valid for no-slip conditions  we may
conclude that
\begin{displaymath}
  \|\phi^{\prime\prime}\|_2\leq C(\beta^{-2/3}[1+|\lambda_\pm|\beta^{1/3}]^{\frac{1}{4}}\| h\|_2 + \delta\beta^{1/6}\|\phi\|_\infty)\,.
\end{displaymath}
Since $|\lambda_\pm|$ is bounded in this step, we obtain
\begin{equation}
\label{eq:274}
  \|\phi^{\prime\prime}\|_2\leq C(\beta^{-7/12} \|h\|_2 + \delta\beta^{1/6}\|\phi\|_\infty)\,.
\end{equation}
Combining the above with \eqref{eq:272} and \eqref{eq:273} yields 
\begin{equation}
\label{eq:275} 
  \|h\|_2\leq (1 + C \beta^{-1/4}) \|f\|_2+ C\delta\beta^{5/6}\|\phi\|_{1,2}\,. 
\end{equation}
For later reference we note that by \eqref{eq:240}, \eqref{eq:274} and
\eqref{eq:275} we have
\begin{equation}
  \label{eq:276}
\Big\|\Big(\frac{U^{\prime\prime}\phi}{U+i\lambda}\Big)^{\prime\prime}\Big\|_2 \leq C (\beta^{-1/4}\|f\|_2+ \delta\beta^{5/6}\|\phi\|_{1,2})\,. 
\end{equation}

\paragraph{Estimates of $\phi$}~\\
Next we set, with  \eqref{eq:269} in mind, 
\begin{displaymath}
  v=v_o + v_+ +v_- \,,
\end{displaymath}
where
\begin{equation}\label{eq:160a}
  v_\pm = A_\pm (h) (U+i\lambda)(\psi_\pm -\tilde{v}_\pm )   \,,
\end{equation}
and
\begin{equation}
\label{eq:277}
  v_o = (U+i\lambda)\Gamma_{[-1,1]}\tilde{u} = (U+i\lambda)\Gamma_{[-1,1]}(\widetilde{\LL}_{\beta,\R}-\beta\lambda)^{-1}\tilde{h}   \,.
\end{equation}
Set further
\begin{displaymath}
  \phi_\pm  = \A_{\lambda,\alpha}^{-1}v_\pm  \quad ; \quad \phi_o = \A_{\lambda,\alpha}^{-1}v_o\,.
\end{displaymath}
Clearly,
\begin{equation}\label{eq:decaa}
   \phi=\phi_o + \phi_+ +\phi_- \,.
\end{equation}

To estimate $\phi_\pm$ we write 
\begin{equation}\label{eq:phipm}
  \phi_\pm  =  A_\pm (h) \A_{\lambda,\alpha}^{-1}  (U+i\lambda)\psi_\pm    -  A_\pm (h) \A_{\lambda,\alpha}^{-1}  (U+i\lambda)  \tilde{v}_\pm\,.
\end{equation}
  For the first term in the right hand side, we use \eqref{eq:244} to obtain
\begin{displaymath}
  \| A_\pm (h) \A_{\lambda,\alpha}^{-1}  (U+i\lambda)\psi_\pm  \|_{1,2} \leq C \beta^{-\frac 12} |A_\pm (h)|\,.
\end{displaymath}
   Using  \eqref{eq:57} with $p=2$ and \eqref{eq:148}, yields for the second term
\begin{displaymath}
  \|   A_\pm (h) \A_{\lambda,\alpha}^{-1}  (U+i\lambda)  \tilde{v}_\pm\|_{1,2}  \leq C \beta^{-1}  |A_\pm (h)| \,.
\end{displaymath}
Hence  using \eqref{eq:apmh}, we obtain 
\begin{equation}
\label{eq:278} 
  \|\phi_\pm \|_{1,2}\leq C(\beta^{-1}\|f\|_2+\delta\beta^{-1/6}\|\phi\|_{1,2})\,.
\end{equation}

To estimate $\phi_0$ we first recall that by \eqref{eq:270} and \eqref{eq:277}
\begin{displaymath}
v_o = (U+i\lambda)\Gamma_{[-1,1]}\tilde{u}=  (U+i\lambda) \Gamma_{[-1,1]} (\widetilde{\LL}_{\beta,\R}-\beta\lambda)^{-1}\tilde{h}  \,.
\end{displaymath}
By  \eqref{eq:111} and  \eqref{eq:112}  we then have
\begin{displaymath}
  \|v_o\|_2\leq\frac{C}{\beta}\|h\|_2 \,,
\end{displaymath}
and hence by \eqref{eq:275}
\begin{equation}
\label{eq:279}
  \|v_o\|_2\leq C \beta^{-1} \|f\|_2+ C\delta\beta^{-1/6}\|\phi\|_{1,2}\,. 
\end{equation}
By \eqref{eq:57} (with $p=2$), and the definition of $\phi_o$ we obtain 
\begin{equation}
\label{eq:280}
  \|\phi_o\|_2\leq C(\beta^{-5/6}\|f\|_2+\delta\|\phi\|_{1,2})\,.
\end{equation}
Substituting  \eqref{eq:280} and \eqref{eq:278}  into
\eqref{eq:decaa}  yields,for sufficiently small $\delta$ 
\begin{displaymath}
\|\phi\|_{1,2} \leq C \beta^{-5/6}\|f\|_2\,,
\end{displaymath}
which is precisely  \eqref{eq:263} established in this step  for all
$\lambda\in\C$ such that $|\Im \lambda |\leq 3r$, and either
$\vartheta_1^r/2\leq \beta^{1/3} {\mathfrak J}_m^{-2/3}\Re\lambda \leq\Upsilon$ or $-3r \leq \beta^{1/3} {\mathfrak J}_m^{-2/3}\Re \lambda \leq-\vartheta_1^r/2$\,. \\~\\

{\em Step 2:} We prove \eqref{eq:263} for  sufficiently small $\alpha_0$ and $\lambda$ satisfying 
$\beta^{1/3} {\mathfrak J}_m^{-2/3}|\Re \lambda|<\vartheta_1^r/2$  and $|\Im\lambda | \leq 3r $. \\
\vspace{1ex}

Let $  \mathfrak z_\pm \in C^2([-1,1])$ be given by \eqref{eq:194}. Note that
\begin{equation}
\label{eq:281}
  \|  \mathfrak z_\pm\|_\infty=1 \quad ; \quad  \|  \mathfrak z_\pm\|_{1,2}\leq C(1+\alpha^{1/2}) \leq C(1+
  \alpha^{1/2}_0\beta^{1/6}) \,.
\end{equation}
By \eqref{eq:10v},  \eqref{eq:194}, and two integrations by parts,  it
holds that
\begin{equation}
\label{eq:282}
  \langle  \mathfrak z_\pm ,-\phi^{\prime\prime}+\alpha^2\phi\rangle=0 \,. 
\end{equation}
As \eqref{eq:223} is still valid, and in view of \eqref{eq:281} and
\eqref{eq:282}, we can apply Lemma \ref{lem:integral-conditions},
assuming that $\alpha_0$ is small enough,  with  
 $\zeta_\pm\,$ replaced by $   \mathfrak z_\pm\,$, to obtain for
\begin{displaymath}
g = f+i\beta U^{\prime\prime}\phi \mbox{ and }  v= -\phi^{\prime\prime}+\alpha^2 \phi\,,
\end{displaymath} 
that
\begin{equation}
\label{eq:283}
  -\phi^{\prime\prime}+\alpha^2\phi=B_+(\psi_+-\tilde{v}_+) +B_-(\psi_--\tilde{v}_-) +
  \Gamma_{(-1,1)}(\widetilde{\LL}_{\beta,\R}-\beta\lambda)^{-1} \tilde g 
  \,,
\end{equation}
where 
\begin{displaymath}
B_\pm =A_\pm (  f + i \beta U^{\prime\prime} \phi )\,,
\end{displaymath}
and
\begin{displaymath}
\tilde g:=  \gamma_\R(f+i\beta U^{\prime\prime}\phi)(x) =
  \begin{cases}
    (f+i\beta U^{\prime\prime}\phi)(x)  & x\in[-1,1] \\
    0 & \text{otherwise}\,.
  \end{cases}
\end{displaymath}
By \eqref{eq:154} and the fact that $\|U^{\prime\prime}\|_\infty\leq \delta$ in this
nearly Couette case, it holds that 
\begin{equation}
\label{eq:284}
  |B_\pm |\leq
  C(\beta^{-1/2}\|f\|_2+\delta\beta^{1/3}\log \beta\,\|\phi\|_\infty)\,.
\end{equation}
By  \eqref{eq:148} together with \eqref{eq:284}, we have
\begin{equation}
  \label{eq:285}
\|B_\pm \tilde{v}_\pm \|_2 \leq C(\beta^{-1}\|f\|_2+\delta\beta^{-1/6}\log \beta\,  \|\phi\|_\infty) \,.
\end{equation}
By \eqref{eq:111}   and \eqref{eq:136}   it holds that
\begin{displaymath}
  \|(\widetilde{\LL}_{\beta,\R}-\beta\lambda)^{-1}\gamma_\mathbb R (f+i\beta U^{\prime\prime}\phi)\|_2\leq C(\beta^{-2/3}\|f\|_2+\delta\beta^{1/6}\|\phi\|_\infty) \,. 
\end{displaymath}
Combining the above with \eqref{eq:283} and \eqref{eq:285} then yields
\begin{equation}
  \label{eq:286}
\|w_0\|_2\leq C(\beta^{-2/3}\|f\|_2+\delta\beta^{1/6}\|\phi\|_\infty) \,,
\end{equation}
where
\begin{displaymath}
  w_0=-\phi^{\prime\prime}+\alpha^2\phi-B_+\psi_+- B_-\psi_-\,.
\end{displaymath}
As in the proof of Proposition \ref{prop:zero-shear-stress-extension}, 
we use the result of the previous step by considering $\lambda
+\Upsilon_0\beta^{-1/3}$ for a suitable value of $\Upsilon_0$.  We choose $\Upsilon_0$ such that
\begin{equation}
\label{eq:287}
  {\mathfrak J}_m^{2/3}\frac{\vartheta_1^r}{2}\leq \beta^{1/3}\Re\lambda +  \Upsilon_0 \leq  {\mathfrak J}_m^{2/3}\Upsilon \,. 
\end{equation}
We now write
\begin{displaymath}
  \B_{\lambda+\Upsilon_0\beta^{-1/3},\alpha}^\Df\phi = f -\Upsilon_0\beta^{2/3}(w_0+B_+\psi_++B_-\psi_-) \,,
\end{displaymath}
and introduce the following decomposition of $\phi$
\begin{equation}
\label{eq:288}
  \phi=\chi_0-\Upsilon_0\beta^{2/3}(B_+\chi_++B_-\chi_-) \,,
\end{equation}
where
\begin{displaymath}
  \chi_0=(\B_{\lambda+\Upsilon_0\beta^{-1/3},\alpha}^\Df)^{-1}(f- \Upsilon_0\beta^{2/3}w_0) \quad \text{and}
  \quad \chi_\pm = (\B_{\lambda+\Upsilon_0\beta^{-1/3},\alpha}^\Df)^{-1}\psi_\pm \,.
\end{displaymath}
For convenience we set 
\begin{displaymath}
\check{\lambda}=\lambda+\Upsilon_0 \beta^{-1/3}\,.
\end{displaymath}
 We may now
apply \eqref{eq:263} (with $\lambda$ replaced by $\check{\lambda}$) to obtain,
with the aid of \eqref{eq:286} that
\begin{equation}
  \label{eq:289}
\|\chi_0\|_2 \leq C(\beta^{-5/6}\|f\|_2+\delta\|\phi\|_\infty) \,.
\end{equation}

\paragraph{Estimate of $\chi_\pm$\,. } ~\\
We seek an estimate of $\|\chi_+\|_{1,2}$. To this end we repeat the
same procedure applied in step 1. For convenience of notation we
consider only $\chi_+$ in the following. The same estimates for $\chi_-$ can be
obtained in a similar manner. Let then
\begin{displaymath}
 \check{w}_+  = -\chi_+^{\prime\prime} +\alpha^2\chi_+ + \frac{U^{\prime\prime}}{U+i\check{\lambda}}\chi_+ \,,
\end{displaymath}
and
\begin{equation}
\label{eq:290}
H_+:=  \Big(-\frac{d^2}{dx^2}+i\beta(U+i\check{\lambda})\Big)\check{w}_+\, .
\end{equation}
It can be easily verified that
\begin{displaymath}
  H_+ =\psi_+ + \Big(\frac{U^{\prime\prime}\chi_+}{U+i\check{\lambda}}\Big)^{\prime\prime} \,.
\end{displaymath}
and that $\langle\zeta_\pm,\check{w}_+\rangle=0$. Consequently, we can use Lemma
\ref{lem:integral-conditions} with $\lambda$ replaced by $\check \lambda$.  With
the notation $\psi_\pm=\psi_\pm(\lambda)$, $\check{\psi}_+=\psi_\pm(\check{\lambda})$ ,and
$\check {v}_\pm={\tilde v}_\pm(\check \lambda)$ (as defined in \eqref{eq:146})
Lemma \ref{lem:integral-conditions} yields
\begin{displaymath}
  \check{w}_+ = \hat{B}_+(\check{\psi_+}-{\check v}_+) +
  \hat{B}_-(\check{\psi}_-- {\check v}_-) +  (\widetilde{\LL}_{\beta,\R}-\beta \tilde \lambda )^{-1}\tilde{H}_+ \,,
\end{displaymath}
where  $\hat B_\pm= A_\pm(H_+)$  and 
\begin{displaymath}
  \tilde{H}_+ (x)=
  \begin{cases}
    H_+(x) & x\in (-1,1)\\
    0 & |x|\geq 1 \,.
  \end{cases}
\end{displaymath}
Moreover, we have that
\begin{equation}\label{hatB}
|\hat B_\pm|\leq C \beta^{-\frac 12} \|H_+\|_2\,.
\end{equation}
We now follow the arguments of the previous step with $(\chi_+,\check \psi_+,H_+, \check{w}_+)$ respectively
replacing  $(\phi, f,h,\tilde{v})$. We then reach by the equivalent  of \eqref{eq:276}
\begin{displaymath}
  \|H_+\|_2 \leq (1+C\beta^{-1/4})\|\psi_+\|_2 +  C\delta\beta^{5/6}\|\chi_+\|_{1,2}\,, 
\end{displaymath}
from which, using \eqref{eq:149}, we get
\begin{displaymath}
  \|H _+\|_2 \leq C\beta^{-1/6}[1+|\lambda_\pm|^{1/4}\beta^{\frac{1}{12}}] +  C\delta\beta^{5/6}\|\chi_+\|_{1,2}\,.
\end{displaymath}
Hence, by \eqref{hatB} for $k=0$, we get
\begin{equation}
\label{eq:291}
  |\hat{B}_\pm| \leq C(\beta^{-2/3} [1+|\lambda_\pm|^{1/4}\beta^{\frac{1}{12}}] +
  \delta\beta^{1/3}\|\chi_+\|_{1,2})\,. 
\end{equation}

As above we set
\begin{subequations}\label{eq:defchi+}
\begin{equation}
  \chi_+ = \hat{\phi}_+ +  \hat{\phi}_- +  \hat{\phi}_0 \,,
\end{equation}
where
\begin{equation} 
  \hat{\phi}_\pm  =
  \hat{B}_\pm\A_{\check{\lambda},\alpha}^{-1}(U+i\check{\lambda})(\check \psi_\pm- {\check v}_\pm)  \quad  ;
  \quad \hat{\phi}_o =
  \A_{\check{\lambda},\alpha}^{-1}(U+i\check{\lambda})\Gamma_{[-1,1]}(\widetilde{\LL}_{\beta,\R}-\beta\check{\lambda})^{-1}\tilde{H}_+ \,. 
\end{equation}
\end{subequations}
By \eqref{eq:244} and \eqref{eq:291} we have
\begin{displaymath}
   |\hat{B}_\pm|\, \|\A_{\check{\lambda},\alpha}^{-1}(U+i\check{\lambda})\check{\psi}_\pm\|_{1,2}\leq C\Big(\beta^{-7/6}+
  \delta\frac{\beta^{-1/6}}{[1+|\check{\lambda}_\pm|^{1/4}\beta^{\frac{1}{12}}]}\|\chi_+\|_{1,2}\Big)\,. 
\end{displaymath}
Note here  that as $|\lambda-\check{\lambda}|=\Upsilon_0\beta^{-1/3}$ we have used that, for some $C>1$,
\begin{displaymath}
  C^{-1} [1+|\lambda_\pm|^{1/4}\beta^{\frac{1}{12}}] \leq[1+|\check \lambda_\pm|^{1/4}\beta^{\frac{1}{12}}] \leq C [1+|\lambda_\pm|^{1/4}\beta^{\frac{1}{12}}] \,.
\end{displaymath}
Moreover, by \eqref{eq:291}, \eqref{eq:148}, and \eqref{eq:57} (with
$p=2$) it holds that
\begin{displaymath}
    |\hat{B}_\pm|\, \|\A_{\check{\lambda},\alpha}^{-1}(U+i\check{\lambda}){\check v}_\pm\|_{1,2}
    \leq C(\beta^{-4/3} [1+|\check{\lambda}_\pm|^{1/4}\beta^{\frac{1}{12}}] +
  \delta\beta^{-1/3}\|\chi_+\|_{1,2})\,. 
\end{displaymath}
Combining the above yields
\begin{equation}
\label{eq:292}
  \|  \hat{\phi}_\pm \|_{1,2}\leq C\Big(\beta^{-7/6}+
  \delta\frac{\beta^{-1/6}}{[1+|\lambda_\pm|^{1/4}\beta^{\frac{1}{12}}]}\|\chi_+\|_{1,2}\Big)\,. 
\end{equation}
Next, we estimate $\| \hat{\phi}_o \|_{1,2}$.  By \eqref{eq:248} we have
\begin{equation}
\label{eq:293} 
  \|\A_{\check{\lambda},\alpha}^{-1}(U+i\check{\lambda})\Gamma_{[-1,1]}
  (\widetilde{\LL}_{\beta,\R}-\beta\check{\lambda})^{-1} \tilde{ \psi}_+ \|_{1,2}
  \leq C\frac{\beta^{-7/6}}{ 1+|\nu-U(+1)| \,\beta^{\frac{1}{3}}}  \,.
\end{equation}
Furthermore, by \eqref{eq:112}, \eqref{eq:111}, and \eqref{eq:57}, we
have 
\begin{equation}
\label{eq:294}
  \Big\|\A_{\check{\lambda},\alpha}^{-1}(U+i\check{\lambda})\Gamma_{[-1,1]}
  (\widetilde{\LL}_{\beta,\R}-\beta\check{\lambda})^{-1}
  \Big(\frac{U^{\prime\prime}\chi_+}{U+i\check{\lambda}}\Big)^{\prime\prime} \Big\|_{1,2} \leq
  \frac{C}{\beta^{5/6}}\Big\|\Big(\frac{U^{\prime\prime}\chi_+}{U+i\check{\lambda}}\Big)^{\prime\prime}\Big\|_2
\end{equation}
By \eqref{eq:276} with $\chi_+$ instead of $\phi$ and $\psi_+$ replacing
$f$, it holds that
\begin{displaymath}
  \Big\|\Big(\frac{U^{\prime\prime}\chi_+}{U+i\check{\lambda}}\Big)^{\prime\prime}\Big\|_2\leq C\Big(\beta^{-5/12}[1+|\lambda_\pm|^{1/4}\beta^{\frac{1}{12}}]+
  \delta\beta^{5/6}\|\chi_+\|_{1,2}\Big)\,. 
\end{displaymath}
Combining the above with \eqref{eq:294} leads to
\begin{displaymath}
   \Big\|\A_{\check{\lambda},\alpha}^{-1}(U+i\check \lambda)\Gamma_{[-1,1]}
  (\widetilde{\LL}_{\beta,\R}-\beta \check \lambda)^{-1}
  \Big(\frac{U^{\prime\prime}\chi_+}{U+i\check{\lambda}}\Big)^{\prime\prime} \Big\|_{1,2} \leq
  C\Big(\beta^{-5/4}[1+|\lambda_\pm|^{1/4} \,\beta^{\frac{1}{12}}] +
  \delta\|\chi_+\|_{1,2}\Big)\,,
\end{displaymath}
which, together with \eqref{eq:293}, yields
\begin{displaymath}
   \|\hat{\phi}_o\|_{1,2} \leq C\Big(\beta^{-7/6}+
  \delta\|\chi_+\|_{1,2}\Big)\,. 
\end{displaymath}
Combining the above  with \eqref{eq:defchi+} and \eqref{eq:292} yields, for $\delta>0$ small
enough, 
\begin{displaymath}
  \|\chi_+\|_2\leq C\beta^{-7/6} \,.
\end{displaymath}
In a similar manner we obtain
\begin{displaymath}
    \|\chi_-\|_2\leq  C\beta^{-7/6} \,.
\end{displaymath}
Substituting the above, \eqref{eq:289}, and \eqref{eq:284} into
\eqref{eq:288} yields that \eqref{eq:263} holds for $|\Re\lambda| \leq  
  \beta^{-1/3}  {\mathfrak J}_m^{2/3}\vartheta_1^r/2\,$.\\ 

 {\em Step 3:} We prove \eqref{eq:263} for $\lambda\in\C$ satisfying $\Re\lambda\leq-3r$.\\
\vspace{2ex}
As $f= \B_{\lambda,\alpha,\beta}^\Df\phi$ and
\begin{displaymath} 
  - \Re\langle\phi,\B_{\lambda,\alpha,\beta}^\Df\phi\rangle =
  \|\phi^{\prime\prime}\|_2^2-(\beta\mu-\alpha^2)\|\phi^\prime\|^2_2 -\beta\mu\alpha^2\|\phi\|_2^2 -
  \beta\Im\langle U^\prime\phi,\phi^\prime\rangle \,,
\end{displaymath}
we can conclude that
\begin{displaymath}
  \beta(3r\|\phi^\prime\|^2_2 -\|U^\prime\|_\infty\|\phi^\prime\|_2 \|\phi\|_2 )\leq\|\phi\|_2\|f\|_2\,.
\end{displaymath}
Since $\phi\in H^1_0(-1,1)$ it holds that
\begin{displaymath}
  \|\phi^\prime\|_2 \geq \frac{\pi}{2}\|\phi\|_2 \,,
\end{displaymath}
and hence (recall that $\|U^\prime\|_\infty<r$ and $3r>3 > 2/\pi$) we can conclude that there
exists $C>0$ such that
\begin{equation}
\label{eq:295}
  \|\phi^\prime\|_2\leq \frac{C}{\beta}\|f\|_2 \,,
\end{equation}
completing, thereby, the proof of \eqref{eq:263} for $\Re\lambda\leq-3r$. \\

 {\em Step 4:} We prove \eqref{eq:263} for $\lambda\in\C$ satisfying  $|\Im\lambda|\geq3r$.\\
 
An integration by parts yields
\begin{displaymath}
  \Im \langle\phi,\B_{\lambda,\alpha,\beta}^\Df\phi\rangle = \beta\nu\left(\|\phi^\prime\|_2^2+\alpha^2\|\phi\|_2^2\right)
  -\beta\left(\langle U\phi^\prime,\phi^\prime\rangle+\alpha^2\langle U\phi,\phi\rangle+\Re\langle U^\prime\phi,\phi^\prime\rangle+\langle U^{\prime\prime}\phi,\phi\rangle\right) \,,
\end{displaymath}
and as 
\begin{displaymath}
  \Re\langle U^\prime\phi,\phi^\prime\rangle=-\frac{1}{2}\langle U^{\prime\prime}\phi,\phi\rangle  \,,
\end{displaymath}
we obtain that
\begin{displaymath}
  \beta\Big(|\nu| -\|U\|_\infty- \frac{1}{2}\|
  U^{\prime\prime}\|_\infty\Big)\|\phi^\prime\|_2^2+\alpha^2\beta(|\nu|-\|U\|_\infty)\|\phi\|_2^2 \leq \|\phi\|_2\|f\|_2\,.
\end{displaymath}
As
\begin{displaymath}
 |\nu|  -\|U\|_\infty- \frac{1}{2}\|
  U^{\prime\prime}\|_\infty\geq |\nu|-\frac{3}{2} r \,,
\end{displaymath}
we establish \eqref{eq:295} whenever  $|\nu|\geq3r$.\\
The proposition is proved. 
\end{proof}

\subsubsection{The case $\alpha\geq \alpha^\# \beta^{1/6}$}
We separately treat the case $\#=\Df$ and the case $\#=\Sf$.
\begin{proposition}
  \label{lem:no-slip-large-alpha}
For any $r>1$ and  any 
    $\varkappa>0$
  there exist $\beta_0>0$,  $\alpha^\Df>0$ and $C$ such that for all $\beta\geq \beta_0$ and $U\in \Sg_r$ 
  \begin{equation}
\label{eq:296}
      \sup_{
        \begin{subarray}{c}
          \Re\lambda\leq\beta^{-1/3}(\hat{\mu}_0(\beta^{-1/3}\alpha)-\varkappa) \\
          \alpha^\Df \beta^{1/6}\leq \alpha) 
        \end{subarray}}\Big((1+\alpha)\big\|(\B_{\lambda,\alpha,\beta}^\Df)^{-1}\big\|+
      \Big\|\frac{d}{dx}\, (\B_{\lambda,\alpha,\beta}^\Df)^{-1}\Big\|\Big)\leq
    C\,\beta^{-5/6}\,.
  \end{equation}
\end{proposition}
\begin{proof}
  Let $ \mathfrak z_\pm \in C^2([-1,1])$ be given by \eqref{eq:194}. By
  \eqref{eq:198}, \eqref{eq:223}, and \eqref{eq:282}, we have
 \begin{equation} 
   \label{eq:297}
\|-\phi^{\prime\prime}+\theta^2\beta^{2/3}\phi\|_2\leq
C(\beta^{1/3}\|\phi\|_2+\beta^{-2/3}\|f\|_2)\,,
 \end{equation}
where $\theta=\alpha\beta^{-1/3}$. \\
Hence, 
\begin{equation}\label{eq:298}
  \|\phi^\prime\|_2^2 +\theta^2\beta^{2/3}\|\phi\|_2^2\, = \langle -\phi^{\prime\prime}+\theta^2\beta^{2/3}\phi,\phi\rangle  \leq
  C(\beta^{1/3}\|\phi\|_2^2+\beta^{-2/3}\|f\|_2\|\phi\|_2)\,.
\end{equation}
As $\theta\geq \alpha^\Df \,\beta^{-1/6}$, we obtain that for sufficiently large
$\alpha^\Df$ and $\beta$,
\begin{displaymath}
  \|\phi\|_2\leq \frac{C}{\theta^2\beta^{4/3}}\|f\|_2 \,,
\end{displaymath}
which implies
\begin{displaymath}
 \alpha  \|\phi\|_2\leq \frac{C}{\theta \beta}\|f\|_2 \leq  \frac{C}{ \alpha^\Df\, \beta^{\frac{5}{6}}}\|f\|_2   \,.
\end{displaymath}
Returning to \eqref{eq:298}, we get
\begin{displaymath}
   \|\phi^\prime\|_2^2 \leq  \beta^{-\frac 23} \|f\|_2 \,\|\phi\| _2\leq \frac{C}{\theta^2\beta^{2}  }\|f\|_2^2 \,,
\end{displaymath}
hence
\begin{displaymath}
\|\phi^\prime\|_2\leq  \frac{C }{\theta \beta}   \|f\|_2 \leq\frac{C}{ \alpha^\Df\, \beta^{\frac{5}{6}}}\|f\|_2  
\end{displaymath}
\end{proof}
\begin{remark}
Using the
definition of $\hat{\mu}_0$ from \eqref{eq:32} yields that
\begin{displaymath}
  \hat{\mu}_0(\beta^{-1/3}\alpha)\beta^{2/3}+ \alpha^2= \beta^{2/3}\min\big(
  J_+^{2/3}\mu_0(J_+^{-1/3}\theta)+ \theta^2,  J_-^{2/3}\mu_0(J_-^{-1/3}\theta)+ \theta^2\big)\,,
\end{displaymath}
where $\theta=\beta^{-1/3}\alpha$. Let $\theta_\pm=J_\pm^{-1/3}\theta$. \\
 Using the definition
of $\hat{\mu}_m$ from \eqref{eq:184} we then conclude
\begin{displaymath}
  J_\pm^{2/3}\mu_0(J_\pm^{-1/3}\theta)+ \theta^2 =  J_\pm^{2/3}(\mu_0(\theta_\pm)+ \theta^2_\pm)\geq
  {\mathfrak J}_m^{2/3}\hat{\mu}_m \,.
\end{displaymath}
Consequently we obtain that
\begin{displaymath}
  \hat{\mu}_0(\beta^{-1/3}\alpha)\beta^{2/3}+ \alpha^2 \geq \beta^{2/3} {\mathfrak J}_m^{2/3}\hat{\mu}_m \,.
\end{displaymath}
By the foregoing discussion, we may conclude from \eqref{eq:296} that
  \begin{equation}\label{eq:31aaz} 
      \sup_{  
        \begin{subarray}{c}
          \Re\lambda \leq\beta^{-1/3}( {\mathfrak J}_m^{2/3}\hat{\mu}_m-\varkappa -\beta^{-\frac 23} \alpha^2)\\
          \alpha^\Df \beta^{1/6}\leq\alpha 
        \end{subarray}}\Big((1+\alpha)\big\|(\B_{\lambda,\alpha,\beta}^\Df)^{-1}\big\|+
      \Big\|\frac{d}{dx}\, (\B_{\lambda,\alpha,\beta}^\Df)^{-1}\Big\|\Big)\leq
    \frac{C}{\beta^{5/6}}\,.
  \end{equation}
\end{remark}

A similar estimate holds true also for $\B_{\lambda,\alpha,\beta}^\Sf$. 
\begin{proposition}
  For any $r>1$ and $\Upsilon<\Re\nu_1$ there exist positive $\beta_0$,
  $\alpha^{\Sf}$ and $C$ such that for all $\beta\geq \beta_0$ and $U\in \Sg_r$
  such that
  \begin{equation}
\label{eq:299}
      \sup_{
        \begin{subarray}{c}
          \Re\lambda\leq\Upsilon  {\mathfrak J}_m^{2/3}\beta^{-1/3} \\
          \alpha^{\Sf}\beta^{1/6}\leq\alpha
        \end{subarray}}\left( (1+\alpha )\big\|(\B_{\lambda,\alpha,\beta}^\Sf)^{-1}\big\|+
      \big\|\frac{d}{dx}\, (\B_{\lambda,\alpha,\beta}^\Sf)^{-1}\big\| \right)\leq
      C \beta^{-5/6}\,.
  \end{equation}
\end{proposition}
  \begin{proof}
To prove \eqref{eq:299} we note that
by \eqref{eq:223} we have
\begin{displaymath}
  \|\phi^{\prime\prime}-\alpha^2\phi\|_2\leq C(\beta^{1/3}\|\phi\|_2+\beta^{-2/3}\|f\|_2) \,.
\end{displaymath}
Consequently,
\begin{displaymath}
  \|\phi^\prime\|_2^2 +\alpha^2\|\phi\|_2^2 \leq
  C(\beta^{1/3}\|\phi\|_2^2+\beta^{-2/3}\|f\|_2\|\phi\|_2)\,.
\end{displaymath}
from which \eqref{eq:299} follows as in the  $\Df$- case.     
  \end{proof}

\subsection{Strictly convex/concave flows}

\subsubsection{Large $\alpha$ or $|\Im\lambda|$}
If we assume large $\alpha$ or $|\nu|$ we may obtain resolvent  estimates for
all $U\in C^4([-1,1])$, without the necessity to assume any further
restrictions on $U$ as in the nearly Couette case or in  the case
$U^{\prime\prime}\neq0\,$. 
\begin{proposition}
  \label{lem:no-slip-large-alpha-1}
  Let $r>1$, $\Upsilon<\min(\Re\nu_1,\vartheta_1^r)$, and
  $\#\in\{\Sf,\Df\}$.  Then, there exist $\tilde \alpha_0>0$, $\tilde \alpha_1>0$, $\tilde \nu>0$, $\beta_0>0$, and
  $C$ such that, for all $\beta\geq \beta_0$ and $U \in \Sg_r$ satisfying
  (\ref{eq:55})-(\ref{eq:57}), it holds that
  \begin{equation}
\label{eq:300}
      \sup_{
        \begin{subarray}{c}
          \Re\lambda\leq\Upsilon  {\mathfrak J}_m^{2/3}\beta^{-1/3} \\
          \tilde \alpha_1\leq\alpha \leq \tilde \alpha_0\beta^{1/3}
        \end{subarray}}\big\|(\B_{\lambda,\alpha,\beta}^\#)^{-1}\big\|+ \big\|\frac{d}{dx}\, (\B_{\lambda,\alpha,\beta}^\#)^{-1}\big\|\leq C\beta^{-5/6}\,,
  \end{equation}
and
\begin{equation}
\label{eq:301}   \sup_{
        \begin{subarray}{c}
          \Re\lambda\leq\Upsilon  {\mathfrak J}_m^{2/3}\beta^{-1/3} \\
         d(\Im\lambda,[U(-1),U(1)])\geq\tilde{\nu}\beta^{-1/3}\\
         0\leq \alpha\leq \tilde \alpha_1
        \end{subarray}}\big\|(\B_{\lambda,\alpha,\beta}^\#)^{-1}\big\|+ \big\|\frac{d}{dx}\, (\B_{\lambda,\alpha,\beta}^\#)^{-1}\big\|\leq C\beta^{-5/6}\,.
  \end{equation}
\end{proposition}
\begin{proof}
  Let $f\in L^2(-1,1)$, $\phi\in D(\B_{\lambda,\alpha,\beta}^\#)$ and $h$ satisfy
\begin{displaymath}
\B_{\lambda,\alpha,\beta}^\#\phi=f \mbox{  and  } h= f+ \left(\frac{U^{\prime\prime}\phi}{U+i    \lambda}\right)^{\prime\prime}\,.
\end{displaymath} 
{\bf Proof of  \eqref{eq:301}.}\\
  Consider first the case
  $d(\Im\lambda,[U(-1),U(1)])\geq\tilde{\nu}\beta^{-1/3}$. Here, following the
  same derivation as in (\ref{eq:237})-(\ref{eq:240}) (for
  $\#=\Sf $) and
  (\ref{eq:273})-(\ref{eq:275}) (for $\#=\Df$) (the role of $\delta$
    being replaced by $\frac{1}{\tilde \nu}$, which is small too)  we obtain  
  \begin{equation}\label{eq:271alt} 
     \|h\|_2\leq \|f\|_2+ \frac{C}{\tilde{\nu}}\beta^{5/6}\|\phi\|_{1,2}\,. 
  \end{equation}
  For sufficiently large $\tilde{\nu}$ we can thus follow the same
  steps of either the proof of
  Proposition~\ref{prop:zero-shear-stress-extension} starting from
  \eqref{eq:abcd} now replaced by \eqref{eq:271alt}, or the proof of
  Proposition~\ref{prop:no-slip} starting from \eqref{eq:275}, now
  replaced by \eqref{eq:271alt}, to obtain
  \eqref{eq:301}.\\

  \noindent {\bf Proof of  \eqref{eq:300}.}\\[1.2ex]
 {\bf Case 1: $|\Re\lambda|> {\mathfrak J}_m^{2/3}\beta^{-1/3}\min(\vartheta_1^r,\Re\nu_1)/2$}. To prove \eqref{eq:300}
  for $\#=\Df$ we repeat the same procedure used in the proof of
  Proposition \ref{prop:no-slip} to establish \eqref{eq:279} with
  $ \delta=r$ (note that we always have $\delta_2(U)\leq r$ and recall that $v_0$ is defined by \eqref{eq:277})
\begin{equation}
\label{eq:302}
  \|v_o\|_2\leq\frac{C}{\beta}\|f\|_2+ C\beta^{-1/6}\|\phi\|_{1,2}\,. 
\end{equation}
In the case  $\#=\Sf$ we repeat the same steps as in the proof of
Proposition \ref{prop:zero-shear-stress-extension} to establish
\eqref{eq:241} with $ \delta=r$. For convenience we use the notation $v_o$
instead of $v$ (which is defined by \eqref{eq:218}). 
\\

We now estimate $\phi_o=\A_{\lambda,\alpha}^{-1}v_o$. To this end we rewrite
\eqref{eq:51a} in the form
\begin{displaymath}
   \Re\Big\langle\frac{\phi_o}{U-\nu+i\mu},v_o\Big\rangle\geq    \frac 12 
     \|\phi^\prime_o\|_2^2 + (\alpha^2  +\gamma_m(\lambda,U))\|\phi_o\|^2_2 \,.
\end{displaymath}
The left-hand-side can be bounded by \eqref{eq:66} (for $p=2$) to obtain
\begin{displaymath} 
   \frac 12  \|\phi_o^\prime\|_2^2 + (\alpha^2  +\gamma_m(\lambda,U))\|\phi_o\|^2  \leq
   C\beta^{1/6}\|\phi_o\|_\infty\|v_o\|_2 \,.
\end{displaymath}
Since by  Lemma \ref{rem:gamma-finite},  $\gamma_m(\lambda,U)\geq \gamma_0>-\infty$   we
pick $\alpha\in\R$ such that 
\begin{displaymath}
\alpha^2\geq \sup (1,-2\gamma_0)\,.
\end{displaymath}
Then we write
\begin{equation}\label{bcde}
 \alpha\|\phi_o\|_\infty^2\leq\alpha\|\phi_o^\prime\|_2\|\phi_o\|_2 \leq \frac 12  (\|\phi_o^\prime\|_2^2 + \alpha^2\|\phi_o\|^2_2) \leq C\beta^{1/6}\|\phi_o\|_\infty\|v_o\|_2 \,.
\end{equation}
Consequently,
\begin{displaymath}
\|\phi_o\|_\infty \leq \frac C \alpha \beta^\frac 16 \|v_o\|_2\,,
\end{displaymath}
and 
\begin{equation}\label{eq:aabb}
\|\phi_o\|_{1,2}^2  \leq 
  \frac{2C}{\alpha}\beta^{1/3}\|v_o\|_2^2 \,. 
\end{equation}
By \eqref{eq:302} and  \eqref{eq:aabb}, we deduce 
\begin{equation}
\label{eq:303}
\|\phi_o\|_{1,2} \leq 
  \frac{\hat C}{\sqrt{\alpha}}\beta^{-5/6}\|f\|_2 +  \frac{\hat C}{\sqrt{\alpha}} \|\phi\|_{1,2}   \,. 
\end{equation}
In the case $\#=\Sf$ we have $\phi_0=\phi$ and hence \eqref{eq:300}
immediately follows. In the case $\#=\Df$ we continue as in the proof
of Proposition \ref{prop:no-slip} to establish \eqref{eq:278} with
$ \delta=r$, or explicitly,
\begin{displaymath}
   \|\phi_\pm \|_{1,2}\leq C(\beta^{-1}\|f\|_2+\beta^{-1/6}\|\phi\|_{1,2})\,.
\end{displaymath}
The above, combined with \eqref{eq:303} and \eqref{eq:decaa} yields
for $\alpha\geq  \tilde \alpha_1$ with $\tilde \alpha_1$ large enough and $\beta\geq \beta_0$ 
 with $\beta_0$ large enough
\begin{equation}
  \label{eq:304}
\|\phi\|_{1,2} \leq 
  C\Big(\frac{1}{\sqrt{\alpha}}+\beta^{-1/6}\Big)\beta^{-5/6}\|f\|_2 \,,
\end{equation}
from which \eqref{eq:300} readily follows.\\

{\bf Case 2: $|\Re\lambda|\leq \min(\Re\nu_1,\vartheta_1^r)\, 2^{-1}  {\mathfrak J}_m^\frac 23 \beta^{-\frac 13}$.} \\

For $\#=\Sf$ we use 
\eqref{eq:243} for $ \delta=r$, i.e., 
\begin{equation}
\label{eq:305}
   \|\phi^{\prime\prime}-\alpha^2\phi\|_2 \leq C(\beta^{-2/3}\|f\|_2+  r \beta^{1/6} \|\phi\|_\infty)\,.
\end{equation}
Then, as 
\begin{displaymath}
 \B_{\lambda+s\beta^{-1/3},\alpha,\beta}^{\mathfrak S} \phi = f -s\beta^{2/3}(\phi^{\prime\prime}-\alpha^2\phi) \,.
\end{displaymath}
we may use \eqref{eq:303},  having in mind that $\phi_0=\phi$,  to obtain that
\begin{displaymath}
  \|\phi\|_{1,2} \leq 
  \frac{\hat
    C}{\sqrt{\alpha}}\beta^{-5/6}(\|f\|_2+\beta^{2/3}\|\phi^{\prime\prime}-\alpha^2\phi\|_2)   \,.
\end{displaymath}
Substituting \eqref{eq:305} into the above yields \eqref{eq:300}. \\

For $\#=\Df$ we first obtain \eqref{eq:286} with $ \delta=r$, which
implies 
\begin{equation}
  \label{eq:306}
\|w_0\|_2\leq C(\beta^{-2/3}\|f\|_2+ r \beta^{1/6}\|\phi\|_\infty) \,,
\end{equation}
where
\begin{displaymath}
  w_0=-\phi^{\prime\prime}+\alpha^2\phi-B_+\psi_+- B_-\psi_-\,,
\end{displaymath}
and 
\begin{displaymath}
B_\pm =A_\pm (  f + i \beta U^{\prime\prime} \phi )\,.
\end{displaymath}
As
\begin{displaymath}
  \B_{\lambda+\Upsilon_0\beta^{-1/3},\alpha,\beta}^\Df\, \phi = f -\Upsilon_0\beta^{2/3}(w_0+B_+\psi_++B_-\psi_-) \,,
\end{displaymath}
where $\Upsilon_0$ satisfies \eqref{eq:287}, we can write as in
\eqref{eq:288} that
\begin{equation}
\label{eq:307}
   \phi=\chi_0-\Upsilon_0\beta^{2/3}(B_+\chi_++B_-\chi_-) \,,
\end{equation}
where
\begin{displaymath}
  \chi_0=(\B_{\lambda+\Upsilon_0\beta^{-1/3},\alpha,\beta}^\Df)^{-1}(f- \Upsilon_0\beta^{2/3}w_0) \quad \text{and}
  \quad \chi_\pm = (\B_{\lambda+\Upsilon_0\beta^{-1/3},\alpha}^\Df)^{-1}\psi_\pm \,.
\end{displaymath}
The estimation of $\chi_0$ can be done with the aid of
\eqref{eq:304} and  \eqref{eq:306}, yielding
\begin{equation}
\label{eq:308}
  \|\chi_0\|_{1,2} \leq
  C\Big(\frac{1}{\sqrt{\alpha}}+\beta^{-1/6}\Big)\beta^{-5/6}(\|f\|_2 +\|\phi\|_\infty)\,.
\end{equation}
To estimate $\chi_+$ (or $\chi_-$) we write
\begin{displaymath}
  \chi_+ = \chi_+^1 + \chi_+^2 \,,
\end{displaymath}
where 
\begin{displaymath}
  \chi_+^1=\A_{\check{\lambda},\alpha}^{-1}(U+i\check{\lambda})\Gamma_{[-1,1]}
  (\widetilde{\LL}_{\beta,\R}-\beta\check{\lambda})^{-1}
  \Big(\frac{U^{\prime\prime}\chi_+}{U+i\check{\lambda}}\Big)^{\prime\prime} \,.
\end{displaymath}
As 
\begin{displaymath}
  \Big\|\Big(\frac{U^{\prime\prime}\chi_+}{U+i\check{\lambda}}\Big)^{\prime\prime}\Big\|_2\leq C\Big(\beta^{-1/3}+
  \beta^{5/6}\|\chi_+\|_{1,2}\Big)\,,
\end{displaymath}
we can conclude from \eqref{eq:111}, \eqref{eq:112}, and
\eqref{eq:aabb} that
\begin{displaymath}
  \|\chi_+^1\|_{1,2}\leq \frac{C}{\sqrt{\alpha}}(\beta^{-7/6}+\|\chi_+\|_{1,2}) \,.
\end{displaymath}
The estimation of $\chi_+^2$ in the proof of Proposition
\ref{prop:no-slip} does not involve $\delta$ at all, but only $r$ and hence we can
conclude by \eqref{eq:292} and \eqref{eq:293}  that
\begin{displaymath}
   \|\chi_+^2\|_{1,2}\leq C(\beta^{-7/6}+\beta^{-1/6}\|\chi_+\|_{1,2}) \,.
\end{displaymath}
Consequently we obtain that for sufficiently large $\alpha$
\begin{displaymath}
  \|\chi_+\|_{1,2} \leq  C\beta^{-7/6}\,.
\end{displaymath}
A similar estimate holds for $\chi_-$, and hence we can conclude
\eqref{eq:300} from \eqref{eq:308}, \eqref{eq:307}, and
\eqref{eq:284}.
\end{proof}

\subsubsection{The case $U^{\prime\prime}\neq 0$}
\begin{lemma}
\label{lem:no-slip-convex-U}
Let $r>1$ and ${\hat \delta}\in (0,\frac 12]$.   Then, there exist $\beta_0>0$, $\Upsilon>0$, and $C>0$
such that for all $\beta\geq \beta_0$ and $U\in \Sg_r$ satisfying \eqref{condsurinf}, it holds that
  \begin{equation}
\label{eq:309}
      \sup_{
        \begin{subarray}{c}
          \Re\lambda\leq\Upsilon  {\mathfrak J}_m^{\frac 23} \beta^{-1/3} \\
          0\leq \alpha\leq \tilde \alpha_1
        \end{subarray}}\big\|(\B_{\lambda,\alpha,\beta}^\Df)^{-1}\big\|+
      \big\|\frac{d}{dx}\, (\B_{\lambda,\alpha,\beta}^\Df)^{-1}\big\|\leq
      \frac{C}{\beta^{1/2\,-{\hat \delta}}}\,,
  \end{equation}
where $\tilde \alpha_1$ is the same as in
Proposition~\ref{lem:no-slip-large-alpha-1}.
\end{lemma}
\begin{proof}
  Let $\phi\in D(\B_{\lambda,\alpha,\beta}^\Df)$, $\alpha\leq \tilde \alpha_1$, $f= \B_{\lambda,\alpha,\beta}^\Df\, \phi$ and  $v_\Df\in H^2(-1,+1)$ defined by
  \begin{equation}
\label{eq:310}
    v_\Df =\A_{\lambda,\alpha}\phi + (U+i\lambda)[\phi^{\prime\prime}(1)\hat{\psi}_+ + \phi^{\prime\prime}(-1)\hat{\psi}_-] \,,
  \end{equation}
  where 
  \begin{equation}\label{eq:311}
  \hat{\psi}_\pm =\psi_\pm \, \Theta_\pm /\psi_\pm (\pm 1)
  \end{equation}
   in which $\psi_\pm$ is defined in (\ref{eq:132a}) and
$\Theta_\pm (x)=1-\tilde{\eta}(1\mp x)$, with $\tilde{\eta}$  given by
  \eqref{eq:120}.  We note that by (\ref{eq:363}d) we have that
    for some $C>0$
  \begin{equation}
\label{eq:312}
    \frac{1}{C} [1+|\lambda_\pm|\beta^{1/3}]^{1/2}  \leq  |\psi_\pm (\pm 1)| \leq C [1+|\lambda_\pm|\beta^{1/3}]^{1/2} \,.
  \end{equation}
Note that $v_\Df\in H^1_0(-1,1)$ and hence we may introduce
\begin{displaymath}
g_\Df:=  (\LL_\beta^\Df-\beta\lambda)v_{\Df} \,.
\end{displaymath}
We have
\begin{equation}
\label{eq:313} 
  g_\Df=(U+i\lambda)( -f + \phi^{\prime\prime}(1)\hat{g}_+ +\phi^{\prime\prime}(-1)\hat{g}_-) -   (U^{\prime\prime}\phi)^{\prime\prime} 
  -2U^\prime \tilde{v}_\Df^\prime-U^{\prime\prime} \tilde{v}_\Df\,,
\end{equation}
wherein
\begin{displaymath}
  \hat{g}_\pm =  \Big(-\frac{d^2}{dx^2} +i\beta U-\beta\lambda\Big)\hat{\psi}_\pm  \,,
\end{displaymath}
and
\begin{displaymath}
   \tilde{v}_\Df = \frac{v_{\Df}-U^{\prime\prime}\phi}{U+i\lambda}=  - \phi^{\prime\prime} + \alpha^2\phi + \phi^{\prime\prime}(1)\hat{\psi}_+
  + \phi^{\prime\prime}(-1)\hat{\psi}_-\,. 
\end{displaymath}
We note that
\begin{equation}
\label{eq:314} 
   (\LL_\beta^\Df-\beta\lambda)\tilde{v}_{\Df}-i\beta U^{\prime\prime}\phi= - f + 
   \phi^{\prime\prime}(1)\hat{g}_+ + \phi^{\prime\prime}(-1)\hat{g}_- \,. 
\end{equation}

As in the proof of Proposition  \ref{prop:zero-shear-stress} (see in particular \eqref{eq:221}) we can
integrate by parts to obtain 
\begin{multline}
\label{eq:315}
  \Re\langle(U^{\prime\prime})^{-1}\tilde{v}_{\Df},(\LL_\beta^\Df-\beta\lambda)\tilde{v}_{\Df}-i\beta U^{\prime\prime}\phi\rangle=
  \|(U^{\prime\prime})^{-1/2}\tilde{v}_\Df^\prime\|_2^2+ \\
  +\Re\langle\big((U^{\prime\prime})^{-1}\big)^\prime\tilde{v}_\Df,\tilde{v}_\Df^\prime\rangle
 -\beta \mu\, \|\tilde{v}_\Df\|_2^2 + \beta
\Re\langle\phi^{\prime\prime}(1)\hat{\psi}_++ \phi^{\prime\prime}(-1)\hat{\psi}_-,i\phi\rangle\,.
\end{multline}

We begin the estimation by obtaining a bound for the last term on the
right-hand-side. \\[1.5ex]
{\bf Estimate of $ \beta
\Re\langle\phi^{\prime\prime}(1)\hat{\psi}_++ \phi^{\prime\prime}(-1)\hat{\psi}_-,i\phi\rangle\,$.}\\[1.5ex]
We first write
\begin{displaymath}
\phi(x) = \int_x^1(\xi-x)\phi^{\prime\prime}(\xi)\,d\xi=\phi^{\prime\prime}(1)\int_x^1(\xi-x)\hat{\psi}_+(\xi)\,d\xi+
\int_x^1(\xi-x)[\phi^{\prime\prime}(\xi)-\phi^{\prime\prime}(1)\hat{\psi}_+(\xi)]\,d\xi\,.
\end{displaymath}
Let 
\begin{displaymath}
  w_+(x)=\int_x^1(\xi-x)\hat{\psi}_+(\xi)\,d\xi\,.
\end{displaymath}
By (\ref{eq:363}c) and \eqref{eq:312} there exists $C>0$ such that
$\|\hat{\psi}_+\|_\infty\leq C$, and hence 
\begin{displaymath}
 |w_+(x)|\leq C\, |1-x|^2\,.
\end{displaymath}
Thus, 
\begin{displaymath}
  |\Re\langle\phi^{\prime\prime}(1)\hat{\psi}_+,i\phi^{\prime\prime}(1)w_+\rangle|\leq C|\phi^{\prime\prime}(1)|^2 \|(1-x)^2\hat{\psi}_+\|_1 \,.
\end{displaymath}
Using (\ref{eq:363}b), \eqref{eq:312}, translation, and dilation  (see
also \eqref{eq:149}),  we may conclude that
\begin{equation}
\label{eq:372}
  \|(1-x)^s\hat{\psi}_+\|_1\leq C\beta^{-(s+1)/3}[1+
    |\lambda_+|^{1/2}\beta^{1/6}]^{-(s+1)} \,,\quad  \forall s \leq3\,,
\end{equation}
and hence
\begin{equation}
\label{eq:373}
   |\Re\langle\phi^{\prime\prime}(1)\hat{\psi}_+,i\phi^{\prime\prime}(1)w_+\rangle|\leq \frac{C}{\beta}[1+
    |\lambda_+|^{1/2}\beta^{1/6}]^{-3}|\phi^{\prime\prime}(1)|^2 
\end{equation}
We then obtain for $x\in (-1,+1)$, using the fact that $\hat \psi_+\hat
\psi_- \equiv0\,$,
\begin{displaymath}
\begin{array}{ll}
  \Big|\overline{\hat{\psi}_+(x)}\int_x^1(\xi-x)[\phi^{\prime\prime}(\xi)-\phi^{\prime\prime}(1)\hat{\psi}_+(\xi)]\,d\xi\Big|&=
  \Big|\overline{\hat{\psi}_+(x)}\int_x^1(\xi-x)[-{\tilde
    v}_\Df(\xi)+\alpha^2\phi(\xi)]\,d\xi\Big|\\[1.5ex] & \leq C \, (1-x)^{5/2}\,
  |\hat{\psi}_+(x)|(\|\tilde{v}_\Df^\prime\|_2+\alpha^2\|\phi^\prime\|_2)\,.  
  \end{array}
\end{displaymath}
Consequently, from \eqref{eq:372} and \eqref{eq:373}, we thus get, as $\alpha\leq \tilde \alpha_1\,$,
\begin{multline}
\label{eq:316}
  \beta|\, \Re\langle\phi^{\prime\prime}(1)\hat{\psi}_+,i\phi\rangle|\leq C\big[1+
    |\lambda_+|\beta^{1/3}]^{-3/2}[|\phi^{\prime\prime}(1)|^2 +
    \\ \quad  [1+|\lambda_\pm|\beta^{1/3}]^{-1/4} \,
    \beta^{-1/6}|\phi^{\prime\prime}(1)|(\|\tilde{v}_\Df^\prime\|_2+ \|\phi^\prime\|_2)\big ]\,.
  \end{multline}
{\bf Estimate of $\|\tilde v^\prime_\Df\|$.}\\[1.5ex]
Next we obtain from \eqref{eq:315} and \eqref{eq:316} that
\begin{multline}
\label{eq:317}
   \|\tilde{v}_\Df^\prime\|_2^2 \leq C\, \|\tilde{v}_\Df\|_2(\|f\|_2
   +|\phi^{\prime\prime}(1)|\,\|\hat{g}_+\|_2
   +|\phi^{\prime\prime}(-1)|\,\|\hat{g}_-\|_2)+
   C \max(0,\mu \beta) \|\tilde{v}_\Df\|_2^2\\ +C\big[1+
    |\lambda_+|\beta^{1/3}]^{-3/2}[|\phi^{\prime\prime}(1)|^2 +
    [1+|\lambda_\pm|\beta^{1/3}]^{-1/4} \,
    \beta^{-1/6}|\phi^{\prime\prime}(1)|(\|\tilde{v}_\Df^\prime\|_2+ \|\phi^\prime\|_2)\big]
   \,.
\end{multline}
To estimate $ |\phi^{\prime\prime}(\pm 1)|$ we first write, as in  \eqref{eq:223}
  \begin{displaymath}
      (\LL_\beta^\zeta-\beta\lambda)(\phi^{\prime\prime}-\alpha^2\phi)=i\beta U^{\prime\prime}\phi+f\,,
  \end{displaymath}
where  $\zeta_\pm = \mathfrak z_\pm$ is given by \eqref{eq:194}. Then,  in
view of Remark~\ref{rem:dirichlet-estimate},
we may use
  \eqref{eq:33} to obtain that with $i\beta U^{\prime\prime}\phi+f$ to obtain
\begin{equation}
\label{eq:318}
  |\phi^{\prime\prime}(\pm 1)|\leq  C[1+|\lambda_\pm|^{1/2}\beta^{1/6}](\beta^{-1/6}\|f\|_2+\beta^{1/3}\log\beta \|\phi\|_\infty)\,.
\end{equation}

Furthermore, using the fact that
$\hat{g}_\pm(x)=g_\pm(x)/\psi_\pm(\pm1)$ , we may use \eqref{eq:150} and
\eqref{eq:312} to obtain that
\begin{equation}
\label{eq:319}
  \|\hat{g}_\pm \|_2\leq C\beta^{1/6}[1+|\lambda_\pm|^{1/2}\beta^{1/6}]^{-5/2}\,.
\end{equation}
which, when substituted into \eqref{eq:317}, yields, with the aid of
Sobolev's embeddings, 
\begin{equation}
\label{eq:320}
  \|\tilde{v}_\Df^\prime\|_2^2\leq C(\|f\|_2^2 + \beta^{1/2}\log\beta
  \|\phi\|_\infty\|\tilde{v}_\Df\|_2+
   \beta^{1/3}\log^2\beta \|\phi\|_{1,2}^2+
  \max (0, \mu \beta)  \, \|\tilde{v}_\Df\|_2^2 )\,.
 \end{equation}

 {\bf Estimate of $\|\tilde v_\Df\|$.}\\[1.5ex]
By  \eqref{eq:126} and  \eqref{eq:314}  we have that 
\begin{displaymath}
  \|\tilde{v}_{\Df}\|_2
  \leq C\big(\beta^{1/6}\|\phi\|_\infty+\beta^{-2/3}[\|f\|_2+|\phi^{\prime\prime}(1)|\,\|\hat{g}_+\|_2
  +|\phi^{\prime\prime}(-1)|\,\|\hat{g}_-\|_2]\big)\,.
\end{displaymath}
Hence, by \eqref{eq:319} and \eqref{eq:318},  we have that
\begin{displaymath}
   \|\tilde{v}_{\Df}\|_2
  \leq C(\beta^{1/6}\|\phi\|_\infty+\beta^{-2/3}\|f\|_2)\,.
\end{displaymath}
Substituting the above into \eqref{eq:320} yields
\begin{equation}
  \label{eq:321}
\|\tilde{v}_\Df^\prime\|_2 \leq C\, \left(\|f\|_2 + ((\max(\mu,0)^{1/2}\beta^{2/3}+\beta^{1/3}\log
\beta)\|\phi\|_{1,2} \right) \,.
\end{equation}
We now combine \eqref{eq:313}, \eqref{eq:321}, \eqref{eq:320}, \eqref{eq:318}, and
\eqref{eq:319} to obtain that
\begin{equation}
\label{eq:322}
  \|g_\Df\|_2\leq C(\|f\|_2 + (\max(\mu,0)^{1/2}\beta^{2/3}+\beta^{1/3}\log\beta)\|\phi\|_{1,2}) \,.
\end{equation}

{\bf Proof of \eqref{eq:309} }\\
We continue as in the proof of Proposition \ref{prop:no-slip}. We
first write, in view of \eqref{eq:310}
\begin{equation}
\label{eq:323}
  \phi = \phi_\Df + \check{\phi}_+ + \check{\phi}_-\,,
\end{equation}
where 
\begin{displaymath}
  \phi_\Df =\A_{\lambda,\alpha}^{-1}v_\Df \quad ; \quad
 \check{\phi}_\pm =-\A_{\lambda,\alpha}^{-1}\big([U+i\lambda]\phi^{\prime\prime}(\pm 1)\hat{\psi}_\pm \big) \,. 
\end{displaymath}
By \eqref{eq:244} and \eqref{eq:318} we have
\begin{displaymath}
  \|\check{\phi}_\pm \|_{1,2}\leq C\, [1+|\lambda_\pm|^{1/2}\beta^{1/6}]^{ -1/2}(\beta^{-2/3}\|f\|_2+\beta^{-1/6}\log \beta\|\phi\|_{1,2})\,,
\end{displaymath}
and hence, by \eqref{eq:323},
\begin{equation}
  \label{eq:324}
 \|\check{\phi}_\pm \|_{1,2}\leq C[1+|\lambda_\pm|^{1/2}\beta^{1/6}]^{-1/2}(\beta^{-2/3}\|f\|_2+\beta^{-1/6}\log\beta\|\phi_\Df\|_{1,2})\,.
\end{equation}
Substituting the above into \eqref{eq:322} yields, with the aid of
\eqref{eq:323} 
\begin{displaymath}
   \|g_\Df\|_2\leq C(\|f\|_2 +[(\max(\mu,0)^{1/2})\beta^{2/3}+\beta^{1/3}\log\beta ]\|\phi_\Df\|_{1,2}) \,.
\end{displaymath}
By (\ref{eq:374}b,c) we then have
for any $q>1$ and $p>2$
(recall that $\mu \leq  {\mathfrak J}_m^{2/3}\Upsilon \beta^{-\frac 13}$)
 \begin{equation}
  \label{eq:325}
  \|g_\Df\|_2\leq C\big(\|f\|_2 +
  \Upsilon^{1/2-1/p}\beta^{\frac{3p+2}{6p}}\|v_\Df\|_p+ \beta^{1/3}\log
  \beta(\|v^\prime_\Df\|_q+\|v_\Df\|_\infty)\big)\,. 
\end{equation}
By \eqref{eq:220} there exists $C>0$ such that for all $p>2$
(including $p=\infty$) it holds that
\begin{equation}
  \label{eq:326}
\|v_\Df\|_p \leq C\beta^{-\frac{3p+2}{6p}}\|g_\Df\|_2 \,. 
\end{equation}
Similarly, by \eqref{eq:116}, for all $1<q<2$, there exists $C_q>0$ such that 
\begin{equation}
\label{eq:327}
  \|v_\Df^\prime\|_q \leq C_q\, \beta^{-\frac{2+q}{6q}}\|g_\Df\|_2 \,. 
\end{equation}

Substituting \eqref{eq:326} and \eqref{eq:327} into \eqref{eq:325}
yields, choosing $\Upsilon>0$ small enough and $\beta_0$ large enough,
  the existence of $C>0$ such that for $\beta \geq \beta_0$
\begin{displaymath}
   \|g_\Df\|_2\leq C\, \|f\|_2\,. 
\end{displaymath}
Using (\ref{eq:374}b,c)  once again upon \eqref{eq:326}
and \eqref{eq:327} and the above inequality readily verifies \eqref{eq:309}.
\end{proof}

We can now conclude.
\begin{proposition}
Let $r>1$ and ${\hat \delta}>0$. Then, there exist $\beta_0>0$, $\Upsilon>0$, and $C>0\,$,
such that for all $\beta\geq \beta_0$ and $U\in \Sg_r$ satisfying \eqref{condsurinf}, it holds that
  \begin{equation}
\label{eq:328}
      \sup_{
        \begin{subarray}{c}
          \Re \lambda +\beta^{-1} \alpha^2 \leq\Upsilon J_{m}^{2/3} \beta^{-1/3} \\
          0<\alpha
        \end{subarray}}\big\|(\B_{\lambda,\alpha,\beta}^\Df)^{-1}\big\|+
      \big\|\frac{d}{dx}\, (\B_{\lambda,\alpha,\beta}^\Df)^{-1}\big\|\leq
      \frac{C}{\beta^{1/2-{\hat \delta}}}\,,
  \end{equation}
\end{proposition}
The proof follows immediately by combining \eqref{eq:31aaz},
\eqref{eq:300}, and \eqref{eq:309} for a sufficiently small value of
$\Upsilon >0$.

\section{Semigroup estimates}
\label{sec:semigroup-estimates}

In this section we prove Theorems \ref{thm:traction} and
\ref{thm:no-slip}.

 \subsection{Preliminaries}
 
For $\sharp\in \{\Sf,\Df \}$, 
  let ${\bf F}=(F_1,F_2)\in H^1_{loc}(\bar{D},\R^2)\cap\Hg$, $\Lambda\in\C$, and ${\mathbf u}\in\Wg_\sharp^0$
  satisfy
  \begin{equation}\label{eq:9.1.a}
    (\Tf_P^\sharp (U,\epsilon,L)-\Lambda){\bf u}=P{\bf F}\,. 
  \end{equation}
Recall that by writing ${\bf u}=\nabla_\perp\psi$ for some $\psi\in D(P_{\Lambda,\epsilon}^\sharp)$  we have established in \eqref{eq:26}
that 
\begin{equation}\label{eq:9.1.b}
\mathcal P_{\Lambda,\epsilon}^\sharp  \,  \psi
= \curl {\mathbf F}  \quad \text{in } D
  \,.
\end{equation}
This gives, for $\Lambda \in \rho ( \Tf_P^\sharp)$,
\begin{equation}\label{eq:9.1.c}
{\bf u} = \nabla_\perp (\mathcal P_{\Lambda,\epsilon}^\sharp )^{-1} \curl  {\mathbf F} =  (\Tf_P^\sharp(U,\epsilon,L) -\Lambda)^{-1} \, P{\bf F}\ \,.
\end{equation}
To estimate  $(\Tf_P^\sharp (U,\epsilon,L) -\Lambda)^{-1}$ we seek therefore a bound  for 
$(\mathcal P_{\Lambda,\epsilon}^\sharp )^{-1} $, for $\Lambda$ in a suitable region of
$\mathbb C$. We later derive the properties of the semi-group
$e^{-t\Tf_P^\sharp}$ from these
resolvent estimates. 

We seek a bound for the
$\LL (L^2_{per}, H^1_{per})$ norm of 
$(\PP_{\Lambda,\epsilon}^\#)^{-1}$ in the domain in  $\Re \Lambda \leq
\epsilon\Upsilon\beta_1^{2/3}$ for some $\Upsilon>0$. Recall that
$\mathcal P_{\Lambda,\epsilon}^\sharp $ depends on $L>0$ through the
periodicity condition appearing in the definition of its domain. As
in Section \ref{sec:orr-somm-oper} we can rewrite \eqref{eq:34}, which
provides an $\LL (L^2_{per})$ bound for
$(\PP_{\Lambda,\epsilon}^\sharp)^{-1}$ to the following $\LL
(L^2_{per}, H^1_{per})$ estimate
\begin{equation}
\label{eq:31a} 
 \sup_{\Re\Lambda\leq
   \epsilon\Upsilon\beta_1^{2/3}}\|(\PP_{\Lambda,\epsilon}^\sharp)^{-1}(I-\Pi)\|_{\LL (L^2_{per}, 
H^1_{per})}\leq  \epsilon^{-1}  B_*^\sharp(\Upsilon,\epsilon,L)  \,,
 \end{equation}
  where, 
  \begin{equation}\label{eq:defB*}
B_*^\sharp(\Upsilon,\epsilon,L) =  \sup_{\beta\geq\beta_1(\epsilon,L)}\sup_{
    \begin{subarray}{c}
      \alpha\geq0 \\
      \Re\lambda\leq \beta^{-1}(\Upsilon\beta^{2/3}-\alpha^2)
    \end{subarray}}\Big[\|(1+\alpha)(\B_{\lambda,\alpha,\beta}^\sharp)^{-1}\|+
   \Big\|\frac{d}{dx}\,(\B_{\lambda,\alpha,\beta}^\#  )^{-1}\Big\|\Big] \,,
  \end{equation}
in which $\beta_1(\epsilon,L)= (2\pi)/ (L\epsilon)$.

\subsection{Proof of Theorem \ref{thm:traction}}

\subsubsection{$(\PP_{\Lambda,\epsilon}^\Sf)^{-1}$ estimates} 

\paragraph{Estimation of $B_*^\Sf( \hat \Upsilon,\epsilon,L)$.} ~\\
We begin by showing that for any ${\hat \delta} \in (0,1/3)$, 
and $ \alpha_M>0$, there exist $\check \beta_0 >0$,  $C>0$ and $\Upsilon_0>0$, such that, for
all $\beta \geq \check \beta_0$ and $\Upsilon\leq \Upsilon_0$
\begin{equation}
\label{eq:329} 
\sup_{
    \begin{subarray}{c}
      0\leq\alpha\leq\alpha_M \beta^{1/6} \\
      \Re\lambda\leq \Upsilon\beta^{-1/3}
    \end{subarray}}\Big[\|(1+\alpha)(\B_{\lambda,\alpha,\beta}^\Sf)^{-1}\|+
   \Big\|\frac{d}{dx}\,(\B_{\lambda,\alpha,\beta}^\Sf)^{-1}\Big\|\Big] \leq C \, \beta^{-1/3+\hat{\delta}}\,.
\end{equation}
Let then
\begin{displaymath}
 S(\beta,\Upsilon)=\{(\Re\lambda,\alpha)\in\R^2 \,| \,  \Re\lambda\leq \beta^{-1/3}\Upsilon \; ; \;
  0\leq\alpha\leq\alpha_M \beta^{1/6} \}\,.
\end{displaymath}
By \eqref{eq:217} we have that for sufficiently small $\Upsilon>0$ there exists $\beta_0^\prime(\Upsilon)  >0$ and
$C>0$ such that
\begin{equation}\label{eq:217a}
  \sup_{(\Re\lambda,\alpha)\in S(\beta,\Upsilon)}[\|(1+\alpha)(\B_{\lambda,\alpha,\beta}^\Sf)^{-1}\|+
   \Big\|\frac{d}{dx}\,(\B_{\lambda,\alpha,\beta}^\Sf)^{-1}\Big\|\Big] \leq C \, \beta^{-1/3+\hat{\delta}}\,,
\end{equation}
for all $\beta\geq \beta_0^\prime(\Upsilon)$.  

On the other hand, it follows from \eqref{eq:299} in a similar manner
that for any $0<\Upsilon<\Re\nu_1$ there exist $\alpha^\Sf>0$, $\beta_{0}^{\prime\prime}>0$, and
$C>0$, such that for all $\beta\geq \beta_0^{\prime\prime}$ we have
\begin{equation}
\label{eq:370}
  \sup_{
    \begin{subarray}{c}
      \alpha^\Sf\beta^{1/6}\leq\alpha \\
      \Re\lambda\leq \beta^{-1/3}\Upsilon 
    \end{subarray}}\Big[\|(1+\alpha)(\B_{\lambda,\alpha,\beta}^\Sf)^{-1}\|+
   \Big\|\frac{d}{dx}\,(\B_{\lambda,\alpha,\beta}^\Sf)^{-1}\Big\|\Big] \leq C\, \beta^{-5/6}\,.
\end{equation}
As a result, for all $L>0$ and $\epsilon >0$ such that  $L \epsilon\leq \epsilon_0:=
2\pi/(\max(\beta_0^{\prime\prime},\beta_0^\prime)$ (implying that $ \beta_1(\epsilon,L) \geq
\max(\beta_0^{\prime\prime},\beta_0^\prime)$) it holds that
\begin{displaymath}
  \sup_{\beta\geq\beta_1(\epsilon,L)}\sup_{
    \begin{subarray}{c}
      \alpha^\Sf\beta^{1/6}\leq\alpha \\
      \Re\lambda\leq \beta^{-1/3}\Upsilon 
    \end{subarray}}\Big[\|(1+\alpha)(\B_{\lambda,\alpha,\beta}^\Sf)^{-1}\|+
   \Big\|\frac{d}{dx}\,(\B_{\lambda,\alpha,\beta}^\Sf)^{-1}\Big\|\Big] \leq C\,\beta_1(\epsilon,L)^{-5/6}\,.
\end{displaymath}
 Combining the above with \eqref{eq:329} for $\alpha_M=\alpha^\Sf$ yields that
for any $\hat{\delta}\in (0,\frac 13)$ there exists $C>0$,  $\Upsilon_0>0$, and
$\epsilon_{0}>0$  such that for  all $L>0$ and $\epsilon >0$ satisfying
$L \epsilon\leq \epsilon_0$ and $\Upsilon\leq \Upsilon_0$, we have
\begin{displaymath}
   \sup_{\beta\geq\beta_1(\epsilon,L)} \sup_{
    \begin{subarray}{c}
      0\leq\alpha \\
      \Re\lambda\leq \Upsilon\beta^{-1/3}
    \end{subarray}}\Big[\|(1+\alpha)(\B_{\lambda,\alpha,\beta}^\Sf)^{-1}\|+
   \Big\|\frac{d}{dx}\,(\B_{\lambda,\alpha,\beta}^\Sf)^{-1}\Big\|\Big] \leq C \, {\beta_1(\epsilon,L)}^{-1/3+\hat{\delta}}\,.
\end{displaymath}
We now observe that 
\begin{displaymath}
B_*^\Sf( \Upsilon,\epsilon,L)\leq    \sup_{\beta\geq\beta_1(\epsilon,L)} \sup_{
    \begin{subarray}{c}
      0\leq\alpha \\
      \Re\lambda\leq \Upsilon\beta^{-1/3}
    \end{subarray}}\Big[\|(1+\alpha)(\B_{\lambda,\alpha,\beta}^\Sf)^{-1}\|+
   \Big\|\frac{d}{dx}\,(\B_{\lambda,\alpha,\beta}^\Sf)^{-1}\Big\|\Big]\,.
\end{displaymath}
   Hence we obtain that for any $\hat{\delta}\in (0,\frac
   13)$\, there exist $\Upsilon_0>0$, $\epsilon_{0} >0$, and $C>0$ such that for all $L>0$,
   $\epsilon >0$ satisfying $0<L \epsilon\leq \epsilon_{0}$ it holds that
\begin{equation}
  \label{eq:331}  
B_*^\Sf( \Upsilon,\epsilon,L) \leq C\, \beta_1(\epsilon,L)^{-1/3+\hat{\delta}} \,.
\end{equation}

In conclusion, we have established the following
\begin{proposition} \label{prop9.1} For any $\hat{\delta}\in (0,\frac 13)$
  there exist $\epsilon_0>0$ and $\Upsilon_0>0$, such that for all positive $L$
  and $\epsilon $ for which $0<L \epsilon\leq \epsilon_0$, $\Upsilon\leq \Upsilon_0$, and $U\in \Sg_r$
  satisfying \eqref{condsurinf}, it holds that
\begin{equation}
\label{eq:31aa}
 \sup_{\Re\Lambda\leq
   \epsilon\Upsilon\beta_1^{2/3}}\|(\PP_{\Lambda,\epsilon}^\Sf)^{-1}(I-\Pi)\|_{\LL(L^2_{per},H^1_{per})}\leq  \frac{C}{\epsilon}\, \beta_1(\epsilon,L)^{-1/3+\hat{\delta}} \,.
 \end{equation} 
\end{proposition}

\subsubsection{Estimation of $(I-\Pi)e^{-t \,\Tf_P^\Sf (U,\epsilon,L)}$.}
\paragraph{Proof of (\ref{eq:23})}~\\
 Let $\check \delta\in(0,1/3)$ and $\Upsilon_0>0$ be as in the statement of
Proposition \ref{prop9.1}.
     We can now combine
\eqref{eq:9.1.b}, \eqref{eq:9.1.c}, and \eqref{eq:31aa}, to obtain for
all $\Upsilon\leq \Upsilon_0$
\begin{multline}
\label{eq:333}
    \sup_{\Re\Lambda\leq
   \epsilon\Upsilon\beta_1^{2/3}}\|(I-\Pi){\mathbf u}\|_2\leq
    \frac{C}{\epsilon}\beta_1(\epsilon,L)^{-1/3+\hat{\delta}} \,
    \|\curl  (I-\Pi){\mathbf
      F}\|_2\\ \leq \frac{C}{\epsilon}\beta_1(\epsilon,L)^{-1/3+\hat{\delta}}  \|
    (I-\Pi)\mathbf F\|_{1,2} \,. 
\end{multline}
 
We now establish a bound on
$\|\nabla(I-\Pi)u\|_2$ for $\Re \Lambda<-\frac 12 \|U^\prime\|_\infty-1$. To this end
we first recall Lemma \ref{lem:commute} to obtain 
\begin{displaymath}
  (\Tf_P^\Sf (U,\epsilon,L)-\Lambda)(I-\Pi) {\bf u} = (I-\Pi){\mathbf F} \,,
\end{displaymath}
Using \eqref{eq:newineq} with ${\bf u}$ replaced by $(I-\Pi){\bf u}$
  (note that, by Lemma \ref{Lemma2.9} $(I-\Pi){\bf u}\in
  W^\Sf_0$) we obtain
 \begin{multline} \label{eq:338}
  \epsilon\|\nabla(I-\Pi){\bf u}\|_2^2 + \langle (I-\mathfrak p )u_1,U^\prime(I-\mathfrak p) u_2\rangle -\Lambda\|(I-\Pi){\bf u}\|_2^2\\ =
  \langle(I-\Pi){\mathbf u},(I-\Pi){\mathbf F}\rangle \,,
\end{multline}
which implies
\begin{equation} \label{eq:338a}
  \epsilon\|\nabla(I-\Pi){\bf u}\|_2^2 \leq   (\Re \Lambda + \frac 12 \|U'\|_\infty) \|(I-\Pi){\bf u}\|_2^2 + \| (I-\Pi) \mathbf u \| \,  \| (I-\Pi) \mathbf F \| 
   \,,
\end{equation}
Since by \eqref{eq:367}  we have 
\begin{displaymath}
  \| e^{-t \,\Tf_P^\sharp}(I-\Pi)\|\leq e^{\frac 12 \|U^\prime\|_\infty t} \,,
\end{displaymath}
we may use the resolvent-semigroup relation to obtain that
\begin{displaymath}
\sup_{\Re\Lambda\leq -\frac 12 \|U^\prime\|_\infty-1}  \|(I-\Pi){\bf u}\|\leq  \| (I-\Pi) \mathbf F \| \,.
\end{displaymath}
We may therefore conclude that
\begin{equation}
  \label{eq:14pi}
\sup_{\Re\Lambda\leq -\frac 12 \|U^\prime\|_\infty-1} \| (I-\Pi){\bf u}\|_{1,2} \leq
\frac{C}{\epsilon^{1/2}} \|(I-\Pi) {\bf F}\|_2\,.
\end{equation}
 
Consequently, 
\begin{equation}
\label{eq:15}
\sup_{\Re\Lambda\leq -\frac 12 \|U^\prime\|_\infty-1}
\|(I-\Pi)(\Tf_P^\Sf-\Lambda)^{-1}\|_{ \LL(\Hg,H^1_{per}(D,\C^2))}\leq \frac{C}{\epsilon^{1/2}}\,.
\end{equation}

To estimate the resolvent for $\Lambda\in\C$ satisfying $-\frac12 \|U^\prime\|_\infty  -1\leq
\Re\Lambda\leq\epsilon\Upsilon\beta_1^{2/3}$, we introduce $\omega=-\frac 12 \|U^\prime\|_\infty \, -1$
and use the resolvent identity, composed on the left by $(I-\Pi)$,
\begin{multline}
\label{eq:369}
  (I-\Pi)(\Tf_P^\Sf-\Lambda)^{-1} = (I-\Pi)(\Tf_P^\Sf-\omega-i\Im\Lambda)^{-1}\\
  + (\Re\Lambda-\omega)
 (I-\Pi) (\Tf_P^\Sf-\Lambda)^{-1}(\Tf_P^\Sf-\omega-i\Im\Lambda)^{-1}\,,
\end{multline}
to obtain, with the aid of Lemma \ref{lem:commute}, that
\begin{displaymath} \begin{array}{l}
   \|(I-\Pi)(\Tf_P^\Sf-\Lambda)^{-1}\|_{ \LL(\Hg)}\\
   \quad  \leq \|(I-\Pi)(\Tf_P^\Sf-\omega-i\Im\Lambda)^{-1}\|_{ \LL(\Hg)} \\ \quad  \quad+
   |\Re\Lambda-\omega|\,\|(I-\Pi)(\Tf_P^\Sf-\Lambda)^{-1}\|_{
     \LL(H^1_{per}(D,\C^2),\Hg)}\,\|(I-\Pi)(\Tf_P^\Sf-\omega-i\Im\lambda)^{-1}\|_{ \LL(\Hg,H^1_{per}(D,\C^2))}\,.
\end{array}
\end{displaymath}
Using \eqref{eq:333} and \eqref{eq:15} we then obtain that for each
$\hat{\delta}\in(0,1/3)$ there exist $\Upsilon_0>0$, $C>0$, and
$\epsilon_0>0$ such that for all $\epsilon\in (0,1]$, $L>0$ satisfying
$0<\epsilon L \leq \epsilon_0$
\begin{equation}\label{eq:9.14}
  \sup_{-\|U^\prime\|_\infty/2-1\leq \Re\Lambda\leq\epsilon\Upsilon\beta_1^{2/3}}
  \|(I-\Pi)(\Tf_P^\Sf(U,\epsilon,L)-\Lambda)^{-1}\|_{ \LL(\Hg)} \leq
  \frac{C}{\epsilon^{3/2}}\beta_1(\epsilon,L)^{-1/3+\hat{\delta}} \,. 
\end{equation}

We can now use \cite[Proposition 2.1]{hesj10} (cf. also  \cite[Proposition
  13.31]{Helbook} or \cite[Theorem 11.3.5]{Sj}), which we repeat here 
for the benefit of the reader
\begin{proposition}\label{propa}~\\
Let $S(t)$ be a  strongly continuous semigroup, defined on a Hilbert
space $H$, which satisfies for some
$\hat{M}\geq1$ and $\hat{\omega}\in\R$,
\begin{equation}
  \label{eq:334}
\|S(t)\|\leq\hat{M}e^{-\hat{\omega}t}\,.
\end{equation}
Let $-A:D(A)\to H$ denote the generator of $S(t)$. Suppose further that
 for some $\omega>\hat{\omega}$ 
\begin{equation}
  \label{eq:368}
\sup_{\Re z\leq \omega}\|(A-z)^{-1}\|:= \frac{1}{r(\omega)} <+\infty \,.
\end{equation}
 Then,
\begin{equation}
\label{eq:335}
\| S(t)\| \leq \widehat M  \Big(1 +  \frac{2 \widehat M  (
  \omega-\hat{\omega} ) }{ r(\omega)}\Big)\; e^{-\omega t}\,.
\end{equation}
\end{proposition}
By Remark 1.4  in \cite{hesj10} (or Remark 11.3.4 in \cite{Sj}), we
may conclude that 
\begin{equation}
\label{eqimprov}
\sup_{\Re z= \omega}\|(A-z)^{-1}\|= \sup_{\Re z\leq \omega}\|(A-z)^{-1}\|\,.
\end{equation}
 Note that since $(A-z)^{-1}$ is bounded by \eqref{eq:368}, we can obtain
  the above from its analyticity and the Phragm\`en-Lindel\"of Theorem. 
We now apply the proposition with $A=(I-\Pi) \Tf_P^\Sf(U,\epsilon,L)$. By
\eqref{eq:367} it follows that \eqref{eq:334} holds for $\hat{M}=1$
and $\hat{\omega}=-\|U^\prime\|_\infty/2$  . By \eqref{eq:9.14} and \eqref{eqimprov}, 
\eqref{eq:368} holds for $\omega=\epsilon\Upsilon\beta_1^{2/3}$ with 
$ r(\omega)\geq \epsilon^{3/2}\beta_1^{1/3-\hat{\delta}}/C$. \\
Consequently by \eqref{eq:335} , we obtain that for
any $\hat \delta\in(0,1/3)$, there exist $\Upsilon>0$, $C>0$, and $\epsilon_0>0$  such that for all $\epsilon
\in (0,1]$ and $L>0$  satisfying
$0<L \epsilon\leq \epsilon_0$ we have
\begin{displaymath}
  \|e^{-t \, \Tf_P^\Sf(U,\epsilon,L)}(I-\Pi)\|\leq \frac{C}{\epsilon^{3/2}}\beta_1(\epsilon,L)^{-1/3+\hat{\delta}} \, e^{-\epsilon
    \Upsilon\beta_1^{2/3}t}\,.
\end{displaymath}
This completes the proof of (\ref{eq:23}).\\

\begin{remark}
  Note that while \eqref{eq:369} allows for an estimate in $\LL(\Hg)$
  of \break $e^{-t \, \Tf_P^\Sf(U,\epsilon,L)}(I-\Pi)$, it contributes
  an additional $\epsilon^{-1/2}$ factor to the coefficient of the
  exponent in (\ref{eq:23}).  If we manage to obtain estimates for
  $e^{-t \, \Tf_P^\Sf(U,\epsilon,L)}(I-\Pi)$ in
  $\LL(H^1_{per}(D,\C^2),\Hg)$ the additional factor could be avoided.
  In fact, in \cite{chen2018transition}, the initial conditions
    are assumed to be in $H^2_{per}$, resulting in improved
  estimates for the semigroup.
\end{remark}

\paragraph{ Proof of Part 2.}~\\
The proof of part 2  is obtained in  the same
manner. 
 We begin by  
 replacing Proposition \ref{prop9.1} by the following result
\begin{proposition} \label{prop9.3}
Let $r>1$. For any $\Upsilon<\Re\nu_1$ there exist $\delta>0$, $\epsilon_0>0$, and $C>0$ such that for all positive $L$
  and $\epsilon $ for which $0<L \epsilon\leq \epsilon_0$, 
and  $U\in\Sg_r$ satisfying $ \delta_2(U) <\delta \,,$ it holds that 
\begin{equation}
\label{eq:31aab}
 \sup_{\Re\Lambda\leq \epsilon\Upsilon\beta_1^{2/3}}\|(\PP_{\Lambda,\epsilon}^\Sf)^{-1}\|_{\LL(L^2_{per},H^1_{per})}\leq   \frac{C}{\epsilon}\beta_1^{-2/3} \,.
 \end{equation} 
\end{proposition}
\begin{proof}
  We use 
  \eqref{eq:233} which gives 
  \begin{equation}\label{eq:217b}
    \sup_{\beta \geq \beta_1(\epsilon,L)}\; \sup_{
    \begin{subarray}{c}
      0\leq\alpha\leq\alpha_M \beta^{1/6} \\
      \Re\lambda\leq \Upsilon\beta^{-1/3}
    \end{subarray}}\Big[\|(1+\alpha)(\B_{\lambda,\alpha,\beta}^\Sf)^{-1}\|+
   \Big\|\frac{d}{dx}\,(\B_{\lambda,\alpha,\beta}^\Sf)^{-1}\Big\|\Big] \leq C \, \beta_1(\epsilon,L)^{-2/3}\,.
\end{equation}
Combining the above with \eqref{eq:370} yields
\begin{displaymath}
 B_*^\sharp(\Upsilon,\epsilon,L)  \leq C\, \beta_1(\epsilon,L)^{-2/3} \,,
\end{displaymath}
which together with \eqref{eq:31a} yields
\eqref{eq:31aab}.
\end{proof}
 The proof of (\ref{eq:337}) proceeds from here in the same
 manner as in Part 1. \\

\subsection{Proof of Theorem \ref{thm:no-slip}}
   Since the proof is similar to the proof of Theorem
  \ref{thm:traction} we address here only it main ingredients.\\
  \paragraph{Proof of Part 1}  For the
  first part of the theorem we use   \eqref{eq:31aaz}
 and  \eqref{eq:328}  to establish  that for any $\hat{\delta}\in (0,\frac
13)$\, there exist $\Upsilon_0>0$, $\epsilon_{0} >0$ and $C>0$ such that for  all positive
$L$ and  $\epsilon$ satisfying $0<L \epsilon\leq
\epsilon_{0}$ and $\Upsilon\leq \Upsilon_0$ we have
  \begin{equation}
\label{eq:25}
    B_*^\Df( \Upsilon ,\epsilon,L) \leq C\, \beta_1(\epsilon,L)^{-1/3+\hat{\delta}} \,.
  \end{equation}
  This is similar to \eqref{eq:331} in the  case $\#=\Sf$.\\

Combining \eqref{eq:25} with \eqref{eq:31a} we may
conclude that
\begin{proposition}
  For any $\hat{\delta}\in(0,1/3)$ there exist
  $\epsilon_0>0$ and $\Upsilon_0>0$, such that for all positive $L$ and $\epsilon >0$
  for which  $0<L \epsilon\leq \epsilon_0$, $\Upsilon\leq \Upsilon_0$,
  and $U\in \Sg_r$ satisfying \eqref{condsurinf}, it holds that
\begin{displaymath}
 \sup_{\Re\Lambda\leq \epsilon\Upsilon\beta_1^{2/3}}\|(I-\Pi)(\PP_{\Lambda,\epsilon}^\Df)^{-1}\|_{\LL(L^2_{per},H^1_{per})}\leq   \frac{C}{\epsilon }\beta_1(\epsilon,L)^{-1/3+\hat{\delta}}  \,.
 \end{displaymath}  
\end{proposition}
\noindent We may now proceed as in the proof of (\ref{eq:23}) to establish
(\ref{eq:24}). \\

\paragraph{Proof of Part 2} 
To establish \eqref{eq:371} in Part 2 we use \eqref{eq:263}  and
\eqref{eq:31aaz} to show that 
  \begin{displaymath}\\
    B_*^\Df (\Upsilon,\epsilon,L)  \leq C\,\beta_1(\epsilon,L)^{-2/3}\,.
  \end{displaymath}
 Then we can continue in  the same manner as in the proof of
(\ref{eq:337}) in the first part of the theorem. \\

{\bf Acknowledgments:} Y. Almog was partially supported by NSF grant
DMS-1613471.  Both authors would like to thank Pierre Bolley and Nader
Masmoudi for for some fruitful discussions with B. Helffer.

\def\cprime{$'$}

\appendix

\section{Basic properties of the Airy function and Wasow's results on $A_0$}  \label{AppA}

\subsection{Airy function  properties}
In this subsection, we summarize some of the basic properties of Airy
function $\Ai(z)$, and the generalized Airy function $A_0(z)$, that
are being used throughout this work (see \cite{abst72} for details)
and establish some new inequalities satisfied by these functions.  We
recall that Airy function is the unique solution of
\begin{equation*}
(D_x^2 + x) u =0\,,
\end{equation*}
on the line such that $u(x)$ tends to $0$ as $x \to +\infty$ and 
$
\Ai (0) = 1/ \left( 3^{\frac 23}\, \Gamma \left(\frac 23\right)\right).$
Standard ODE theory shows that Airy function is entire
and strictly decreasing on $\mathbb R_+$, but has an
infinite number of zeros in $\mathbb R_-$. \\
Airy function  satisfies different asymptotic
expansions as $|z|\to\infty$ depending on $\arg z$. 
We bring two of them here
\begin{subequations}\label{5} 
  \begin{align} 
&\Ai(z) &= \frac 12 \pi^{-\frac 12}z^{-1/4} \,
\exp\left(-\frac{2}{3}z^{3/2}\right) \bigl(1 + \mathcal O
(|z|^{-\frac 32}) \bigr), & \quad \rm{ for } \ |\arg z|<\pi\,,  \\
&\Ai(-z) &= \pi^{-\frac 12}z^{-1/4} \,
\sin\left(\frac{2}{3}z^{3/2}+\frac{\pi}{4}\right) \bigl(1 + \mathcal O
(|z|^{-\frac 32}) \bigr), & \quad \rm{ for } \ |\arg z|<\frac{2\pi}{3}\,.
\end{align}
\end{subequations}
Moreover the $\mathcal O(|z|^{-3/2})$ estimate is, for any ${\hat \delta}>0$, uniform when $|\arg
z| \leq \pi -{\hat \delta}$ in (\ref{5}a) or $|\arg
z| \leq 2\pi/3 -{\hat \delta}$ in (\ref{5}b)\,.  In particular, $\Ai(z)$ is rapidly decreasing at $\infty$ if $z$ belongs to
a sector $|\arg z | \leq \frac \pi 3 -{\hat \delta}$, with ${\hat \delta} >0$.

The following moment estimates are needed in Subsection
\ref{sec:no-slip-schrodinger}.
\begin{proposition}
\label{prop:airy-funct-prop}
Let $<\lambda> :=\sqrt{1+|\lambda|^2}$.  For every $\mu_0>0$ there exists $C>0$
such that, for any $\lambda$ with $\Re \lambda \leq \mu_0\,$, we have,
\begin{equation}  \label{eq:284a} 
\int_0^\infty x^{2k}  |\Ai(e^{i\pi/6}[x+i\lambda])|^2 \,dx \leq
C <\lambda>^{-k-1}\, e^{-\frac{4}{3}\Re\{(e^{i2\pi/3}\lambda)^{3/2}\}}\,,
\mbox{ for }  0 \leq k \leq 4 \,,  \end{equation}
and 
\begin{equation}\label{eq:284b}
    \int_0^\infty x^s|\Ai(e^{i\pi/6}[x+i\lambda])| \,dx \leq
 C <\lambda>^{-\frac{3+2 s}{ 4} }e^{-\frac 23  \Re\{(e^{i2\pi/3}\lambda)^{3/2}\}} \mbox{
   for } 0\leq s\leq3  \,.
 \end{equation}
\end{proposition}
\begin{proof}~\\
  If $|\lambda|\leq3\mu_0$ then, by (\ref{5}a), all the estimates of the
  proposition are satisfied for some $C=C(\mu_0)$.  Hence, we can
  consider from now on the case where $|\lambda|>3\mu_0\,.$ Note that by
  interpolation, it is sufficient to consider
  $k\in\Z$.\\[1.5ex]
  {\bf Proof of \eqref{eq:284a} for $k=0$\,.}\\
  Let
\begin{equation}
\label{eq:339}
  u (x) =\frac{\Ai(e^{i\pi/6}[x+i\lambda])}{\Ai(e^{i2\pi/3}\lambda)} \,, 
\end{equation}
which is well-defined for 
\begin{displaymath}
  \mathcal V(\mu_0):=  \{ \Re \lambda \leq \mu_0\} \cap \{|\lambda|>3\mu_0\}\,,
\end{displaymath}
since the
denominator can vanish only when $\arg \lambda=\pi/3\,$.\\
 It can be easily
verified that
\begin{equation}
\label{eq:340}
  \begin{cases}
    -u^{\prime\prime}+(ix-\lambda)u=0 & \mbox{ for }  x>0 \\
    u(0)=1\,. &
  \end{cases}
\end{equation}
Let further
\begin{equation}
\label{eq:341} 
  w=u-e_\lambda (x) \mbox{ with } e_\lambda(x):= e^{-(-\lambda)^{1/2}x} \,,
\end{equation}
where $(-\lambda)^\frac 12$ is well-defined for $\lambda \in \mathcal V (\mu_0)$  as
\begin{displaymath}
  -\pi+\arccos\Big(\frac{1}{3}\Big)<\arg(-\lambda)< \pi-\arccos\Big(\frac{1}{3}\Big)
\end{displaymath}
Note that  for $\lambda \in \mathcal V (\mu_0)$
\begin{equation}
\label{lowb}
\Re (-\lambda)^{\frac 12} \geq |\lambda|^{\frac 12 }\,   \sin \Big( \frac 12
\arccos\Big(\frac{1}{3}\Big)\Big)\geq   |3\mu_0|^\frac 12 \,    \sin \Big(
  \frac 12 \arccos\Big(\frac{1}{3}\Big)\Big)  >0\,,
\end{equation}
implying, for any $\mu_0>0$, the existence of $C(\mu_0)>0$,  such that,
for all $\lambda \in \mathcal V (\mu_0)$, it holds that
\begin{equation}
\label{uppb1}
\| e_\lambda \|_2 \leq C(\mu_0) |\lambda|^{-\frac 14}\,.
\end{equation}
Substituting into \eqref{eq:340} yields
\begin{equation}
\label{eq:342}
(\LL_+-\lambda)w= -  ixe_\lambda (x) \,,
\end{equation}
where $\LL_+$ is associated with the differential operator
$-d^2/dx^2 +ix$ and is defined on the domain
\begin{displaymath}
  D(\LL_+) = \{ u\in H^2(\R_+)\cap H^1_0(\R_+) \, | \, xu\in L^2(\R_+)\,\} \,.
\end{displaymath}
It has been established in \cite[\S5]{Hel2}  that, for any
$\hat \mu\in\R^+$, there exist $\nu_0(\hat \mu)>0$ and $C(\hat \mu)$ such
that 
\begin{equation}
\label{eq:127a}
  \|(\LL_+-\lambda)^{-1}\|\leq C(\hat \mu)\,,\, \mbox{ for }  |\Im \lambda| \geq \nu_0 (\hat \mu) \mbox{ and }| \Re \lambda|  \leq \hat \mu\,.
\end{equation}
In addition, we have 
\begin{equation}
\label{eq:127b}
  \|(\LL_+-\lambda)^{-1}\|\leq C \,,\, \mbox{ for }  \Re \lambda \leq 0\,.
\end{equation}

Denote by $\{ \nu_n\}_{n=1}^\infty$ the eigenvalues of $\LL_+$, and recall
that they are located on the ray $\arg \lambda=\pi/3$ (see
\cite[\S2.2]{al08}).  Observing that $\mathcal V(\mu_0)$ does not
contain any eigenvalue (a consequence of the fact that $\arccos 1/3 >\pi/3$) and combining the above with
\eqref{eq:127a}-\eqref{eq:127b} yield the existence of $C(\mu_0)>0$
such that
\begin{equation}
\label{eq:343}
  \sup_{\lambda \in \mathcal V(\mu_0)}
\|(\LL_+-\lambda)^{-1}\|\leq C(\mu_0)\,.
\end{equation}
Hence, for $\lambda \in \mathcal V (\mu_0)$, 
\begin{displaymath}
  \|w\|_2 \leq C(\mu_0)\|xe_\lambda \|_2\,.
\end{displaymath}
Using \eqref{lowb}, we obtain that
\begin{equation}\label{uppb2}
  \|x e_\lambda \|_2 \leq \hat C(\mu_0)|\lambda|^{- 3/4}\,,
\end{equation}
and hence
\begin{equation}\label{uppb3} 
  \|w\|_2 \leq \check C(\mu_0) |\lambda|^{- 3/4}\,.
\end{equation}
By the above, \eqref{eq:341}, and \eqref{uppb1},   we thus have
\begin{equation}\label{uppb4}
  \|u\|_2 \leq C(\mu_0) |\lambda|^{-1/4} \,.
\end{equation}
To prove (\ref{eq:284a}) for $k=0$, we need to establish yet an upper
bound for $\Ai (e^{2i\frac \pi 3}\,\lambda)$, for $\lambda \in \mathcal V(\mu_0)$. This
is an immediate consequence of (\ref{5}a). We observe indeed that for
any $\epsilon >0$, $\arg \lambda \not \in (-\frac \pi 2 + \epsilon, \frac \pi 2-\epsilon)$ as $|\lambda|\to +\infty$
with $\Re \lambda \leq \mu_0$. This implies that, for any $\epsilon >0\,$, there exists
$\lambda_0(\epsilon,\mu_0)>0$ such that for $|\lambda|\geq \lambda_0(\epsilon)$ and $\Re \lambda \leq \mu_0\,$,  $\arg
\Big(\lambda \exp {\frac{2i \pi}{3}}\Big)\in
(- \frac {5\pi}{ 6} -\epsilon, \frac{ \pi}{6} +\epsilon)$. \\
Consequently, for any $\mu_0 >0$, there exists a constant $C(\mu_0)$, such
that for all $\lambda \in \mathcal V(\mu_0)$
\begin{equation} \label{uppb5}
|\Ai (e^{2i\frac \pi 3}\,\lambda)| \leq C(\mu_0) <\lambda>^{-\frac 14} e^{-\frac{2}{3}\Re\{(e^{i2\pi/3}\lambda)^{3/2}\}}\,.
\end{equation}
Together with \eqref{uppb4}, \eqref{uppb5} yields   \eqref{eq:284a} for $k=0$ and $\lambda \in \mathcal V(\mu_0)$.\\

{\bf Proof of (\ref{eq:284a}) for $k=1$.}\\
We begin by deriving an estimate of $x w$.  To this end we observe
that
\begin{displaymath}
(D_x^2 + ix -\lambda) (xw) = -i x^2 e_\lambda -  2 i  w^\prime\,.
\end{displaymath}
Following the same procedure applied in the previous proof, we need an
estimate of $w^\prime$.  To achieve this end, we first observe that
\begin{displaymath}
\Re \langle (D_x^2 + ix -\lambda) w\,,\, w\rangle = \|w^\prime\|^2  - \Re \lambda \| w\|_2^2\,,
\end{displaymath}
which leads to the estimate
\begin{displaymath}
 \|w^\prime\|^2 \leq \mu_0 \| w\|_2^2 + \| x e_\lambda\| \,  \|w\|_2\,.
\end{displaymath}
 Implementing \eqref{uppb2} and \eqref{uppb3} leads to
\begin{equation}
\label{eq:344}
 \|w^\prime\|_2  \leq C(\mu_0) |\lambda|^{-\frac 34}\,.
\end{equation}
 We observe in addition that
 \begin{displaymath}
\| x^2 e_\lambda \|_2 \leq C(\mu_0) |\lambda|^{-\frac 54}\,.
\end{displaymath}
Proceeding as in the proof of (\ref{eq:284a}) for $k=0$, we obtain (\ref{eq:284a}) for $k=1\,$.\\

{\bf Proof of  $(\ref{eq:284a})$ for $k=2\,$.}\\
In this case the approximation of $u$ by $e_\lambda$ is unsatisfactory.  To
improve it, in light of \eqref{eq:342}, we solve in $H^2(\mathbb
R^+)$ the problem
\begin{equation}
  \begin{cases}
    \Big(-\frac{d^2}{dx^2}-\lambda\Big) f_\lambda= - i x e_\lambda & \mbox{ for } x>0\,,\,\\
 f_\lambda (0) =0\,.&
  \end{cases}
\end{equation}
We look for $f_\lambda$ in the form $f_\lambda = p_\lambda e_\lambda$ which means that $p_\lambda$
must satisfy
\begin{displaymath}
 p_\lambda^{\prime\prime}(x)- 2 (-\lambda)^{\frac 12} p_\lambda^\prime (x) = ix \,,\, p_\lambda (0)=0\,.
\end{displaymath}
We search for a polynomial solution.
A simple computation leads to
\begin{equation}
f_\lambda = p_\lambda e_\lambda \mbox{ with } p_\lambda (x) =  \frac{i}{4}  \left((-\lambda)^{- 1/2}x^2 -\lambda^{-1} x \right)\,.
\end{equation}
We now write 
\begin{equation}
\label{eq:345}
u= e_\lambda +   w=  e_\lambda + f_\lambda +w_1\,,
\end{equation}
to obtain 
\begin{equation}
\label{eq:346}
  (\LL_+-\lambda)w_1= - i x f_\lambda = - ix \, p_\lambda  e_\lambda  \,.
\end{equation}
In order to prove \eqref{eq:284a}  ($k=2$) we observe first that
\begin{equation} \label{eq:315a} 
 \sum_{\ell=0}^3 <\lambda>^{\frac \ell 2} \| x^\ell f_\lambda \|  \leq  C <\lambda> ^{-\frac 7 4  }\,.
\end{equation}
Hence, it remains necessary to obtain an estimate of $x^2 w_1$ in $L^2$. Here we
use the fact that \eqref{eq:342} is similar to \eqref{eq:346}, the
only difference being that the right-hand-side is given by $-ix f_\lambda$
instead of $-i x e_\lambda$. Note that $\|xf_\lambda\|_2$ is much smaller than
$\|xe_\lambda\|_2$ as $|\lambda|\to\infty$.

By \eqref{eq:343}, \eqref{eq:346} and \eqref{eq:315a} (for $\ell=1$), we
then have
\begin{equation}
\label{eq:347}
  \|w_1\|_2\leq \frac{C}{|\lambda|^{9/4}} \,,
\end{equation}
which is significantly smaller than the bound provided by
\eqref{uppb3} for $w$. We continue as in the case
$k=1$.  We first use the identity
 \begin{displaymath}
   \Re\langle w_1,(\LL_+-\lambda)w_1\rangle = \|w_1^\prime\|_2^2 - \mu\|w_1\|_2^2 \,, 
 \end{displaymath}
 to conclude with the aid of \eqref{eq:346}, \eqref{eq:315a}, 
 \eqref{eq:347},  and recalling that $\mu\leq \mu_0$, 
 \begin{equation}
\label{eq:348}
     \|w_1^\prime\|_2\leq \frac{C}{|\lambda|^{9/4}} \,.
 \end{equation}
Then we write
\begin{equation}
\label{eq:349}
  (\LL_+-\lambda)(xw_1) = x(\LL_+-\lambda)w_1-2w_1^\prime= -i x^2 f_\lambda - 2 w_1^\prime \,.
\end{equation}
Combining \eqref{eq:348}, \eqref{eq:346}, and \eqref{eq:343} yields
\begin{equation}
\label{eq:350}
  \|xw_1\|_2\leq  \frac{C}{|\lambda|^{9/4}} \,.
\end{equation}
 From \eqref{eq:350} and \eqref{eq:347} we get in addition
\begin{equation}
\label{eq:351}
  \|x^{1/2}w_1\|_2\leq \|xw_1\|_2^{1/2} \|w_1\|_2^{1/2} \leq \frac{C}{|\lambda|^{9/4}} \,.
\end{equation}
Next, we write
 \begin{displaymath}
     \Re\langle xw_1,(\LL_+-\lambda)(xw_1)\rangle = \|(xw_1)^\prime\|_2^2 -\mu\,  \|xw_1\|_2^2\,,
 \end{displaymath}
to conclude from \eqref{eq:349} and \eqref{eq:350} that 
\begin{equation}
\label{eq:352}
  \|(xw_1)^\prime\|_2 \leq  \frac{C}{|\lambda|^{9/4}} \,.
\end{equation}
Upon writing
\begin{displaymath}
    (\LL_+-\lambda)(x^2w_1) = x^2(\LL_+-\lambda)w_1-2(xw_1)^\prime= -i x^3f_\lambda  -2(xw_1)^\prime\,, 
\end{displaymath}
we use \eqref{eq:352}, \eqref{eq:346}, and \eqref{eq:343} to obtain
\begin{equation}
\label{eq:353}
  \|x^2w_1\|_2\leq  \frac{C}{|\lambda|^{9/4}} \,,
\end{equation}
which can be used, together with \eqref{eq:345} and \eqref{eq:341} to
obtain (\ref{eq:284a}) for $k=2\,$.\\

In a similar manner we obtain the estimates for $k\in\{3,4\}\,$.  We
note in particular that
\begin{displaymath}
 \|x^4w_1\|_2 +  \|(x^4w_1)^\prime\|_2\leq  \frac{C}{|\lambda|^{9/4}} \,,
\end{displaymath}
and hence, by Sobolev embeddings,
\begin{displaymath}
   \|x^4w_1\|_\infty \leq  \frac{C}{|\lambda|^{9/4}}\,.
\end{displaymath}
Combining the above with \eqref{eq:341}  and  \eqref{eq:345} yields, for $k\in [0,4]\,$, 
\begin{equation}
\label{eq:354}
  \|x^k u\|_\infty \leq  \frac{C}{|\lambda|^{\frac k 2}}\,.
\end{equation}

{\bf Weighted $L^1$ estimates.\\} Recall that in deriving
\eqref{eq:284a} we needed to establish that, for $k=0,\dots,4$, it
holds that
\begin{equation}  \label{eq:284aa} 
 \|x^k u \|_{L^2} \leq 
C\, <\lambda>^{-\frac {2k+1}{4}} \,.
 \end{equation}
 By interpolation it is enough to treat the case when $s$ is an
 integer.  Then we use H\"older inequality and \eqref{eq:284aa} for $k=s$ and $k=s+1$ to obtain, for $s\leq 3$, 
\begin{equation} \label{eq:284ab}
\begin{array}{ll}
  \|x^s u \|_{L^1} &  \leq   \|[x+i<\lambda>^{-\frac 12} ]x^s u \|_2 \|[x+i<\lambda>^{-\frac 12}]^{-1}\|_2  \\ & 
   \leq   \left( \|x^{s+1} u \|_2  +  <\lambda>^{-\frac 12}  \|x^s u \|_2\right) \|[x+i<\lambda>^{-\frac 12}]^{-1}\|_2\\ & 
    \leq  C\, <\lambda>^{-\frac {s+1}{2}}  \,.
    \end{array}
 \end{equation}
 We then recover \eqref{eq:284b} by using \eqref{uppb5}.
 \end{proof}

\subsection{Definition of $A_0$ and the locus of its zeroes.}
\label{sec:definition-a_0-locus}
Let $A_0:\C\to\C$ be given (see  \eqref{eq:144}) by the holomorphic
extension to $\mathbb C$ of 
\begin{displaymath}
\mathbb R\ni   z \mapsto  A_0(z)=e^{i\pi/6}\int_z^{+\infty}\Ai(e^{i\pi/6}t)\,dt\,.
\end{displaymath}
To use the results of Wasow \cite{wa53} (see \cite[Eq. (39)]{wa53})
and justify this holomorphic extension we observe the following
relation 
\begin{lemma}
\begin{equation}\label{eq:apsi}
  A_0(z)=-\psi(-e^{i\pi/6}z)\,.
\end{equation}
where $\psi$ is the holomorphic extension of the real function
\begin{equation}\label{defpsi}
\mathbb R \ni x \mapsto \psi (x):=\int_{-\infty}^x \Ai(-t) dt = \frac 13 + \int_{0}^x \Ai(-t) dt\,.
\end{equation}
\end{lemma}

It has been proved by Wasow in \cite[Section 3]{wa53} that the zeroes of
$\psi$ are all located in the sector $|\arg z|< \frac \pi 6$.\\ 
\begin{proposition}\label{propwas}
 The zeroes of $\psi$ belongs to  $ -\pi/6<\arg z<\pi/6 \,$. Moreover $\psi$ has no real zeroes.
\end{proposition}
\begin{proof}[Sketch of the proof]
To establish that result is a combination of the
argument principle and Rouch\'e's. The change of $\arg \psi$ is estimated along the
path $\Gamma=\Gamma_1\cup\Gamma_2\cup\Gamma_3$ where, for some $R>0$
\begin{displaymath}
 \Gamma_1=\{z=re^{i\pi/6}\,|\, 0<r\leq R\} \quad ; \quad  \Gamma_2= \{z=Re^{i\theta}\,|\, \pi/6<\theta<\pi\} \quad ; \quad  \Gamma_3=[-R,0] \,.
\end{displaymath}
Since $\Ai(z)$ is real and positive for $z\in\R_+$ it follows that
$\arg \psi$ does not change along $\Gamma_3$. Along $\Gamma_2$ one uses
\eqref{5}, with a bound on the remainder. Finally, to estimate $\Delta\arg
\psi$ along $\Gamma_1$ one uses the power series of $\Ai(t)$, for $0<|t|<9\,$,
and \eqref{5} for $|t|>9$. The  tails of the ensuing power series of $\Re\psi$ and
$\Im\psi$ are Leibniz series with terms of alternating sign and
decreasing moduli. Thus, one may  truncate the series into
finite sums, and the remainders can be easily estimated. Once the
above procedure is applied, one can establish that $|\Delta \arg
\psi|<2\pi$ and hence that $\Delta \arg \psi=0$ along $\Gamma$. Since
$\psi(z)=0\Rightarrow\psi(\bar{z})=0$, the first statement of the proposition
follows. \\ The second statement is proved in \cite[p. 199]{wa53}.
\end{proof}
\begin{corollary}\label{corwas}
  Let $A_0(iz)=0$. Then,  $ \pi/6<\arg z<\pi/2\,$. 
\end{corollary}
We continue with the following  result stated in \cite{wa53} (Eq. (35)) which allows us to obtain
additional information on the location of the zeroes of $A_0$.  It is
also serves as a useful tool in some of the proofs in Subsection
\ref{sec:no-slip-schrodinger}.
\begin{lemma}\label{lemwas}
  Let $\tilde{\psi}(z)=\psi(-z)$. For any $0<{\hat \delta}<\pi$ there exists $C_{\hat \delta} >0$ and
  $r_0({\hat \delta})>0$ such that for all $|z|>r_0({\hat \delta})$   in the sector 
 $|\arg z|<\pi-{\hat \delta}$ it holds that
\begin{equation}
\label{eq:355}
  \Big|\tilde{\psi}(z) - \frac{1}{2}\pi^{-\frac 12}z^{-3/4} \,
  \exp\Big(-\frac{2}{3}z^{3/2}\Big) \Big| \leq \frac{C_{\hat \delta}}{|z|^{7/4}}
  \Big|\exp\Big(-\frac{2}{3}z^{3/2}\Big) \Big| \,.
\end{equation}
\end{lemma}
The proof is a rather standard application of the method of steepest
descent method \cite[Chapter 4]{mi06} and is therefore being skipped.

\subsection{Asymptotic of the zeroes}
In Subsection \ref{sec:no-slip-schrodinger} we also need to establish the
following lemma about the asymptotic behavior of the zeroes of $A_0$. 
A similar statement for $\psi$ is made in \cite[\S{} 3]{wa53} without a clear proof.
\begin{proposition}
 \label{lem:zeroes-a_0}
 Let $\Sg_0$ denote the set of points $\lambda\in\C$ satisfying $A_0(i\lambda)=0$.
 Then, for any $R>0$, $\Sg_0\cap B(0,R)$ is finite and its cardinality tends to $+\infty$ as $R$ tends to $ +\infty$. In particular $ \Sg_0$ is non empty.
 Moreover, for any $\epsilon >0$, there exists $R$ such that, for $\lambda \in \Sg_0\cap B(0,R)^c$, 
   \begin{equation}
\label{eq:356}
|\arg \lambda -\pi/3| \leq \epsilon\,.
  \end{equation}
\end{proposition}
\begin{proof}~\\
We apply Jensen's formula \cite[Theorem 1.7]{la87} to  $\tilde{\psi}(z)=\psi(-z)$. 
\begin{equation}
\label{eq:357}
  \log \tilde{\psi}(0) + N_{\tilde{\psi}}(R,0) = \frac{1}{2\pi}\int_{-\pi}^\pi \log
  |\tilde{\psi}(Re^{i\theta})| \,d\theta \,.
\end{equation}
In the above $N_{\tilde{\psi}}(R,0)$ is the Nevanlinna counting
function,
\begin{displaymath}
  N_{\tilde{\psi}}(R,0) =\sum_{
    \begin{subarray}{c}
      |a|<R \\
\tilde{\psi}(a)=0
    \end{subarray}}\log \Big|\frac{R}{a}\Big| \,,
\end{displaymath}
where we have used the fact that all zeroes of $\tilde{\psi}$ are
simple. We recall from \cite{wa53} that none of them is real, and all
the zeroes of $\tilde{\psi}^\prime = - \Ai $ lie on the negative real axis.
Note that if $\Sg_0=\emptyset$ then $N_{\tilde{\psi}}(R,0)\equiv0$ for all $R>0$.
Recall from the definition that $\tilde{\psi}(0)=1/3$ and hence if we
show that the first term on the right-hand-side of \eqref{eq:357}
is unbounded as $R\to\infty$, we may conclude the first statement of the proposition.

It follows from \eqref{eq:355} that for any $C_0 \in (0,\frac 89)$
there exist ${\hat \delta} >0$ and $r_0>0$ such that for all $R>r_0 $ we
have
\begin{equation}\label{eq:358}
  \int_{-\pi+{\hat \delta}}^{\pi-{\hat \delta}} \log  |\tilde{\psi}(Re^{i\theta})| \,d\theta  \geq C_0 R^{3/2} \,.
\end{equation}
Indeed, using Lemma \ref{lemwas}, we eastablish the existence for any
${\hat \delta} >0$ of $r_0({\hat \delta})$ and $C$ such that, for all $R\geq
r_0({\hat \delta})$,
\begin{equation*}
\begin{array}{ll}
  \int_{-\pi+{\hat \delta}}^{\pi-{\hat \delta}} \log  |\tilde{\psi}(Re^{i\theta})| \,d\theta & \geq \frac{2}{3}R^{3/2}(1-CR^{-1})
  \int_{-\pi+{\hat \delta}}^{\pi-{\hat \delta}} -\cos\Big( \frac{3\theta}{2}\Big)
  \,d\theta\\
  & \geq \frac{8}{9}R^{3/2}(1-CR^{-1}) \cos ( 3 {\hat \delta}/2) \,.
 \end{array}
\end{equation*}
To estimate the integral for $|\arg z|\in (\pi-{\hat \delta},\pi)$ we write, owing to the
concavity of $\mathbb R^+\ni x \mapsto \log x $, 
\begin{equation}
\label{eq:359}
  \int_{|\theta|>\pi-{\hat \delta}} \log  |\tilde{\psi}(Re^{i\theta})| \,d\theta = 2\int_{\pi-{\hat \delta}}^\pi
  \log  |\tilde{\psi}(Re^{i\theta})| \,d\theta \leq 2{\hat \delta}\log \Big(
  \frac{1}{{\hat \delta}}\int_{\pi-{\hat \delta}}^\pi|\tilde{\psi}(Re^{i\theta})| \,d\theta \Big)\,. 
\end{equation}
Since by \eqref{defpsi} we have, for any $0<{\hat \delta}<\pi/3\,$ and $R> R_0$, 
\begin{displaymath}
  \int_{\pi-{\hat \delta}}^\pi|\tilde{\psi}(Re^{i\theta})| \,d\theta  \leq \int_0^{R_0} \int_{\pi-{\hat \delta}}^\pi
  |\Ai(se^{i\theta})| \,d\theta ds + \int_{R_0}^R\int_{\pi-{\hat \delta}}^\pi
  |\Ai(se^{i\theta})| \,d\theta ds \,,
\end{displaymath}
where $R_0$ is fixed, but sufficiently large so that
$\Ai(se^{i\theta})$ obeys (\ref{5}b) for all $s>R_0$. For the first term on
the right-hand side there exists $C(R_0)>0$ such that
\begin{displaymath}
  \int_0^{R_0} \int_{\pi-{\hat \delta}}^\pi
  |\Ai(se^{i\theta})| \,d\theta ds \leq C(R_0) \,{\hat \delta}\,.
\end{displaymath}
For the second term we use (\ref{5}b) to obtain the rather crude
estimate  for $R > R_0$
\begin{displaymath}
  \int_{R_0}^R\int_{\pi-{\hat \delta}}^\pi
  |\Ai(se^{i\theta})| \,d\theta ds \leq 
  C{\hat \delta} (R-R_0) e^{-\frac{2}{3}R^{3/2}\cos(3(\pi-{\hat \delta})/2)} \,.
\end{displaymath}
Consequently, for every $C_1>1$ there exists $R_1$, such that, for
all $R>R_1$,
\begin{equation}
\label{eq:360}
  \log \Big(
  \frac{1}{{\hat \delta}}\int_{\pi-{\hat \delta}}^\pi|\tilde{\psi}(Re^{i\theta})| \,d\theta \Big)\leq \frac{2}{3}
  C_1 R^{3/2}|\cos(3(\pi-{\hat \delta})/2)|\leq C_1R^{3/2}{\hat \delta} \,.
\end{equation}
This implies, by \eqref{eq:359}, the existence, for any $0<{\hat \delta}<\pi/3\,$,  of 
$R_2>0$, such that for $R \geq R_2$ 
\begin{displaymath} 
  \int_{|\theta|>\pi-{\hat \delta}} \log  |\tilde{\psi}(Re^{i\theta})| \,d\theta \leq 2{\hat \delta}^2 C_1R^{3/2} \,.
\end{displaymath}

Combining the above with \eqref{eq:358} yields, by fixing ${\hat \delta}$
small enough,  the existence of  $R_3$  and $\hat C_0 >0$ such that
for all $R>R_3$ 
\begin{displaymath}
  \frac{1}{2\pi}\int_{-\pi}^\pi \log
  |\tilde{\psi}(Re^{i\theta})| \,d\theta \geq \hat C_0 R^{3/2} \,.
\end{displaymath}
The proof shows that the above bound holds true for any $\hat C_0 \in
(0,\frac 89)$ for $R_3(\hat{C}_0)$ large enough.
Hence, we get a lower bound for $  N_{\tilde{\psi}}(R,0) $ which implies
the first statement of the proposition.\\ 

  To prove \eqref{eq:356} we notice that for any ${\hat \delta}>0$, it holds by
  \eqref{eq:355} that there exists $r_0>0$ such that $\tilde{\psi}$
  cannot have any zeroes, for $|z|>r_0$ in the sector
  $ |\arg z|<\pi -{\hat \delta}$. We then observe that $A_0(i\lambda) =-\tilde \psi (e^{2i\pi/3} z)\,.$
\end{proof}
 We can now immediately draw the following conclusion. 
 \begin{corollary}
 \label{cor:zeroes-a_0}
   \begin{equation}
 \label{eq:361}
     \inf \Re \Sg_0 =\vartheta_1^r>0 \,.
   \end{equation}
\end{corollary}

\subsection{Normalized Airy functions}
 
We complete the appendix with some corollaries of
\eqref{eq:284a}-\eqref{eq:284b} and \eqref{eq:355} needed in
Subsection \ref{sec:no-slip-schrodinger} and with other estimates
needed in Subsection \ref{sec:large-alpha-no-slip}.
\begin{proposition}
  Let $\Psi_\lambda \in L^2(\R_+)$ be defined by
  \begin{equation}
    \label{eq:362}
\Psi_\lambda (x)= \frac{\Ai(e^{i\pi/6}(x+i\lambda))}{A_0(i\lambda)} \,.
  \end{equation}
  Then, for any ${\hat \delta}_1>0$  there exists $C>0$ such that 
for all $\Re \lambda \leq \vartheta_1^r-{\hat \delta}_1$
\begin{subequations}
\label{eq:363}
  \begin{align}
 &    \| x^k \Psi_\lambda\|_2 \leq C<\lambda>^{\frac{1 - 2 k}{4}} \mbox{ for } k\in [0,4]\,, \\
&    \| x^s \Psi_\lambda\|_1 \leq C  <\lambda>^{-  \frac{s}{2}}  \mbox{ for } s \in
[0,3] \,, \\
&  \| x^s \Psi_\lambda\|_\infty  \leq C  <\lambda>^{-  \frac{s-1}{2}} \mbox{ for } s \in
[0,4] \,,\\
& \frac{1}{C}  <\lambda>^{1/2} \leq \Psi_\lambda(0)  \leq C  <\lambda>^{1/2}  \,.
  \end{align}
\end{subequations}
\end{proposition}
\begin{proof}
   Let $|\lambda|>3\vartheta_1^r$ (the proof for
  $|\lambda|\leq3\vartheta_1^r$ follows by continuity as the denominator is bounded
  away from zero). As $A_0(i\lambda) =
  -\tilde{\psi}(e^{i2\pi/3}\lambda)$ we obtain from \eqref{eq:355} that there
  exists $C>0$ such that for
  any ${\hat \delta}_1>0$ and $|\arg \lambda+2\pi/3|<\pi$
  \begin{displaymath}
   \frac{1}{C<\lambda>^{3/4}}\exp
    \Big\{-\frac{2}{3}\Re\{(e^{i2\pi/3}\lambda)^{3/2}\}\Big\}\leq |A_0(i\lambda)| \leq
    \frac{C}{<\lambda>^{3/4}}\exp
    \Big\{-\frac{2}{3}\Re\{(e^{i2\pi/3}\lambda)^{3/2}\}\Big\} \,,
  \end{displaymath}
  which combined with (\ref{eq:284a}) and \eqref{eq:284b} yields
  (\ref{eq:363}a,b), and with the aid of \eqref{5} gives
  (\ref{eq:363}d)\,. The proof of (\ref{eq:363}c) follows from
  \eqref{eq:354}.
\end{proof}

We also need, in Subsection \ref{sec:large-alpha-no-slip} the following estimate
\begin{lemma}
  \label{lem:more-estimates-F}
Let $\varkappa>0$ and,  for $\theta >0$, $\Upsilon(\theta) =\mu_0(\theta)-\varkappa$ where 
\begin{displaymath}
  \mu_0(\theta):=\inf \Re\sigma( \LL^\theta)\,. 
\end{displaymath}
Let  $F(\lambda,\theta)$ be defined by (\ref{eq:178}). 
There exists $C(\varkappa) >0$ such that, for all $\theta >0$, 
\begin{equation}
  \label{eq:364}
\sup_{\Re \lambda\leq\Upsilon(\theta)}\Big|\frac{A_0(i\lambda)}{F(\lambda,\theta)}\Big| \leq C (\varkappa) (1+\theta) \,.\end{equation}
\end{lemma}
\begin{proof}
Let $R>0$ be determined later.
\paragraph{The case $|\lambda|<R$.}~\\
Let  $\theta_1(\varkappa)$ satisfy
\begin{displaymath}
  \sup_{\theta>\theta_1(\varkappa)}\mu_0(\theta)\leq \Re\nu_1+\frac{\varkappa}{2}
\end{displaymath}
By the
continuity of $(\lambda,\theta) \mapsto A_0(i\lambda)/F(\lambda,\theta)$ on $ B(0,R) \times [0, \theta_1(\varkappa)]$, we have  
\begin{equation}
  \label{eq:365}
\sup_{
  \{  \Re \lambda\leq\Upsilon(\theta)\}\cap B(0,R)\,,\,
    0\leq\theta\leq\theta_1}
\Big|\frac{A_0(i\lambda)}{F(\lambda,\theta)}\Big| \leq C(\varkappa,R) \,.
\end{equation}
In the case $\theta>\theta_1(\varkappa)$ we write, as in Subsection \ref{sec:large-alpha-no-slip} (see \eqref{eq:3.14b}),
\begin{equation*}
   \theta  F(\lambda,\theta)- \Ai(e^{i2\pi/3}\lambda)=  \theta \int_{\R_+}e^{-\theta x}[\Ai(e^{i2\pi/3}\lambda+e^{i\pi/6}x)-\Ai(e^{i2\pi/3}\lambda)]\,dx\,,
  \end{equation*}
  and deduce (see \eqref{eq:314f}) that, for $|\lambda| \leq R$,  we have
  \begin{displaymath}
   | \theta  F(\lambda,\theta)- \Ai(e^{i2\pi/3}\lambda)| \leq C(R) \theta^{-\frac 12}\,.
  \end{displaymath}
We may now use the fact that $\Upsilon(\theta) <\Re\nu_1-\varkappa/2$ to
 get the existence of $\theta(\varkappa,R)\geq \theta_1(\varkappa)$ and $C(\varkappa,R)>0\,$  such that, 
 for $\theta \geq \theta(\varkappa,R)$, $|\lambda|\leq R$  and $\Re \lambda \leq  \Upsilon(\theta)$ we have
\begin{displaymath}
  |F(\lambda,\theta)| \geq \frac{C(\varkappa,R)}{1+\theta}\,.
\end{displaymath}
For $\theta \in [\theta_1(\varkappa)\,, \,\theta(\varkappa,R) ]$, we use the continuity of
$|F(\lambda,\theta)|$ to get a uniform lower bound for it. 
Consequently, we obtain that 
\begin{equation}
\label{eq:366}
\sup_{\{\Re \lambda\leq\Upsilon(\theta)\} \cap|\lambda|\leq R }
\Big|\frac{A_0(i\lambda)}{F(\lambda,\theta)}\Big| \leq C(\varkappa,R)(1+\theta) \,.
\end{equation}

\paragraph{The case $|\lambda| \geq R$.}~\\
If $|\lambda|\geq R$ we may use \eqref{eq:193} which reads
\begin{equation*}
  \Big|\frac{[\theta+(-\lambda)^{1/2}] F(\lambda,\theta)}{\Ai(e^{i2\pi/3}\lambda)}-1\Big| \leq C\, |\lambda|^{-1/4} \,, 
\end{equation*}
to obtain for $R$ large enough with the aid of \eqref{5},
\eqref{eq:apsi}, and \eqref{eq:355} that,
\begin{displaymath}
   \Big|\frac{A_0(i\lambda)}{F(\lambda,\theta)}\Big| \leq  
     \Big| \frac{A_0(i\lambda)}{\Ai(e^{i2\pi/3}\lambda)}\Big|\,\Big| \frac{\Ai(e^{i2\pi/3}\lambda)}{F(\lambda,\theta)}\Big| \leq 
   C\, \frac{\theta+|\lambda|^{1/2}}{|\lambda|^{1/2}} \,.
\end{displaymath}
Combining the above with \eqref{eq:365} and \eqref{eq:366} yields
\eqref{eq:364}. 
\end{proof}


\end{document}